\PassOptionsToPackage{dvipsnames}{xcolor} % to enable dvipsnames option in xcolor package
\documentclass[acmsmall,screen,nonacm]{acmart}
\copyrightyear{\the\year}

\setcopyright{cc}
\setcctype{by}

\usepackage{amsmath}
\usepackage{amsthm}
\usepackage{xparse}
\usepackage{mathtools}
\usepackage{mathpartir}
\usepackage{xspace}
\usepackage{suffix}
\usepackage{listings}
\usepackage{cleveref}
\usepackage{cprotect}
\usepackage{tikz}
\usetikzlibrary{decorations.pathreplacing,positioning,calc,arrows,arrows.meta,matrix,shapes.geometric,fit}

\usepackage{thm-restate}
\usepackage[normalem]{ulem}
\usepackage{subcaption}
\usepackage{enumitem}
\usepackage{langrules}
\usepackage{stmaryrd}

\colorlet{linkcolor}{ACMPurple}
\creflabelformat{section}{#2{\color{linkcolor}#1}#3}
\creflabelformat{figure}{#2{\color{linkcolor}#1}#3}
\creflabelformat{appendix}{#2{\color{linkcolor}#1}#3}
\creflabelformat{theorem}{#2{\color{linkcolor}#1}#3}
\crefdefaultlabelformat{#2{\color{linkcolor}#1}#3}

\hypersetup{
  colorlinks,
  citecolor=linkcolor,
  linkcolor=linkcolor,
  urlcolor=linkcolor}

\makeatletter
\newcommand{\citeposessive}[2][]{%
  \def\@yearlist{#2\ifx\empty#1\else,#1\fi}%
  \citeauthor{#2}'s~[\citeyear{\@yearlist}]}
\makeatother

\newcounter{numlevels}
\makeatletter
\NewDocumentCommand{\newMLP}{sO{l}mO{0}om}{
  \IfBooleanTF{#1}{
    \IfValueTF{#5}{
      \WithSuffix\newcommand#3*[#4][#5]
    }{
      \WithSuffix\newcommand#3*[#4]
    }
  }{
    \IfValueTF{#5}{
      \newcommand{#3}[#4][#5]
    }{
      \newcommand{#3}[#4]
    }
  }
  {\ifthenelse{\value{numlevels} > 0}{\begin{array}[t]{@{}#2@{}}}{\begin{array}{#2}}\addtocounter{numlevels}{1}#6\addtocounter{numlevels}{-1}\end{array}}
}
\makeatother

\makeatletter
\newcommand{\para}[1]{%
  \@startsection{paragraph}{4}{\z@}%
                {1ex plus 0.25ex minus 0.25ex}{-1ex plus -0.25ex minus -0.5ex}%
                {\normalfont\upshape\normalsize\bfseries}*{#1.}%
}
\makeatother

\newcommand{\helamsm}{\rotatebox[origin=c]{15}{\textphnc{e}}\kern-1pt\textsubscript{$\lambda$small}}
\newcommand{\zkstrudul}{\textsc{zkStruDul}\xspace}
\newcommand{\zkstrudel}{\zkstrudul}
\newcommand{\corelang}{\ensuremath{\textrm{FJ}_\textsc{nizk}}\xspace}

\newcommand{\DefineTypingRule}[3]{\DefineRule[T#1Rule]{T-#1}[T-#1]{#2}{#3}}
\newcommand{\DefineCTypingRule}[3]{\DefineRule[TC#1Rule]{TC-#1}[TC-#1]{#2}{#3}}
\newcommand{\DefineEvalRule}[3]{\DefineRule[E#1Rule]{E-#1}[E-#1]{#2}{#3}}
\newcommand{\DefineCEvalRule}[3]{\DefineRule[EC#1Rule]{EC-#1}[EC-#1]{#2}{#3}}

\newcommand{\programfont}[1]{\ensuremath{\mathsf{#1}}\xspace}
\makeatletter
\def\subst@inner#1#2#3\_nil{#1 \mapsto #2\ifx\relax#3\relax\else,\mkern2mu\subst@inner#3\_nil\fi}%
\newcommand{\raw@subst}[3]{
  \mathchoice{#1{\left[#2\subst@inner#3\_nil\right]}}
             {#1[#2\subst@inner#3\_nil]}
             {#1[#2\subst@inner#3\_nil]}
             {#1[#2\subst@inner#3\_nil]}%
}
\newcommand{\subst}[3]{\raw@subst{#1}{}{{#2}{#3}}}
\WithSuffix\newcommand\subst*[2]{\raw@subst{#1}{}{#2}}
\makeatother

\newMLP*[@{}l@{}]{\Newline}[2]{#1 \\ #2}

\DeclareSymbolFont{bbsymbol}{U}{bbold}{m}{n}
\DeclareMathSymbol{\bbexcl}{\mathbin}{bbsymbol}{"21}
\newcommand{\fatexclamation}{\bbexcl}
\DeclareMathSymbol{\bbE}{\mathbin}{bbsymbol}{`E}

\DeclareMathSymbol{\bbV}{\mathbin}{bbsymbol}{`V}

\DeclareMathSymbol{\bbU}{\mathbin}{bbsymbol}{`U}

\DeclareMathSymbol{\bbColon}{\mathbin}{bbsymbol}{"3A}

\newcommand{\alt}{\mkern3mu\mid\mkern3mu}
\newcommand{\divi}{\mkern1mu\mid\mkern1mu}
\newcommand{\ty}{\mkern2mu{:}\mkern2mu}
\newcommand{\proves}{\vdash}
\newcommand{\cproves}{\Vdash}

\newcommand{\produces}{\dashv}
\newcommand{\cproduces}{~\scalebox{1}{\rotatebox[origin = c]{180}{$\Vdash$}}~}
\newcommand{\dom}{\operatorname{dom}}
\newcommand{\defeq}{\ensuremath{\stackrel{\text{\normalfont\tiny def}}{=}}}
\newcommand{\ow}{\text{o.w.}}
\newcommand{\nestedok}{~\operatorname{nested-ok}}
\newcommand{\bang}{\mkern1mu{!}\mkern1mu}
\newcommand{\fatbang}{\mkern1mu{\fatexclamation}\mkern1mu}

\newcommand{\link}{\bowtie}
\newcommand{\subsetl}{\subseteq}

\newcommand{\diverges}{\uparrow}

\newcommand{\witassign}{\leftarrow}
\newcommand{\witderef}{\fatbang}
\newcommand{\refunreachN}{\ensuremath{\operatorname{\it no-refs}}\xspace}
\newcommand{\refunreach}[1]{\refunreachN(#1)}
\newcommand{\mtype}{\ensuremath{\operatorname{\it mtype}}\xspace}
\newcommand{\mbody}{\ensuremath{\operatorname{\it mbody}}\xspace}
\newcommand{\fields}{\ensuremath{\operatorname{\it fields}}\xspace}
\newcommand{\ptypes}{\ensuremath{\operatorname{\it p-types}}\xspace}
\newcommand{\canoverride}{\ensuremath{\operatorname{\it can-override}}\xspace}
\newcommand{\EvalN}{E}
\newcommand{\Eval}[1]{\EvalN[#1]}
\newcommand{\cupdot}{\mathbin{\mathaccent\cdot\cup}}

\newcommand{\tplab}[2]{{#1}^{#2}}

\newcommand{\targetconcsem}{\mathrel{\TrgCol{\Rightarrow}}}
\newcommand{\targetconcsemSW}{\mathrel{\TrgCol{\hspace*{0.6em}\scalebox{1}{\rotatebox[origin=c]{225}{$\targetconcsem$}}}}}
\newcommand{\targetconcsemSE}{\mathrel{\TrgCol{\hspace*{0.6em}\scalebox{1}{\rotatebox[origin=c]{315}{$\targetconcsem$}}}}}

\makeatletter
\newlength{\@overset@width}
\newcommand{\xtwoheadedrightarrow}[2][]{\mathrel{\mathord{%
  \setbox0=\hbox{$\scriptstyle#2\mkern2mu$}%
  \setlength{\@overset@width}{\wd0}
  \ifdim\@overset@width<6pt
    \xrightarrow[#1]{\makebox[6pt]{$\scriptstyle#2\mkern2mu$}}%
  \else
    \xrightarrow[#1]{#2\mkern2mu}%
  \fi%
  }\mathllap{\rightarrow~}%
}}
\newcommand{\def@stretch@arrow}[5][0.75em]{%
  \NewDocumentCommand{#2}{so}{\mathrel{#3{
    \IfValueTF{##2}{%
      \setbox0=\hbox{$\scriptstyle##2$}%
      \setlength{\@overset@width}{\wd0}%
      \ifdim\@overset@width<#1
        \mathord{\overset{{\color{black}##2}\mkern2mu}{#4}}%
      \else
        \mathord{#5{{\color{black}##2}}}%
      \fi%
  }{#4}}\IfBooleanT{##1}{{}^*}}}
}
\def@stretch@arrow[0.67em]{\cpstep}{\SrcCol}{\twoheadrightarrow}{\xtwoheadedrightarrow}
\def@stretch@arrow[0.67em]{\tstep}{\TrgCol}{\rightarrow}{\xrightarrow}
\def@stretch@arrow[0.67em]{\pstep}{\TrgCol}{\Rightarrow}{\xRightarrow}

\DeclareRobustCommand{\rvdots}{%
  \vbox{
    \baselineskip4\p@\lineskiplimit\z@
    \kern-\p@
    \hbox{.}\hbox{.}\hbox{.}
  }}
\makeatother

\definecolor{cpcolor}{HTML}{b51963}
\newcommand\Compute{\texttt{\color{cpcolor}C}\xspace}
\newcommand\Prove{\texttt{\color{cpcolor}P}\xspace}
\newcommand\ComputeProve{\texttt{\color{cpcolor}CP}\xspace}

\colorlet{sourcecolor}{RoyalBlue}
\colorlet{targetcolor}{RedOrange}

\newcommand{\SrcCol}[1]{{\color{sourcecolor}#1}}
\newcommand{\TrgCol}[1]{{\color{targetcolor}#1}}

\newcommand{\SrcProg}[1]{\programfont{\SrcCol{#1}}}
\newcommand{\TrgProg}[1]{\ensuremath{\mathtt{\TrgCol{#1}}}\xspace}

\definecolor{darkgreen}{RGB}{30, 130, 30}
\definecolor{purpleish}{RGB}{30, 130, 30}
\definecolor{cornellred}{rgb}{0.7, 0.11, 0.11}
\colorlet{middlecolor}{purpleish}
\newcommand{\annotated}[1]{%
  {\color{darkgreen}#1}%
}
\newcommand\annell{\annotated{\ell}}
\newcommand\annellp{\annotated{\ell'}}
\newcommand\annellpp{\annotated{\ell''}}
\newcommand\annc{\annotated{\texttt{C}}}
\newcommand\annp{\annotated{\texttt{P}}}
\newcommand\anncp{\annotated{\texttt{CP}}}

\newcommand{\conf}[2]{\mathchoice%
  {\left\langle #1 \mid #2 \right\rangle}
  {\langle #1 \mid #2 \rangle}
  {\langle #1 \mid #2 \rangle}
  {\langle #1 \mid #2 \rangle}}
\WithSuffix\newcommand\conf*[1]{\conf{#1}{\sigma}}
\newcommand{\cpconf}[4]{\conf{#1}{(#2, #3, #4)}}
\WithSuffix\newcommand\cpconf*[1]{\cpconf{#1}{\sigmaC}{\sigmaP}{\rho}}
\newcommand{\pconf}[5]{\conf{(#1, #3) \mid (#2, #4)}{#5}}
\WithSuffix\newcommand\pconf*[2]{\pconf{#1}{#2}{\sigmaC}{\sigmaP}{\rho}}

\newcommand{\sigmaC}{\sigma_\Compute}
\newcommand{\sigmaP}{\sigma_\Prove}
\newcommand{\SigmaC}{\Sigma_\Compute}
\newcommand{\SigmaP}{\Sigma_\Prove}
\newcommand{\SigmaCP}{\Sigma_\ComputeProve}

\newcommand{\Lower}[2]{\lfloor{#2}\rfloor_{#1}}
\WithSuffix\newcommand\Lower*[1]{\Lower{\ell}{#1}}
\newcommand{\Lift}[2]{\lceil{#2}\rceil^{#1}}
\WithSuffix\newcommand\Lift*[1]{\Lift{\ell}{#1}}
\newcommand{\LiftAnn}[2]{\Lift{#1}{#2}_\texttt{ann}}
\WithSuffix\newcommand\LiftAnn*[1]{\LiftAnn{\ell}{#1}}
\newcommand{\LiftTheta}[2]{\Lift{#1}{#2}_\Theta}
\WithSuffix\newcommand\LiftTheta*[1]{\LiftTheta{\ell}{#1}}
\newcommand{\LiftAnnTheta}[2]{\Lift{#1}{#2}_{\texttt{ann}\Theta}}
\WithSuffix\newcommand\LiftAnnTheta*[1]{\LiftAnnTheta{\ell}{#1}}
\newcommand{\FAnn}[1]{\hat{f}(#1)}

\newcommand\storetuple{\overline{\sigma}}
\newcommand\lessthan{\mathrel{\preceq}}
\newcommand\traceq{\equiv_\Nulltrace}

\newcommand\stuckpred[2]{\operatorname{stuck-allowed}(#1, #2)}
\newcommand\stuckpredzk[4]{\operatorname{stuck-allowed-combined}(#1, #2, #3, #4)}

\newcommand{\den}[1]{{\color{black} \llbracket #1 \rrbracket}}
\newcommand{\denC}[2][\Compute]{\den{#2}_{#1}}
\newcommand{\denP}[2][\Prove]{\den{#2}_{#1}}
\newcommand{\denL}[2][\ell]{\den{#2}_{#1}}
\newcommand{\denG}[2][\ell]{\den{#2}_{#1}}
\newcommand{\denFull}[1]{\denG{#1}}
\newcommand{\fullComp}[2][]{\mathcal{F}\den{#2}^{#1}}

\newcommand{\witden}[1]{{\color{black} \{\mkern-4mu \mid \mkern-2mu #1 \mkern-2mu \mid \mkern-4mu\}}}
\newcommand{\witdenR}[2][\varphi]{\witden{#2}^{#1}}
\newcommand{\witdenC}[2][\Compute]{\witdenR{#2}_{#1}}
\newcommand{\witdenP}[2][\Prove]{\witdenR{#2}_{#1}}
\newcommand{\witdenL}[2][\ell]{\witdenR{#2}_{#1}}

\newcommand{\seq}{\mathrel{;}}

\newcommand{\opalloc}{\operatorname{alloc}}
\newcommand{\opset}{\operatorname{set}}
\newcommand{\opgen}{\operatorname{gen}}
\newcommand{\opverif}{\operatorname{verif}}
\newcommand{\emit}{\mathcal{E}}
\newcommand{\trace}{t}
\newcommand{\Nulltrace}{\bullet}

\newcommand{\sopalloc}{\SrcCol{\operatorname{alloc}}}
\newcommand{\sopset}{\SrcCol{\operatorname{set}}}
\newcommand{\sopgen}{\SrcCol{\operatorname{gen}}}
\newcommand{\sopverif}{\SrcCol{\operatorname{verif}}}
\newcommand{\strace}{\SrcCol{\trace}}
\newcommand{\semit}{\SrcCol{\emit}}

\newcommand{\topalloc}{\TrgCol{\operatorname{alloc}}}
\newcommand{\topset}{\TrgCol{\operatorname{set}}}
\newcommand{\topgen}{\TrgCol{\operatorname{gen}}}
\newcommand{\topverif}{\TrgCol{\operatorname{verif}}}
\newcommand{\ttrace}{\TrgCol{\trace}}
\newcommand{\temit}{\TrgCol{\emit}}

\newcommand{\TContextN}{\mathbf{\TrgCol{C_T}}}

\newcommand{\SContextN}{\SrcCol{C_S}}

\newcommand{\BehavN}{\programfont{Behav}}
\newcommand{\Behav}[1]{\BehavN(#1)}
\newcommand{\TrgBehavN}{\TrgCol{\mathbf{Behav}}}
\newcommand{\SrcBehavN}{\SrcCol{\BehavN}}
\newcommand{\TBehav}[1]{\TrgBehavN(#1)}
\newcommand{\SBehav}[1]{\SrcBehavN(#1)}
\newcommand{\PartialP}{\SrcCol{P}}
\newcommand{\FullProg}{W}

\newcommand{\SFullProg}{\SrcCol{\FullProg}}

\newcommand{\Tlink}{\mathbin{\TrgCol{\link}}}

\newcommand{\TTau}{\tau}
\newcommand\TBool{\programfont{bool}}
\newcommand\TInt{\programfont{int}}
\newcommand\TUnit{\programfont{unit}}
\newcommand\TRefN{\programfont{ref}}
\newcommand\TRefTau[1]{\TRefN\mkern4mu{#1}}
\WithSuffix\newcommand\TRefTau*[1][\TTau]{\TRefTau{#1}}
\newcommand\TProofOfN{\TrgProg{proof \mhyphen of}}
\newcommand\TProofOf[1]{\ProofOfN~{#1}}
\WithSuffix\newcommand\TProofOf*[1][\alpha]{\TProofOf{#1}}
\newcommand\TProof{\programfont{proof}}

\newcommand{\TVal}{v}
\newcommand\TVals[1]{\overline{#1}}
\WithSuffix\newcommand\TVals*{\TVals{\TVal}}
\newcommand\TVar[1]{#1}
\WithSuffix\newcommand\TVar*[1][x]{\TVar{#1}}
\newcommand\TUnitVal{\TrgCol{()}}
\newcommand\TTrue{\TrgProg{true}}
\newcommand\TFalse{\TrgProg{false}}
\newcommand\TLoc{{\iota}}
\newcommand\TNewN{\TrgProg{new}}
\newcommand\TNew[2]{\TNewN\mkern4mu{#1}(#2)}
\WithSuffix\newcommand\TNew*[1][\TVals*]{\TNew{C}{#1}}
\newcommand\TUsingN{\TrgProg{using}}
\newcommand\TProofOfUsing[2]{\TProofOfN~#1~\TUsingN~#2}
\WithSuffix\newcommand\TProofOfUsing*{\TProofOfUsing{\alpha}{\TVals*}}

\newcommand{\TExp}{e}
\newcommand\TRefVal[1]{\TrgProg{ref}\mkern4mu{#1}}
\WithSuffix\newcommand\TRefVal*[1][\TVal]{\TRefVal{#1}}
\newcommand\TDeref[1]{\TrgProg{\bang}{#1}}
\WithSuffix\newcommand\TDeref*[1][\TVal]{\TDeref{#1}}
\newcommand\Tcoloneq{\TrgCol{\coloneqq}}
\newcommand\TAssign[2]{#1~\Tcoloneq~#2}
\WithSuffix\newcommand\TAssign*{\TAssign{v}{w}}
\newcommand\TCast[2]{\TrgCol{(}#1\TrgCol{)}#2}
\WithSuffix\newcommand\TCast*{\TCast{C}{\TVal}}
\newcommand\TField[2]{#1.#2}
\WithSuffix\newcommand\TField*{\TField{\TVal}{f}}
\newcommand\TCall[3]{#1.#2(#3)}
\WithSuffix\newcommand\TCall*{\TCall{\TVal}{m}{\overline{w}}}
\newcommand\TAlloc[1]{\TrgProg{alloc}(#1)}
\newcommand\TLetN{\TrgProg{let}}
\newcommand\TInN{\TrgProg{in}}
\newcommand\TLetIn[3]{\TLetN~{#1}~\TrgCol{=}~{#2}~\TInN~{#3}}
\newMLP\TLetInMlp[3]{\TLetN~{#1}~\TrgCol{=}~{#2} \\\TInN~{#3}}
\WithSuffix\newcommand\TLetIn*{\TLetIn{x\ty\TTau}{\TExp_1}{\TExp_2}}
\newcommand\TIfN{\TrgProg{if}}
\newcommand\TThenN{\TrgProg{then}}
\newcommand\TElseN{\TrgProg{else}}
\newcommand\TIfThenElse[3]{\TIfN~#1~\TThenN~#2~\TElseN~#3}
\WithSuffix\newcommand\TIfThenElse*{\TIfThenElse{\TVal}{\TExp_1}{\TExp_2}}
\newcommand\TProveN{\TrgProg{prove}}
\newcommand\TProveUsing[3]{\TProveN~{#1}~\TrgCol{=}~{#2}~\TUsingN~{#3}}
\WithSuffix\newcommand\TProveUsing*[1][\TExp]{\TProveUsing{\alpha}{\CircuitDefault[#1]}{\overline{v_x}[\overline{v_y}]}}
\newcommand\TVerifyN{\TrgProg{verify}}
\newcommand\TProvesN{\TrgProg{proves}}
\newcommand{\TVerify}[3]{\TVerifyN~{#1}~\TProvesN~{#2}~\TUsingN~{#3}}
\WithSuffix\newcommand\TVerify*[1][\TVal]{\TVerify{#1}{\alpha}{\TVals*}}
\newcommand{\TCnpN}{\TrgProg{cnp}}
\newcommand{\TCnpAdmin}[8][\alpha]{\TCnpN^{(#1,#2,#3)}_{(#4,#5,#6)} \{#7\}~\TUsingN~{#8}}
\WithSuffix\newcommand\TCnpAdmin*[3]{\TCnpAdmin{#2}{#3}{\overline{x_p}}{\overline{x_s}}{\varphi}{#1}{\overline{v_p}[\overline{v_s}]}}

\newcommand{\TComputeN}{\TrgProg{compute}}
\newcommand{\TWithN}{\TrgProg{with}}
\newcommand{\TCnp}[8]{\TComputeN~({#1}[{#2}])~\TrgProg{and}~\TProveN~{#3}~\TWithN~({#4}[{#5}])~\{{#6}\}~\TUsingN~{#7}[{#8}]}
\WithSuffix\newcommand\TCnp*[1][\SExp]{%
  \TCnp{\overline{x_p}\ty\overline{\TTau_{x_p}}}{\overline{x_s}\ty\overline{\TTau_{x_s}}}{\alpha}
  {\overline{y_p}\ty\overline{\TTau_{y_p}}}{\overline{y_s}\ty\overline{\TTau_{y_s}}}
  {#1}{\overline{v_p}}{\overline{v_s}}%
}
\newcommand\Tseq{~\TrgCol{\seq}~}

\newcommand\TClassN{\TrgProg{class}}
\newcommand\TExtendsN{\TrgProg{extends}}
\newcommand\TSuperN{\TrgProg{super}}
\newcommand\TThisN{\mathit{this}}
\newcommand\TClassListSyntax{\TClassN~C_{\{\ell\}}~\TExtendsN~C_{\{\ell\}}~\{\overline{f} \ty \overline{\TTau} \seq K \seq \overline{M} \}}

\newcommand\TConstructor{C(\overline{f} \ty \overline{\TTau}) \{ \TSuperN(\overline{f}) \seq \TThisN.\overline{f} \coloneq \overline{f} \}}

\newcommand{\Slink}{\mathbin{\SrcCol{\link}}}

\newcommand{\STau}{\tau}
\newcommand\SBool{\programfont{bool}}
\newcommand\SInt{\programfont{int}}
\newcommand\SUnit{\programfont{unit}}
\newcommand\SRefN{\programfont{ref}}
\newcommand\SRefTau[1]{\SRefN\mkern4mu{#1}}
\WithSuffix\newcommand\SRefTau*[1][\TTau]{\SRefTau{#1}}
\newcommand\SProofOfN{\programfont{proof \mhyphen of}}
\newcommand\SProofOf[1]{\SProofOfN~{#1}}
\WithSuffix\newcommand\SProofOf*[1][\alpha]{\SProofOf{#1}}
\newcommand\SProof{\programfont{proof}}

\newcommand\SWit{x}

\newcommand{\SVal}{v}
\newcommand\SVals[1]{\overline{#1}}
\WithSuffix\newcommand\SVals*{\SVals{\SVal}}
\newcommand\SVar[1]{#1}
\WithSuffix\newcommand\SVar*[1][x]{\SVar{#1}}
\newcommand\SUnitVal{\SrcCol{()}}
\newcommand\STrue{\SrcProg{\True}}
\newcommand\SFalse{\SrcProg{\False}}
\newcommand\SNewN{\SrcProg{new}}
\newcommand\SNew[3][\ell]{\SNewN\mkern4mu{#2}(#3)}
\WithSuffix\newcommand\SNew*[1][\ell]{\SNew[#1]{C}{\SVals*}}
\newcommand\SUsingN{\SrcProg{using}}
\newcommand\SProofOfUsing[2]{\SrcProg{\SProofOfN}~#1~\SUsingN~#2}
\WithSuffix\newcommand\SProofOfUsing*{\SProofOfUsing{\alpha}{\SVals*}}
\newcommand\SLocC[1]{\SrcProg{r}_\Compute(#1)}
\WithSuffix\newcommand\SLocC*{\SLocC{\iota}}
\newcommand\SLocP[1]{\SrcProg{r}_\Prove(#1)}
\WithSuffix\newcommand\SLocP*{\SLocP{\iota}}
\newcommand\SLocCP[2]{\SrcProg{r}_\ComputeProve(#1, #2)}
\WithSuffix\newcommand\SLocCP*{\SLocCP{\iota_1}{\iota_2}}

\newcommand{\SExp}{e}
\newcommand\SRefVal[2]{\SrcCol{\SRefN}_{#1}\mkern4mu{#2}}
\WithSuffix\newcommand\SRefVal*[1][\ell]{\SRefVal{#1}{\SVal}}
\newcommand\SDeref[1]{\SrcProg{\bang}{#1}}
\WithSuffix\newcommand\SDeref*[1][\SVal]{\SDeref{#1}}
\newcommand\Scoloneq{\SrcCol{\coloneqq}}
\newcommand\SAssign[2]{#1~\Scoloneq~#2}
\WithSuffix\newcommand\SAssign*{\SAssign{\SVal_1}{\SVal_2}}
\newcommand\SCast[2]{\SrcCol{(}#1\SrcCol{)}#2}
\WithSuffix\newcommand\SCast*{\SCast{C}{\SVal}}
\newcommand\SField[2]{#1.#2}
\WithSuffix\newcommand\SField*{\SField{\SVal}{f}}
\newcommand\SCall[4]{#1.#2_{#3}(#4)}
\WithSuffix\newcommand\SCall*[1][\ell]{\SCall{\SVal}{m}{#1}{\overline{w}}}
\newcommand\SAlloc[1]{\SrcCol{\AllocN}(#1)}
\newcommand\SLetN{\SrcProg{let}}
\newcommand\SInN{\SrcProg{in}}
\newcommand\SLetIn[3]{\SLetN~{#1}~\SrcCol{=}~{#2}~\SInN~{#3}}
\WithSuffix\newcommand\SLetIn*[1][\ell]{\SLetIn{x\ty\STau}{\SExp_1}{\SExp_2}}
\newMLP*{\SLetInBreak}[3]{\SLetN~{#1} = {#2}~\SInN \\ {#3}}
\newcommand\SIfN{\SrcProg{if}}
\newcommand\SThenN{\SrcProg{then}}
\newcommand\SElseN{\SrcProg{else}}
\newcommand\SIfThenElse[3]{\SIfN~#1~\SThenN~#2~\SElseN~#3}
\WithSuffix\newcommand\SIfThenElse*{\SIfThenElse{\SVal}{\SExp_1}{\SExp_2}}
\newMLP*{\SIfThenElseBreak}[3]{\SIfN~{#1}~\SThenN \\ \mkern30mu{#2} \\ \SElseN \\ \mkern30mu{#3}}
\newMLP*{\SSeqBreak}[2]{#1~\Sseq \\ #2}
\newcommand\SProveN{\SrcProg{prove}}
\newcommand\SProveUsing[3]{\SProveN~{#1}~\SrcCol{=}~{#2}~\SUsingN~{#3}}
\WithSuffix\newcommand\SProveUsing*[1][\SExp]{\SProveUsing{\alpha}{\CircuitDefault[#1]}{\overline{v_x}[\overline{v_y}]}}
\newcommand\SVerifyN{\SrcProg{verify}}
\newcommand\SProvesN{\SrcProg{proves}}
\newcommand{\SVerify}[3]{\SVerifyN~{#1}~\SProvesN~{#2}~\SUsingN~{#3}}
\WithSuffix\newcommand\SVerify*[1][\SVal]{\SVerify{#1}{\alpha}{\SVals*}}
\newcommand\SWitAssign[2]{#1~\SrcCol{\witassign}~#2}
\WithSuffix\newcommand\SWitAssign*{\SWitAssign{\SWit}{\SVal}}
\newcommand\SWitDeref[1]{\SrcCol{\witderef}#1}
\WithSuffix\newcommand\SWitDeref*{\SWitDeref{\SWit}}
\newcommand{\SCnpN}{\SrcProg{cnp}}
\newcommand{\SCnpAdmin}[8][\alpha]{\SCnpN^{(#1,#2,#3)}_{(#4,#5,#6)} \{#7\}~\SUsingN~{#8}}
\WithSuffix\newcommand\SCnpAdmin*[3]{\SCnpAdmin{#2}{#3}{\overline{x_p}}{\overline{x_s}}{\varphi}{#1}{\overline{v_p}[\overline{v_s}]}}

\newcommand{\SComputeN}{\SrcProg{compute}}
\newcommand{\SWithN}{\SrcProg{with}}
\newcommand{\SCnp}[8]{\SComputeN~({#1}[{#2}])~\SrcProg{and}~\SProveN~{#3}~\SWithN~({#4}[{#5}])~\{{#6}\}~\SUsingN~{#7}[{#8}]}
\WithSuffix\newcommand\SCnp*[1][\SExp]{%
  \SCnp{\overline{x_p}\ty\overline{\TTau_{x_p}}}{\overline{x_s}\ty\overline{\TTau_{x_s}}}{\alpha}
  {\overline{y_p}\ty\overline{\TTau_{y_p}}}{\overline{y_s}\ty\overline{\TTau_{y_s}}}
  {#1}{\overline{v_p}}{\overline{v_s}}%
}
\newcommand\Sseq{~\SrcCol{\seq}~}

\newcommand\Bool{\programfont{bool}}
\newcommand\Int{\programfont{int}}
\newcommand\Unit{\programfont{unit}}
\newcommand\RefN{\programfont{ref}}
\newcommand\Reft{\RefN\mkern4mu}
\newcommand\Refsub[1]{\RefN_{#1}\mkern4mu}
\newcommand\Refl{\Refsub{\ell}}

\newcommand\ThisN{\programfont{this}}
\newcommand\This{\mkern4mu\ThisN}
\newcommand\True{\programfont{true}}
\newcommand\False{\programfont{false}}
\newcommand\NewN{\programfont{new}}
\newcommand\AllocN{\programfont{alloc}}
\newcommand\New{\NewN\mkern4mu}
\newcommand\Newsub[1]{\NewN\mkern4mu}
\newcommand\Newl{\Newsub{\ell}}
\newcommand\IfN{\programfont{if}}
\newcommand\ThenN{\programfont{then}}
\newcommand\ElseN{\programfont{else}}
\newcommand\If{\IfN\mkern4mu}
\newcommand\Then{\mkern4mu\ThenN\mkern4mu}
\newcommand\Else{\mkern4mu\ElseN\mkern4mu}
\newcommand\LetN{\programfont{let}}
\newcommand\Let{\LetN\mkern4mu}
\newcommand\InN{\programfont{in}}
\newcommand\In{\mkern4mu\InN\mkern4mu}
\newcommand\IfThenElse[4][]{\IfN\if\relax\detokenize{#1}\relax\else_{#1}\fi~#2~\ThenN~#3~\ElseN~#4}
\newcommand{\LetIn}[3]{\LetN~{#1} = {#2}~\InN~{#3}}
\newMLP*{\LetIn}[3]{\LetN~{#1} = {#2}~\InN \\ {#3}}
\newcommand{\Alloc}[1]{\AllocN(#1)}

\newcommand{\UsingN}{\programfont{using}}
\newcommand{\ComputeN}{\programfont{compute}}
\newcommand\Proof{\programfont{proof}}
\newcommand\ProveN{\programfont{prove}}
\newcommand\Provet{\ProveN\mkern4mu}
\newcommand{\ProveWith}[3]{\ProveN~{#1} = {#2} ~\UsingN~ {#3}}
\WithSuffix\newcommand\ProveWith*[1][e]{\ProveWith{\alpha}{\CircuitDefault[#1]}{\overline{v_x}[\overline{v_y}]}}
\newcommand\ProvesN{\programfont{proves}}
\newcommand\Proves{\mkern4mu\ProvesN\mkern4mu}
\newcommand\ProofOfN{\programfont{proof \mhyphen of}}
\newcommand\ProofOf[1]{\ProofOfN~{#1}}
\WithSuffix\newcommand\ProofOf*{\ProofOf{\alpha}}
\newcommand{\ProofWith}[2]{\ProofOfN~{#1}~\UsingN~{#2}}
\WithSuffix\newcommand\ProofWith*{\ProofWith{\alpha}{\overline{v}}}
\newcommand\CircuitDefault[1][e]{\exists \overline{x} \ty \overline{\TTau_x}[\overline{y} \ty \overline{\TTau_y}].#1}
\newcommand\WithN{\programfont{with}}
\newcommand\With{\mkern4mu\WithN\mkern4mu}
\newcommand\VerifyN{\programfont{verify}}
\newcommand\Verify{\VerifyN\mkern4mu}
\newcommand\Verifyt{\Verify}
\mathchardef\mhyphen="2D
\newcommand{\VerifyWith}[3]{\VerifyN~{#1}~\ProvesN~{#2}~\UsingN~{#3}}
\WithSuffix\newcommand\VerifyWith*[1]{\VerifyWith{#1}{\alpha}{\overline{v}}}
\newcommand\CombInputs[2]{#1[#2]}
\WithSuffix\newcommand\CombInputs*{\CombInputs{\overline{x}\ty\overline{\tau_x}}{\overline{y}\ty\overline{\tau_y}}}
\newcommand{\CnpN}{\programfont{cnp}}
\newcommand{\CnpAdmin}[8][\alpha]{\CnpN^{(#1,#2,#3)}_{(#4,#5,#6)} \{#7\}~\UsingN~{#8}}
\WithSuffix\newcommand\CnpAdmin*[3]{\CnpAdmin{#2}{#3}{\overline{x_p}}{\overline{x_s}}{\varphi}{#1}{\overline{v_p}[\overline{v_s}]}}

\newcommand{\Cnp}[8]{\ComputeN~({#1}[{#2}])~\programfont{and}~\ProveN~{#3}~\WithN~({#4}[{#5}])~\{{#6}\}~\UsingN~{#7}[{#8}]}
\WithSuffix\newcommand\Cnp*[1][e]{%
  \Cnp{\overline{x_p}\ty\overline{\tau_{x_p}}}{\overline{x_s}\ty\overline{\tau_{x_s}}}{\alpha}
  {\overline{y_p}\ty\overline{\tau_{y_p}}}{\overline{y_s}\ty\overline{\tau_{y_s}}}
  {#1}{\overline{v_p}}{\overline{v_s}}%
}

\newcommand{\rone}{(\emph{i})\xspace}
\newcommand{\rtwo}{(\emph{ii})\xspace}
\newcommand{\rthree}{(\emph{iii})\xspace}
\newcommand{\rfour}{(\emph{iv})\xspace}
\newcommand{\rfive}{(\emph{v})\xspace}
\newcommand{\rsix}{(\emph{vi})\xspace}
\newcommand{\rseven}{(\emph{vii})\xspace}

\newcommand{\ProveHL}[1]{\colorbox{pink}{{$\displaystyle\scriptstyle{#1}$}}}
\newcommand{\ComputeHL}[1]{\colorbox{yellow}{{$\displaystyle\scriptstyle{#1}$}}}

\newcommand{\changes}[1]{#1}
\newcommand{\undercolor}[1]{%
  \bgroup
  \markoverwith{%
    \textcolor{#1}{\rule[-0.5ex]{0.2pt}{0.6pt}}%
  }%
  \ULon
}
\newcommand{\undergreen}{\undercolor{green}}
\newcommand{\underred}{\undercolor{red}}
\newcommand{\underblue}{\undercolor{blue}}

\newcommand{\codesize}{\footnotesize}

\SetRuleLabelVCenter

\definecolor{codeBlue}{rgb}{0.1,.23,0.85}
\definecolor{orangeGold}{rgb}{0.67, 0.33, 0.1}
\colorlet{numColor}{Red}

\lstdefinelanguage{exlang}{
  keywords=[1]{
    class, interface, extends, implements,
    if, then, else, for, while,
    compute, and_prove, with, using,
    assert,
    return,
    prove, verify, proves,
    switch, case, break, continue,
    fail,
  },
  keywordstyle=[1]\color{codeBlue},
  keywords=[2]{
    hash, int, proof, proofOf,
    new,
    witness
  },
  keywordstyle=[2]\color{orangeGold},
  keywords=[3]{
    input,
    sha3,
    pub, priv,
  },
  keywordstyle=[3]\color{Plum},
  keywords=[4]{C, P, CP},
  keywordstyle=[4]\color{cpcolor},
  literate={0}{{{\color{numColor}0}}}{1}
           {1}{{{\color{numColor}1}}}{1}
           {2}{{{\color{numColor}2}}}{1}
           {[-0-]}{0}{1},
  sensitive=true,
  morecomment=[l][\color{ForestGreen}]{//}, % l is for line comment
  morecomment=[s]{/*}{*/}, % s is for start and end delimiter
  string=[b]", % defines that strings are enclosed in double quotes
  stringstyle=\upshape\color{Red},
  escapeinside={(*}{*)},
}
\tikzset{
  mapsto/.style={{Bar[width=4pt]}-{Stealth[length=5pt,width=4pt]},semithick},
  src-arrow/.style={-{{Straight Barb}[scale=0.85]{Straight Barb}[scale=0.85]}, draw=sourcecolor, semithick},
  conc-arrow/.style={-implies, double equal sign distance, draw=targetcolor, semithick},
  trg-arrow/.style={-{{Straight Barb}[scale=0.85]}, draw=targetcolor, semithick},
}

\theoremstyle{theorem}

\theoremstyle{definition}

\title{\zkstrudul: Programming zkSNARKs with Structural Duality}

\author{Rahul Krishnan}
\orcid{0000-0003-0230-5185}
\affiliation{
  \institution{University of Wisconsin--Madison}
 \city{Madison}
 \state{Wisconsin}
 \country{USA}
}
\email{rahulk@cs.wisc.edu}

\author{Ashley Samuelson}
\orcid{0009-0001-8800-2590}
\affiliation{
  \institution{University of Wisconsin--Madison}
  \city{Madison}
  \state{Wisconsin}
  \country{USA}
}
\email{ashley.samuelson@wisc.edu}

\author{Emily Yao}
\orcid{0009-0001-4728-4093}
\affiliation{
  \institution{University of Wisconsin--Madison}
 \city{Madison}
 \state{Wisconsin}
 \country{USA}
}
\email{eyao3@wisc.edu}

\author{Ethan Cecchetti}
\orcid{0000-0001-7900-8328}
\affiliation{
  \institution{University of Wisconsin--Madison}
 \city{Madison}
 \state{Wisconsin}
 \country{USA}
}
\email{cecchetti@wisc.edu}

\begin{document}

%% Abstract
\begin{abstract}
Non-Interactive Zero Knowledge (NIZK) proofs, such as zkSNARKS, let one prove knowledge of private data without revealing it or interacting with a verifier. 
While existing tooling focuses on specifying the predicate to be proven, real-world applications optimize predicate definitions
to minimize proof generation overhead, but must correspondingly transform predicate inputs.
Implementing these two steps separately duplicates logic that must precisely match to avoid catastrophic security flaws.
We address this shortcoming with \zkstrudul, a language that unifies input transformations and predicate definitions into a single \emph{combined abstraction}
from which a compiler can \emph{project} both procedures, eliminating duplicate code and problematic mismatches.
\zkstrudul provides a high-level abstraction to layer on top of existing NIZK technology and supports important features like recursive proofs.
We provide a source-level semantics and prove its behavior is identical to the projected semantics, allowing straightforward standard reasoning.

\end{abstract}

\maketitle

\DefineEvalRule{Eval} {
  \conf*{e} \tstep[\ttrace] \conf{e'}{\sigma'}
}{
  \conf*{\Eval{e}} \tstep[\ttrace] \conf{\Eval{e'}}{\sigma'}
}

\DefineEvalRule{Let}{ }{\conf*{\TLetIn{x}{v}{e}} \tstep \conf*{\subst{e}{x}{v}}}

\DefineEvalRule{IfT} { } {
  \conf*{\TIfThenElse{\TTrue}{e_1}{e_2}} \tstep \conf*{e_1}
}
\DefineEvalRule{IfF} { } {
  \conf*{\TIfThenElse{\TFalse}{e_1}{e_2}} \tstep \conf*{e_2}
}
\DefineEvalRule{Ref} {
    \iota \not\in \dom(\sigma)
}{\conf*{\TRefVal*} \tstep[\topalloc(\iota, v)] \conf{\iota}{\subst{\sigma}{\iota}{v}}}

\DefineEvalRule{Deref}{
  v = \sigma(\iota) \\
  v \neq \bot
}{
  \conf*{\TDeref{\iota}} \tstep \conf*{v}
}

\DefineEvalRule{Assign} {
    \iota \in \dom(\sigma)
}{\conf*{\TAssign{\iota}{\TVal}} \tstep[\topset(\iota,v)] \conf{\TUnitVal}{\subst{\sigma}{\iota}{v}}}

\DefineEvalRule{Cast} {
    D <: C
} {
  \conf*{\TCast{C}{(\TNew{D}{\TVals*})}} \tstep \conf*{\TNew{D}{\TVals*}}
}
\DefineEvalRule{Field}{ }{\conf*{\TField{\TNew*}{f_i}} \tstep \conf*{v_i}}

\DefineEvalRule{Call} {
    \operatorname{mbody}(m,C) = (e, \overline{x}, \tau) \\
}{
  \conf*{\TCall{(\TNew*)}{m}{\overline{w}}} \tstep \conf*{\subst*{e}{{\overline{x}}{\overline{w}}{\TThisN}{\TNew*}}}
}
\DefineEvalRule{Alloc} {
  \iota \not\in \dom(\sigma)
} {
  \conf*{\TAlloc{\TTau}} \tstep[\topalloc(\iota, \bot)] \conf{\iota}{\sigma[\iota \mapsto \bot]}
}
\DefineEvalRule{Prove} {
  \conf{\subst*{e}{{\overline{x}}{\overline{v_x}}{\overline{y}}{\overline{v_y}}}}{\varnothing} \tstep* \conf{\TTrue}{\_}
} {
  \conf*{\TProveUsing{\alpha}{\exists \overline{x}\ty\overline{\tau_x}[\overline{y}\ty\overline{\tau_y}].e}{\overline{v_x}[\overline{v_y}]}}
  \tstep[\topgen(\alpha, \overline{v_x}, \overline{v_y})]
  \conf*{\TProofOfUsing*}
}
\DefineEvalRule{VerifyT}{ }{\conf*{\TVerify*[(\TProofOfUsing*)]} \tstep[\topverif(\alpha,\overline{v}, \top)] \conf*{\TTrue}}

\DefineEvalRule{VerifyF}{
  \TProofOfUsing{\beta}{\overline{w}} \neq \TProofOfUsing*
}{\conf*{\TVerify*[(\TProofOfUsing{\beta}{\overline{w}})]} \tstep[\topverif(\alpha,\overline{w},\bot)] \conf*{\TFalse}}

\newtoggle{EComputeAndProveToggle}
\DefineEvalRule{ComputeAndProveInit} {
    \varphi = \subst*{}{{\overline{x_p}}{\overline{\iota_{x_p}}}{\overline{x_s}}{\overline{\iota_{x_s}}}} \\
    \overline{\iota_{x_p}}, \overline{\iota_{x_s}} \not\in \dom(\sigma) \\
    \sigma' = \sigma[\overline{\iota_{x_p}} \mapsto \bot, \overline{\iota_{x_s}} \mapsto \bot]
} {
  {\begin{array}{c}
    \conf*{\iftoggle{EComputeAndProveToggle}{\SCnp*}{\TCnp*}} \\
    \tstep[\topalloc(\overline{\iota_{x_p}}, \bot),~\topalloc(\overline{\iota_{x_s}}, \bot)]
    \conf{\iftoggle{EComputeAndProveToggle}{\SCnpAdmin*{\subst*{e}{{\overline{y_p}}{\overline{v_p}}{\overline{y_s}}{\overline{v_s}}}}{\varnothing}{\varnothing}}{\TCnpAdmin*{\subst*{e}{{\overline{y_p}}{\overline{v_p}}{\overline{y_s}}{\overline{v_s}}}}{\varnothing}{\varnothing}}}{\sigma'}
  \end{array}}
}
\DefineEvalRule{ComputeAndProveStep} {
    \cpconf{e}{\sigma}{\sigmaP}{\rho} \cpstep[\strace] \cpconf{e'}{\sigma'}{\sigmaP'}{\rho'}
} {
    \conf*{\iftoggle{EComputeAndProveToggle}{\SCnpAdmin*{e}{\sigmaP}{\rho}}{\TCnpAdmin*{e}{\sigmaP}{\rho}}} \tstep[\witdenR{\strace}] \conf{\iftoggle{EComputeAndProveToggle}{\SCnpAdmin*{e'}{\sigmaP'}{\rho'}}{\TCnpAdmin*{e'}{\sigmaP'}{\rho'}}}{\sigma'}
}
\DefineEvalRule{ComputeAndProveTrue} {
  \sigma' = \subst*{\sigma}{{\varphi(\overline{x_p})}{\rho(\overline{x_p})}{\varphi(\overline{x_s})}{\rho(\overline{x_s})}}
} {
  \conf*{\iftoggle{EComputeAndProveToggle}{\SCnpAdmin*{\STrue}{\sigmaP}{\rho}}{\TCnpAdmin*{e}{\sigmaP}{\rho}}} \tstep[\topgen(\alpha, \overline{v_p} :: \rho(\overline{x_p}), \overline{v_s} :: \rho(\overline{x_s}))] \conf{\TProofOfUsing{\alpha}{\overline{v_p}{::}\rho(\overline{x_p})}}{\sigma'}
}

\DefineCEvalRule{Lift}{
  \conf{e}{\varnothing} \tstep[\ttrace] \conf{e'}{\varnothing}
}{\cpconf*{\Lift{\ell}{e}} \cpstep[\strace] \cpconf*{\Lift{\ell}{e'}}}

\DefineCEvalRule{Eval} {
  \cpconf*{e} \cpstep[\strace] \cpconf{e'}{\sigmaC'}{\sigmaP'}{\rho'}
} {
  \cpconf*{\Eval{e}} \cpstep[\strace] \cpconf{\Eval{e'}}{\sigmaC'}{\sigmaP'}{\rho'}
}
\DefineCEvalRule{Let} { } {
  \cpconf*{\SLetIn{x\ty\TTau}{v}{e}} \cpstep \cpconf*{\subst{e}{x}{v}}
}
\newtoggle{ECRefCToggle}
\DefineCEvalRule{RefC} {
  \iota \notin \dom(\sigmaC)
} {
  {\begin{array}{@{}l@{}}
    \cpconf*{\SRefVal{\Compute}{v}}
    \iftoggle{ECRefCToggle}{ \\ \hspace*{1em} }{}
    \cpstep[\sopalloc(\iota, \Lower{\Compute}{v})] \cpconf{\SLocC*}{\subst{\sigmaC}{\iota}{\Lower{\Compute}{v}}}{\sigmaP}{\rho}
    \end{array}}  
}
\DefineCEvalRule{RefP} {
  \iota \notin \dom(\sigmaP)
} {
  \cpconf*{\SRefVal{\Prove}{v}} \cpstep \cpconf{\SLocP*}{\sigmaC}{\subst{\sigmaP}{\iota}{\Lower{\Prove}{v}}}{\rho}
}
\DefineCEvalRule{RefCP} {
  \iota_1 \notin \dom(\sigmaC) \\
  \iota_2 \notin \dom(\sigmaP) \\
} {
  \cpconf*{\SRefVal{\ComputeProve}{v}} \cpstep[\sopalloc(\iota_1, \Lower{\Compute}{v})]
  \cpconf{\SLocCP*}{\subst{\sigmaC}{\iota_1}{\Lower{\Compute}{v}}}{\subst{\sigmaP}{\iota_2}{\Lower{\Prove}{v}}}{\rho}
}
\DefineCEvalRule{DerefC} {
  v = \sigmaC(\iota) \\
  v \neq \bot
} {
  \cpconf*{\SDeref{\SLocC*}} \cpstep \cpconf*{\Lift{\Compute}{v}}
}
\DefineCEvalRule{DerefP} {
} {
  \cpconf*{\SDeref{\SLocP*}} \cpstep \cpconf*{\Lift{\Prove}{\sigmaP(\iota)}}
}
\DefineCEvalRule{DerefCP} {
  v = \Lift{\Compute}{\sigmaC(\iota_1)} \sqcup \Lift{\Prove}{\sigmaP(\iota_2)}
} {
  \cpconf*{\SDeref{\SLocCP*}} \cpstep \cpconf*{v}
}
\DefineCEvalRule{Alloc} {
  \iota \not\in \dom(\sigmaC)
} {
  \cpconf*{\SAlloc{\TTau}} \cpstep[\sopalloc(\iota,\bot)] \cpconf{\SLocC*}{\subst{\sigmaC}{\iota}{\bot}}{\sigmaP}{\rho}
}
\DefineCEvalRule{AssignC} {
  \iota \in \dom(\sigma_\Compute)
} {
  \cpconf*{\SAssign{\SLocC*}{v}} \cpstep[\sopset(\iota, \Lower{\Compute}{v})] \cpconf{\SUnitVal}{\subst{\sigmaC}{\iota}{\Lower{\Compute}{v}}}{\sigmaP}{\rho}
}
\DefineCEvalRule{AssignP} {
  \iota \in \dom(\sigma_\Prove)
} {
  \cpconf*{\SAssign{\SLocP*}{v}} \cpstep \cpconf{\SUnitVal}{\sigmaC}{\subst{\sigmaP}{\iota}{\Lower{\Prove}{v}}}{\rho}
}
\DefineCEvalRule{AssignCP} {
  \iota_1 \in \dom(\sigma_\Compute)\\
  \iota_2 \in \dom(\sigma_\Prove)
} {
  \cpconf*{\SAssign{\SLocCP*}{v}}
  \cpstep[\sopset(\iota_1, \Lower{\Compute}{v})]
  \cpconf{\SUnitVal}{\subst{\sigmaC}{\iota_1}{\Lower{\Compute}{v}}}{\subst{\sigmaP}{\iota_2}{\Lower{\Prove}{v}}}{\rho}
}
\DefineCEvalRule{InputAssign} {
  x \notin \dom(\rho)
} {
  \cpconf*{\SWitAssign{x}{v}} \cpstep[\sopset(x,v)] \cpconf{\SUnitVal}{\sigmaC}{\sigmaP}{\subst{\rho}{x}{v}}
}
\DefineCEvalRule{InputDeref} {
} {
  \cpconf*{\SWitDeref{x}} \cpstep \cpconf*{\rho(x)}
}
\DefineCEvalRule{ComputeAndProveInit} {
    \varphi = \subst*{}{{\overline{x_p}}{\overline{\iota_{xp}}}{\overline{x_s}}{\overline{\iota_{xs}}}} \\
    \overline{\iota_{xp}}, \overline{\iota_{xs}} \not\in \dom(\sigmaC) \\
    \sigmaC' = \sigmaC[\overline{\iota_{xp}} \mapsto \bot, \overline{\iota_{xs}} \mapsto \bot]
} {
  {\begin{array}{c}
    \cpconf*{\SCnp*} \\
    \cpstep[\sopalloc(\overline{\iota_{xp}}, \bot), \sopalloc(\overline{\iota_{xs}}, \bot)]
    \cpconf{\SCnpAdmin*{\subst*{e}{{\overline{y_p}}{\overline{v_p}}{\overline{y_s}}{\overline{v_s}}}}{\varnothing}{\varnothing}}{\sigmaC'}{\sigmaP}{\rho}
  \end{array}}
}
\DefineCEvalRule{ComputeAndProveStep} {
  \cpconf*{e} \cpstep[\strace] \cpconf{e'}{\sigmaC'}{\sigmaP'}{\rho'}
} {
  \cpconf{\SCnpAdmin*{e}{\sigmaP}{\rho}}{\sigmaC}{\tilde{\sigma}_\Prove}{\tilde{\rho}}
  \cpstep[\witdenR{\strace}]
  \cpconf{\SCnpAdmin*{e'}{\sigmaP'}{\rho'}}{\sigmaC'}{\tilde{\sigma}_\Prove}{\tilde{\rho}}
}
\DefineCEvalRule{ComputeAndProveTrue} {
  \sigmaC' = \subst*{\sigmaC}{{\varphi(\overline{x_p})}{\rho(\overline{x_p})}{\varphi(\overline{x_s})}{\rho(\overline{x_s})}}
} {
  {\begin{array}{l}
    \cpconf{\SCnpAdmin*{\STrue}{\sigmaP}{\rho}}{\sigmaC}{\tilde{\sigma}_\Prove}{\tilde{\rho}} \\
    \qquad\qquad \cpstep[\sopgen(\alpha,\overline{v_p}::\rho(\overline{x_p}), \overline{v_s}::\rho(\overline{x_s}))]
    \cpconf{\SProofOfUsing{\alpha}{\overline{v}{::}\rho(\overline{x_p})}}{\sigmaC'}{\tilde{\sigma}_\Prove}{\tilde{\rho}}
  \end{array}}
}

\DefineTypingRule{Var} {
    \Gamma(x) = \TTau
} {
    \Sigma ; \Gamma \proves x : \TTau
}
\DefineTypingRule{Unit} { } {
    \Sigma ; \Gamma \proves \TUnitVal : \Unit
}
\DefineTypingRule{True} { } {
    \Sigma ; \Gamma \proves \TTrue : \Bool
}
\DefineTypingRule{False} { } {
    \Sigma ; \Gamma \proves \TFalse : \Bool
}
\DefineTypingRule{Location} {
    \Sigma(\iota) = \TTau
}{
    \Sigma ; \Gamma \proves \iota : \TRefTau*
}
\DefineTypingRule{ProofOf} {
    \ptypes(\alpha) = \overline{\TTau} \\
    \proves \overline{v} : \overline{\TTau}
} {
    \Sigma ; \Gamma \proves \TProofOfUsing* : \ProofOf*
}
\DefineTypingRule{New} {
    \Sigma ; \Gamma \proves \overline{v} : \overline{\TTau} \\\\
    \fields(C) = \overline{f} : \overline{\TTau}
} {
    \Sigma ; \Gamma \proves \TNew* : C
}
\DefineTypingRule{Val} {
    \Sigma ; \Gamma \proves v : \TTau
} {
    \Sigma ; \Gamma \proves v : \TTau \changes{\produces \ell}
}
\DefineTypingRule{Cast} {
    \Sigma ; \Gamma \proves v : D
} {
    \Sigma ; \Gamma \proves \TCast* : C \changes{\produces \ell}
}
\DefineTypingRule{Field} {
    \Sigma ; \Gamma \proves v : C \\\\
    \fields(C) = \overline{f} : \overline{\TTau} 
} {
    \Sigma ; \Gamma \proves \TField{v}{f_i} : \TTau_i \changes{\produces \ell}
}
\DefineTypingRule{Call} {
    \mtype(m, C) = \overline{\TTau} \xrightarrow{\changes{\ell}} \TTau_1 \\\\
    \Sigma ; \Gamma \proves v : C \\
    \Sigma ; \Gamma \proves \overline{w} : \overline{\TTau}
} {
    \Sigma ; \Gamma \proves \TCall{v}{m}{\overline{w}} : \TTau_1 \changes{\produces \ell}
}
\DefineTypingRule{Ref} {
    \Sigma ; \Gamma \proves v : \TTau
} {
    \Sigma ; \Gamma \proves \TRefVal* : \TRefTau* \changes{\produces \ell}
}
\DefineTypingRule{Deref} {
    \Sigma ; \Gamma \proves v : \TRefTau*
} {
    \Sigma ; \Gamma \proves \TDeref* : \TTau \changes{\produces \ell}
}
\DefineTypingRule{Assign} {
    \Sigma ; \Gamma \proves v_1 : \TRefTau* \\\\
    \Sigma ; \Gamma \proves v_2 : \TTau
} {
    \Sigma ; \Gamma \proves \TAssign{v_1}{v_2} : \Unit \changes{\produces \ell}
}
\DefineTypingRule{If} {
    \Sigma ; \Gamma \proves v : \Bool \\
    \changes{\ell_1 \cap \ell_2 \neq \varnothing} \\\\
    \Sigma ; \Gamma \proves e_1 : \TTau \changes{\produces \ell_1} \\
    \Sigma ; \Gamma \proves e_2 : \TTau \changes{\produces \ell_2}
} {
    \Sigma ; \Gamma \proves \TIfThenElse* : \TTau \changes{\produces \ell_1 \cap \ell_2}
}
\DefineTypingRule{Let}{
    \Sigma ; \Gamma \proves e_1 : \TTau_1 \changes{\produces \ell_1} \\
    \changes{\ell_1 \cap \ell_2 \neq \varnothing} \\\\
    \Sigma ; \Gamma, x \ty \TTau_1 \proves e_2 : \TTau_2 \changes{\produces \ell_2}
} {
    \Sigma ; \Gamma \proves \TLetIn{x\ty\TTau_1}{e_1}{e_2} \changes{\produces \ell_1 \cap \ell_2}
}
\DefineTypingRule{Subtype} {
    \Sigma ; \Gamma \proves e : \TTau_1 \changes{\produces \ell}\\
    \TTau_1 <: \TTau_2 
} {
    \Sigma ; \Gamma \proves e : \TTau_2 \changes{\produces \ell}
}
\DefineTypingRule{Prove} {
    \cdot ; \overline{x} \ty \overline{\TTau_x}, \overline{y} \ty \overline{\TTau_y} \proves e : \Bool \produces \Prove \\
    \Sigma ; \Gamma \proves \overline{v} : \overline{\TTau_x} \\\\
    \Sigma ; \Gamma \proves \overline{w} : \overline{\TTau_y} \\
    \refunreach{\overline{\TTau_x}} \\
    \refunreach{\overline{\TTau_y}} \\
} {
    \Sigma ; \Gamma \proves \TProveUsing* : \ProofOf* \produces \Compute
}
\DefineTypingRule{Verify} {
    \Sigma ; \Gamma \proves v : \Proof\\\\
    \ptypes(\alpha) = \overline{\TTau} \\
    \Sigma ; \Gamma \proves \overline{v} : \overline{\TTau}
} {
    \Sigma ; \Gamma \proves \TVerify*[v] : \Bool \produces \ell
}
\DefineTypingRule{Alloc} {
} {
    \Sigma ; \Gamma \proves \TAlloc{\TTau} : \TRefTau{\TTau} \produces \Compute
}
\newtoggle{TCombinedToggle}
\DefineTypingRule{ComputeAndProve} {
    \refunreach{\overline{\TTau_{x_p}}, \overline{\TTau_{x_s}}, \overline{\TTau_{y_p}}, \overline{\TTau_{y_s}}} \\
    \Sigma ; \Gamma \proves \overline{v_p} : \overline{\TTau_{y_p}} \\
    \Sigma ; \Gamma \proves \overline{v_s} : \overline{\TTau_{y_s}} \\\\
    \Gamma_\ComputeProve = \overline{y_p} \ty \tplab{\overline{\TTau_{y_p}}}{\ComputeProve}, \overline{y_s} \ty \tplab{\overline{\TTau_{y_s}}}{\ComputeProve} \\
    \Delta = \overline{x_p} \ty \overline{\TTau_{x_p}}, \overline{x_s} \ty \overline{\TTau_{x_s}} \\
    (\Sigma, \varnothing, \varnothing) ; \tplab{\Gamma}{\Compute}, \Gamma_\ComputeProve ; \Delta ; \varnothing \cproves e : \tplab{\Bool}{\Prove} \cproduces \{\overline{x_p}, \overline{x_s}\}
} {
    \Sigma ; \Gamma \proves \iftoggle{TCombinedToggle}{\SCnp*}{\TCnp*} : \TProofOf* \produces \Compute
}
\DefineTypingRule{Combined} {
    \dom(\Delta) = \overline{w_p},\overline{w_s} \\
    \Sigma'; \varnothing; \Delta; A' \cproves e : \tplab{\Bool}{\Prove} \cproduces \{\overline{w_p},\overline{w_s}\}
} {
    \Sigma ; \Gamma \proves \TCnpAdmin*{e}{\sigmaP}{\rho} : \TProofOf* \produces \Compute
}

\DefineRule{SubClass} {
    CT(C) = \TClassN~C~\TExtendsN~D~\{\cdots\}
} {
    C <: D
}
\DefineRule{SubProof}{ }{\ProofOfN~\alpha <: \Proof}
\DefineRule{SubTrans} {
    \TTau_1 <: \TTau_2 \\
    \TTau_2 <: \TTau_3
} {
    \TTau_1 <: \TTau_3
}

\DefineRule[MethodOk]{Method-Ok}{
    \Sigma ; \overline{x} \ty \overline{\TTau_1}, \TThisN \ty C \proves e : \TTau \changes{\produces \ell_m} \\
    \canoverride(D, m, \overline{\TTau_1} \xrightarrow{\changes{\ell}} \TTau) \\\\
    \changes{\ell \subsetl \ell_m \cap \ell_C} \\
    CT(C) = \TClassN~C_{\changes{\{\ell_C\}}}~\TExtendsN~D_{\changes{\{\ell_D\}}}~\{\cdots\}
} {
    \Sigma \proves \TTau\changes{\{\ell\}}~m(\overline{x} \ty \overline{\TTau_1}) \{ e \} \text{ ok in }C
}
\DefineRule{Class-Ok-Ref}[ClassOkRef] {
    \fields(D) = \overline{g} : \overline{\TTau_g} \\
    \Sigma \proves \overline{M} \text{ ok in }C \\\\
    K = C(\overline{g} \ty \overline{\TTau_g}, \overline{f} \ty \overline{\TTau_f})\{ \TSuperN (\overline{g}) \seq \This .\overline{f} = \overline{f} \} 
} {
    \Sigma \proves \TClassN~C_{\changes{\{\Compute\}}}~\TExtendsN~D_{\changes{\{\Compute\}}}~\{\overline{f} \ty \overline{\TTau_f} \seq K \seq \overline{M}\} \text{ ok}
}
\DefineRule{Class-Ok-NoRef}[ClassOkNoRef] {
    \changes{\refunreach{\overline{\TTau_g}, \overline{\TTau_f}}} \\
    \changes{\ell_D \subsetl \ell_C} \\\\
    \fields(D) = \overline{g} : \overline{\TTau_g} \\
    \Sigma \proves \overline{M} \text{ ok in }C \\\\
    K = C(\overline{g} \ty \overline{\TTau_g}, \overline{f} \ty \overline{\TTau_f})\{ \TSuperN (\overline{g}) \seq \This .\overline{f} = \overline{f} \} 
} {
    \Sigma \proves \TClassN~C_{\changes{\{\ell_C\}}}~\TExtendsN~D_{\changes{\{\ell_D\}}}~\{\overline{f} \ty \overline{\TTau_f} \seq K \seq \overline{M} \} \text{ ok}
}
\DefineRule{CT-Ok}[CTOk] {
    C \text{ referenced } \Rightarrow C \in \dom(CT) \\\\
    \forall C \in \dom(C).~\Sigma \proves CT(C) \text{ ok}
} {
    \Sigma \proves CT \text{ ok}
}

\DefineCTypingRule{LiftVal} {
    \ell \in \{\Compute, \Prove\} \\\\
    \Sigma_\ell ; \Gamma |_\ell \proves v : \TTau 
} {
    \Sigma ; \Gamma \cproves \Lift{\ell}{v} : \tplab{\TTau}{\ell}
}

\DefineCTypingRule{LiftValCP} {
    \Theta : loc \times loc \rightharpoonup loc \\\\
    \Theta \text{ injective on} \dom(\SigmaCP) \subseteq \dom(\Theta) \\\\
    \Theta(\SigmaCP) ; \Gamma |_\ComputeProve \proves v : \TTau
} {
    \Sigma ; \Gamma \cproves \Lift{\ComputeProve}{v}_\Theta : \tplab{\TTau}{\ComputeProve}
}

\DefineCTypingRule{Val}{
    \Sigma ; \Gamma \cproves v : \tplab{\TTau}{\ell}
}{
    \Sigma ; \Gamma ; \Delta ; A \cproves v : \tplab{\TTau}{\ell} \cproduces A
}

\DefineCTypingRule{Lift} {
    \Sigma_\ell ; \Gamma |_\ell \proves e : \TTau \produces \ell \\
    \ell \in \{\Compute, \Prove\}
} {
    \Sigma ; \Gamma ; \Delta ; A \cproves \Lift{\ell}{e} : \tplab{\TTau}{\ell} \cproduces A
}

\newtoggle{LiftCPToggle}
\DefineCTypingRule{LiftCP} {
    {\begin{array}{@{}c@{}}
        \Theta(\SigmaCP) ; \Gamma |_\ComputeProve \proves e : \TTau \produces \ComputeProve \\
        \Theta : loc \times loc \rightharpoonup loc
        \iftoggle{LiftCPToggle}{ ~\text{injective on}~ \dom(\SigmaCP)}{}
        \iftoggle{LiftCPToggle}{ }{\\ \Theta ~\text{injective on}~ \dom(\SigmaCP) \subseteq \dom(\Theta)}
    \end{array}}
} {
    \Sigma ; \Gamma ; \Delta ; A \cproves \Lift{\ComputeProve}{e}_\Theta : \tplab{\TTau}{\ComputeProve} \cproduces A
}

\DefineCTypingRule{InputAssign} {
    \Delta(x) = \TTau \\
    x \not\in A \\
    \Sigma ; \Gamma \cproves v : \tplab{\TTau}{\Compute}
} {
    \Sigma ; \Gamma ; \Delta ; A \cproves \SWitAssign{x}{v} : \tplab{\Unit}{\ell} \cproduces A \cup \{x\}
}
\DefineCTypingRule{InputDeref} {
    \Delta(x) = \TTau \\ 
    x \in A
} {
    \Sigma ; \Gamma ; \Delta ; A \cproves \SWitDeref{x} : \tplab{\TTau}{\ComputeProve} \cproduces A
}

\DefineCTypingRule{If} {
    \Sigma ; \Gamma \cproves v : \tplab{\Bool}{\ell'} \\\\
    \Sigma ; \Gamma ; \Delta ; A \cproves e_1 : \tplab{\TTau}{\ell} \cproduces A' \\
    \Sigma ; \Gamma ; \Delta ; A \cproves e_2 : \tplab{\TTau}{\ell} \cproduces A'
} {
    \Sigma ; \Gamma ; \Delta ; A \cproves \SIfThenElse* : \tplab{\TTau}{\ell} \cproduces A'
}
\DefineCTypingRule{Let} {
    \tplab{\TTau_1'}{\ell_1'} <: \tplab{\TTau_1}{\ell_1} \\
    \Sigma ; \Gamma ; \Delta ; A \cproves e_1 : \tplab{\TTau_1'}{\ell_1'} \cproduces A_1 \\\\
    \Sigma ; \Gamma, x \ty \tplab{\TTau_1}{\ell_1} ; \Delta ; A_1 \cproves e_2 : \tplab{\TTau_2}{\ell_2} \cproduces A_2
} {
    \Sigma ; \Gamma ; \Delta ; A \cproves \SLetIn{x\ty\tplab{\TTau_1}{\ell_1}}{e_1}{e_2} : \tplab{\TTau_2}{\ell_2} \cproduces A_2
}

\DefineCTypingRule{ComputeAndProve} {
    \refunreach{\overline{\TTau_{x_p}}, \overline{\TTau_{x_s}}, \overline{\TTau_{y_p}}, \overline{\TTau_{y_s}}} \\
    \cdot ; \Gamma \cproves \overline{v_p} : \tplab{\overline{\TTau_{y_p}}}{\overline{\ell_{y_p}}} \\
    \cdot ; \Gamma \cproves \overline{v_s} : \tplab{\overline{\TTau_{y_s}}}{\overline{\ell_{y_s}}} \\\\
    \Compute \subsetl \overline{\ell_{y_p}} \cap \overline{\ell_{y_s}} \\
    \Gamma_\Compute = \tplab{(\Gamma |_\Compute)}{\Compute} \\
    \Gamma_\ComputeProve = \overline{y_p} \ty \tplab{\overline{\TTau_{y_p}}}{\ComputeProve}, \overline{y_s} \ty \tplab{\overline{\TTau_{y_s}}}{\ComputeProve} \\\\
    \Delta' = \overline{x_p} \ty \overline{\TTau_{x_p}}, \overline{x_s} \ty \overline{\TTau_{x_s}} \\
    (\SigmaC, \varnothing, \varnothing) ; \Gamma_\Compute, \Gamma_\ComputeProve ; \Delta' ; \varnothing \cproves e : \tplab{\Bool}{\Prove} \cproduces \{\overline{x_p}, \overline{x_s}\}
} {
  \Sigma;\Gamma;\Delta; A \cproves \SCnp* : \tplab{\ProofOf*}{\Compute} \cproduces A
}
\DefineCTypingRule{Combined}{
    \dom(\Delta') = \overline{x_p},\overline{x_s} \\
    (\SigmaC, \SigmaP', \SigmaCP'); \varnothing; \Delta'; A' \cproves e : \tplab{\Bool}{\Prove} \cproduces \{\overline{x_p},\overline{x_s}\}
}{
    \Sigma ; \Gamma ; \Delta ; A \cproves \SCnpAdmin*{e}{\sigmaP}{\rho} : \tplab{\SProofOf*}{\Compute} \cproduces A
}

\DefineRule[EPLetOneC]{EP-Let1C}{
    \pconf*{e_{\Compute 1}}{e_\Prove} \pstep[\trace] \pconf{e_{\Compute 1}'}{e_\Prove}{\sigmaC'}{\sigmaP}{\rho'}
} {
    \pconf*{\LetIn{x\ty\TTau}{e_{\Compute 1}}{e_{\Compute 2}}}{e_\Prove} \pstep[\trace] \pconf{\LetIn{x\ty\TTau}{e_{\Compute 1}'}{e_{\Compute 2}}}{e_\Prove}{\sigma_\Compute'}{\sigma_\Prove}{\rho'}
}
\DefineRule[EPLetOneP]{EP-Let1P}{
    \pconf*{e_\Compute}{e_{\Prove 1}} \pstep \pconf{e_\Compute}{e_{\Prove 1}'}{\sigmaC}{\sigmaP'}{\rho}
} {
    \pconf*{e_\Compute}{\LetIn{x\ty\TTau}{e_{\Prove 1}}{e_{\Prove 2}}} \pstep \pconf{e_\Compute}{\LetIn{x\ty\TTau}{e_{\Prove 1}'}{e_{\Prove 2}}}{\sigmaC}{\sigmaP'}{\rho}
}
\DefineRule[EPStepC]{EP-StepC}{
    \conf{e_\Compute}{\sigmaC} \tstep[\trace] \conf{e_\Compute'}{\sigmaC'} \\
} {
    \pconf*{e_\Compute}{e_\Prove} \pstep[\trace] \pconf{e_\Compute'}{e_\Prove}{\sigmaC'}{\sigmaP}{\rho}
}
\DefineRule[EPStepP]{EP-StepP}{
    \conf{e_\Prove}{\sigmaP} \tstep[\trace] \conf{e_\Prove'}{\sigmaP'} \\
} {
    \pconf*{e_\Compute}{e_\Prove} \pstep \pconf{e_\Compute}{e_\Prove'}{\sigmaC}{\sigmaP'}{\rho}
}
\newtoggle{EPWitAssignLinebreak}
\DefineRule[EPWitAssign]{EP-WitAssign}{ 
    w \not\in \dom(\rho)
} {
    {\begin{array}{@{}l@{}}
    \pconf*{w \witassign v}{e_\Prove}
    \iftoggle{EPWitAssignLinebreak}{ \\ \hspace*{3em} }{}
    \pstep[\opset(w, v)]
    \pconf{()}{e_\Prove}{\sigmaC}{\sigmaP}{\rho[w \mapsto v]}
    \end{array}}
}
\DefineRule[EPWitDerefC]{EP-WitDerefC}{
} {
    \pconf*{\witderef w}{e_\Prove} \pstep \pconf*{\rho(w)}{e_\Prove}
}
\DefineRule[EPWitDerefP]{EP-WitDerefP}{
} {
    \pconf*{e_\Compute}{\witderef w} \pstep \pconf*{e_\Compute}{\rho(w)}
}

\DefineRule{EA-Let1}[EALetOne] {
    \cpconf*{e_1} \cpstep[\trace] \cpconf{e_1'}{\sigmaC'}{\sigmaP'}{\rho'}
} {
    \cpconf*{\LetIn{x\ty\tplab{\TTau}{\ell}}{e_1}{e_2}} \cpstep[\trace] \cpconf{\LetIn{x\ty\tplab{\TTau}{\ell}}{e_1'}{e_2}}{\sigmaC'}{\sigmaP'}{\rho'}
}

\DefineRule{EA-Let2}[EALetTwo] { } {
    \cpconf*{\SLetIn{x\ty\tplab{\TTau}{\ell}}{\SVal_\annellp}{\SExp}} \cpstep \cpconf*{\SExp[x_\annell \mapsto \SVal_\annell]}
}

\DefineRule{EA-IfT}[EAIfT] { } {
    \cpconf*{\IfThenElse{(\True)_\annellp}{e_1}{e_2}} \cpstep \cpconf*{e_1}
}

\DefineRule{EA-IfF}[EAIfF] { } {
    \cpconf*{\IfThenElse{(\False)_\annellp}{e_1}{e_2}} \cpstep \cpconf*{e_2}
}

\DefineRule{EA-Ref-C}[EARefC] {
    \iota \not\in \dom(\sigma_\Compute)
} {
    \cpconf*{\Refsub{\Compute} (v)_\annell} \cpstep[\opalloc(\iota, \Lower{\Compute}{v})] \cpconf{r_\Compute(\iota)_\annc}{\subst{\sigmaC}{\iota}{\Lower{\Compute}{v}}}{\sigmaP}{\rho}
}

\DefineRule{EA-Ref-P}[EARefP] {
    \iota \not\in \dom(\sigma_\Prove)
} {
    \cpconf*{\Refsub{\Prove} (v)_\annell} \cpstep \cpconf{r_\Prove(\iota)_\annp}{\sigmaC}{\subst{\sigmaP}{\iota}{\Lower{\Prove}{v}}}{\rho}
}

\DefineRule{EA-Ref-CP}[EARefCP] { 
    \iota_\Compute \not\in \dom(\sigma_\Compute) \\ 
    \iota_\Prove \not\in \dom(\sigma_\Prove)
} {
    \cpconf*{\Refsub{\ComputeProve} (v)_\annell} \cpstep[\opalloc(\iota_\Compute, \Lower{\Compute}{v})] \cpconf{r_\ComputeProve(\iota_\Compute, \iota_\Prove)}{\subst{\sigmaC}{\iota_\Compute}{\Lower{\Compute}{v}}}{\subst{\sigmaP}{\iota_\Prove}{\Lower{\Prove}{v}}}{\rho}
}

\DefineRule{EA-Deref-C}[EADerefC] {
} {
    \cpconf*{\bang (r_\Compute(\iota))_\annc} \cpstep \cpconf*{(\Lift{\Compute}{\sigmaC(\iota)})_\annc}
}

\DefineRule{EA-Deref-P}[EADerefP] {
} {
    \cpconf*{\bang (r_\Prove(\iota))_\annp} \cpstep \cpconf*{(\Lift{\Prove}{\sigmaP(\iota)})_\annp}
}

\DefineRule{EA-Deref-CP}[EADerefCP] {
    v = \Lift{\Compute}{\sigma_\Compute(\iota_\Compute)} \sqcup \Lift{\Prove}{\sigma_\Prove(\iota_\Prove)}
} {
    \cpconf*{\bang (r_\ComputeProve(\iota_\Compute, \iota_\Prove))_\anncp} \cpstep \cpconf*{(v)_\anncp}
}

\DefineRule{EA-Assign-C}[EAAssignC] {
    \iota \in \dom(\sigma_\Compute)
} {
    \cpconf*{(r_\Compute(\iota))_\annc \coloneq (v)_\annell} \cpstep[\opset(\iota, \Lower{\Compute}{v})] \cpconf{()_\annellp}{\subst{\sigmaC}{\iota}{\Lower{\Compute}{v}}}{\sigmaP}{\rho}
}

\DefineRule{EA-Assign-P}[EAAssignP] {
    \iota \in \dom(\sigma_\Prove)
} {
    \cpconf*{(r_\Prove(\iota))_\annp \coloneq (v)_\annell} \cpstep \cpconf{()_\annellp}{\sigmaC}{\subst{\sigmaP}{\iota}{\Lower{\Prove}{v}}}{\rho}
}

\DefineRule{EA-Assign-CP}[EAAssignCP] {
    \iota_\Compute \in \dom(\sigma_\Compute)\\
    \iota_\Prove \in \dom(\sigma_\Prove)
} {
    \cpconf*{(r_\ComputeProve(\iota_\Compute, \iota_\Prove))_\anncp \coloneq (v)_\anncp} \cpstep[\opset(\iota_\Compute, \Lower{\Compute}{v})] \cpconf{()_\annellp}{\subst{\sigmaC}{\iota_\Compute}{\Lower{\Compute}{v}}}{\subst{\sigmaP}{\iota_\Prove}{\Lower{\Prove}{v}}}{\rho}
}

\DefineRule{EA-Cast}[EACast] {
    D <: C
} {
    \cpconf*{(C)(\Newl D(\overline{v}))_\annell} \cpstep \cpconf*{(\Newl D(\overline{v}))_\annell}
}

\DefineRule{EA-Field}[EAField] { } {
    \cpconf*{(\Newl C(\overline{v}))_\annell.f_i} \cpstep \cpconf*{(v_i)_\annell}
}

\DefineRule{EA-Call}[EACall] {
    \operatorname{mbody}(m,C) = (e, \overline{x}, \TTau)
} {
    \cpconf*{(\Newl C(\overline{v}))_\annell.m_{\ell'}(\overline{u})} \cpstep \cpconf*{\subst*{\LiftAnn{\ell'}{e}}{{\overline{x}}{(\overline{u})_\annell}{\ThisN}{(\Newl C(\overline{v}))_\annell}}}
}

\DefineRule{EA-InputAssign}[EAInputAssign] {
    x \not\in \dom(\rho)
} {
    \cpconf*{x \witassign (v)_\annell} \cpstep[\opset(x,v)] \cpconf{()_\annellp}{\sigmaC}{\sigmaP}{\rho[x \mapsto v]}
}

\DefineRule{EA-InputDeref}[EAInputDeref] {
} {
    \cpconf*{\witderef x} \cpstep \cpconf*{\rho(x)_\anncp}
}

\DefineRule{EA-Prove}[EAProve] { 
    \conf{\subst*{e}{{\overline{x}}{\overline{v}}{\overline{y}}{\overline{w}}}}{\varnothing} \tstep^* \conf{\True}{\_} 
} {
    \cpconf*{\Provet \alpha = \CircuitDefault \With (\overline{v})_\annell[(\overline{w})_\annell]} \cpstep[\opgen(\alpha, \overline{v}, \overline{w})] \cpconf*{(\ProofOf* ~\UsingN~ \overline{v})_\annc}
}

\DefineRule{EA-VerifyT}[EAVerifyT] {
} {
    \cpconf*{\Verify (\ProofOf* ~\UsingN~ \overline{v})_\annell \Proves \alpha ~\UsingN~ \overline{v}} \cpstep[\opverif(\alpha, \overline{u}, \top)] \cpconf*{(\True)_\annell}
}
\DefineRule{EA-VerifyF}[EAVerifyF] {
    \ProofOf{\beta}~\UsingN~{\overline{u}} \neq \ProofOf*~\UsingN~\overline{V}
} {
    \cpconf*{\Verify (\ProofOf{\beta} ~\UsingN~ {\overline{u}})_\annell \Proves \alpha ~\UsingN~ \overline{v}} \cpstep[\opverif(\alpha, \overline{u}, \bot)] \cpconf*{(\True)_\annell}
}

\section{Introduction}
\label{sec:introduction}

Non-Interactive Zero Knowledge~(NIZK) proofs are powerful cryptographic primitives
that allow a prover to construct a \emph{non-interactive proof}---a stand-alone value that anyone can verify---% DO NOT DELETE, COMMENT IS IMPORTANT FOR SPACING
that a predicate~$P$ holds on inputs~$\overline{x}$ in \emph{zero-knowledge}---hiding any elements of~$\overline{x}$.
The last decade has seen rapid advancements in the theory and practice of NIZK proofs
with the introduction of zero-knowledge Succinct Non-interactive Arguments of Knowledge~(zkSNARKs),
the first NIZK proof systems to support practical proofs for an arbitrary polynomial-time~$P$~\citep{snarks14}.
Efficient implementations of this powerful primitive~\citep{libsnark,gnark,circomCode} have led to an explosion of use
across both academia~\citep{zkCNN21,SinghDP22,cinderella16,zkcreds23,snarkblock22,FangDNZ21,GreenHHKPV23,GrubbsAZBW22}
and industry~\citep{zcash,railgun,noauthor_tornado_2024,aleo,aztec}.

Existing tools, while immensely effective, focus on how to efficiently generate cryptographic operations from a high-level description of the predicate~$P$.
As an isolated implementation of a cryptographic primitive, this focus is entirely sensible.
In a full system, however, there is an additional critical step: computing the inputs~$\overline{x}$.

While zkSNARKs are far more efficient than other general-purpose NIZK proof systems,
generating a proof that~$P(\overline{x})$ holds remains orders of magnitude slower than computing~$P(\overline{x})$ natively.
Developers therefore go to great lengths to optimize their predicates,
even when the more efficient predicates require different inputs that incur substantial preprocessing to construct.
For example, say an application needs to prove $P(\sqrt{x})$.
Taking squares is much faster than square roots,
so it is much faster to precompute~$y$ such that $y = \sqrt{x}$ and then prove the logically-equivalent statement $x = y^2 \land P(y)$.
Unlike this toy example, input transformations in real systems often require nontrivial control flow and careful reasoning.
By separating the input computation code from the predicate definition code, existing tooling drastically complicates development and maintenance.
Developers must keep two entirely separate implementations perfectly in sync to avoid critical security and correctness flaws in their systems.

A key insight of this work is that the input computation and predicate definition regularly exhibit a \emph{structural duality}:
one builds up an input, while the other tears it down to ascertain its validity.
These operations follow nearly identical control flow and perform dual operations at each step. 
To leverage this connection, we introduce the zero-knowledge Structural Duality~(\zkstrudul) calculus for programming with zkSNARKs.
The core feature of \zkstrudul is the \emph{compute-and-prove} block,
a new language abstraction that combines updates to the system state, input construction, and input verification
into a single piece of code without needing to duplicate shared instructions or control flow.
To enable this abstraction, \zkstrudul includes lightweight annotations tracked by a type system.
The annotations specify which operations---including control flow---are relevant to which procedures: \emph{compute}, for system updates and input construction, \emph{prove}, for specifying the predicate to prove, or both.
The type system ensures ``compute'' can execute entirely without the ``prove'' code,
and ``prove'' needs only the computed inputs from ``compute.''
Interleaving the logic of building and verifying inputs allows developers to reason about these dual operations more locally.

Importantly, \zkstrudul does not replace existing tools for building and compiling zkSNARKs.
It provides a cleaner programming abstraction to be layered \emph{on top} of cryptographic compilers
like Circom~\citep{circom23,circomCode} or \texttt{gnark}~\citep{gnark}
and supports important features like recursive proofs---generating a proof that $P(\overline{x})$ holds
when~$P$ includes verifying a proof of~$P(\overline{y})$~\citep{recursiveSnark14}.
To accomplish this goal, we show how to compile compute-and-prove blocks to a lower-level representation.
The core of \zkstrudul's compilation procedure is a \emph{projection} operation
that leverages the procedure annotations mentioned above to project a full compute-and-prove block to either the ``compute'' or ``prove'' code.
The full semantics of a compute-and-prove block is to first run the ``compute'' projection, executing standard code and constructing inputs, 
then pass the ``prove'' code to a zkSNARK library, which constructs a NIZK proof that the predicate defined by the ``prove'' projection holds on the constructed input.

This projection-based semantics models what would execute in a real system, but it remains a challenge to reason about directly.
Executing the ``compute'' procedure to completion before beginning ``prove''
does not correspond to any reasonable order in the source code and involves executing some instructions twice.
To simplify reasoning, we also provide an in-order semantics for compute-and-prove blocks
and prove that it defines the same behavior as the projection-based semantics.
We formalize this connection in two ways.
First, we show the two semantics are sound and complete relative to each other; they always produce the same outputs.
Second, we show the compiler defined by the projection-based translation satisfies
Robust Relational Hyperproperty Preservation~(RrHP), the strongest property in the robust compilation hierarchy~\citep{journey-beyond19},
and strictly stronger than full abstraction.

Notably, the non-local structure of the projection-based compilation
prevents us from using the ``lock step'' correspondence traditionally underpinning translation adequacy proofs.
Instead, we define a concurrent semantics where separated ``compute'' and ``prove'' procedures can interleave operations.
We then show that this concurrent semantics can exhibit the behavior of both the in-order source semantics and projected target semantics,
and that every concurrent execution produces the same result, regardless of ordering choices.

The core contributions of this work are as follows.
\begin{itemize}[leftmargin=*]
  \item Section~\ref{sec:core-calculus} presents \corelang, a core calculus with simple abstractions designed to reason about systems using NIZK proofs to which \zkstrudul compiles.
  \item Section~\ref{sec:overview} describes a novel approach to interleaving code for ``compute'' and ``prove'' procedures so that they can be cleanly separated by a compiler,
    which we instantiate with compute-and-prove blocks in the \zkstrudul calculus in Section~\ref{sec:combined-abstraction}.
  \item Using techniques inspired by concurrency theory,
    we prove in Section~\ref{sec:adequacy} that \zkstrudul's projection-based semantics are sound and complete with respect to its in-order semantics, and that the compilation between them satisfies RrHP.
\end{itemize}
We then present related work in Section~\ref{sec:related-work} and Section~\ref{sec:conclusion} concludes.

\section{Background and Motivating Examples}
\label{sec:motivation}
Before specifying the technical details of \zkstrudul, we provide some background
on zkSNARKs and NIZK proofs in general and two larger motivating examples.

\subsection{Non-Interactive Zero Knowledge (NIZK) Proofs}
\label{sec:nizk-bg}
A cryptographic zero-knowledge proof system allows a prover to convince a verifier that a \emph{statement} is true without revealing the \emph{witness} proving why.
In general, the prover can demonstrate that any polynomial-time predicate~$P$ holds on inputs~$\overline{x}$ while hiding any or all elements of~$\overline{x}$.
The hidden (or \emph{secret}) inputs form the witness, while the open (or \emph{public}) inputs define the instance of the problem.
A \emph{non-interactive} zero-knowledge proof is one where the prover generates a proof object~$\pi$ without interacting with a verifier.
Anyone can then verify~$\pi$ given~$P$, any public inputs, and appropriate public cryptographic parameters.

The most common general-purpose NIZK proofs are zero-knowledge Succinct Non-interactive Arguments of Knowledge~(zkSNARKs).
They have garnered wide adoption due to their support for arbitrary polynomial-time predicates~$P$
and highly efficient verification---proof sizes and verification times can be constant
with concrete values under 250 bytes and 3ms, respectively~\citep{Groth16,bellman}.
Proof \emph{generation} performance, however, remains a concern,
as predicates must be compiled to a restrictive set of low-level circuit operations, generally represented as a set of constraints~\citep{pinocchio16,Buterin16}.
Predicates not tailored to these specific operations can quickly require enormous constraint sets,
pushing proof generation to use gigabytes of memory and minutes of time~\citep{ErnstbergerCK+24}.
As a result, applications rarely use the version of a predicate that is the simplest to write and understand.
Instead developers look for the most efficient logically-equivalent predicate,
even if it requires additional inputs that the application must compute locally before generating the proof.

\para{Probabilistic Proofs}
Like most cryptographic objects, zkSNARKs are probabilistic---relying on randomness to ensure both soundness and zero-knowledge---%%%
but their security and correctness can be specified in relation to an \emph{ideal functionality}, $\mathcal{F}$, with information theoretic properties~\citep{DolevYao83,Canetti-UC01}.
Implementing $\mathcal{F}$ would usually require a fully trusted third party,
but the cryptographic protocol is built to be indistinguishable from~$\mathcal{F}$ to a computationally-bounded attacker.
Instead of reasoning directly about the implementation of zkSNARKs, we abstract away the probabilistic details and reason about this idealized version~$\mathcal{F}$.
In particular, we model two proofs as being equivalent if they prove the same proposition with the same inputs, even if their bit-level representations may differ due to randomness.

\subsection{Example: Merkle Proofs}
\label{sec:merkle-example}

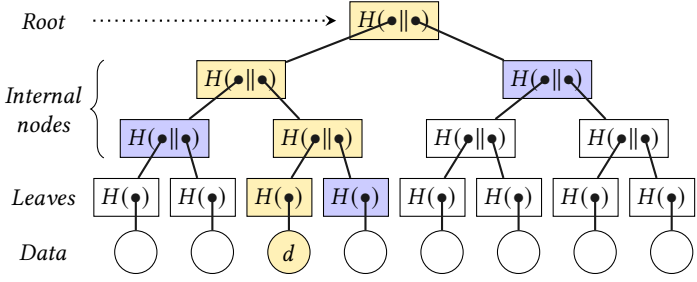
\begin{figure}
  \definecolor{bullet-col}{gray}{0.1}
  \newcommand{\bul}{{\color{bullet-col}\bullet}}
  \newlength{\lvlgap}
  \setlength{\lvlgap}{0.75em}
  \newlength{\datagap}
  \setlength{\datagap}{0.5em}
  \newcommand{\treeLabFont}{\itshape\small}
  \newlength{\treeLabWidth}
  \setlength{\treeLabWidth}{\widthof{\treeLabFont Internal\kern3pt}}
  \begin{tikzpicture}[
      every node/.style={transform shape,font=\small},
      node/.style={rectangle,draw=black,inner sep=3pt},
      path/.style={fill=Goldenrod!40},
      copath/.style={fill=blue!20},
      data/.style={circle,inner sep=0pt,minimum height=1.6em,draw=black},
      tree edge/.style={draw,bullet-col,thick},
      label/.style={align=center, font=\treeLabFont, text width=\treeLabWidth},
    ]
    \node[node] (l3-0) {$H(\bul)$};
    \node[node,right=0.4em of l3-0] (l3-1) {$H(\bul)$};
    \node[path,node,right=0.4em of l3-1] (l3-2) {$H(\bul)$};
    \node[copath,node,right=0.4em of l3-2] (l3-3) {$H(\bul)$};
    \node[node,right=0.4em of l3-3] (l3-4) {$H(\bul)$};
    \node[node,right=0.4em of l3-4] (l3-5) {$H(\bul)$};
    \node[node,right=0.4em of l3-5] (l3-6) {$H(\bul)$};
    \node[node,right=0.4em of l3-6] (l3-7) {$H(\bul)$};

    \node[data,below=\datagap of l3-0,xshift=0.36em] (d0) {};
    \node[data,below=\datagap of l3-1,xshift=0.36em] (d1) {};
    \node[data,below=\datagap of l3-2,xshift=0.36em,path] (d2) {$d$};
    \node[data,below=\datagap of l3-3,xshift=0.36em] (d3) {};
    \node[data,below=\datagap of l3-4,xshift=0.36em] (d4) {};
    \node[data,below=\datagap of l3-5,xshift=0.36em] (d5) {};
    \node[data,below=\datagap of l3-6,xshift=0.36em] (d6) {};
    \node[data,below=\datagap of l3-7,xshift=0.36em] (d7) {};

    \node[copath,node,anchor=south] (l2-0) at ($(l3-0.north)!0.5!(l3-1.north)+(0,\lvlgap)$) {$H(\bul\|\bul)$};
    \node[path,node,anchor=south] (l2-1) at ($(l3-2.north)!0.5!(l3-3.north)+(0,\lvlgap)$) {$H(\bul\|\bul)$};
    \node[node,anchor=south] (l2-2) at ($(l3-4.north)!0.5!(l3-5.north)+(0,\lvlgap)$) {$H(\bul\|\bul)$};
    \node[node,anchor=south] (l2-3) at ($(l3-6.north)!0.5!(l3-7.north)+(0,\lvlgap)$) {$H(\bul\|\bul)$};

    \node[path,node,anchor=south] (l1-0) at ($(l2-0.north)!0.5!(l2-1.north)+(0,\lvlgap)$) {$H(\bul\|\bul)$};
    \node[copath,node,anchor=south] (l1-1) at ($(l2-2.north)!0.5!(l2-3.north)+(0,\lvlgap)$) {$H(\bul\|\bul)$};

    \node[path,node,anchor=south] (root) at ($(l1-0.north)!0.5!(l1-1.north)+(0,\lvlgap)$) {$H(\bul\|\bul)$};

    \draw[tree edge] (d0) -- ($(l3-0)+(0.36em,0)$);
    \draw[tree edge] (d1) -- ($(l3-1)+(0.36em,0)$);
    \draw[tree edge] (d2) -- ($(l3-2)+(0.36em,0)$);
    \draw[tree edge] (d3) -- ($(l3-3)+(0.36em,0)$);
    \draw[tree edge] (d4) -- ($(l3-4)+(0.36em,0)$);
    \draw[tree edge] (d5) -- ($(l3-5)+(0.36em,0)$);
    \draw[tree edge] (d6) -- ($(l3-6)+(0.36em,0)$);
    \draw[tree edge] (d7) -- ($(l3-7)+(0.36em,0)$);

    \draw[tree edge] (l3-0) -- ($(l2-0)+(-0.07em,0)$);
    \draw[tree edge] (l3-1) -- ($(l2-0)+(0.8em,0)$);
    \draw[tree edge] (l3-2) -- ($(l2-1)+(-0.07em,0)$);
    \draw[tree edge] (l3-3) -- ($(l2-1)+(0.8em,0)$);
    \draw[tree edge] (l3-4) -- ($(l2-2)+(-0.07em,0)$);
    \draw[tree edge] (l3-5) -- ($(l2-2)+(0.8em,0)$);
    \draw[tree edge] (l3-6) -- ($(l2-3)+(-0.07em,0)$);
    \draw[tree edge] (l3-7) -- ($(l2-3)+(0.8em,0)$);

    \draw[tree edge] (l2-0) -- ($(l1-0)+(-0.07em,0)$);
    \draw[tree edge] (l2-1) -- ($(l1-0)+(0.8em,0)$);
    \draw[tree edge] (l2-2) -- ($(l1-1)+(-0.07em,0)$);
    \draw[tree edge] (l2-3) -- ($(l1-1)+(0.8em,0)$);

    \draw[tree edge] (l1-0) -- ($(root)+(-0.07em,0)$);
    \draw[tree edge] (l1-1) -- ($(root)+(0.8em,0)$);

    \node[label,anchor=east] (root-label) at (root-|l3-0.west) {Root};
    \draw[dotted,thick,-stealth] (root-label) -- ($(root.west)-(0.5em,0)$);
    \draw[decorate,decoration={brace,amplitude=5pt,raise=-0.4em}] (l2-0.south-|l3-0.west)
      -- node[label,left] {Internal nodes} (l1-0.north-|l3-0.west);
    \node[label,anchor=east] at (l3-0.west) {Leaves};
    \node[label,anchor=east] at (d0-|l3-0.west) {Data};
  \end{tikzpicture}
  \caption{A Merkle tree with that path to data~$d$ in yellow and the siblings highlighted in blue.}
  \label{fig:merkle-tree}
\end{figure}

To motivate our approach, the rest of the section presents two applications where we can leverage the structural duality between computing and verifying inputs.
Our first example involves building Merkle Proofs, which verify a desired property of some data within a larger dataset.
To demonstrate, we will consider a Merkle Proof for correctly updating some datum in the dataset.

The full set of data is stored as an ordered list in a commonly-used authenticated data structure called a Merkle tree~\citep{merkle89}.
In a Merkle tree, data is stored at the leaves of a binary tree,
and each node is given a value by hashing the values of its two children (or the associated data in the case of a leaf),
as shown visually in Figure~\ref{fig:merkle-tree}.
The value assigned to the root, called the ``root hash,'' is a succinct identifier to authenticate all data in the tree.

To verify a datum~$d$ is at a given location in the tree,
one can iteratively hash~$d$ with the sibling of each node along the path from~$d$ to the root,
and check that the result matches the known root hash~$h$.
The \emph{sibling path}, shown in blue in Figure~\ref{fig:merkle-tree},
serves as an efficient proof that~$d$ lies at the specified location in the tree with root hash~$h$.

One common combination of Merkle trees and zkSNARKs is to prove in zero-knowledge that data is updated one value at a time
and in accordance with some application-specific policy~\citep[e.g.,][]{zerocash14,solidus17,zexe20}.
More precisely, every update produces a new root hash~$h'$ and a proof~$\pi$
that the trees identified by~$h$ and~$h'$, respectively, differ only at a single location and that difference comports with the update policy.
While one could theoretically construct a zkSNARK to prove this fact directly,
checking every value in a large set to construct the proof would be prohibitively expensive.
The Merkle tree structure provides a critical optimization:
verifying~$d$ hashes to~$h$ and~$d'$ hashes to~$h'$ with the \emph{same sibling path}
is sufficient to show (with overwhelming probability) that no other value has changed.
To see why, notice that, when updating~$d$ to~$d'$, the hash values in the tree change only
for nodes along the path from~$d$ to the root---the yellow nodes in Figure~\ref{fig:merkle-tree}---%
leaving the blue sibling's values unchanged.
If any other value changes, however, one of the sibling node's values must change as well.
An optimized zkSNARK will thus check that (1)~$d$ and~$d'$ hash to~$h$ and~$h'$, respectively, using the same sibling list,
and (2)~changing~$d$ to~$d'$ follows the application's policies.

\begin{figure}
  \centering
  \mbox{
\begin{lstlisting}
compute pub (new_root_hash) priv (siblist[DEPTH])(*\label{lst:li:merkle:cnp-open}*)
    and_prove valid_merke_update
    with pub (ROOT_HASH, DEPTH) priv (PATH, OLD_DATA, NEW_DATA) {(*\label{lst:li:merkle:inputs}*)
  // Initialize hashes and current node to leaves.
  P hash old_hash = sha3(OLD_DATA);(*\label{lst:li:merkle:old-hash-def}*)
  CP hash new_hash = sha3(NEW_DATA);(*\label{lst:li:merkle:new-hash-def}*)
  C Node node = getLeaf(ROOT, PATH);(*\label{lst:li:merkle:node-def}*)
  node.data = NEW_DATA;(*\label{lst:li:merkle:data-update}*)
  node.data_hash = new_hash;(*\label{lst:li:merkle:hash-set-1}*)

  for(CP int i = 0; i < DEPTH; i++) {(*\label{lst:li:merkle:loop-start}*)
    // Update the node and get the next sibling.
    node = node.parent;
    C Node sibling = (PATH[i] == LEFT) ? node.right : node.left;(*\label{lst:li:merkle:sib-select}*)
    // Set the sibling hash for the computed input.
    siblist[i] <- sibling.data_hash;(*\label{lst:li:merkle:wit-assign}*)
    // Update the new and old hashes.
    old_hash = (PATH[i] == LEFT) ? sha3(old_hash, siblist[i])(*\label{lst:li:merkle:wit-use-start}*)
                                 : sha3(siblist[i], old_hash);
    new_hash = (PATH[i] == LEFT) ? sha3(new_hash, siblist[i])
                                 : sha3(siblist[i], new_hash);(*\label{lst:li:merkle:wit-use-end}*)
    // Update the computed hash.
    node.data_hash = new_hash;(*\label{lst:li:merkle:hash-set-2}*)
  } (*\label{lst:li:merkle:loop-end}*)
  // Assign new public identifier after compute.
  new_root_hash <- new_hash;(*\label{lst:li:merkle:new-root-hash}*)
  // Verify old and new root hashes and update validity.
  assert old_hash == ROOT_HASH && new_hash == new_root_hash(*\label{lst:li:merkle:assert-1}*)
      && validUpdate(OLD_DATA, NEW_DATA);(*\label{lst:li:merkle:assert-2}*)
}
\end{lstlisting}}
  \caption{Combined code to compute a Merkle sibling path and verify that two root hashes correspond to the same tree except for a change at the leaf at \lstinline{PATH}.}
  \label{fig:merkle-code}
\end{figure}

A critical input to this proof is the sibling path.
The system must construct it before building the proof, and the proof specification must tear it down---dual operations.
Since this proof is associated with a data update, the system must also update its state to the new value.
Compute-and-prove blocks allow all three pieces of code to be implemented together, as shown in Figure~\ref{fig:merkle-code}.
Variables written in all caps are assumed to be in scope outside of the compute-and-prove block:
\lstinline{DEPTH} is the depth of the Merkle tree,
\lstinline{ROOT_HASH} is the previous root hash,
\lstinline{PATH} specifies the location of the data in the tree and is ordered leaf-to-root,
and \lstinline{OLD_DATA} and \lstinline{NEW_DATA} are the old and new data values, respectively.

Lines~\ref{lst:li:merkle:cnp-open}--\ref{lst:li:merkle:inputs} open a compute-and-prove block
that computes two inputs: the new root hash, which is public, and the sibling list, which is secret.
It also includes the old root hash, tree depth, data path, and old and new data values as already in-scope inputs.
These values are available to all parts of the computation,
while in-scope variables not specified are unavailable in the proof specification.

To update the system state, the code fetches the leaf containing the old data (line~\ref{lst:li:merkle:node-def}),
updates it (line~\ref{lst:li:merkle:data-update}),
and updates the hash value of each node from that leaf to the root (lines~\ref{lst:li:merkle:hash-set-1} and~\ref{lst:li:merkle:hash-set-2}).
Notice the \Compute annotation on line~\ref{lst:li:merkle:node-def}, which indicates that the \lstinline{node} variable,
and subsequently all other operations that use it, occur only in the \emph{compute} code and are not directly part of the proof.

To compute the sibling list, the code loops over the provided input path (line~\ref{lst:li:merkle:loop-start})
moving one layer up the tree from the leaf to the root each loop.
\lstinline{PATH} indicates if the node was a right or left child, so line~\ref{lst:li:merkle:sib-select} selects the correct sibling
and line~\ref{lst:li:merkle:wit-assign} puts its hash value into the computed sibling list.
The notation ``\lstinline{<-}'' on line~\ref{lst:li:merkle:wit-assign} is a special assignment operator used for computed inputs.
It is the only way to assign to computed inputs listed on line~\ref{lst:li:merkle:cnp-open} and each input must be assigned to exactly once and before it is used.
As we will see in Section~\ref{sec:combined-abstraction},
these restrictions are critical for ensuring that in-order execution behaves the same as the projected code
that computes updates and inputs to completion before beginning the proof.

Lastly, the code specifies how to use the computed sibling list to recompute the root hashes and verify they match the expected values.
Notably, the old root hash is already available, while the new one must be computed as a fresh input.
Therefore, the derivation of the old hash is only needed for the proof itself, as indicated by the \Prove annotation on line~\ref{lst:li:merkle:old-hash-def},
while both the compute and prove code must construct the new hash, indicated by the \ComputeProve annotation on line~\ref{lst:li:merkle:new-hash-def}.
Both are initialized to the hash of the respective values (lines~\ref{lst:li:merkle:old-hash-def} and~\ref{lst:li:merkle:new-hash-def}),
and then updated each iteration of the loop based on their previous value and the computed sibling hash (lines~\ref{lst:li:merkle:wit-use-start}--\ref{lst:li:merkle:wit-use-end}).
To verify the hashes match, line~\ref{lst:li:merkle:new-root-hash} sets the new root hash on the compute side,
and line~\ref{lst:li:merkle:assert-1} asserts that the hash matches on the prove side.
Finally, line~\ref{lst:li:merkle:assert-2} asserts the update follows the application-specific update policies.

Note that both the compute and prove code need the exact same \lstinline{for} loop.
By marking the loop variable with \ComputeProve, we indicate that the loop will execute in both computations, allowing us to only write it once.
The compiled code will have two copies of the loop---and any other \ComputeProve computations---one in each generated block.
Importantly, notice how input computation (line~\ref{lst:li:merkle:wit-assign}) and input usage in proof generation (lines~\ref{lst:li:merkle:wit-use-start}-\ref{lst:li:merkle:wit-use-end}) are now adjacent.
Writing this compute-and-prove block thus enables local reasoning between the two computations.

\subsection{Example: Zero Knowledge Virtual Machines}
\label{sec:zkvm-example}
Our second motivating example is building a zero-knowledge Virtual Machine~(zkVM), which is an environment that runs a program and produces a NIZK proof that the execution was correct.
Users can prove the correctness of their computation without revealing its inputs, internal state, or control flow.
This powerful tool, which allows for general zero-knowledge computation that is impractical for monolithic circuits,
is part of a broader trend to develop high-level languages and interpreters for generating proofs with zkSNARKs~\citep[e.g.,][]{lurk23,circ22}.

Typically, a zkVM simulates a low-level instruction set and proves that each instruction adheres to its defined semantics.
Proving a single step at a time is straightforward.
As with a conventional interpreter, it updates its internal state (program counter, registers, etc.),
and then produces a zkSNARK that the new state is the correct result of applying the last instruction to the prior state.

\begin{figure}
  \centering
  \mbox{
\begin{lstlisting}[mathescape=true]
Configuration conf = new Configuration(C_INIT);(*\label{lst:li:zkvm:conf-init}*)
proofOf[valid_step] pf = ... // Invalid dummy proof for initial step
Configuration C_prev = conf.immutableCopy();
for (int steps = 1; steps < MAX_STEPS && !terminated(conf, CODE); steps++) {(*\label{lst:li:zkvm:outer-loop}*)
  Configuration C_in = conf.immutableCopy();
  pf = compute pub (C_out) and_prove valid_step (*\label{lst:li:zkvm:cnp-open}*)
          with pub (C_in, C_INIT, CODE, steps) priv (C_prev, pf) {(*\label{lst:li:zkvm:cnp-inputs}*)
    assert steps < MAX_STEPS;
    // Recursively check previous steps were valid.
    assert (C_in == C_INIT && steps == 1)(*\label{lst:li:zkvm:recurs-init}*)
      || verify pf proves valid_step using C_in, C_prev, C_INIT, CODE, steps-1;(*\label{lst:li:zkvm:recurs-step}*)

    switch (CODE[C_in.pc].opcode) { (*\label{lst:li:zkvm:switch-open}*)
      case ADD: (*\label{lst:li:zkvm:switch-add}*)
        // Update current state to add registers.
        CP Arg[ARG_LEN] args = CODE[C_in.pc].args; (*\label{lst:li:zkvm:conf-args}*)
        conf.regs[args[0]] = conf.regs[args[1]] + conf.regs[args[2]]; (*\label{lst:li:zkvm:regs-assign}*)
        conf.pc++; (*\label{lst:li:zkvm:pc-update}*)
        // Copy updated state as computed input.
        C_out <- conf.immutableCopy(); (*\label{lst:li:zkvm:input-assign}*)

        // Verify operation was executed correctly.
        for (P int i = 0; i < REG_LEN; i++) { (*\label{lst:li:zkvm:assert-loop-open}*)
          assert (C_out.regs[i] == C_in.regs[i] && i != args[0])
            || (C_out.reg[i] == C_out.regs[args[1]] + C_out.regs[args[2]]
                && i == args[0]);
        }(*\label{lst:li:zkvm:assert-loop-close}*)
        break;
      case ...
    }
  }
  C_prev = C_in;
}(*\iffalse % Taking out the rest of this for now

prove exec ∃init,code,pf,[conf,steps].terminated(conf, code)
    && (verify pf proves is_init using init, conf, steps
        || verify pf proves valid_step using init, conf, code, steps);
\fi*)
\end{lstlisting}}
  \caption{Verifying correct interpretation by recursively verifying previous steps and then one more.}
  \label{fig:zkvm-code}
\end{figure}

To handle a bounded number of steps, one could simply unroll the loop,
but zkSNARK generation cost scales with the \emph{maximum} number of steps, regardless of when the program terminates.
With any bound large enough for interesting programs, this approach becomes prohibitively expensive.
Instead, the literature suggests using \emph{recursive proofs}~\citep{recursiveSnark14} to allow incremental proof generation.
Each step inductively \emph{verifies} a proof that all previous steps were correct before proving the next step is correct.
In our example, step~$n$ produces a proof~$\pi_n$ that a program~$e$ run for~$n$ steps from configuration~$C_0$ produces configuration~$C_n$.
A simple structure for this proof would take~$e$, $n$, $C_0$, and~$C_n$ as public inputs, and~$C_{n-1}$ and $\pi_{n-1}$ as hidden inputs,
and \emph{verify}~$\pi_{n-1}$ with inputs~$e$, $n-1$, $C_0$, and~$C_{n-1}$ and then prove that a single step takes~$C_{n-1}$ to~$C_n$.

This structure, while simple, is highly inefficient.
Which step occurs determines how to transform the configuration, but that step depends on the \emph{input} configuration~$C_{n-1}$, which is a hidden input.
To assure zero-knowledge, if control flow depends on hidden inputs,
proof generation must consider every branch that could be taken in its circuit~\citep{Groth16},
incuring the computational cost of executing every one.
To avoid this overhead, the input configuration~$C_{n-1}$ must be a \emph{public} input.
This change means that verifying~$\pi_{n-1}$ requires~$C_{n-2}$, but that can safely be provided as a \emph{secret} input when constructing~$\pi_n$.

This inductive proof structure follows precisely the logic of a standard multi-step semantics and correctly verifies the whole execution.
A system wishing to hide the number of steps or (part of) the final configuration~$C_\mathit{fin}$
could simply produce a zkSNARK with \emph{hidden} inputs~$n$, $\pi_n$, the second-to-last configuration~$C_{\mathit{fin} - 1}$, and (part of)~$C_\mathit{fin}$
that checks that~$C_\mathit{fin}$ is a terminated configuration and verifies~$\pi_n$ with inputs~$e$, $n$, $C_0$, $C_{\mathit{fin} - 1}$, and $C_\mathit{fin}$.

Currently, developers of zkVMs must specify execution semantics (including state update) and proof generation separately.
As with sibling paths in the Merkle tree example above, the configuration changes are dual operations;
they are built when updating the state and examined during proof generation.
Compute-and-prove blocks allow these to be written together,
with each operation's verification logic directly adjacent to the system update code.
We show a partial implementation of precisely the algorithm above in Figure~\ref{fig:zkvm-code}.
For simplicity, we assume global constants \lstinline{MAX_STEPS}, \lstinline{ARG_LEN}, and \lstinline{REG_LEN},
specifying the maximum number of execution steps, argument length, and register length, respectively.

Note a few interesting structures.
First, the compute-and-prove block is computing the configuration~\lstinline{C_out},
as the configuration resulting from a step is precisely what is both computed and proven.
Second, line~\ref{lst:li:zkvm:recurs-step} does the recursive proof checking.
The base case (line~\ref{lst:li:zkvm:recurs-init}) is when this is the first step and $\text{\lstinline{C_in}} = \text{\lstinline{C_INIT}}$.
At each subsequent step, there must be a valid recursive proof of the previous steps.
The performance-based choice to mark~\lstinline{C_in} as public impacts this line, but otherwise adds no complication to the code.

Finally, the compute-and-prove structure allows the code to execute an instruction and the code to verify that execution to be placed immediately adjacent.
Inside the \lstinline{ADD} branch of the switch-case, line~\ref{lst:li:zkvm:input-assign} assigns the newly-computed configuration to the computed input variable~\lstinline{C_out}.
Immediately afterward is the code to verify this update, specific to addition.
Notice the \Prove label in the loop on line~\ref{lst:li:zkvm:assert-loop-open} indicates the loop only corresponds to proof generation.

As in the Merkle tree example, a compiler can take this combined code and extract two procedures, one to execute updates and construct inputs, and the second to generate a zkSNARK.
\section{Designing Compute-and-Prove}
\label{sec:overview}

The primary goal of our compute-and-prove abstraction is to allow developers to interleave the instructions for two dual procedures:
a \emph{compute} procedure to update state and compute proof inputs, and a \emph{prove} procedure defining the predicate to prove.
A compiler then needs to extract each procedure from this combined code, a process we call \emph{projection}.
Supporting this process while ensuring the projected programs do not attempt to, for example, use data they do not have access to poses several key challenges.

First, to support not only the interleaving of code, but the use of the same code across both processes, we require that both use the same syntax.
However, many NIZK proof systems restrict the language features available inside predicates, and the prove procedure must define a valid predicate.
To allow the procedures to support different sets of language features while using a common syntax, we turn to a type system.
We augment typing judgments with a label~$\ell$ indicating where a given expression can execute:
\Compute if the expression is valid in general-purpose (compute) code,
\Prove if the expression is valid inside a proof predicate (prove code),
or \ComputeProve if it is valid in both.
An expression mixing features available only in general-purpose code with features available only in proof predicates
will have no valid label and will be considered ill-typed.
This approach is highly modular with respect to different NIZK backends;
when compiling against a more restrictive backend, more operations will use a~\Compute label, rather than~\ComputeProve.

Second, to project a compute-and-prove block, a compiler must be able to determine which code corresponds to which procedure.
In the absence of a syntactic distinction, we again turn to the same procedure labels.
Inside a compute-and-prove block, however, placing labels on the judgments is no longer effective, as code for both procedures is interleaved by design.
Instead, we push the procedure labels from the judgment level into the types themselves.
Outside compute-and-prove blocks, the labels represent where an operation can execute.
Inside compute-and-prove blocks, they indicate where a value must appear, and thus where code should run.
As we will see in Section~\ref{sec:combined-types}, this structure also allows
us to build the type system inside compute-and-prove blocks directly from the type system outside, incorporating its permissions and restrictions automatically.
If an expression type-checks with type~$\tau$ and label~$\ell$ outside, it type-checks with type~$\tau^\ell$ inside.

These labels also make it easy to avoid either procedure attempting to use data only available in the other.
The type system can easily verify that a computation producing a value of type~$\tau^\Compute$ cannot use data with label~$\Prove$.
Indeed, the only interaction between the procedures must come through explicitly declared \emph{compute inputs}---precisely the inputs compute-and-prove blocks are designed to support.
Computed inputs begin uninitialized, are constructed by the compute procedure, and are examined by the prove procedure.
The type system verifies that each computed input is initialized exactly once by the compute procedure.

Finally, while general-purpose code may depend on or modify contextual information---%%%
like memory state, which node in a system runs the code, or the time and date---proof predicates cannot.
To prove a meaningful statement anyone can verify, they must produce the same result in every context for a given input.
We therefore create a notion of \emph{context-aware} types whose behavior might depend on the context.
These types could include references into a heap or data structures containing context-aware types,
while non--context-aware types might include primitive integers and booleans or data structures containing only non--context-aware types.
Inputs to proofs---computed by compute-and-prove or provided directly---are then restricted to not be context-aware, ensuring that the predicate will mean the same thing in every context.

The calculus presented in Sections~\ref{sec:core-calculus} and~\ref{sec:combined-abstraction} shows how these principles enable
compute-and-prove blocks with semantics matching standard in-order execution
and language features sufficient to express the examples detailed above in Section~\ref{sec:motivation}, including objects and mutable references.
\section{\corelang: A Core Calculus for NIZK Systems}
\label{sec:core-calculus}
We begin by introducing the Featherweight Java NIZK~(\corelang) calculus, an object-oriented core calculus for describing systems using NIZK proofs.
We choose an object-oriented structure for two reasons.
First, it can model zero-knowledge systems by representing each party as an object with methods providing interfaces for parties to interact.
Second, it demonstrates how our approach can handle complex language features used by zero-knowledge systems including objects with inheritance and mutable state~\citep{cairo,zokrates18,o1js}.

\begin{figure}[b]
  \rulefiguresize
  \[
    \begin{array}{rcl}
      f,m,x,\alpha & \in & \mathcal{V} \quad \text{(field, method, variable, proof names)}
      \\[2pt]
      \TTau & \Coloneqq & \TInt \alt \TBool \alt \TUnit \alt \TRefTau* \alt C \alt \TProofOf* \alt \TProof
      \\[2pt]
      \mathit{CL} & \Coloneqq & \TClassListSyntax
      \\[2pt]
      K & \Coloneqq & \TConstructor
      \\[2pt]
      M & \Coloneqq & \TTau\{\ell\}~ m(\overline{x} \ty \overline{\TTau}) \{ e \}
      \\[2pt]
      \TVal & \Coloneqq & \TVar* \alt \TUnitVal \alt \TTrue \alt \TFalse \alt \TNew* \alt \iota \alt \TProofOfUsing*
      \\[2pt]
      \TExp & \Coloneqq & v \alt \TRefVal* \alt \TDeref* \alt \TAssign{\TVal_1}{\TVal_2} \alt \TCast* \alt \TField* \alt \TCall* \\
      & \alt & \TLetIn* \alt \TIfThenElse* \\
      & \alt & \TProveUsing* \alt \TVerify*
    \end{array}
  \]
    \caption{\corelang Syntax, where $\ell \in \{\Compute, \Prove, \ComputeProve\}$ is a procedure label.}
    \label{fig:core-syntax}
\end{figure}

The syntax of \corelang, shown in Figure~\ref{fig:core-syntax}, is based on Featherweight Java~(FJ)~\citep{featherweightJava01}
with standard mutable references~\citep[Chapter 13]{tapl} and two primitives abstracting the essential operations of NIZK proofs: \TProveN and \TVerifyN.
To simplify analysis, subexpressions are (open) values except for explicit sequencing through \TLetN expressions and the branches of \TIfN statements.
We use \TrgProg{orange~teletype} font for \corelang syntax to visually differentiate from the combined code of compute-and-prove blocks presented below.

Recall a NIZK proof demonstrates that~$P(\overline{x})$ holds for some predicate~$P$ while hiding any combination of values in~$\overline{x}$.
The \TProveN primitive constructs such a proof of a predicate defined by expression~$e$.
The predicate uses the same syntax as the rest of the language, which will allow us to interleave its definition with general-purpose code
in the compute-and-prove blocks defined in Section~\ref{sec:combined-abstraction}.
The inputs to the predicate are bound as formal arguments~$\overline{x}$ and~$\overline{y}$, representing public and secret inputs, respectively,
and the predicate is evaluated on inputs~$\overline{v_x}$ and~$\overline{v_y}$.
The identifier~$\alpha$ provides a name to this predicate that must be globally unique.
In addition to being convenient, the predicate body can reference this name, supporting recursive proofs like those used in the zkVM example above (Section~\ref{sec:zkvm-example}).
To verify a NIZK proof, we include the term \TVerifyN,
which takes a proof~$v$, a predicate identifier~$\alpha$, and a set of \emph{public} inputs~$\overline{v}$.

\subsection{Semantics}
\begin{figure}
    \begin{ruleset}
        \EProveRule
        \and
        \EVerifyTRule
        \and
        \EVerifyFRule
    \end{ruleset}
    \caption{Semantic Rules for NIZK Proof Terms}
    \label{fig:selected-core-semantics}
\end{figure}

We equip \corelang with a small-step operational semantics with traces.
The rules for all FJ constructs and mutable references are standard,
with creation and assignment to references creating trace events
that we will use to demonstrate that our compute-and-prove blocks behave as expected.
These rules can be found in \Cref{app:core-semantics}.
Figure~\ref{fig:selected-core-semantics} presents the semantic rules for our new \TProveN and \TVerifyN operations,
each of which also produce appropriate trace events.

To ensure \TProveN models the behavior of cryptographic proof generation,
\ruleref{E-Prove} evaluates the predicate defined by~$e$ on the supplied inputs and produces a proof, represented $\TProofOfUsing*$, if the predicate terminates with \TTrue.
This value records the information needed to properly verify the proof: the predicate identifier~$\alpha$ and public inputs~$\overline{v}$.
There are two interesting properties to note about this rule.
First, evaluation of the predicate~$e$ begins with an empty heap.
This models the need for proof generation to be context independent;
a verifier will not have access to the prover's heap, so the proof cannot safely use references.
Second, if~$e$ fails to evaluate to~\TTrue---either because it diverges or produces a different value---the expression will get stuck.
This structure correctly models the underlying cryptography which can only generate a proof when the predicate holds.

While \TProveN can only construct valid proofs, a verifier may attempt to check \emph{any} purported proof.
We therefore allow \TVerifyN to step with any proof object,
returning \TTrue if the provided proof proves the expected statement---predicate~$\alpha$ with public inputs~$\overline{v}$---(\ruleref{E-VerifyT})
and \TFalse otherwise (\ruleref{E-VerifyF}).

As noted in Section~\ref{sec:nizk-bg}, both rules treat proofs entirely symbolically, abstracting away the probabilistic nature of the underlying cryptography.

\subsection{Type System}
Recall from Section~\ref{sec:overview} that the unified syntax for general-purpose (\emph{compute}) code and predicate (\emph{prove}) code
causes us to push two important checks into the type system.
First, the typing judgment tracks where any given expression can occur: compute, prove, or both.
Second, we must prevent proofs from relying on inputs with context-aware types.

The type system for \corelang has two typing judgments, one for values and one for expressions.
The value judgment takes the form $\Sigma; \Gamma \proves v : \TTau$, where $\Sigma$ is a heap context mapping references to types, and $\Gamma$ is a typing context mapping variables to types.
The rules are entirely standard, with an additional rule giving $\TProofOfUsing*$ the type~$\ProofOf*$.

For the expression judgment, recall from Section~\ref{sec:overview} that the unified syntax
for general-purpose (\emph{compute}) code and predicate (\emph{prove}) code causes us to push two important checks into the type system.
First, the typing judgment tracks the procedure(s) in which an expression can occur: compute (\Compute), prove (\Prove), or both (\ComputeProve).
The expression judgment thus takes the form $\Sigma; \Gamma \proves e : \TTau \produces \ell$ where $\ell \in \{\Compute, \Prove, \ComputeProve\}$ is a procedure label.
We consider these labels sets of procedures with $\ComputeProve = \Compute \cup \Prove$.

Second, we must prevent proofs from relying on inputs with context-aware types.
In \corelang, the only context-aware values are memory locations, so we define the predicate $\refunreach{\TTau}$
that holds precisely when no value of type~$\TTau$ can store a pointer into the heap, directly or indirectly.
Since \corelang includes objects, that includes reference types, classes with fields of reference types,
and classes with fields of a class that might itself store a pointer into the heap.

Notably, inheritance complicates this check substantially.
Just because a class~$C$ does not declare any context-aware fields does not mean it cannot be extended by a class~$D$ that does.
In this circumstance, it is impossible to know if a value of type~$C$ is context-aware, as its runtime type could be~$D$.
We address this concern by tagging each class declaration with whether or not it is context-aware.
A class with context-aware fields may only extend another context-aware class.
Note that, given a complete full class table, these tags are simple to infer with a simple subtyping graph and reachability analysis.

\begin{figure}
    \begin{ruleset}
      \SetRuleLabelLoc{lab}
        \TProveRule
        \hspace*{0.25em}
        \TVerifyRule
        \and
        \MethodOk
        \and
        \TCallRule
      \end{ruleset}
    \caption{Selected Typing Rules for \corelang, where notable changes to existing rules from Featherweight Java with mutable references~\citep{featherweightJava01, tapl} are colored in \changes{red}.}
    \label{fig:selected-core-types}
\end{figure}

A selection of typing rules in Figure~\ref{fig:selected-core-types} illustrate how the type system combines these features
(the full type system is available in \Cref{app:core-types}).
In \ruleref{T-Prove}, since~$e$ defines the predicate to prove, it must be well-typed to return a boolean when given the specified inputs,
using only operations available to predicates---the label~$\ell$ is~$\Prove$.
It will be passed inputs~$\overline{v_x}$ and~$\overline{v_y}$, so those must have appropriate types.
Also, none of the input types can be context-aware.
Finally, we disallow \TProveN inside proof predicates, as computing a proof inside a proof would be entirely pointless, so the label on the \TProveN computation is~\Compute.

\ruleref{T-Verify} is simpler, ensuring it is given an actual proof---note that $\ProofOf*$ is a subtype of~$\Proof$ for all~$\alpha$---and that the public inputs $\overline{v}$ have the correct types.
As both general-purpose code and predicates can verify proofs, the procedure label~$\ell$ is unconstrained.
Note that it uses the auxiliary lookup function $\ptypes(\alpha)$ to identify the types of the public inputs, just as FJ uses similar lookup functions to identify the types of methods and fields of classes.

The procedure separation tracked by labels~$\ell$ also extends to methods, as the body of any method must be able to execute in whichever procedure calls it.
To enforce this restriction, each method definition carries a syntactic label~$\ell$ and its method body~$e$ is type checked against that label, as formalized by \ruleref{Method-Ok}.
The condition $\ell \subseteq \ell_m \cap \ell_C$ ensures that~$\ell$ is restricted by both the operations in its body, and the context-awareness tag on the class itself,
ensuring that context-aware fields do not leak into proof predicates.
\ruleref{T-Call} then properly restricts the caller of this method to the labeled procedure.

This type system is sound with respect to the operational semantics,
with precise exceptions for invalid casts (inherited from FJ), or \TProveN terms that do not evaluate to true.
These cases are included in the predicate $\stuckpred{e'}{\sigma}$, defined in \Cref{app:config-safety-def}.
\begin{restatable}[Type Soundness of \corelang]{theorem}{coresoundness}
  If $\cdot ; \cdot \proves e : \TTau \produces \Compute$ and $\conf{e}{\varnothing} \tstep^* \conf*{e'}$,
  then either (1)~$e'$ is a value, (2)~$\conf*{e'} \tstep \conf{e''}{\sigma'}$ can step, or (3)~$\stuckpred{e'}{\sigma}$.
\end{restatable}
The proof of this theorem follows from a standard progress and preservation argument.
\section{The \zkstrudul Language}
\label{sec:combined-abstraction}
In this section, we instantiate the principles from~\Cref{sec:overview} in \zkstrudul, which extends \corelang with the compute-and-prove block abstraction.

\subsection{Syntax}
\label{sec:combined-syntax}
A compute-and-prove block interleaves a general purpose ``compute'' procedure and a predicate ``prove'' procedure within a single programming construct that generates a resulting proof.
Formally, the syntax of the abstraction performs three functions: \rone declare the set of computed inputs, \rtwo define the inputs to the proof, and \rthree specify the interleaved procedures in a unified syntax.
To capture these requirements, we offer the following syntax, which we write in $\SrcProg{blue~sans~serif}$.
\begin{align*}
    \TrgProg{\TExp} &\Coloneqq \cdots \divi \SCnp*[\SrcProg{\SExp}]
\end{align*}
For \rone, each compute-and-prove block begins by declaring a set of computed inputs, written as $\overline{x_p}, \overline{x_s}$, along with their types.
For \rtwo, the block associates a unique proof name $\alpha$ and specifies bindings for the (non-computed) proof inputs $\overline{y_p}, \overline{y_s}$ and their values $\overline{v_p}, \overline{v_s}$.
We separate these to specify the input names to prove separately from the concrete values where substitution applies.
Lastly for \rthree, the expression $\SrcProg{\SExp}$ denotes the body of the compute-and-prove block.

Because the general-purpose and predicate languages in \corelang share a common syntax, we can express code belonging to either (or both) procedure(s) in a single syntactic form.
This unified syntax extends that of \corelang by adding operations for computed inputs.
Since computed inputs are handled separately from other references,
we include assignment $\SWitAssign*$ and dereferencing $\SWitDeref*$ operators
exclusive to computed inputs.

Certain terms require attaching a procedure label $\ell$ to the syntax.
Let bindings include the label in the type of the bound variable, e.g., $\SLetIn{x\ty\tplab{\Bool}{\Prove}}{\True}{\SExp}$ binds a boolean variable scoped exclusively to the prove procedure.
Notably, a variable annotated with $\ComputeProve$ represents a pair of variables, one in each procedure, that will be bound to matching values.
We also add labels to the syntax of reference allocations to determine where the reference should reside.
A reference with label $\ComputeProve$ models a pair of locations that we require remain in sync, so they can operate as a single value in the intuitive in-order semantics presented in~\Cref{sec:combined-semantics}.
We also add a label to the syntax of method calls, $\SCall*$, for a subtle technical reason also explained in~\Cref{sec:combined-semantics}.

\subsection{Type System}
\label{sec:combined-types}
Recall from \Cref{sec:overview} that the type system for our interleaved code pushes procedure labels into types,
allowing the type system to separate the data in each procedure.
To perform this reasoning, the typing judgment for values takes the form $(\SigmaC, \SigmaP, \SigmaCP) ; \Gamma \cproves \SVal : \tplab{\TTau}{\ell}$.
The variable context $\Gamma$ maps variables to labeled types, and the heap contexts $(\SigmaC, \SigmaP, \SigmaCP)$, abbreviated collectively as $\Sigma$, map raw locations to unlabeled types.
The first two heap contexts track locations exclusive to the compute and prove procedures.
Because \ComputeProve references represent pairs of memory locations, $\SigmaCP$ maps pairs $(\iota_\Compute, \iota_\Prove)$ to a type $\TTau$ and tracks when locations must be synchronized.

The typing judgment $\Sigma ; \Gamma ; \Delta ; A_\text{in} \cproves \SExp : \tplab{\TTau}{\ell} \cproduces A_\text{out}$ for expressions extends value judgments to track the behavior of computed inputs.
Specifically, the context~$\Delta$ ascribes computed input names to their types, and the two parameters $A_\text{in}$ and $A_\text{out}$ track the set of computed inputs which have been assigned to before and after executing $e$, respectively.

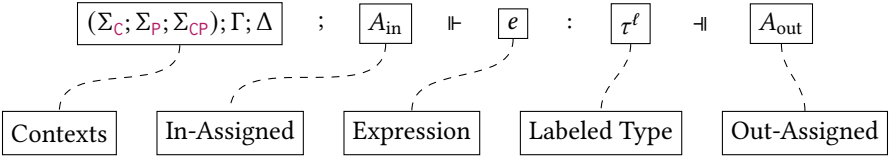
\begin{center}
  \begin{tikzpicture}[
      boxed/.style={rectangle,draw=black},
    ]
    \matrix[column sep=1em] (rule) {
      \node[boxed] (ctx) {$(\SigmaC ; \SigmaP ; \SigmaCP) ; \Gamma ; \Delta$}; &
      \node (semi) {$;$}; &
      \node[boxed] (Ain) {$A_\text{in}$}; &
      \node (proves) {$\cproves$}; &
      \node[boxed] (e) {$\SExp$}; &
      \node (ty) {$\ty$}; &
     \node[boxed] (tau) {$\tplab{\TTau}{\ell}$}; &
      \node (produces) {$\cproduces$}; &
      \node[boxed] (Aout) {$A_\text{out}$}; \\
    };

    \matrix[below=1.75em of rule,column sep=1.5em,minimum height=1.35\baselineskip]{
      \node[boxed] (ctxText) {Contexts}; &
      \node[boxed] (AinText) {In-Assigned}; &
      \node[boxed] (eText) {Expression}; &
     \node[boxed] (tauText) {Labeled Type}; &
      \node[boxed] (AoutText) {Out-Assigned}; \\
    };

    \draw[dashed] (ctx) to[out=-90,in=90] (ctxText);
    \draw[dashed] (Ain) to[out=-90,in=90] (AinText);
    \draw[dashed] (e) to[out=-90,in=90] (eText);
    \draw[dashed] (tau) to[out=-90,in=90] (tauText);
    \draw[dashed] (Aout) to[out=-90,in=90] (AoutText);
  \end{tikzpicture}
\end{center}

Since types now have labels, these labels inform subtyping.
Specifically, the following two rules lift the subtyping of \corelang types to labeled types.
\begin{ruleset}
    \infer*{
        \TTau_1 <: \TTau_2
    } {
        \tplab{\TTau_1}{\ell} <: \tplab{\TTau_2}{\ell}
    }
    \and
    \infer*{
        \TTau_1 <: \TTau_2 \\
        \ell_2 \subsetl \ell_1 \\
        \refunreach{\TTau_1}
    } {
        \tplab{\TTau_1}{\ell_1} <: \tplab{\TTau_2}{\ell_2}
    }
\end{ruleset} 
One might expect labels to be covariant for all types: if a value exists at \ComputeProve then it is available to \Compute and~\Prove.
However, we should statically reject the expression~$\SAssign{x}{v}$ if~$x$ is a~\ComputeProve reference and~$x$ is \Compute-only,
since the one-sided assignment would desynchronize a shared reference.
We therefore require procedure labels to be \emph{invariant} for context-aware types, but allow covariance for types satisfying \refunreachN.

A selection of typing rules are presented in \Cref{fig:selected-combined-types}, with the complete type system found in \Cref{app:combined-types}.
This type system is closely tied to that of \corelang, linking the labels on types to the labels in the typing judgment of \corelang.
Specifically, we lift typing derivations from \corelang to \zkstrudul for expressions which belong to a single procedure using the following rule:
\begin{mathpar}[\rulefiguresize]
  \TCLiftRule
\end{mathpar}
Formally, we capture this set of expressions labeled with a single procedure using the lifting operator $\Lift*{e}$ which attaches the label~$\ell$ throughout all subterms.
Two rules for the lifting operator are presented below, and the rest are in \Cref{app:expression-lifting-def}.

\begin{mathpar}
\Lift*{\TRefVal*} = \SRefVal{\ell}{\Lift*{v}}
\and
\Lift*{\TLetIn{x\ty\TTau}{e_1}{e_2}} = \SLetIn{x\ty\tplab{\TTau}{\ell}}{\Lift*{e_1}}{\Lift*{e_2}}
\end{mathpar}

Another aspect of lifting typing derivations in this manner is that we must enforce separation of data: values labeled exclusive to \Compute or \Prove remain isolated, and so the results of operations exclusive to \Compute or \Prove must also remain isolated.
This separation is accomplished by defining a context restriction function $\Gamma |_\ell$ that returns a subset of the variables available to the relevant procedure:
\[
  (\Gamma, x \ty \tplab{\TTau}{\ell'})|_\ell = \begin{cases}
        \Gamma |_\ell , x \ty \TTau & \tplab{\TTau}{\ell'} <: \tplab{\TTau}{\ell} \\
        \Gamma |_\ell & \ow
    \end{cases}
\]
Rather than including all variables in the namespace, this definition includes only those variables whose types can safely be used by $\ell$ without violating synchronization of \ComputeProve references.
This corresponds exactly to the reason why we disallow covariant subtyping for context-aware types, so we directly leverage subtyping to enforce this restriction.

The reuse of \corelang typing rules via \ruleref{TC-Lift} works cleanly for contexts involving just \Compute or \Prove, but shared references labeled \ComputeProve require a more careful treatment.
Specifically, we cannot na\"ively use $\SigmaCP$ as a context when leveraging \corelang's type system, as its domain (pairs of locations) does not match the domain of \corelang's heap context.
Instead, we substitute all $\SLocCP{\iota_1}{\iota_2}$ references with \emph{ghost locations},
which exist for the express purpose of enabling typing for \ComputeProve: they do not appear in any of the stores, and simply represent a synchronized pair of locations as a single location.
Using ghost locations, we define a second typing rule \ruleref{TC-LiftCP} that handles \ComputeProve typing separately:
\[\rulefiguresize
    \toggletrue{LiftCPToggle}
    \TCLiftCPRule
\]
The basic premise is the same as with the single-side rule: provide the namespace by projecting the contexts and ensure the operations in expression $e$ are available to \ComputeProve.
The differences involve substituting synchronized location pairs in both the expression $e$ and the context $\SigmaCP$ with ghost locations, defined by a map $\Theta$ from synchronized pairs of locations to unique individual ghost locations.
This map must be injective to prevent spurious typing proofs from reference collisions.
We then provide to the sub-derivation a transformed context $\Theta(\SigmaCP)$ that maps ghost locations to the type corresponding to the original location pair in $\SigmaCP$.
When lifting $e$ to \ComputeProve, the operation $\Lift{\ComputeProve}{e}_\Theta$, defined compositionally, then replaces each ghost location in $\Theta$ with its corresponding \ComputeProve reference object and otherwise lifts the expression by augmenting its syntax with labels when necessary.

These two typing rules enable lifting expressions to a single desired procedure,
but the type system must additionally permit interleaved operations and computed input operations within a single expression.
Therefore we explicitly include typing rules for branching and sequencing.
The typing rule for sequencing, \ruleref{TC-Let}, is similar to its corresponding rule in \corelang.
To interleave operations from different procedures, $e_1$ can have type $\tplab{\TTau_1'}{\ell_1'}$ and $e_2$ type $\tplab{\TTau_2}{\ell_2}$ where, unlike before, label $\ell_1'$ need not necessarily be equal to $\ell_2'$.
Note that we do not have an explicit rule for subtyping.
We instead embed it in \ruleref{TC-Let}: the expression $e_1$ can have a subtype of the declared type of the bound variable.
We require explicit subtyping using \SLetN to simplify projection, but note that it does not reduce expressive power.
The typing rule for branching, \ruleref{TC-If}, also mirrors its corresponding rule in \corelang but with the new labels on types.

\begin{figure}
    \begin{ruleset}
      \SetRuleLabelLoc{lab}
        \TCLetRule
        \and
        \TCIfRule
        \and
        \TCInputAssignRule
        \and
        \TCInputDerefRule
        \and
        \toggletrue{TCombinedToggle}
        \TComputeAndProveRule
    \end{ruleset}
    \caption{Selected Typing Rules for \zkstrudel}
    \label{fig:selected-combined-types}
\end{figure}

These rules also facilitate tracking assigned computed inputs.
Sequencing will thread the output set $A_1$ of the first expression $e_1$ to the input of $e_2$.
If statements use the same incoming tracking set for each branch and require the outgoing tracking set to be equal in each branch to ensure assignments are tracked precisely: $A$ represents the exact set of assigned computed inputs.

The type system tracks computed inputs in this way to permit assignment to computed inputs exactly once.
Requiring assignment occurs at least once guarantees any dereference in prove is defined.
Requiring assignment at most once serves to keep the behavior reasonable.
Otherwise, a computed input could be dereferenced before it is assigned.
Since compute runs in its entirety before prove, the two procedures would get different values, drastically complicating reasoning.

In \ruleref{TC-InputAssign}, we first check the computed input $x$ has not been assigned to yet (i.e., $x \not\in A$).
Then, if the data type $\TTau$ of the assigned value matches the expected one in $\Delta(x)$ and the value is available in the compute procedure, we can type check the whole assignment and add the computed input to the outgoing tracking set: $A \cup \{x\}$.
Once a computed input $x$ has been assigned, meaning the computed input is in the incoming context $A$, the type system permits a dereference with \ruleref{TC-InputDeref}.
The resulting value has the data type $\Delta(x)$ and label $\ComputeProve$, as it exists in \Prove by definition and was originally assigned in \Compute.

Finally, the \ruleref{T-ComputeAndProve} rule defines how to type check a compute-and-prove block.
All external variables and memory is available to compute, while only explicitly-defined inputs (both computed inputs $\overline{x_p}, \overline{x_s}$ and explicit inputs $\overline{y_p}, \overline{y_s}$) are available to prove.
Since proof generation must not use external references, the type system checks that all inputs to the proof must satisfy \refunreachN.
The explicit input names serve as binders for the externally-provided values $\overline{v_p}, \overline{v_s}$, so we also check that these values type with the types of the binders using the \corelang typing judgement.
Finally, we verify that $e$ types to a boolean available to \Prove in the combined type system, ensuring valid unified syntax to generate a proof.

The following contexts are used to type the proof predicate:
Since all memory locations are available at compute, they populate $\SigmaC$, while $\SigmaP$ and $\SigmaCP$ remain empty to start.
We then assign the label of proof inputs $\overline{y_p}$ and $\overline{y_s}$ to be \ComputeProve, making them available to the prove procedure in addition to compute.
We also assign the label $\Compute$ to all variables in $\Gamma$, with the order of contexts giving precedence to $\Gamma_\ComputeProve$.
The computed input context $\Delta$ maps all computed input variables to their types.
Then, to prove assignment to all computed inputs, the rule assumes an empty inbound assignment set $\varnothing$ and checks the outgoing $A$ must exactly match the set of computed inputs to ensure all computed inputs have been assigned.
If all of these conditions are satisfied, the block type checks.
Like with $\TProveN$ statements, the programmer assigns a unique name $\alpha$ to all compute-and-prove blocks, and so the type of a combined block is then a $\TProofOf*$ with label \Compute.

\subsection{Translating \zkstrudel to \corelang}
\label{sec:compilation}
Now that types have labels, we can separate the two procedures in a compute-and-prove block and execute them sequentially: first system updates and computed input construction, then proof generation.
We therefore define a projection operator that, given an augmented expression, extracts a \corelang expression for each procedure.
We then define how to combine these projected code procedures together to define a translation from a \zkstrudul program to a \corelang program.

Because determining an operation's procedure fundamentally relies on labels, projection is a type-directed transformation.
The projection operator $\denG{\cdot}$ takes a typing proof as input and extracts the relevant operations for procedure $\ell \in \{\Compute, \Prove\}$.
We abbreviate the input typing proof to just the expression in the relevant typing judgment $\denG{\SExp}$ for shorthand.

Intuitively, projection extracts the relevant code in $\SExp$ that belongs to $\ell$, retaining operations from~$\ell$ and ignoring operations not from $\ell$.
Consider
$\SExp = \SAssign{x_c}{1} \Sseq \SAssign{x_p}{2} \Sseq \SAssign{x_\mathit{cp}}{3}$,
where $x_c$, $x_p$, and $x_{cp}$ are mutable references with labels \Compute, \Prove, and \ComputeProve, respectively.
Projection to~$\Compute$ should retain the assignments to~$x_c$ and~$x_\mathit{cp}$ because they exist in the compute namespace,
but it should ignore~$x_p$, which is not.
Projection at \Prove is symmetric, projecting only the assignments to $x_p$ and $x_{cp}$.
\begin{mathpar}
    \denG[\Compute]{\SExp} = \TAssign{x_c}{1} \Tseq \TAssign{x_{cp}}{3}
    \and
    \denG[\Prove]{\SExp} = \TAssign{x_p}{2} \Tseq \TAssign{x_{cp}}{3}
\end{mathpar}

\begin{figure}
  \begin{mathpar}[\rulefiguresize]
    \denFull{\Lift{\ell'}{e}} = \begin{cases}
      e & \ell \subsetl \ell' \\
      \TUnitVal & \ow
    \end{cases}
    \and
    \denFull{\SWitAssign{x}{v}} = \begin{cases}
      \TAssign{x}{v} & \ell = \Compute \\
      \TUnitVal & \ell = \Prove
    \end{cases}
    \and
    \denFull{\SWitDeref{x}} = \begin{cases}
      \TDeref{x} & \ell = \Compute \\
      x & \ell = \Prove
    \end{cases}
    \and
    \denFull{\SIfThenElse*} = \begin{cases}
      \TIfThenElse{v}{\denG{e_1}}{\denG{e_2}} & \denFull{v} = v \\
      \denFull{e_1} & \denFull{v} = \TUnitVal \land \denFull{e_1} = \denFull{e_2} \\
      \text{undefined} & \text{otherwise}
    \end{cases}
    \and
    \denFull{\SLetIn{x\ty\tplab{\TTau}{\ell'}}{\SExp_1}{\SExp_2}} = \begin{cases}
      \TLetIn{x\ty\TTau}{\denG{e_1}}{\denG{e_2}} & \ell \subsetl \ell' \\
      \denFull{e_1} \Tseq \denFull{e_2} & \ell \not\subsetl \ell' \land \denFull{e_1} \neq v \\
      \denFull{e_2} & \ell \not\subsetl \ell' \land \denFull{e_1} = v 
    \end{cases}
  \end{mathpar}
  \caption{Selected Projection Rules}
  \label{fig:selected-projection}
\end{figure}

\Cref{fig:selected-projection} presents a selected set of projection rules, with the full rules in \Cref{app:compilation}.
All projection rules rely on the typing judgment to determine inclusion or exclusion.
We capture projection for most operations by the first rule:
if the typing proof types $\Lift{\ell'}{e}$ via a \corelang subderivation (\ruleref{TC-Lift} or \ruleref{TC-LiftCP}) and $\ell \subsetl \ell'$, then the operation $e$ corresponds to $\ell$.
Otherwise, it projects to an inert unit value~$\TUnitVal$.

We handle computed inputs differently in each procedure.
Compute must be able to assign to them, so we project them to mutable references;
input assignment $\SWitAssign{x}{v}$ translates to reference assignment $\TAssign{x}{v}$
and input usage $\SWitDeref{x}$ translates to dereferencing~$\TDeref{x}$.
However, mutable references cannot be used as proof inputs, so our full compilation procedure will dereference all computed inputs before feeding them to a $\TProveN$ statement.
The prove code therefore treats $\SWitDeref{x}$ as simply~$x$, as the proof will be given a value of the underlying type.
Since witness assignment can only happen in the compute code, the \Prove projection ignores input assignments and produces~$\TUnitVal$.

Branching and sequencing have subexpressions that could contain operations from either procedure, and so we recursively project the subterms.
For branching, if the condition~$v$ is visible to~$\ell$---$\denFull{v} = v$---then we project the entire if statement.
However, if~$v$ is not visible---$\denFull{v} = \TUnitVal$---then~$\ell$ cannot know which branch to take.
If both branches contain the same operations---$\denFull{\SExp_1} = \denFull{\SExp_2}$---then the direction does not matter and we simply use that projection.
Otherwise the projection of the branching statement is ill-formed and projection fails.

The rule for \SLetN similarly differentiates based on the procedural visibility.
If the label~$\ell'$ of the bound variable includes~$\ell$, then the output of~$\SExp_1$ should be available in the projection,
and the rule recursively projects the subterms into a let binding.
If $\ell \not\subsetl \ell'$, we do not project the variable binding, but still recursively project the subterms, which could contain meaningful computation for~$\ell$.
However, if additionally~$\SExp_1$ projects to a value (such as $\TUnitVal$), then~$\SExp_1$ has no meaningful effect in this projection and we simply return $\denFull{\SExp_2}$.
While this may seem like just an optimization, collapsing inert values has the effect of mapping semantically equivalent expressions to syntactically equivalent ones,
thereby increasing expressivity by allowing more branching statements to properly project.

We now turn to defining the full compilation pass of a compute-and-prove block.
The behavior is to first allocate memory for computed inputs as mutable reference, then execute the compute projection, and finally generate a proof determined by the prove projection.
Since computed inputs start uninitialized, we need a separate construct, $\TAlloc{\TTau}$, to allocate but not initialize a new reference of a given type.
This feature is not intended for direct use by the programmer, as a dereference step can now get stuck on an uninitialized value. 
However, because the type system ensures assignment to all computed inputs, the semantics will not get stuck by uninitialized references allocated through compilation.
We then dereference all computed inputs before entering the $\TProveN$ statement to provide the computed inputs without using context-aware values as proof inputs.
The formal definition of this compilation is shown below:
\[
  \begin{array}{l}
    \fullComp{\SCnp*[e]} = \\
    \qquad \TLetInMlp{\overline{w_p}, \overline{w_s} \ty \Reft \overline{\TTau_{x_p}}, \Reft \overline{\TTau_{x_s}}}%
                 {\TAlloc{\overline{\TTau_{x_p}}, \overline{\TTau_{x_s}}}}{%
                   \begin{array}[t]{@{}l@{}}
                    \subst*{\denC{e}}{{\overline{y_p}}{\overline{v_p}}{\overline{y_s}}{\overline{v_s}}{\overline{x_p}}{\overline{w_p}}{\overline{x_s}}{\overline{w_s}}}
                    \Tseq \\
                    \TProveUsing{\alpha}{\exists \overline{y_p}, \overline{x_p} [\overline{y_s}, \overline{x_s}]. \denP{e}}{\overline{v_p}, \TDeref{\overline{w_p}}[\overline{v_s}, \TDeref{\overline{w_s}}]}
                   \end{array}
                   }
    \end{array}
\]
Compilation of a \zkstrudul expression is then the identity function for all subexpressions except compute-and-prove blocks, which are compiled as defined above.
For a whole \zkstrudul program, we compile the main expression and all method bodies, resulting in a \corelang program.

\subsection{Operational Semantics}
\label{sec:combined-semantics}
The projection-based compilation above defines a behavior for compute-and-prove blocks that
precisely describes how a compiler should behave, but it performs a highly non-local transformation, making it difficult for programmers to reason about.
To address this shortcoming, we define an in-order operational semantics that directly describes the execution of the combined program.
We prove in \Cref{sec:adequacy} that these semantics are deeply equivalent.

To model interleaved procedures that may have independent state along with shared computed inputs,
the in-order semantics tracks three distinct stores:
\rone~a compute store~$\sigmaC$, corresponding to the heap $\sigma$ from \corelang;
\rtwo~a prove store $\sigmaP$ representing an ephemeral (or ``ghost'') heap used within the scope of a proof predicate; and
\rthree~a mapping~$\rho$ of computed input names to their assigned values.
We thus define the operational semantics as a small-step relation ($\cpstep$) over configurations of the form $\cpconf*{\SExp}$.

\begin{figure}
    \begin{ruleset}
      \SetRuleLabelLoc{lab}
      \toggletrue{ECRefCToggle}
        \ECLiftRule
        \and
        \ECRefCRule
        \and
        \ECDerefPRule
        \and
        \ECDerefCPRule
        \and
        \ECInputAssignRule
        \and
        \ECInputDerefRule
    \end{ruleset}
    \caption{Selected Semantic Rules for \zkstrudel}
    \label{fig:selected-combined-semantics}
\end{figure}

\Cref{fig:selected-combined-semantics} presents selected semantic rules, with the full ruleset in \Cref{app:combined-semantics}.
The \ruleref{EC-Lift} rule follows the pattern of the \ruleref{TC-Lift} typing rule and its corresponding projection rule
and allows most steps to lift from \corelang to the compute-and-prove semantics.
Specifically, if~$e$ can step in \corelang with empty heaps before and after---it cannot create, read, or update a reference---%%%
the lifted expression $\Lift{\ell}{e}$ can step identically in the combined language without modifying any stores.

This rule also drives our inclusion of a syntactic procedure label on method calls.
Without that label, \ruleref{EC-Lift} would produce non-deterministic semantics:
a method lifted to different procedures could allocate a reference in its body,
causing meaningfully different behavior depending on the label.
Including the syntactic label ensures there can only be one~$\ell$ such that $\SCall* = \Lift{\ell}{e}$,
ensuring $\cpstep$ remains deterministic.

We do not lift rules involving references, because we must properly track which heap each reference belongs to.
We track this in the syntax with administrative terms that wrap raw references to indicate which heap they belong to:
$\SLocC*$ and $\SLocP*$ belong to~$\sigmaC$ and~$\sigmaP$, respectively, and $\SLocCP*$ consists of two locations, one in each heap, that must remain in sync.
The semantics include three corresponding rule sets to create, access, and update these references.

\ruleref{EC-RefC} allocates a fresh location~$\iota$ in the compute store~$\sigmaC$ and returns it properly wrapped.
There are corresponding \ruleref{EC-RefP} and \ruleref{EC-RefCP}.
rules, with the latter allocating two references, one in each heap, and returning the paired~$\SLocCP*$.
Also note that \ruleref{EC-RefC} does not simply place~$v$ into~$\sigmaC$, but instead a \emph{lowered}~$\Lower{\Compute}{v}$.
Because the compute store~$\sigmaC$ corresponds to the heap in \corelang, the semantics must reconcile the wrapped references with the underlying heap representation.
The lowering function~$\Lower{\ell}{v}$ bridges this gap by stripping away the wrappers identifying the heap and converting to a standard \corelang value.
Specifically, it is a structurally recursive partial function where $\Lower{\Compute}{\SLocC*} = \iota$, $\Lower{\Compute}{\SLocCP*} = \iota_1$,
$\Lower{\Compute}{\SLocP*}$ is undefined, and other base values are returned unmodified (and symmetrically for $\Lower{\Prove}{v}$).

When reading from a heap, these wrappers must be put back in place, which is accomplished by the lifting operator~$\Lift{\ell}{v}$, as in \ruleref{EC-DerefP}.
However, when reading \ComputeProve references, there are now two values that must match.
\ruleref{EC-DerefCP} therefore joins the lifted references back together ($\sqcup$),
an idempotent structurally-recursive operation defining $\SLocC{\iota_1} \sqcup \SLocP{\iota_2} = \SLocCP*$.

The semantics for the computed input operations $\SWitAssign*$ and $\SWitDeref*$ operate on the store~$\rho$ and enforce a single-assignment property that mirrors the property of the type system.
Specifically, \ruleref{EC-InputAssign} allows assigning to~$x$ if it is not already assigned,
but there is no rule for $\SWitAssign*$ when $x \in \dom(\rho)$, meaning such an operation will get stuck.
Once~$x$ has been assigned, reads behave as expected using~\ruleref{EC-InputDeref}.
No lifting or lowering occurs in $\rho$, as the type system prevents any references from appearing within computed input values.

\begin{figure}
    \begin{ruleset}
        \SetRuleLabelLoc{lab}
        \toggletrue{EComputeAndProveToggle}
        \SetRuleLabelVCenter
        \EComputeAndProveInitRule
        \and
        \EComputeAndProveStepRule
        \and
        \EComputeAndProveTrueRule
    \end{ruleset}
    \caption{Semantic Rules for Compute-And-Prove}
    \label{fig:computeandprove}
\end{figure}

The semantics of compute-and-prove blocks, presented in~\Cref{fig:computeandprove}, follow the step-by-step execution of its body~$e$.
To start, \ruleref{E-ComputeAndProveInit} sets up the execution of~$e$, stepping to an administrative term $\SCnpAdmin*{\SExp}{\varnothing}{\varnothing}$ that represents an intermediate state of execution.
This term stores information to continue executing $e$---the ghost heap and computed input, which start empty---and to finish execution---the proof name~$\alpha$, the public and private computed input names~$\overline{x_p}$ and~$\overline{x_s}$, respectively, and a map~$\varphi$ from computed input names to their fresh memory locations. 
Since computed inputs compile to mutable references, this step also allocates memory locations for each.
These newly created locations are uninitialized, represented as~$\bot$, where dereferencing a location pointing to~$\bot$ gets stuck.
However, we ensure it is not possible to get stuck on dereferencing an uninitialized reference created by this step.

Within this administrative term, \ruleref{E-ComputeAndProveStep} steps the body~$e$ using~$\sigma$ as the compute heap
and the tracked stores~$\sigmaP$ and~$\rho$, updating the administrative syntax as needed.
These updates allow the semantics to properly lower a sequence of $\cpstep$ steps to the outer~$\tstep$ semantics.
Finally, if the inner expression evaluates to $\STrue$, then the \ruleref{E-ComputeAndProveTrue} rule yields the proof $\TProofOfUsing{\alpha}{\overline{v_p}::\rho(\overline{x_p})}$ carrying the proof's public inputs and updates the store $\sigma$ with the final values of computed inputs.

The type system for \zkstrudul is sound with respect to its operational semantics: if an expression type checks and is not value, it either can take a step or is stuck in an expected way.
An expected stuck state includes invalid runtime casts and false proof predicates as in \corelang, or a \SCnpN term evaluating to false,
captured by the predicate $\stuckpredzk{\SExp'}{\sigmaC'}{\sigmaP'}{\rho'}$
\begin{restatable}[Type Soundness Inside Combined Blocks]{theorem}{abssoundness}
\label{thm:abs-soundness}
  If $\Sigma ; \varnothing ; \Delta ; A \cproves \SExp : \tplab{\TTau}{\ell} \cproduces A'$ and $\cpconf*{\SExp} \cpstep^* \cpconf{\SExp'}{\sigmaC'}{\sigmaP'}{\rho'}$,
  then either
  (1)~$\SExp'$ is a value,
  (2)~$\cpconf{\SExp'}{\sigmaC'}{\sigmaP'}{\rho'}$ can step, or
  (3)~$\stuckpredzk{\SExp'}{\sigmaC'}{\sigmaP'}{\rho'}$.
\end{restatable}

Proving type soundness in \zkstrudul presents more challenges than in \corelang, as we cannot prove progress and preservation using only the typing judgment.
We instead use a stronger notion of \emph{configuration safety} that tracks locations in the heap to ensure that any \ComputeProve reference stays in sync.
In particular, configuration safety partitions each heap into three sets:
locations for only that procedure, locations synced with the other store, and other untracked locations e.g. that are outside of scope.
References in one partition may not alias or point to references in another---similar to the separating conjunction in separation logic.
Along with a correspondence between~$\Delta$, $A$, and~$\rho$ ensuring $\dom(\rho) = A$ and~$\rho$ agrees with~$\Delta$ on types,
configuration safety provides a sufficiently strong inductive invariant to prove progress and preservation.

\subsubsection*{A Note on Nesting Compute-And-Prove}
Our language also supports nesting compute-and-prove blocks.
Since $\TProveN$ statements in \corelang cannot nest, the inner nested compute-and-prove block only corresponds to the outer block's compute procedure.
Formalizing nested blocks involves adding very similar rules to those presented in this section but permitting only compute-side values and locations in a nested block.
We include the formal semantic and typing rules in \Cref{app:combined-semantics} and \Cref{app:combined-types}, and we prove all of the presented theorems even with nesting.

\section{Correctness of Compilation}\label{sec:adequacy}
While the combined semantics of \zkstrudul defined in~\Cref{sec:combined-semantics} executes operations in-order,
interleaving compute and prove steps,
in reality the projection-based semantics of~\Cref{sec:compilation} will compile the source program into a target program that
runs the entirety of the compute procedure first, and then generates a proof specified by the prove procedure.
Since these semantics provide two different interpretations of the same program, it is important
that they are equivalent, allowing use of the in-order semantics to reason about compiled programs.

First, we will show that our semantics are \emph{adequate} with respect to each other.
That is, executing a program using the source semantics~($\cpstep$)
is equivalent to executing the compiled program using the target semantics~($\tstep$).
\begin{definition}[Adequacy]
  A compilation procedure $\den{\cdot}$ is adequate when, for any source program~$e$, $e$ evaluates to the value~$v$ if and only if the compiled program~$\den{e}$ evaluates to~$\den{v}$.
  \[
    \text{Adequate}(\den{\cdot}) ~\defeq~ \forall \SExp, \SVal \ldotp \SExp \cpstep* \SVal \Longleftrightarrow \den{\SExp} \tstep* \den{\SVal}
  \]
\end{definition}

Second, we prove that compilation satisfies Robust Relational Hyperproperty Preservation (RrHP)---the strongest property in the robust compilation hierarchy---%% DO NOT REMOVE
ensuring that our source and target programs have fundamentally identical meanings.
Specifically, it says that, for any context in the target language, there exists a context in the source language,
such that \emph{every} source program produces the same behavior when linked with the source context
or compiled and linked with the target context.
This condition subsumes nearly every other secure and correct compilation criteria, including full abstraction~\citep{journey-beyond19}.

We first describe the proof strategy we use for adequacy.
The proof of RrHP follows the same strategy with some extra reasoning about traces, which we discuss in \Cref{sec:rrhp}.

\subsection{Proving Adequacy}
Classic adequacy proofs demonstrate a lock-step correspondence between two semantics.
However, the in-order semantics of \zkstrudul interleaves operations performed by each procedure,
whereas the projected semantics executes~\Compute before~\Prove.
This means a lock-step correspondence is impossible;
to simulate a single in-order step corresponding to~\Prove,
the projected semantics would have to fully execute the expression corresponding to~\Compute,
taking an arbitrary number of steps in the process.

To bridge the gap between these two semantics, we define a new \emph{concurrent} semantics~($\pstep$)
inspired by choreographies~\cite{Montesi23}
which allows for the operations performed by~\Compute and~\Prove to be interleaved,
so long as the order of events for each procedure is consistent.
The diagram below summarizes our proof strategy.
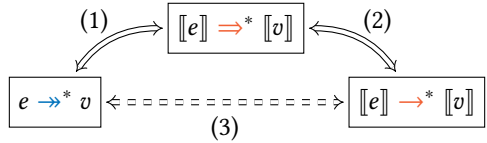
\begin{center}
  \begin{tikzpicture}[
      node/.style={outer sep=2pt,draw,black},
    ]
    \node[node] (SurfSequential) {$\SExp \cpstep^* \SVal$\strut};
    \node[node,right=9em of SurfSequential] (Sequential) {$\den{\SExp} \tstep^* \den{\SVal}$\strut};
    \coordinate (midpoint) at ($(SurfSequential.north)!.5!(Sequential.north)$);
    \node[node,above=0.5em of midpoint] (Concurrent) {$\den{\SExp} \pstep^* \den{\SVal}$\strut};

    \draw[implies-implies,double equal sign distance] (SurfSequential) to[out=55,in=180] node[outer sep=0pt,midway,above left = -2pt]{(1)} (Concurrent);
    \draw[implies-implies,double equal sign distance] (Concurrent) to[out=0,in=125] node[outer sep=0pt,midway,above right = -2pt]{(2)} (Sequential);
    \draw[implies-implies,double equal sign distance,dashed] (SurfSequential) -- (Sequential) node[outer sep=0pt,midway,below = 2pt]{(3)};
  \end{tikzpicture}
\end{center}

We (1)~establish a bisimulation between the source and concurrent target semantics (\Cref{thm:bisim-complete,thm:bisim-sound}),
then (2)~show that the sequential and concurrent target semantics yield equivalent outputs (\Cref{thm:bubblesort}),
allowing us to (3)~conclude our compilation is adequate (\Cref{thm:adequacy}).

\subsubsection{Concurrent Target Semantics}\label{sec:conc-sem}

\begin{figure}
  \begin{ruleset}
    \SetRuleLabelLoc{lab}
    \SetRuleLabelVCenter
    \EPStepC
    \and
    \EPStepP
     \and
    \toggletrue{EPWitAssignLinebreak}
    \EPWitAssign
    \and
    \EPWitDerefP
  \end{ruleset}
  \caption{Selected Concurrent Target Operational Semantics}
  \label{fig:selected-target-concurrent-semantics}
\end{figure}

The concurrent semantics, which we denote with the relation~$\pstep$, is similar to the projected semantics, with two main differences.
First, rather than tracking a single expression which is being executed, we track two expressions~$e_\Compute$ and~$e_\Prove$ that will execute concurrently.
Second, since execution of these programs is concurrent, the corresponding memory stores~$\sigma_\Compute$ and $\sigma_\Prove$
as well as the computed input memory store~$\rho$---which is treated as shared memory---must be separately accounted for.
While the expressions~$e_\Compute$ and~$e_\Prove$ can each independently update their corresponding memory stores,
\Prove can only read from the computed input memory~$\rho$, while~\Compute can both read and write.
These principles are captured by the rules shown in \Cref{fig:selected-target-concurrent-semantics}, with the remaining rules presented in \Cref{app:target-conc-lang}.

\ruleref{EP-StepC} and \ruleref{EP-StepP} lift the standard target semantics to allow~\Compute and~\Prove to execute independently.
\ruleref{EP-WitAssign} allows~\Compute to write to the computed input memory, and
\ruleref{EP-WitDerefP} allows~\Prove to read the computed input memory.
While there is an analogous rule \ruleref{EP-WitDerefC} to allow~\Compute
to read the computed input memory, \Prove cannot write to the computed input memory.

\subsubsection{Bisimilarity of Source and Concurrent Target Semantics}\label{sec:concurrent-bisimilarity}
To show that this concurrent semantics is bisimilar to the in-order semantics of \zkstrudul,
we must first address two technical complications needed to formally define our bisimulation relation.

\para{Annotated Source Programs}
To prove a bisimulation, one might hope that reduction would preserve equivalence between a program and its projection.
Unfortunately, this is not always true.
Specifically, substituting into the body of a let expression removes the syntactic type of the variable,
meaning values that may have been bound in only one procedure could suddenly appear in the other.
For example, the expression $\SLetIn{x\ty\tplab{\Bool}{\Prove}}{\STrue}{(\SIfThenElse{x}{2}{2})}$ projects to~$2$ for~\Compute,
but after stepping to~$\SIfThenElse{\STrue}{2}{2}$, it projects to~$\TIfThenElse{\TTrue}{2}{2}$.

To prevent this unintended behavior, we annotate our source language semantics and give each abstract syntax node its own procedure label.
For instance, the value~$\STrue_\annp$ denotes that $\STrue$ is only utilized in~\Prove's program,
while the value~$2_\anncp$ denotes the integer will be utilized by both~\Compute and~\Prove.
These annotations are simple to add with a type-directed elaboration pass over our regular syntax.
The original source program in the example above would be re-written as $\SLetIn{x\ty\tplab{\Bool}{\Prove}}{\STrue_\annp}{(\SIfThenElse{x}{2_\anncp}{2_\anncp})}$,
and step to~$\SIfThenElse{\STrue_\annp}{2_\anncp}{2_\anncp}$.

Projection becomes a simpler syntax-directed operation over these annotated programs, relying on the annotations rather than the types.
The above example then correctly projects to~$2$ for~\Compute, preventing the if-expression from materializing.

\para{Less-Sequencing Relation}
We might also hope for a bisimulation relating each annotated source program only to its projection.
However, this relation is too strict.
During projection we collapse expressions of the form~$v ~\Tseq e$ to $e$ (e.g., in the rules for let- and if-expressions)
which can break the strict equality between a \zkstrudul program and its projection after a reduction occurs.
This problem also arises in the choreographic literature.
We use the approach of \citet{SamuelsonHC25} and loosen the bisimulation with
a \emph{less-sequencing} relation~$\lessthan$ that relates $e \lessthan v \Tseq e$.
Formally, $\lessthan$ is the smallest structurally compatible pre-order generated by admitting the rule~$e \lessthan v \Tseq e$.

\para{Soundness and Completeness}
With these two adjustments, we can prove our compilation procedure is sound and complete with respect to our concurrent semantics.
Completeness is straightforward: for each sequence of steps in a \zkstrudul program, there is a sequence of concurrent steps in the projected program which simulate the source steps.
\begin{restatable}[Completeness of Concurrent Semantics]{theorem}{bisimcomplete}\label{thm:bisim-complete}
  If $\cpconf*{e} \cpstep^* \cpconf{e'}{\sigmaC'}{\sigmaP'}{\rho'}$ then there is some $e_\Compute'$ and $e_\Prove'$ such that $\pconf*{\denC{e}}{\denP{e}} \pstep \pconf{e_\Compute'}{e_\Prove'}{\sigmaC'}{\sigmaP'}{\rho'}$ where $\denC{e'} \lessthan e_\Compute'$ and $\denP{e'} \lessthan e_\Prove'$.
\end{restatable}

Soundness, however, is complicated by infinite loops.
If one procedure diverges, the source program may never be able to simulate a step made in the other procedure in the concurrent semantics.
However, if the target programs do not both terminate (in polynomial time, no less), then the security guarantees of zkSNARKs no longer hold,
nullifying the purpose of \zkstrudul.
We are therefore concerned only with terminating programs.
Fortunately, these termination-based complications are the only areas impeding soundness;
assuming the termination of both target programs yields soundness.
\begin{restatable}[Soundness of Concurrent Semantics]{theorem}{bisimsound}\label{thm:bisim-sound}
  If~$\denC{\SExp}$ and~$\denP{\SExp}$ both terminate and $\pconf*{\denC{\SExp}}{\denP{\SExp}} \pstep^* \pconf{e'_\Compute}{e'_\Prove}{\sigmaC'}{\sigmaP'}{\rho'}$, then there is some
  $\SExp'$, $e''_\Compute$, and $e''_\Prove$ such that
  $\cpconf*{\SExp} \cpstep^* \cpconf{\SExp'}{\sigmaC''}{\sigmaP''}{\rho''}$ and
  $\pconf{e'_\Compute}{e'_\Prove}{\sigmaC'}{\sigmaP'}{\rho'} \pstep^* \pconf{e''_\Compute}{e''_\Prove}{\sigmaC''}{\sigmaP''}{\rho''}$, where
  $\denC{\SExp'} \lessthan e''_\Compute$ and $\denP{\SExp'} \lessthan e''_\Prove$.
\end{restatable}

\subsubsection{Equivalence of Sequential and Concurrent Target Semantics}\label{sec:target-soundness}
The goal is to show the equivalence of the sequential and concurrent target semantics.
That is, executing a (terminating) program with the concurrent semantics produces the same results as executing with the projected semantics.
One direction is obvious: any sequential execution is a valid concurrent one.

To show the target semantics can replicate the concurrent semantics,
we consider an arbitrary interleaving of steps in the concurrent semantics
and progressively rearrange the steps, preserving the validity of the entire sequence along the way.
In particular, we note that the two procedures can only interact by the compute procedure setting a computed input which the prove procedure reads.
If the interleaving has a prove step immediately followed by a compute step, swapping their order must always produce the same result.
By repeating this operation until all compute operations come first---essentially bubble sorting the interleaving---%%%
we produce a sequence corresponding precisely to a sequential execution of the target program.

The following diagram shows this sorting process visually.
\begin{center}
  \begin{tikzpicture}[
      conf/.style={inner sep=4pt},
      label/.style={color=black, font=\small, inner sep=2pt, outer sep=0pt},
      highlighted/.style={draw,circle,fill=yellow!20},
    ]
    \matrix[column sep=3em, row sep=2.5em] {
      \node[conf] (e1) {$\conf{e_{\Compute 1}}{e_{\Prove 1}}$}; &
      \node[conf] (e2) {$\conf{e_{\Compute 1}}{e_{\Prove 2}}$}; &
      \node[conf] (e3) {$\conf{e_{\Compute 1}}{e_{\Prove 3}}$}; &
      \node[conf] (e4) {$\conf{e_{\Compute 2}}{e_{\Prove 3}}$}; &
      \node[conf] (e5) {$\conf{e_{\Compute 3}}{e_{\Prove 3}}$}; \\

      \node[conf] (e1') {$\conf{e_{\Compute 1}}{e_{\Prove 1}}$}; &
      \node[conf] (e2') {$\conf{e_{\Compute 1}}{e_{\Prove 2}}$}; &
      \node[conf] (e3') {$\conf{e_{\Compute 2}}{e_{\Prove 2}}$}; &
      \node[conf] (e4') {$\conf{e_{\Compute 2}}{e_{\Prove 3}}$}; &
      \node[conf] (e5') {$\conf{e_{\Compute 3}}{e_{\Prove 3}}$}; \\

      \node[conf] (e1'') {$e_{\Compute 1} \seq e_{\Prove 1}$}; &
      \node[conf] (e2'') {$e_{\Compute 2} \seq e_{\Prove 1}$}; &
      \node[conf] (e3'') {$e_{\Compute 3} \seq e_{\Prove 1}$}; &
      \node[conf] (e4'') {$e_{\Prove 2}$}; &
      \node[conf] (e5'') {$e_{\Prove 3}$}; \\
    };

    \draw[conc-arrow] (e1) -- (e2) node[midway,above=3pt,label]{$\Prove_1$};
    \draw[conc-arrow] (e2) -- (e3) node[midway,above=3pt,label,highlighted](swap1){$\Prove_2$};
    \draw[conc-arrow] (e3) -- (e4) node[midway,above=3pt,label](swap2){$\Compute_1$};
    \draw[conc-arrow] (e4) -- (e5) node[midway,above=3pt,label]{$\Compute_2$};

    \draw[stealth-stealth,bend right,shorten >=3pt,shorten <=3pt] (swap1.north east) to[out=45,in=135] (swap2.north west);

    \draw[mapsto] (e3) -- (e3');

    \draw[conc-arrow] (e1') -- (e2') node[midway,above=3pt,label]{$\Prove_1$};
    \draw[conc-arrow] (e2') -- (e3') node[midway,above=3pt,label]{$\Compute_1$};
    \draw[conc-arrow] (e3') -- (e4') node[midway,above=3pt,label,highlighted](swap3){$\Prove_2$};
    \draw[conc-arrow] (e4') -- (e5') node[midway,above=3pt,label](swap4){$\Compute_2$};

    \draw[stealth-stealth,bend right,shorten >=3pt,shorten <=3pt] (swap3.north east) to[out=45,in=135] (swap4.north west);

    \draw[mapsto] (e3') -- (e3'') node[midway,fill=white] {$\rvdots$};

    \draw[trg-arrow] (e1'') -- (e2'') node[midway,above=2pt,label]{$\Compute_1$};
    \draw[trg-arrow] (e2'') -- (e3'') node[midway,above=2pt,label]{$\Compute_2$};
    \draw[trg-arrow] (e3'') -- (e4'') node[midway,above=2pt,label]{$\Prove_1$};
    \draw[trg-arrow] (e4'') -- (e5'') node[midway,above=2pt,label]{$\Prove_2$};
  \end{tikzpicture}
\end{center}

Combining this result with \Cref{thm:bisim-sound,thm:bisim-complete}
proves our desired adequacy result.
\begin{restatable}[Adequacy of $\fullComp{\cdot}$]{theorem}{adequacy}\label{thm:adequacy}
  For any expression $e$ that may contain a combined block,
  $\langle e \divi \sigma \rangle \tstep^* \langle v \divi \sigma' \rangle$
  if and only if~$\langle \fullComp{e} \divi \sigma \rangle \tstep^* \langle \fullComp{v} \divi \sigma' \rangle$.
\end{restatable}

\subsection{Robust Relational Hyperproperty Preservation}\label{sec:rrhp}
With adequacy proved, we now show that our semantics and compilation satisfy another equivalence condition: Robust Relational Hyperproperty Preservation~(RrHP)~\citep{journey-beyond19}.
To reason about observable behavior within the robust compilation framework, we consider the traces produced by our semantics.
Our semantics emit nontrivial trace events from the operations that are relevant to the state of the system and NIZK proof operations:
allocating and assigning memory, generating proofs, and verifying proofs.
Note that executing the body (predicate) of a \TProveN statement does not produce trace events,
as proof generation occurs within a cryptographic proof system and only the final result is visible.

Robust compilation, and RrHP in particular, is defined in terms of linking \emph{contexts} with \emph{partial programs} to form \emph{whole programs}.
Partial programs~$P$ and contexts~$C$ both consist of a class list and an expression to evaluate,
though contexts include a single hole $[\cdot]$ in their expression.
To form a whole program~$W$, we can link a partial program and a context, denoted $C \link P$,
combining their class lists and placing the partial program's expression in the context's hole.
RrHP then compares the possible behaviors of source-language programs linked with source-language contexts---in our case \zkstrudul programs and contexts---%%%
with the compilations of those source programs linked with target-language contexts---here \corelang contexts.

Because RrHP compares the behavior of a partial source program~$\PartialP$ with its compilation~$\fullComp{\PartialP}$, we need a notion of compiling partial programs.
However, our projection process is type-directed, meaning compiling partial programs requires a typing derivation.
There are two challenges to this: the partial program expression may not be closed, and any part of~$\PartialP$ might use classes defined in the linked context.
We therefore allow~$\PartialP$ to be typed in a non-empty variable context~$\Gamma$ and class context~$\Phi$, defined as a set of class lookup functions available in the typing derivation denoted $\Phi; \Gamma \cproves \PartialP$.
These free variables and undefined classes must then come from the linked context to give a well-typed whole program, $\cproves \SContextN \Slink \PartialP$.

We then define the behavior of a full program~$\FullProg$ as the set of traces it can produce.
Due to the security of the underlying NIZK proof systems only holding when programs terminate in polynomial time,
we restrict our definition to \emph{terminating} behavior.
That is, $\Behav{\FullProg} \defeq \{ \trace \mid \FullProg \tstep*[\trace] v\}$.

With this notion of linking and behavior, we slightly modify \citeposessive{journey-beyond19} notion of terminating RrHP for type-directed compilation.
\begin{definition}[Typed RrHP]
  A type-directed compiler $\den{\cdot}$ satisfies \emph{typed RrHP} if,
  \[
    \forall \TContextN\ldotp \exists \SContextN\ldotp \forall \PartialP\ldotp
    {}\proves \TContextN \Tlink \den{\PartialP} \Longleftrightarrow {\cproves \SContextN \Slink \PartialP}
    \hspace{0.5em} \text{and} \hspace{0.5em}
    \TBehav{\TContextN \Tlink \den{\PartialP}} = \SBehav{\SContextN \Slink \PartialP}
  \]
\end{definition}
We can now prove that our compilation satisfies RrHP with our terminating notion of behavior.
\begin{restatable}[Robust Relational Hyperproperty Preservation]{theorem}{rrhp}\label{thm:rrhp}
  $\fullComp{\cdot}$ satisfies typed RrHP.
\end{restatable}
To prove this theorem, given an arbitrary target context $\TContextN$, we define a back-translation to $\SContextN$ as the identity function, which is allowed since \corelang is a subset of \zkstrudul. 
The proof then follows by updating the adequacy theorem to include traces, which shows that the sets of traces, and therefore the behaviors, are equal.

\section{Related Work}
\label{sec:related-work}

\para{Languages for Cryptographic Proofs}
\label{sec:crypto-languages}
There are a numerous language-based tools for NIZK proofs.
Many provide high-level language support for specifying zkSNARK predicates and compile those predicates to efficient circuits
and generate cryptographic proofs~\citep{pinocchio16,gnark,libsnark,zokrates18,circ22,circom23,o1js,leo-lang,zinc-lang,xjsnark18}.
Other work, like \textsc{Coda}~\citep{LiuKLT+24}, aims to verify important properties about individual circuits,
such as requiring that only one set of hidden inputs will allow them to verify.
All of these works, however, operate at the level of individual zkSNARK proofs and presume system updates and input computation is already complete.
They are therefore entirely complementary to \zkstrudul,
as our results show how to tie these powerful existing languages and libraries into larger systems.

Other work builds support for \emph{authenticated data structures}, like Merkle Trees, that allow verifiers to check the correctness of operations performed by untrusted provers.
\citet{Miller14} present a purely functional language for building them,
where programmers use special ``authentication'' types to designate differing but related behavior for the prover and verifier.
Recent work~\citep{Gregersen25} mechanically proves the correctness of this idea in an OCaml-like language using a variant of separation logic.
Both rely on a duality similar to our Merkle tree example (Section~\ref{sec:merkle-example}), but their techniques are specific to authenticated data structures.

\para{Languages for MPC Systems}
There exist languages designed to support full-system construction leveraging cryptographic primitives,
but they focus on \emph{interactive} multi-party computation~(MPC) and ignore the complexities of efficiently building NIZK proofs.
Languages like Viaduct~\citep{viaduct21} and Jif/Split~\citep{jifSplit02,ZhengCMZ03} transform high-level programs
with information-flow type systems into secure distributed implementations by selecting appropriate cryptographic mechanisms, including (interactive) zero knowledge proofs.
Wysteria~\citep{wysteria14,wysstar18} mixes local and secure multi-party computations using a refinement type system expressing the requirements of the computation that is compiled to circuits to be executed.
Symphony~\citep{symphony23} allows for more complex coordination by adding new first-class constructs, shares, and party-sets.
Our work does not explicitly model multiple parties.
While they can be approximated in \zkstrudul by representing each party as a separate object, including such explicit separation would be valuable future work.

\para{Choreographies}
\label{sec:choreographies}

Our approach to proving translation adequacy (Section~\ref{sec:adequacy}) pulls from \emph{choreographic programming}~\citep{Montesi23},
a concurrent programming paradigm
that unifies the code for all processes in a concurrent system into one combined program.
A projection operator then extracts code for each process from this top-level program
just as \zkstrudul compiles compute-and-prove blocks to separate procedures.
Recent choreographies allow processes to share data and control flow using multiply-located operations~\citep{BatesK+25,SamuelsonHC25},
similar to the \ComputeProve annotation in \zkstrudul.
These similarities are enough to make the proof techniques highly applicable to our work.
Choreographies have also previously proved useful for language-based analysis of synthesized cryptography~\citep{AcayGRM24} for interactive multi-party computation.
\section{Conclusion}
\label{sec:conclusion}
This work introduced \zkstrudul, a new programming abstraction for NIZK systems that unifies predicate definition and input transformation into a single coherent expression---the compute-and-prove block.
By leveraging the structural duality between these two procedures,
compute-and-prove blocks eliminate code duplication and enable local reasoning about normally-separate operations involving input construction and usage.
\zkstrudul projects the predicate and input transformation procedures from this combined abstraction,
and supports statically restricting the set of available predicate operations for compatibility with different potential NIZK backends.  
Leveraging ideas from concurrency theory, we established two correspondences between the source semantics and projected semantics;
first proving the semantics are adequate, and second that our compilation satisfies Robust Relational Hyperproperty Preservation.
As a result of these features and properties, \zkstrudul enables writing NIZK applications---including Merkle Proofs and zkVMs---that are more concise, easier to reason about, and less error-prone than existing approaches.

%% Acknowledgments
\section*{Acknowledgments}

Thanks to Ian Miers for helping to formulate motivating examples and,
along with Leonidas Lampropulous, discussions and direction that solidified the problem framing.
Support for this work was provided by NSF grant \#2504579
and by by the University of Wisconsin--Madison Office of the Vice Chancellor for Research with funding from the Wisconsin Alumni Research Foundation.

%% Bibliography
\bibliography{ethan,ephemeral}

%% Appendix
\appendix
\section{Syntax}
\label{app:syntax}

\subsection{Full Syntax}
\label{app:full-syntax}
The extension from \corelang to \zkstrudul is underlined in \underred{red}.
Terms not available in the surface language are underlined in \undergreen{green}.
{\small
  \[
    \begin{array}{rcl}
      f,m,x,\alpha & \in & \mathcal{V} \quad \text{(field, method, variable, proof names)} \\[2pt]
      \TTau & \Coloneqq & \TInt \alt \TBool \alt \TUnit \alt \TRefTau* \alt C \alt \TProofOf* \alt \TProof \\[2pt]
      \ell & \Coloneqq & \Compute \alt \Prove \alt \ComputeProve \quad (\Compute, \Prove \subsetl \ComputeProve) \\[2pt]
      \mathit{CL} & \Coloneqq & \TClassListSyntax \\ [2pt]
      K & \Coloneqq & \TConstructor \\[2pt]
      M & \Coloneqq & \TTau\{\ell\}~ m(\overline{x} \ty \overline{\TTau}) \{ e \} \\[2pt]
      \TVal & \Coloneqq & \TVar* \alt \TUnitVal \alt \TTrue \alt \TFalse \alt \TNew* \alt \undergreen{\TLoc} \alt \undergreen{\TProofOfUsing*} \\[2pt]
      \TExp & \Coloneqq & v \alt \TRefVal* \alt \TDeref* \alt \TAssign* \alt \TCast* \alt \TField* \alt \TCall* \alt \TAlloc{\TTau} \\[2pt]
      & \alt & \TLetIn* \alt \TIfThenElse* \\[2pt]
      & \alt & \TProveUsing* \alt \TVerify* \\[2pt]
      & \alt & \underred{\TCnp*} \\[2pt]
      & \alt & \undergreen{\underred{\TCnpAdmin*{e}{\sigmaP}{\rho}}}
    \end{array}
  \]
}

\subsection{Inner \zkstrudel Syntax}
\label{app:combined-syntax}
We underline changes from \corelang involving adding labels in \underred{red}, and involving computed inputs in \underblue{blue}.
Terms not available in the surface language are underlined in \undergreen{green}.
{\small
  \[
    \begin{array}{rcl}
      x,\alpha & \in & \mathcal{V} \quad \text{(variable, proof names)} \\[2pt]
      \TTau & \Coloneqq & \SInt \alt \SBool \alt \SUnit \alt {\SRefTau*}  \alt C \alt \SProofOf* \alt \SProof \\[2pt]
      \ell & \Coloneqq & \Compute \alt \Prove \alt \ComputeProve \quad (\Compute, \Prove \subsetl \ComputeProve) \\[2pt]
      \STau & \Coloneqq & \tplab{\TTau}{\ell} \\[2pt]
      r & \Coloneqq & \underred{\SLocC* \alt \SLocP* \alt \SLocCP*}\\[2pt]
      \SVal & \Coloneqq & \SVar* \alt \SUnitVal \alt \STrue \alt \SFalse \alt \SNew* \alt \undergreen{r} \alt \undergreen{\SProofOfUsing*} \\[2pt]
      \SExp & \Coloneqq & \SVal \alt \underred{\SRefVal*} \alt \SDeref* \alt \SAssign* \alt \underblue{\SWitAssign*} \alt \underblue{\SWitDeref*} \alt \SCast* \alt \SField* \alt \underred{\SCall*} \alt \SAlloc{\TTau} \\[2pt]
      & \alt & \underred{\SLetIn*} \alt \SIfThenElse*  \\[2pt]
      & \alt & \SProveUsing* \alt \SVerify* \\[2pt]
      & \alt & \SCnp* \\[2pt]
      & \alt & \undergreen{\SCnpAdmin*{e}{\sigmaP}{\rho}}
    \end{array}
  \]
}

\section{Operational Semantics}
\label{app:semantics}

\subsection{\corelang Operational Semantics}
\label{app:core-semantics}
\begin{align*}
    E \Coloneqq [\cdot] \mid \TLetIn{x\ty\TTau}{\EvalN}{e}
\end{align*}
\begin{rulesetpagebreakable}
    \SetRuleLabelLoc{lab}
    \SetRuleLabelVCenter

    \EEvalRule
    \and
    \ELetRule
    \and
    \EIfTRule
    \and
    \EIfFRule
    \and
    \ERefRule
    \and
    \EDerefRule
    \and
    \EAssignRule
    \and
    \ECastRule
    \and
    \EFieldRule
    \and
    \ECallRule
    \and
    \EAllocRule
    \and
    \EProveRule
    \and
    \EVerifyTRule
    \and
    \EVerifyFRule
\end{rulesetpagebreakable}

\subsection{Additional Semantic Rules for Combined Blocks}
\begin{rulesetpagebreakable}
    \SetRuleLabelLoc{lab}
    \SetRuleLabelVCenter
    \EComputeAndProveInitRule
    \and
    \EComputeAndProveStepRule
    \and
    \EComputeAndProveTrueRule
\end{rulesetpagebreakable}

\subsection{\zkstrudel Operational Semantics}
\label{app:combined-semantics}
\begin{align*}
    E \Coloneqq [\cdot] \mid \SLetIn{x\ty\STau}{\EvalN}{e}
\end{align*}

\begin{rulesetpagebreakable}
    \SetRuleLabelLoc{lab}
    \SetRuleLabelVCenter
    \ECLiftRule
    \and
    \ECEvalRule
    \and
    \ECLetRule
    \and
    \ECRefCRule
    \and
    \ECRefPRule
    \and
    \ECRefCPRule
    \and
    \ECDerefCRule
    \and
    \ECDerefPRule
    \and
    \ECDerefCPRule
    \and
    \ECAllocRule
    \and
    \ECAssignCRule
    \and
    \ECAssignPRule
    \and
    \ECAssignCPRule
    \and
    \ECInputAssignRule
    \and
    \ECInputDerefRule
    \and
    \ECComputeAndProveInitRule
    \and
    \ECComputeAndProveStepRule
    \and
    \ECComputeAndProveTrueRule
\end{rulesetpagebreakable}

\subsection{Trace Element Syntax and Equivalence Definition}
\subsubsection*{\corelang Traces}
{
    \begin{align*}
        \temit & \Coloneqq \Nulltrace \alt \topalloc(\TLoc, v) \alt \topset(\TLoc, v) \divi \topgen(\alpha, \overline{v}, \overline{u}) \divi \topverif(\alpha, \overline{v}, b) \\
        \ttrace & \Coloneqq \epsilon \alt \temit :: \ttrace \\
        b & \Coloneqq \bot \mid \top
    \end{align*}
}

\subsubsection*{\zkstrudel Traces}
{
    \begin{align*}
        \semit & \Coloneqq \Nulltrace \alt \sopalloc(\TLoc, v) \alt \sopset(\TLoc, v) \divi \sopset(x,v) \divi \sopgen(\alpha, \overline{v}, \overline{u}) \divi \sopverif(\alpha, \overline{v}, b) \\
        \strace & \Coloneqq \epsilon \alt \semit :: \strace \\
        b & \Coloneqq \bot \mid \top
    \end{align*}
}

\subsection{Trace Projection Definition}
\begin{mathparpagebreakable}
    \witdenR{\semit :: \strace} = \witdenR{\semit} :: \witdenR{\strace}
    \and
    \witdenR{\Nulltrace} = \Nulltrace
    \and
    \witdenR{\sopalloc(\iota, v)} = \topalloc(\iota, v)
    \and
    \witdenR{\sopset(\iota, v)} = \topset(\iota, v)
    \and
    \witdenR{\sopgen(\alpha, \overline{v}, \overline{w})} = \topgen(\alpha, \overline{v}, \overline{w})
    \and
    \witdenR{\sopverif(\alpha, \overline{v}, b)} = \topverif(\alpha, \overline{v}, b)
    \and
    \witdenR{\sopset(x,v)} = \topset(\varphi(x), v)
\end{mathparpagebreakable}

\subsection{Value Lifting and Lowering Value Definitions}\label{app:lifting-lowering-def}
\begin{mathparpagebreakable}
    \Lift*{x} = x 
    \and
    \Lift*{()} = ()
    \and
    \Lift*{\TTrue} = \STrue
    \and
    \Lift*{\TFalse} = \SFalse
    \and
    \Lift*{\TNew*} = \SNew{C}{\Lift*{\overline{v}}}
    \and
    \Lift*{\TProofOfUsing*} = \SProofOfUsing{\alpha}{\Lift*{\overline{v}}}
    \\
    \Lift*{\iota} = \begin{cases}
        r_\Compute(\iota) & \ell = \Compute \\
        r_\Prove(\iota) & \ell = \Prove \\
        \uparrow & \ell = \ComputeProve
    \end{cases}
    \and
    \SVal_1 \sqcup \SVal_2 = \begin{cases}
        \SVal_1 & \SVal_1 = \SVal_2 \\
        r_\ComputeProve(\iota_1, \iota_2) & \SVal_1 = r_\Compute(\iota_1) \land \SVal_2 = r_\Prove(\iota_2) \\
        \SNew{C}{\overline{v_1} \sqcup \overline{v_2}} & \SVal_1 = \SNew{C}{\overline{v_1}} \land \SVal_2 = \SNew{C}{\overline{v_2}} \\
        \SProofOfUsing{\alpha}{\overline{v_1} \sqcup \overline{v_2}} & v_1 = \SProofOfUsing{\alpha}{\overline{v_1}} \land v_2 = \SProofOfUsing{\alpha}{\overline{v_2}} \\
        \uparrow & \ow
    \end{cases}
    \\\\
    \Lower*{x} = x 
    \and
    \Lower*{()} = ()
    \and
    \Lower*{\STrue} = \TTrue
    \and
    \Lower*{\SFalse} = \TFalse
    \and
    \Lower*{\SNew*} = \TNew{C}{\Lower*{\overline{v}}}
    \and
    \Lower*{\SProofOfUsing*} = \TProofOfUsing{\alpha}{\Lower*{\overline{v}}}
    \\
    \Lower*{r_\Compute(\iota)} = \begin{cases}
        \iota & \ell = \Compute \\
        \uparrow & \ow
    \end{cases}
    \and
    \Lower*{r_\Prove(\iota)} = \begin{cases}
        \iota & \ell = \Prove \\
        \uparrow & \ow
    \end{cases}
    \and
    \Lower*{r_\ComputeProve(\iota_1, \iota_2)} = \begin{cases}
        \iota_1 & \ell = \Compute \\
        \iota_2 & \ell = \Prove \\
        \uparrow & \ell = \ComputeProve
    \end{cases}
\end{mathparpagebreakable}

\subsection{Expression Lifting Definitions}
\label{app:expression-lifting-def}
\begin{mathparpagebreakable}
    \Lift*{\TRefVal*} = \SRefVal{\ell}{\Lift*{v}}
    \and
    \Lift*{\TDeref*} = \SDeref{\Lift*{v}}
    \and
    \Lift*{\TAssign{v_1}{v_2}} = \SAssign{\Lift*{v_1}}{\Lift*{v_2}}
    \and
    \Lift*{\TCast*} = \SCast{C}{\Lift*{v}}
    \and
    \Lift*{v.f} = \Lift*{v}.f
    \and
    \Lift*{v.m(\overline{u})} = \Lift*{v}.m_\ell(\Lift*{\overline{u}})
    \and
    \Lift*{\TAlloc{\TTau}} = \SAlloc{\TTau}
    \and
    \Lift*{\TLetIn{x\ty\TTau}{e_1}{e_2}} = \SLetIn{x\ty\tplab{\TTau}{\ell}}{\Lift*{e_1}}{\Lift*{e_2}}
    \and
    \Lift*{\TIfThenElse{v}{e_1}{e_2}} = \SIfThenElse{\Lift*{v}}{\Lift*{e_1}}{\Lift*{e_2}}
    \and
    \Lift{\ell}{\TProveUsing*} = \SProveUsing* 
    \and
    \Lift*{\TVerify*} = \SVerify*[v]
    \and
    \Lift{\ell}{\TCnp*} = \SCnp*
    \and 
    \Lift{\ell}{\TCnpAdmin*{e}{\sigmaP}{\rho}} = \SCnpAdmin*{e}{\sigmaP}{\rho}
\end{mathparpagebreakable}
\section{Type System}
\label{app:types}

\subsection{\corelang Type System}
\label{app:core-types}
Notable changes to existing rules from Featherweight Java with mutable references~\citep{featherweightJava01, tapl} are colored in \changes{red}.

\begin{rulesetpagebreakable}
    \SetRuleLabelLoc{lab}
    \SetRuleLabelVCenter
    \TVarRule
    \and
    \TUnitRule
    \and
    \TTrueRule
    \and
    \TFalseRule
    \and
    \TLocationRule
    \and
    \TProofOfRule
    \and
    \TNewRule
    \and
    \TValRule
    \and
    \TCastRule
    \and
    \TFieldRule
    \and
    \TCallRule
    \and
    \TAllocRule
    \and
    \TRefRule
    \and
    \TDerefRule
    \and
    \TAssignRule
    \and
    \TIfRule
    \and
    \TLetRule
    \and
    \TProveRule
    \and
    \TVerifyRule
    \and
    \TSubtypeRule
    \\
    \SubClass
    \and
    \SubProof
    \and
    \SubTrans
    \and
    \MethodOk
    \and
    \ClassOkRef
    \and
    \ClassOkNoRef
    \and
    \CTOk
    \and 
    \infer*{
        \TTau\changes{\{\ell\}}~m(\overline{x} \ty \overline{\TTau_\alpha}) \{ e \} \in M\\\\
        CT(C) = \TClassN~C_{\changes{\{\ell_C\}}}~\TExtendsN~D_{\changes{\{\ell_D\}}}~\{\overline{f} \ty \overline{\TTau_f} ; K ; \overline{M} \}
    } {
        \mtype(m, C) = \overline{\TTau_\alpha} \xrightarrow{\changes{\ell}} \TTau \\\\
        \mbody(m, C) = (e, \overline{x}, \TTau)
    }
    \and 
    \infer*{
        m \text{ not defined in } \overline{M} \\\\
        CT(C) = \TClassN~C_{\changes{\{\ell_C\}}}~\TExtendsN~D_{\changes{\{\ell_D\}}}~\{\overline{f} \ty \overline{\TTau_f} ; K ; \overline{M} \}
    } {
        \mtype(m,C) = \mtype(m,D) \\\\
        \mbody(m,C) = \mbody(m,D)
    }
    \and
    \infer*{
        \fields(D) = \overline{g} \ty \overline{\TTau_g} \\\\
        CT(C) = \TClassN~C_{\changes{\{\ell_C\}}}~\TExtendsN~D_{\changes{\{\ell_D\}}}~\{\overline{f} \ty \overline{\TTau_f} \seq K \seq \overline{M} \}
    } {
        \fields(C) = \overline{g} \ty \overline{\TTau_g} ; \overline{f} \ty \overline{\TTau_f}
    }
    \and
    \infer*{
        (m,D) \in \dom(\mtype) \Rightarrow \mtype(m,D) = \overline{\TTau_\alpha} \xrightarrow{\changes{\ell}} \TTau
    } {
        \canoverride(D,m,\overline{\TTau_\alpha} \xrightarrow{\changes{\ell}} \TTau)
    }
    \and
    \infer*{
        PT(\alpha) = \CircuitDefault
    } {
        \ptypes(\alpha) = \overline{\TTau_x}
    }
\end{rulesetpagebreakable}

\subsection{Additional Typing Rules for Combined Blocks}
\begin{rulesetpagebreakable}
    \SetRuleLabelLoc{lab}
    \SetRuleLabelVCenter
    \TComputeAndProveRule
    \and
    \infer*{
        PT(\alpha) = \TCnp*
    } {
      \ptypes(\alpha) = \overline{\TTau_{x_p}}, \overline{\TTau_{y_p}}
    }
    \and
    \infer*{
        PT(\alpha) = \SCnp*
    } {
      \ptypes(\alpha) = \overline{\TTau_{x_p}}, \overline{\TTau_{y_p}}
    }
    \and
    \TCombinedRule
\end{rulesetpagebreakable}

\subsection{\zkstrudel Type System}
\label{app:combined-types}
\begin{mathparpagebreakable}
    (\Gamma, x \ty \tplab{\TTau}{\ell'})|_\ell = \begin{cases}
        \Gamma |_\ell , x \ty \tplab{\TTau}{\ell'} & \tplab{\TTau}{\ell'} <: \tplab{\TTau}{\ell} \\
        \Gamma |_\ell & \ow
    \end{cases}
\end{mathparpagebreakable}

\begin{rulesetpagebreakable}
    \SetRuleLabelLoc{lab}
    \SetRuleLabelVCenter
    \infer*{
        \TTau_1 <: \TTau_2
    } {
        \tplab{\TTau_1}{\ell} <: \tplab{\TTau_2}{\ell}
    }
    \and
    \infer*{
        \TTau_1 <: \TTau_2 \\
        \ell_2 \subsetl \ell_1 \\
        \refunreach{\TTau_1}
    } {
        \tplab{\TTau_1}{\ell_1} <: \tplab{\TTau_2}{\ell_2}
    }
    \and
    \TCLiftValRule
    \and
    \TCLiftValCPRule
    \and
    \TCValRule
    \and
    \TCLiftRule
    \and
    \TCLiftCPRule
    \and
    \TCInputAssignRule
    \and
    \TCInputDerefRule
    \and
    \TCIfRule
    \and
    \TCLetRule
    \and
    \TCComputeAndProveRule
    \and
    \TCCombinedRule
\end{rulesetpagebreakable}

\subsection{Theta-Lifting Definition}
\begin{mathparpagebreakable}
    \LiftTheta*{x} = x 
    \and
    \LiftTheta*{()} = ()
    \and
    \LiftTheta*{\TTrue} = \STrue
    \and
    \LiftTheta*{\TFalse} = \SFalse
    \and
    \LiftTheta*{\iota} = \SrcProg{r}_\ComputeProve(\Theta^{-1}(\iota)) 
    \and
    \LiftTheta*{\TNew*} = \SNew{C}{\LiftTheta*{\overline{v}}}
    \and
    \LiftTheta*{\TProofOfUsing*} = \SProofOfUsing{\alpha}{\LiftTheta*{\overline{v}}}
    \\
    \LiftTheta*{\TRefVal*} = \SRefVal{\ell}{\LiftTheta*{v}}
    \and
    \LiftTheta*{\TDeref*} = \SDeref{\LiftTheta*{v}}
    \and
    \LiftTheta*{\TAssign{v_1}{v_2}} = \SAssign{\LiftTheta*{v_1}}{\LiftTheta*{v_2}}
    \and
    \LiftTheta*{\TCast*} = \SCast{C}{\LiftTheta*{v}}
    \and
    \LiftTheta*{v.f} = \LiftTheta*{v}.f
    \and
    \LiftTheta*{v.m(\overline{u})} = \LiftTheta*{v}.m_\ell(\LiftTheta*{\overline{u}})
    \and
    \LiftTheta*{\TAlloc{\TTau}} = \SAlloc{\TTau}
    \and
    \LiftTheta*{\TLetIn{x\ty\TTau}{e_1}{e_2}} = \SLetIn{x\ty\tplab{\TTau}{\ell}}{\LiftTheta*{e_1}}{\LiftTheta*{e_2}}
    \and
    \LiftTheta*{\TIfThenElse{v}{e_1}{e_2}} = \SIfThenElse{\LiftTheta*{v}}{\LiftTheta*{e_1}}{\LiftTheta*{e_2}}
    \and
    \LiftTheta{\ell}{\TProveUsing*} = \SProveUsing* 
    \and
    \LiftTheta*{\TVerify*} = \SVerify*[v]
    \and
    \LiftTheta{\ell}{\TCnp*} = \SCnp*
    \and
    \LiftTheta{\ell}{\TCnpAdmin*{e}{\sigmaP}{\rho}} = \SCnpAdmin*{e}{\sigmaP}{\rho}
\end{mathparpagebreakable}

\section{Compilation}
\label{app:compilation}

\subsection{Projection Definition}
Recall the inputs to projection are typing judgments: $\denL{\Sigma ; \Gamma \cproves v : \tplab{\STau}{\ell}}$ for values and $\denL{\Sigma ; \Gamma ; \Delta ; A \cproves e : \tplab{\STau}{\ell'} \cproduces A'}$ for expressions, which we shorten to [\text{rulename}]$\denL{e}$.
\begin{mathparpagebreakable}[\small]
    [\ruleref{TC-LiftVal}]~\denFull{\Lift{\ell'}{v}} = \begin{cases}
        v & \ell \subsetl \ell' \\
        \TUnitVal & \ow
    \end{cases} 
    \and
    [\ruleref{TC-LiftValCP}]~\denFull{\Lift{\ComputeProve}{v}_\Theta} = v
    \and
    [\ruleref{TC-Val}]~\denFull{\Sigma ; \Gamma ; \Delta ; A \cproves v : \tplab{\STau}{\ell'} \cproduces A} = \begin{cases}
        v & \ell \subsetl \ell' \\
        \TUnitVal & \ow
    \end{cases}
    \and
    [\ruleref{TC-Lift}]~\denFull{\Lift{\ell'}{e}} = \begin{cases}
        e & \ell \subsetl \ell' \\
        \TUnitVal & \ow
    \end{cases} 
    \and
    [\ruleref{TC-LiftCP}]~\denFull{\Lift{\ComputeProve}{e}_\Theta} = e
    \and
    [\ruleref{TC-InputAssign}]~\denFull{\SWitAssign{x}{v}} = \begin{cases}
        \TAssign{x}{v} & \ell = \Compute \\
        \TUnitVal & \ell = \Prove
    \end{cases}
    \and
    [\ruleref{TC-InputDeref}]~\denFull{\SWitDeref{x}} = \begin{cases}
        \TDeref{x} & \ell = \Compute \\
        x & \ell = \Prove
    \end{cases}
    \and
    [\ruleref{TC-Let}]~\denFull{\SLetIn{x\ty\tplab{\TTau}{\ell'}}{\SExp_1}{\SExp_2}} = \begin{cases}
        \TLetIn{x\ty\TTau}{\denG{e_1}}{\denG{e_2}} & \ell \subsetl \ell' \\
        \denFull{e_2} & \ell \not\subsetl \ell' \land \denFull{e_1} = v \\
        \denFull{e_1} \Tseq \denFull{e_2} & \ow
    \end{cases}
    \and
    [\ruleref{TC-If}]~\denFull{\SIfThenElse*} = \begin{cases}
        \TIfThenElse{v}{\denG{e_1}}{\denG{e_2}} & \denFull{v} = v \\
        \denFull{e_1} & \denFull{v} = \TUnitVal \land \denFull{e_1} = \denFull{e_2} \\
        \diverges & \ow
    \end{cases}
    \and
    [\ruleref{TC-ComputeAndProve}]~\denFull{\SCnp*[e]}\\
    = \begin{cases}
        \Lift{\Compute}{\fullComp{\TCnp*[e]}} & \ell = \Compute \\
        () & \ell = \Prove
    \end{cases}
\end{mathparpagebreakable}

\subsection{Full Compilation Definition}
\begin{mathparpagebreakable}[\small]
    \fullComp{\TClassN~C_{\{\ell_C\}} ~\TExtendsN~D_{\{\ell_D\}} \{\overline{f} \ty \overline{\TTau} \seq K \seq \overline{M}\}} = \TClassN~C_{\{\ell_C\}}~\TExtendsN~D_{\{\ell_D\}} \{\overline{f} \ty \overline{\TTau} \seq K \seq \fullComp{\overline{M}} \}
    \and
    \fullComp{\TTau\{\ell\}~ m(\overline{x} \ty \overline{\TTau}) \{ e \}} = \TTau\{\ell\}~ m(\overline{x} \ty \overline{\TTau}) \{ \fullComp{e} \} 
    \and
    \fullComp{v} = v
    \and
    \fullComp{\TRefVal{v}} = \TRefVal{v}
    \and
    \fullComp{\TDeref{v}} = \TDeref{v}
    \and
    \fullComp{\TAssign{v_1}{v_2}} = \TAssign{v_1}{v_2}
    \and
    \fullComp{\TCast*} = \TCast*
    \and
    \fullComp{\TField*} = \TField*
    \and
    \fullComp{\TCall*} = \TCall*
    \and
    \fullComp{\TLetIn{x\ty\TTau}{e_1}{e_2}} = \TLetIn{x\ty\TTau}{\fullComp{e_1}}{\fullComp{e_2}}
    \and
    \fullComp{\TIfThenElse{v}{e_1}{e_2}} = \TIfThenElse{v}{\fullComp{e_1}}{\fullComp{e_2}}
    \and
    \fullComp{\TProveUsing*} = \TProveUsing*
    \and
    \fullComp{\TVerify*} = \TVerify*
    \\
    \begin{array}{rl}
        & \fullComp{\TCnp*[e]} = \\
        &\hspace{2em} \TLetIn{\overline{w_p}, \overline{w_s} \ty \Reft \overline{\TTau_{x_p}}, \Reft \overline{\TTau_{x_s}}}{\TAlloc{\overline{\TTau_{x_p}}, \overline{\TTau_{x_s}}}}{\big( 
            \subst*{\denC{e}}{{\overline{y_p}}{\overline{v_p}}{\overline{y_s}}{\overline{v_s}}{\overline{x_p}}{\overline{w_p}}{\overline{x_s}}{\overline{w_s}}} \Tseq
        } \\
        &\hspace{4em} \TProveUsing{\alpha}{\exists \overline{y_p}, \overline{x_p} [\overline{y_s}, \overline{x_s}]. \denP{e}}{\overline{v_p}, \TDeref{\overline{w_p}}[\overline{v_s}, \TDeref{\overline{w_s}}]} \big)
    \end{array}
\end{mathparpagebreakable}

\section{Proof Artifacts}
\label{app:proof-artifacts}

\subsection{Type Soundness Predicates}
\label{app:config-safety-def}
\begin{align*}
    \stuckpred{e}{\sigma} &= e = \Eval{\TCast{C}{(\TNew{D}{\TVals*})}} \text{ where } D \not<: C \\ 
    & \lor e = \Eval{\TProveUsing*[e']} \\
    &\hspace{10em}\text{ where } \conf{\subst*{e'}{{\overline{x}}{\overline{v}}{\overline{y}}{\overline{u}}}}{\varnothing} \text{ diverges or it steps to } \False \\
    & \lor e = \Eval{\TDeref{\iota}} \text{ where } \sigma(\iota) = \bot 
\end{align*}

\begin{align*}
    \stuckpredzk{\SExp}{\sigmaC}{\sigmaP}{\rho} &= \SExp = \Eval{\SCast{C}{(\SNew{D}{\SVals*})}} \text{ where } D \not<: C \\
    & \lor \SExp = \Eval{\SProveUsing*[e']} \\
    &\hspace{4em}\text{ where } \conf{\subst*{e'}{{\overline{x}}{\overline{v}}{\overline{y}}{\overline{u}}}}{\varnothing} \text{ diverges or it steps to } \False \\
    & \lor \SExp = \Eval{\SDeref{\SLocC*}} \text{ where } \sigmaC(\iota) = \bot \\
    & \lor \SExp = \Eval{\SCnpAdmin*{\SFalse}{\tilde{\sigmaP}}{\tilde{\rho}}} \\
    & \lor \SExp = \Eval{\SCnpAdmin*{\SExp'}{\tilde{\sigmaP}}{\tilde{\rho}}} \\
    &\hspace{4em}\text{ where } \SExp' \text{ itself gets stuck}
\end{align*}

\subsection{\zkstrudul Configuration Safety Definition}
We define $\Sigma \cproves \sigmaC, \sigmaP$ to be the following are true (note that $\SigmaCP |_\Compute$ and $\SigmaCP |_\Prove$ return the projected context where each pair $(a,b)$ mapping to a type $\TTau$ now maps $a$ and $b$ respectively to the type $\TTau$):
\begin{enumerate}[label=(\roman*)]
    \item $\dom(\sigmaC) \supseteq \dom(\SigmaC) \cupdot \dom(\SigmaCP |_\Compute)$ and $\dom(\sigmaP) \supseteq \dom(\SigmaP) \cupdot \dom(\SigmaCP |_\Prove)$
    \item for any $\iota \in \dom(\SigmaC)$, $\SigmaC \proves \sigmaC(\iota) : \SigmaC(\iota)$
    \item for any $\iota \in \dom(\SigmaP)$, $\SigmaP \proves \sigmaP(\iota) : \SigmaP(\iota)$
    \item for any $(\iota_1, \iota_2) \in \dom(\SigmaCP)$, $\SigmaCP |_\Compute \proves \sigmaC(\iota_1) : \SigmaCP(\iota_1, \iota_2)$
    \item for any $(\iota_1, \iota_2) \in \dom(\SigmaCP)$, $\SigmaCP |_\Prove \proves \sigmaP(\iota_2) : \SigmaCP(\iota_1, \iota_2)$;
    \item for any $(\iota_1, \iota_2) \in \dom(\SigmaCP)$: $\sigmaC ; \sigmaP \models \sigmaC(\iota_1) \approx \sigmaP(\iota_2)$; and
    \item for any $(\iota_1, \iota_2) \in \dom(\SigmaCP)$ such that $\sigmaC(\iota_1) = \iota_1'$ and $\sigmaP(\iota_2) = \iota_2'$: $(\iota_1', \iota_2') \in \dom(\SigmaCP)$
\end{enumerate}

We also define $\Delta ; A \cproves \rho$ to be that:
\begin{enumerate}[label=(\roman*)]
    \item $A = \dom(\rho) \subseteq \dom(\Delta)$
    \item for any $x \in \dom(\rho). \proves \rho(x) : \Delta(x)$
\end{enumerate}

Value synchronization $\approx$ is defined as follows:
\begin{mathparpagebreakable}
    \infer* {
        v \in \{\STrue, \SFalse, \SUnitVal \}
    } {
        \sigmaC, \sigmaP \models v \approx v
    }
    \and
    \infer* {
        \sigmaC, \sigmaP \models \sigmaC(\iota_1) \approx \sigmaP(\iota_2)
    } {
        \sigmaC, \sigmaP \models \iota_1 \approx \iota_2
    }
    \and
    \infer* {
        \sigmaC, \sigmaP \models \overline{v_1} \approx \overline{v_2}
    } {
        \sigmaC, \sigmaP \models \SProofOfUsing{\alpha}{\overline{v_1}} \approx \SProofOfUsing{\alpha}{\overline{v_2}}
    }
    \and
    \infer* {
        \sigmaC, \sigmaP \models \overline{v_1} \approx \overline{v_2}
    } {
        \sigmaC, \sigmaP \models \SNew{C}{\overline{v_1}} \approx \SNew{C}{\overline{v_2}}
    }
\end{mathparpagebreakable}

\subsection{Nested Preservation Property}
\label{app:nested-preservation}
\begin{mathparpagebreakable}
    \infer*{
        \TCnpN \not\in e \\
        \SCnpN \not\in e \\
        \SigmaP' = \SigmaCP' = \varnothing
    } {
        (e, \sigmaC, \SigmaC, \SigmaP', \SigmaCP', \Delta', A', S) \nestedok
    }
    \and
    \infer*{
        e = \TCnpAdmin*{e'}{\sigmaP'}{\rho'} \\
        (e', \sigmaC, \SigmaC, \SigmaP'', \SigmaCP'', \Delta'', A'', S \cup \dom(\SigmaCP' |_\Compute)) \nestedok \\\\
        \dom(\Delta') = \{\overline{x_p}, \overline{x_s}\} \\
        \SigmaC, \SigmaP', \SigmaCP' \cproves \sigmaC, \sigmaP' \\
        \Delta' ; A' \cproves \rho' \\\\
        \dom(\SigmaCP' |_\Compute) \cap S = \varnothing \\
        (\SigmaC, \SigmaP', \SigmaCP'); \varnothing ; \Delta' ; A' \cproves e' : \tplab{\SBool}{\Prove} \cproduces \{\overline{x_p}, \overline{x_s}\}
    } {
        (e, \sigmaC, \SigmaC, \SigmaP', \SigmaCP', \Delta', A', S) \nestedok
    }
    \and
    \infer*{
        e = \SCnpAdmin*{e'}{\sigmaP'}{\rho'} \\
        (e', \sigmaC, \SigmaC, \SigmaP'', \SigmaCP'', \Delta'', A'', S \cup \dom(\SigmaCP' |_\Compute)) \nestedok \\\\
        \dom(\Delta') = \{\overline{x_p}, \overline{x_s}\} \\
        \SigmaC, \SigmaP', \SigmaCP' \cproves \sigmaC, \sigmaP' \\
        \Delta' ; A' \cproves \rho' \\\\
        \dom(\SigmaCP' |_\Compute) \cap S = \varnothing \\
        (\SigmaC, \SigmaP', \SigmaCP'); \varnothing ; \Delta' ; A' \cproves e' : \tplab{\SBool}{\Prove} \cproduces \{\overline{x_p}, \overline{x_s}\}
    } {
        (e, \sigmaC, \SigmaC, \SigmaP', \SigmaCP', \Delta', A', S) \nestedok
    }
\end{mathparpagebreakable}

\subsection{Parallel Projected Language}
\label{app:target-conc-lang}

\subsection*{Syntax}
{
    \[
    \begin{array}{rcl}
        x,\underred{w},\alpha &\in& \mathcal{V} \quad \text{(variable, computed input, proof names)} \\
        \TTau & \Coloneqq & \Int \alt \Bool \alt \Unit \alt \Reft \TTau  \alt C \alt \ProofOf* \alt \Proof \\
        v & \Coloneqq & x \alt () \alt \True \alt \False \alt \New C(\overline{v}) \alt \iota \alt \ProofOfN~\alpha~\WithN~\overline{v} \\
        e & \Coloneqq & v \alt \Reft v \alt \bang v \alt v_1 \coloneq v_2 \alt \underred{x \leftarrow v} \alt \underred{\witderef x} \alt (C) v \alt v.f \alt v_1.m(\overline{v}) \alt \Alloc{\TTau} \\
        & \alt & \Let x\ty\TTau = e_1 \In e_2 \alt \If v \Then e_1 \Else e_2 \\ 
        & \alt & \ProveWith* \alt \VerifyWith*{v}
    \end{array}
  \]
}

\subsection*{Operational Semantics}

\begin{rulesetpagebreakable}
    \SetRuleLabelLoc{lab}
    \SetRuleLabelVCenter
    \EPLetOneC
    \and
    \EPLetOneP
    \and
    \EPStepC
    \and
    \EPStepP
    \and
    \EPWitAssign
    \and
    \EPWitDerefC
    \and
    \EPWitDerefP
\end{rulesetpagebreakable}

\subsection{Annotated \zkstrudel Operational Semantics}
\label{app:annotated-semantics}
This semantics is very similar to the semantics in \Cref{app:combined-semantics} but with labels annotated, which are used for defining an annotated projection in \Cref{app:intermediate-projections}.

\begin{rulesetpagebreakable}
    \SetRuleLabelLoc{lab}
    \SetRuleLabelVCenter
    \EALetOne
    \and
    \EALetTwo
    \and
    \EAIfT
    \and
    \EAIfF
    \and
    \EARefC
    \and
    \EARefP
    \and
    \EARefCP
    \and
    \EADerefC
    \and
    \EADerefP
    \and
    \EADerefCP
    \and
    \EAAssignC
    \and
    \EAAssignP
    \and
    \EAAssignCP
    \and
    \EACast
    \and
    \EAField
    \and
    \EACall
    \and
    \EAInputAssign
    \and
    \EAInputDeref
    \and
    \EAProve
    \and
    \EAVerifyT
    \and
    \EAVerifyF
\end{rulesetpagebreakable}

\subsection{Type-Directed Annotations Definition}\label{app:ann-types-def}
\begin{mathparpagebreakable}\\
    [\ruleref{TC-LiftVal}]~\FAnn{\Lift*{v}} = \LiftAnn*{v}
    \and
    [\ruleref{TC-LiftValCP}]~\FAnn{\LiftTheta{\ComputeProve}{v}} = \LiftAnnTheta*{v}
    \and
    [\ruleref{TC-Val}]~\FAnn{v : \tplab{\STau}{\ell}} = \LiftAnn*{v}
    \and
    [\ruleref{TC-Lift}]~\FAnn{\Lift*{e}} = \LiftAnn*{e}
    \and
    [\ruleref{TC-LiftCP}]~\FAnn{\LiftTheta*{e}} = \LiftAnnTheta*{e}
    \and
    [\ruleref{TC-InputAssign}]~\FAnn{x \witassign v} = (x \witassign \FAnn(v))_\ComputeProve
    \and
    [\ruleref{TC-InputDeref}]~\FAnn{\witderef x} = (\witderef x)_\ComputeProve
    \and
    [\ruleref{TC-Let}]~\FAnn{\LetIn{x\ty\tplab{\STau_1}{\ell_1}}{e_1}{e_2} : \tplab{\STau_2}{\ell_2}} = (\LetIn{x\ty\tplab{\STau_1}{\ell_1}}{\FAnn{e_1}}{\FAnn{e_2}})_{\annell_2}
    \and
    [\ruleref{TC-If}]~\FAnn{\IfThenElse{v}{e_1}{e_2} : \tplab{\STau}{\ell}} = (\IfThenElse{\FAnn{v}}{\FAnn{e_1}}{\FAnn{e_2}})_\annell
    \and
    [\ruleref{TC-ComputeAndProve}]~\FAnn{\Cnp*} \\\hspace{4em}= \Cnp{\overline{x_p}\ty\overline{\tau_{x_p}}}{\overline{x_s}\ty\overline{\tau_{x_s}}}{\alpha}
  {\overline{y_p}\ty\overline{\tau_{y_p}}}{\overline{y_s}\ty\overline{\tau_{y_s}}}
  {e}{\FAnn{\overline{v_p}}}{\FAnn{\overline{v_s}}}_\Compute
    \and
    [\ruleref{TC-Combined}]~\FAnn{\CnpAdmin*{e}{\sigmaP}{\rho}} = \CnpAdmin{\sigmaP}{\rho}{\overline{x_p}}{\overline{x_s}}{\varphi}{e}{\FAnn{\overline{v_p}}[\FAnn{\overline{v_s}}]}_\Compute
\end{mathparpagebreakable}

\subsection{Annotated Lifting Definition}\label{app:ann-lifting-def}
\begin{mathparpagebreakable}
    \LiftAnn*{x} = (x)_\annell
    \and
    \LiftAnn*{\True} = (\True)_\annell
    \and
    \LiftAnn*{\False} = (\False)_\annell
    \and
    \LiftAnn*{\New C(\overline{v})} = (\New C(\LiftAnn*{\overline{v}}))_\annell
    \and
    \LiftAnn*{\ProofOfN~\alpha~\UsingN~\overline{v}} = (\ProofOfN~\alpha~\UsingN~\LiftAnn*{\overline{v}})_\annell
    \and
    \LiftAnn*{\Reft v} = (\Refl \LiftAnn*{v})_\annell
    \and
    \LiftAnn*{\bang v} = (\bang \LiftAnn*{v})_\annell
    \and
    \LiftAnn*{v_1 \coloneq v_2} = (\LiftAnn*{v_1} \coloneq \LiftAnn*{v_2})_\annell
    \and
    \LiftAnn*{(C)v} = ((C)\LiftAnn*{v})_\annell
    \and
    \LiftAnn*{v.f} = (\LiftAnn*{v}.f)_\annell
    \and
    \LiftAnn*{v.m(\overline{u})} = (\LiftAnn*{v}.m_\ell(\LiftAnn*{\overline{u}}))_\annell
    \and
    \LiftAnn*{\LetIn{x\ty\TTau}{e_1}{e_2}} = (\LetIn{x\ty\tplab{\TTau}{\ell}}{\LiftAnn*{e_1}}{\LiftAnn*{e_2}})_\annell
    \and
    \LiftAnn*{\Alloc{\TTau}} = (\Alloc{\TTau})_\annell
    \and
    \LiftAnn*{\IfThenElse{v}{e_1}{e_2}} = (\IfThenElse{\LiftAnn*{v}}{\LiftAnn*{e_1}}{\LiftAnn*{e_2}})_\annell
    \and
    \LiftAnn*{\ProveWith*} = (\Provet \alpha = \CircuitDefault \With \LiftAnn*{\overline{v}}[\LiftAnn*{\overline{u}}])_\annell 
    \and
    \LiftAnn*{\VerifyWith*{v}} = (\Verify \LiftAnn*{v} \Proves \alpha \With \LiftAnn*{\overline{u}})_\annell
    \and
    \LiftAnn*{\Cnp*} = (\Cnp{\overline{x_p}\ty\overline{\TTau_{x_p}}}{\overline{x_s}\ty\overline{\TTau_{x_s}}}{\alpha}
                                 {\overline{y_p}\ty\overline{\TTau_{y_p}}}{\overline{y_s}\ty\overline{\TTau_{y_s}}}
                                 {e}{\LiftAnn{\Compute}{\overline{v_p}}}{\LiftAnn{\Compute}{\overline{v_s}}})_\annell
\end{mathparpagebreakable}

\subsection{Annotated Theta Lifting Definition}
\begin{mathparpagebreakable}
    \LiftAnnTheta*{x} = (x)_\annell
    \and
    \LiftAnnTheta*{\True} = (\True)_\annell
    \and
    \LiftAnnTheta*{\False} = (\False)_\annell
    \and
    \LiftAnnTheta*{\New C(\overline{v})} = (\New C(\LiftAnnTheta*{\overline{v}}))_\annell
    \and
    \LiftAnnTheta*{\ProofOfN~\alpha~\UsingN~\overline{v}} = (\ProofOfN~\alpha~\UsingN~\LiftAnnTheta*{\overline{v}})_\annell
    \and
    \LiftAnnTheta*{\iota} = (\SrcProg{r}_\ComputeProve(\Theta^{-1}(\iota)))_\ComputeProve
    \and
    \LiftAnnTheta*{\Reft v} = (\Refl \LiftAnnTheta*{v})_\annell
    \and
    \LiftAnnTheta*{\bang v} = (\bang \LiftAnnTheta*{v})_\annell
    \and
    \LiftAnnTheta*{v_1 \coloneq v_2} = (\LiftAnnTheta*{v_1} \coloneq \LiftAnnTheta*{v_2})_\annell
    \and
    \LiftAnnTheta*{(C)v} = ((C)\LiftAnnTheta*{v})_\annell
    \and
    \LiftAnnTheta*{v.f} = (\LiftAnnTheta*{v}.f)_\annell
    \and
    \LiftAnnTheta*{v.m(\overline{u})} = (\LiftAnnTheta*{v}.m_\ell(\LiftAnnTheta*{\overline{u}}))_\annell
    \and
    \LiftAnnTheta*{\LetIn{x\ty\TTau}{e_1}{e_2}} = (\LetIn{x\ty\tplab{\TTau}{\ell}}{\LiftAnnTheta*{e_1}}{\LiftAnnTheta*{e_2}})_\annell
    \and
    \LiftAnnTheta*{\Alloc{\TTau}} = (\Alloc{\TTau})_\annell
    \and
    \LiftAnnTheta*{\IfThenElse{v}{e_1}{e_2}} = (\IfThenElse{\LiftAnnTheta*{v}}{\LiftAnnTheta*{e_1}}{\LiftAnnTheta*{e_2}})_\annell
    \and
    \LiftAnnTheta*{\ProveWith*} = (\Provet \alpha = \CircuitDefault \With \LiftAnnTheta*{\overline{v}}[\LiftAnnTheta*{\overline{u}}])_\annell 
    \and
    \LiftAnnTheta*{\VerifyWith*{v}} = (\Verify \LiftAnnTheta*{v} \Proves \alpha \With \LiftAnnTheta*{\overline{u}})_\annell
    \and
    \LiftAnnTheta*{\Cnp*} = (\Cnp{\overline{x_p}\ty\overline{\TTau_{x_p}}}{\overline{x_s}\ty\overline{\TTau_{x_s}}}{\alpha}
                                 {\overline{y_p}\ty\overline{\TTau_{y_p}}}{\overline{y_s}\ty\overline{\TTau_{y_s}}}
                                 {e}{\LiftAnnTheta{\Compute}{\overline{v_p}}}{\LiftAnnTheta{\Compute}{\overline{v_s}}})_\annell
\end{mathparpagebreakable}

\subsection{Annotated Projection Definitions}
\label{app:intermediate-projections}
\begin{mathparpagebreakable}[\small]
    \denL{(v)_\annellp} = \begin{cases}
        v & \ell \subsetl \annellp \\
        () & \ow
    \end{cases}
    \and
    \denL{\Refsub{\ell'} (v)_\annellpp} = \begin{cases}
        \Reft v & \ell \subsetl \ell' \\
        () & \ow
    \end{cases}
    \and
    \denL{\bang (v)_\annellp} = \begin{cases}
        \bang v & \ell \subsetl \annellp \\
        () & \ow
    \end{cases}
    \and
    \denL{(r_{\ell'}(\iota))_\annellp} = \begin{cases}
        \iota & \ell \subsetl \ell' \\
        () & \ow
    \end{cases}
    \and
    \denL{(r_\ComputeProve(\iota_\Compute, \iota_\Prove))_\anncp} = \begin{cases}
        \iota_\Compute & \ell = \Compute \\
        \iota_\Prove & \ell = \Prove
    \end{cases}
    \and
    \denL{r \coloneq (v)_\annellp} = \begin{cases}
        \iota \coloneq v & r = (r_\ell(\iota))_\annell \\
        \iota_\Compute \coloneq v & \ell = \Compute \land r = (r_\ComputeProve(\iota_\Compute, \iota_\Prove))_\anncp \\
        \iota_\Prove \coloneq v & \ell = \Prove \land r = (r_\ComputeProve(\iota_\Compute, \iota_\Prove))_\anncp \\
        () & \ow
    \end{cases}
    \and
    \denL{(C)(v)_\annellp} = \begin{cases}
        (C) v & \ell \subsetl \annellp \\
        () & \ow
    \end{cases}
    \and
    \denL{(v)_\annellp .f} = \begin{cases}
        v.f & \ell \subsetl \annellp \\
        () & \ow
    \end{cases}
    \and
    \denL{(v)_\annellp.m(\overline{v})} = \begin{cases}
        v.m(\overline{v}) & \ell \subsetl \annellp \\
        () & \ow
    \end{cases}
    \and
    \denL{\Alloc{\TTau}} = \begin{cases}
        \Alloc{\TTau} & \ell = \Compute \\
        () & \ell = \Prove
    \end{cases}
    \and
    \denL{\witderef x} = \witderef x
    \and
    \denL{x \witassign (v)_\annellp} = \begin{cases}
        x \witassign v & \ell = \Compute \\
        () & \ell = \Prove
    \end{cases}
    \and
    \denL{\Let x\ty\tplab{\tau}{\ell'} = e_1 \In e_2} = \begin{cases}
        \Let x\ty\tau = \denL{e_1} \In \denL{e_2} & \ell \subsetl \ell' \\
        \denL{e_2} & \ell \not\subsetl \ell' \land \denL{e_1} = v \\
        \denL{e_1} ; \denL{e_2} & \ow
    \end{cases}
    \and
    \denL{\If (v)_\annellp \Then e_1 \Else e_2} = \begin{cases}
        \If v \Then \denL{e_1} \Else \denL{e_2} & \ell \subsetl \annellp \\
        \denL{e_1} & \ell \not\subsetl \annellp \land \denL{e_1} = \denL{e_2} \\
        \diverges & \ow
    \end{cases}
    \and
    \denL{\Provet \alpha = \CircuitDefault \With \overline{v}[\overline{w}]} = \begin{cases}
        \Provet \alpha = \CircuitDefault \With \overline{v}[\overline{w}] & \ell = \Compute \\
        () & \ell = \Prove
    \end{cases}
    \and
    \denL{\Verify (v)_\annellp \Proves \alpha \With \overline{u}} = \begin{cases}
        \Verify v \Proves \alpha \With \overline{u} & \ell \subsetl \annellp \\
        () & \ow
    \end{cases}
    \and
    \denL{e_{comb} = \Cnp*} = \begin{cases}
        \LiftAnn{\Compute}{\fullComp{e_{comb}}} & \ell = \Compute \\
        () & \ell = \Prove
    \end{cases}
\end{mathparpagebreakable}

\subsection{Computed Input Translation Definition}\label{app:witness-translation-def}
\begin{mathparpagebreakable}
    \witdenL{v} = v
    \and
    \witdenL{\Reft v} = \Reft v
    \and
    \witdenL{\bang v} = \bang v
    \and
    \witdenL{v_1 \coloneq v_2} = v_1 \coloneq v_2
    \and
    \witdenL{(C)v} = (C)v
    \and
    \witdenL{v.f} = v.f
    \and
    \witdenL{v.m(\overline{w})} = v.m(\overline{w})
    \and
    \witdenL{\LetIn{x\ty\tau}{e_1}{e_2}} = \LetIn{x\ty\tau}{\witdenL{e_1}}{\witdenL{e_2}}
    \and
    \witdenL{\IfThenElse{v}{e_1}{e_2}} = \IfThenElse{v}{\witdenL{e_1}}{\witdenL{e_2}}
    \and
    \witdenL{\Alloc{\TTau}} = \Alloc{\TTau}
    \\
    \witdenL{x \witassign v} = \begin{cases}
        \varphi(x) \coloneq v & \ell = \Compute \\
        () & \ow
    \end{cases}
    \and
    \witdenL{\witderef x} = \begin{cases}
        \bang \varphi(x) & \ell = \Compute \\
        \varphi(x) & \ell = \Prove
    \end{cases}
    \and
    \witdenL{\Provet \alpha = \CircuitDefault \With \overline{v}[\overline{u}]} = \Provet \alpha = \CircuitDefault \With \overline{v}[\overline{u}]
    \and
    \witdenL{ \Verifyt v \Proves \alpha \With \overline{u}} = \Verifyt v \Proves \alpha \With \overline{u}
\end{mathparpagebreakable}

\subsection{Less-Than Relation on Parallel Projected Syntax}\label{app:lessthan-def}
\begin{mathparpagebreakable}
    \infer*{
        e_1 \lessthan e_2
    }{
        e_1 \lessthan v ; e_2
    }
    \and
    \infer*{ }{
        v \lessthan v
    }
    \and
    \infer*{ }{
        \Reft v \lessthan \Reft v
    }
    \and
    \infer*{ }{
        \bang v \lessthan \bang v
    }
    \and
    \infer*{ }{
        v_1 \coloneq v_2 \lessthan v_1 \coloneq v_2
    }
    \and
    \infer*{ }{
        x \witassign v \lessthan x \witassign v
    }
    \and
    \infer*{ }{
        \witderef x \lessthan \witderef x
    }
    \and
    \infer*{ }{
        (C)v \lessthan (C)v
    }
    \and
    \infer*{ }{
        v.f \lessthan v.f
    }
    \and
    \infer*{ }{
        v.m(\overline{v}) \lessthan v.m(\overline{v})   
    }
    \and
    \infer*{ } {
        \Alloc{\TTau} \lessthan \Alloc{\TTau}
    }
    \and
    \infer*{
        e_1 \lessthan e_1'\\
        e_2 \lessthan e_2'
    } {
        \LetIn{x\ty\tau}{e_1}{e_2} \lessthan \LetIn{x\ty\tau}{e_1'}{e_2'}
    }
    \and
    \infer*{
        e_1 \lessthan e_1' \\
        e_2 \lessthan e_2'
    }{
        \IfThenElse{v}{e_1}{e_2} \lessthan \IfThenElse{v}{e_1'}{e_2'}
    }
    \and
    \infer*{ }{
        \Provet \alpha = \CircuitDefault \With \overline{v}[\overline{w}] \lessthan \Provet \alpha = \CircuitDefault \With \overline{v}[\overline{w}]
    }
    \and
    \infer*{ }{
        \Verifyt p \Proves \alpha \With \overline{v} \lessthan \Verifyt p \Proves \alpha \With \overline{v}
    }
\end{mathparpagebreakable}

\subsection{Trace Equivalence Definition}
\begin{mathparpagebreakable}
    \infer* { } {
        t \traceq t
    }
    \and
    \infer* { } {
        t \traceq t :: \Nulltrace
    }
    \and
    \infer* { 
        t_1 \traceq t_2
    } {
        t_2 \traceq t_1
    }
    \and
    \infer* {
        t_1 \traceq t_2
    } {
        t_1 :: t \traceq t_2 :: t
    }
    \and
    \infer* {
        t_1 \traceq t_2 \\
        t_2 \traceq t_3
    } {
        t_1 \traceq t_3
    }
\end{mathparpagebreakable}

\section{Proofs}
We order the proofs based on dependencies: if proof A is used in proof B, then proof A appears before proof B.

\subsection{Type Soundness Theorems}

\begin{lemma}[Cnp Always in Head]\label{lemma:cnp-always-head}
    A $\CnpN$ statement will always appear in head position.
\end{lemma}
\begin{proof}
    By induction on the semantic steps.

    The base case is the surface syntax, which means no $\CnpN$ terms can appear.

    The inductive case is assuming after $k$ steps, proving for $k+1$ step.
    The inductive hypothesis mean there is a $\CnpN$ statement in the head position or no $\CnpN$ statement at all.
    If it's in head position, then either \ruleref{E-ComputeAndProveStep} steps to another $\CnpN$ statement in head position, or \ruleref{E-ComputeAndProveTrue} steps to a proof and so no $\CnpN$ statements appear.
    If there is no $\CnpN$ statement, onyl the \ruleref{E-ComputeAndProveInit} can create one, which is in head position.
\end{proof}

\begin{lemma}[Simultaneous Substitution Preserves \corelang Typing]\label{lemma:subst-pres-core-typing}
    If $\Sigma ; \Gamma_1 \proves e : \TTau$ and $\sigma$ is a variable substitution such that
    $\Sigma ; \Gamma_2 \proves \sigma(x) : \TTau'$ for all variables $x : \TTau' \in \Gamma_1$ and $e[\sigma]$ is a defined substitution,
    then $\Sigma ; \Gamma_2 \proves e[\sigma] : \TTau$.
\end{lemma}
\begin{proof}
    By induction.
    The only case that isn't standard is substituting into a $\TProveN$ statement, where substitution only changes the input values and the body of the expression, which follows from the inductive hypothesis.
\end{proof}

\begin{corollary}[Substitution Preserves \corelang Typing]\label{cor:core-safe-sub}
    If $\Sigma ; \Gamma, x : \TTau' \proves e : \TTau$ and $\Sigma ; \Gamma \proves v : \TTau'$ and $e[x \mapsto v]$ is a defined substitution, then $\Sigma ; \Gamma \proves e[x \mapsto v] : \TTau$.
\end{corollary}

\begin{theorem}[Preservation of \corelang]
\label{thm:core-preservation}
    If $\Sigma ; \Gamma \proves e : \TTau \produces \ell$ and $\Sigma \proves \sigma$, and $\langle e \divi \sigma \rangle \tstep \langle e' \divi \sigma' \rangle$, then for some $\Sigma' \supseteq \Sigma$, we know $\Sigma' ; \Gamma \proves e' : \TTau' \produces \ell'$ where $\TTau <: \TTau'$ and $\Sigma' \proves \sigma'$.
\end{theorem}
\begin{proof}
    By induction on $e$.
    \begin{itemize}[itemsep=0.6em] 
        \item If $e$ is a value, then the statement is vacuously true.
        
        \item In most cases, the proof follows by inverting the typing rule and using the available premises to prove the typing rule for $e'$.
        In all cases, the updated $\Sigma'$ is either $\Sigma$ itself or $\Sigma \cup \{\iota \mapsto \TTau \}$ updated with a single new reference.
        Also note that all substitutions performed are type-preserving by \Cref{lemma:subst-pres-core-typing}.

        \item If $e = \TCall*$ steps, then it remains well-typed by the method body being well-typed (\ruleref{Method-Ok}) and the substitution lemma (\Cref{cor:core-safe-sub}).

        \item If $e = \LetIn{x\ty\TTau}{e_1}{e_2}$ steps to $\LetIn{x\ty\TTau}{e_1'}{e_2}$, then we apply the inductive hypothesis on the subexpression $e_1$ taking a step to $e_2$ and use that implied typing judgment to show $\Sigma' ; \Gamma \proves \LetIn{x\ty\TTau}{e_1'}{e_2} : \TTau_2$.
        
        \item If $e = \LetIn{x\ty\TTau}{e_1}{e_2}$ where $e_1$ is a value $v$, then the stepped expression $e_2[x \mapsto v]$ is a well-typed substitution by \Cref{cor:core-safe-sub}.
        
        \item In both cases for if, the typing rule gives us that either stepping to $e_1$ or $e_2$ is well-typed.
        
        \item Suppose $e = \ProveWith*[e']$ and $\conf*{e} \tstep \conf*{\ProofOf* \With \overline{v}}$.
        Note that $\Sigma ; \Gamma \proves e : \ProofOf*$ by the premise, and $\Sigma ; \Gamma \proves \ProofOf* \With \overline{v} : \ProofOf*$ directly.
        
        \item Suppose $e = \VerifyWith*{e'}$ and $\Sigma ; \Gamma \proves e : \Bool$.
        For $e$ to step means \ruleref{E-VerifyT} must have stepped, which would result in $\True$ which types to a bool,
        or \ruleref{E-VerifyF} which also types to a bool ($\False$).

    \end{itemize}
\end{proof}

\begin{theorem}[Progress of \corelang]
\label{thm:core-progress}
    Suppose that $\Sigma ; \varnothing \proves e : \TTau \cproduces \ell$.
    Then either 
    \begin{enumerate}[label=(\roman*)]
        \item $e$ is a value;
        \item for any store $\sigma$ where $\Sigma \proves \sigma$, there is some $e'$ and $\sigma'$ such that $\conf*{e} \tstep \conf{e'}{\sigma'}$;
        \item $e = \Eval{\TProveUsing*[e']}$ where $\conf{\subst*{e'}{{\overline{x}}{\overline{v}}{\overline{y}}{\overline{u}}}}{\varnothing} \tstep^* \conf{\False}{\_}$ or it diverges;
        \item $e = \Eval{\TCast{C}{(\TNew{D}{\TVals*})}}$ where $D \not<: C$; or
        \item $e = \Eval{\TDeref{\iota}}$ where $\sigma(\iota) = \bot$
    \end{enumerate}
\end{theorem}
\begin{proof}
    By induction on $e$.
    \begin{itemize}[itemsep=0.6em]  
        \item If $e$ is a value, we fall into case \rone and are done.
        
        \item In most cases, $e$ can take a step by simply proving it satisfies the semantic form.
        For instance, since the context $\Gamma$ is assumed empty, all values who type to class $C$ must be a class object, and any value typing to $\iota$ must a location $\iota$ in $\dom{\Sigma}$ and therefore in $\sigma$.

        \item If $e = \TDeref{\iota}$, by $\Sigma \proves \sigma$ we know $\iota \in \dom(\sigma)$.
        Either $\sigma(\iota) = \bot$, where we're in stuck case \rfive, or it's a value in which case we take a step (case \rtwo).
        
        \item If $e = (C)(\New D(\overline{v}))$, then a step can be taken if $D <: C$ (satisfying \rtwo), or we get stuck in case \rfour. 
        
        \item If $e = \LetIn{x\ty\TTau}{e_1}{e_2}$, by its typing rule $\Sigma ; \varnothing \proves e_1 : \TTau_1$.
        By the inductive hypothesis on $e_1$, either \rone $e_1$ is a value (in which case we apply the second Let rule), \rtwo $e_2$ can take a step (in which case we apply the first Let rule), or it gets stuck and so we're in \rthree.
        
        \item If $e = \IfThenElse{v}{e_1}{e_2}$, we know that $\Sigma ; \emptyset \proves v : \Bool$, meaning it's either $\True$ or $\False$.
        In the former case, we apply the IfT rule and in the latter the IfF rule.
        In either case, $\Sigma ; \varnothing \proves e_1 : \TTau$ and $\Sigma ; \varnothing \proves e_2 : \TTau$, so we can apply the inductive hypothesis in either of these cases.
        
        \item If $e = \ProveWith*[e']$, then $\cdot ; \overline{x} : \overline{\TTau_x}, \overline{y} : \overline{\TTau_y} \proves e' : \Bool$.
        This then means that $\subst*{e'}{{\overline{x}}{\overline{v}}{\overline{y}}{\overline{u}}}$ which is a well-typed substitution by \Cref{lemma:subst-pres-core-typing}.
        By the inductive hypothesis, $e'$ is either \rone a value, \rtwo can take a step, \rthree gets stuck on a nested prove, \rfour gets stuck on an invalid cast, or \rfive gets stuck dereferencing an uninitialized location.
        For \rone, the expression is either $\TTrue$ or $\TFalse$ since it types to a bool, so it either steps returning a proof (\ruleref{E-Prove}) or gets stuck in \rthree.
        \ruleref{T-Prove} gives that $\cdot ; \Gamma \proves e : \TBool \produces \Prove$, and $\cdot \proves \emptyset$ so we can invoke the inductive hypothesis and take a step, inductively turning the crank until we reach a value (which is a bool) or is a failure condition.
        For \rthree, \rfour, or \rfive it diverges.
        
        \item If $e = \VerifyWith*{v}$, then either the \ruleref{E-VerifyT} or \ruleref{E-VerifyF} rules must be possible, since the only other option is the proof names don't match up which is disallowed by the type system.
    \end{itemize}
\end{proof}

\coresoundness*
\begin{proof}
    We outline the cases explicitly here:
    \begin{enumerate}[nosep,label={(\roman*)},leftmargin=*]
        \item $e'$ is a value;
        \item there is some~$e''$ and~$\sigma'$ such that $\conf*{e'} \tstep \conf{e''}{\sigma'}$;
        \item $e' = \Eval{\TCast{C}{(\TNew{D}{\TVals*})}}$ where $D \not<: C$;
        \item $e' = \Eval{\TProveUsing*[e'']}$ where either $\conf{\subst*{e''}{{\overline{x}}{\overline{v}}{\overline{y}}{\overline{u}}}}{\varnothing}$ diverges or it steps to~$\False$; or
        \item $e' = \Eval{\TDeref{\iota}}$ where $\sigma(\iota) = \bot$
    \end{enumerate}

    The proof is by induction, applying progress (\Cref{thm:core-progress}) and preservation (\Cref{thm:core-preservation}).
\end{proof}

\begin{lemma}[Simultaneous Substitution Preserves \zkstrudel Typing]\label{thm:subst-pres-surf-typing}
    If $\Sigma ; \Gamma_1 ; \Delta ; A \proves e : \tplab{\TTau}{\ell} \produces A'$ and $\sigma$ is a variable substitution such that
    $\Sigma ; \Gamma_2 \proves \sigma(x) : \tplab{\TTau'}{\ell'}$ for all variables $x : \tplab{\TTau'}{\ell'} \in \Gamma_1$ and $e[\sigma]$ is a defined substitution,
    then $\Sigma ; \Gamma_2 ; \Delta ; A \proves e[\sigma] : \tplab{\TTau}{\ell} \produces A'$.
\end{lemma}
\begin{proof}
    By induction.
    The non-standard cases are substituting into a $\ComputeN$ statement and a $\CnpN$ statement.
    In the former, substitution only changes the input values and the body of the expression, which can be proven to be well-typed.
    In the latter, the substitution is undefined, but can never occur by \Cref{lemma:cnp-always-head}.
\end{proof}

\begin{corollary}[Substitution Preserves \zkstrudel Typing]\label{cor:safe-sub}
    If $\Sigma ; \Gamma, x : \tplab{\TTau'}{\ell'} ; \Delta ; A \proves e : \tplab{\TTau}{\ell} \produces A'$ and $\Sigma ; \Gamma ; \Delta ; A \proves v : \tplab{\TTau'}{\ell'}$ and $e[x\mapsto v]$ is a defined substitution, then $\Sigma ; \Gamma ; \Delta ; A \proves e[x \mapsto v] : \tplab{\TTau}{\ell} \produces A$.
\end{corollary}

\begin{lemma}\label{lemma:raise-and-lower-types}
    Suppose that $\Sigma ; \cdot \cproves v : \tplab{\TTau}{\ell}$.
    Then $\Sigma ; \cdot \cproves \Lift{\Compute}{\Lower{\Compute}{v}} : \tplab{\TTau}{\Compute}$, and $\Sigma ; \Gamma \cproves \Lift{\Prove}{\Lower{\Prove}{v}} : \tplab{\TTau}{\Prove}$.
\end{lemma}
\begin{proof}
    By induction on $v$.

    The lemma is vacuously true for $v$ being a variable.
    For any constant, lowering and lifting results in itself $v$: $\Lift{\Compute}{\Lower{\Compute}{v}} = v$, and their respective typing rules can be typed with both $\Compute$ or $\Prove$.
    For $v = \New C(\overline{v})$, the inductive hypothesis applies to $\overline{v}$ and the result for $v$ follows.
    Similar reasoning proves the case for $v = \ProofOf* \With \overline{v}$.
    If $v = r_\Compute(\iota)$, then $\Sigma ; \cdot \cproves v : \tplab{(\Reft \tplab{\TTau}{\ell})}{\ell}$. 
    Then $\Lift{\Compute}{\Lower{\Compute}{r_\Compute(\iota)}} = \Lift{\Compute}{\iota} = r_\Compute(\iota)$ and so $\Sigma ; \cdot \cproves r_\Compute(\iota) : \tplab{(\Reft \tplab{\TTau}{\ell})}{\ell}$.
    Similar reasoning holds for proving the case for lifting/lowering to $\Prove$.
    Finally, if $v = r_\ComputeProve(\iota_1, \iota_2)$, then $\Lift{\Compute}{\Lower{\Compute}{v}} = \Lift{\Compute}{\iota_1} = r_\Compute(\iota_1)$ (and similarly for $r_\Prove(\iota_2)$).

    Note that even though we proved it, this very last case is actually disallowed by our type system, and so would never actually occur.
\end{proof}

\begin{lemma}\label{lemma:c-and-p-types}
    If $\cdot ; \cdot \cproves v : \tplab{\TTau}{\Compute}$ and $\cdot ; \cdot \cproves v : \tplab{\TTau}{\Prove}$, then $\cdot ; \cdot \cproves v : \tplab{\TTau}{\ComputeProve}$.
\end{lemma}
\begin{proof}
    By induction on $v$.

    The lemma is vacuously true when $v$ is a variable or a location.
    For any constant, they can be typed with label $\ComputeProve$ without the premises.
    We then get the result for $\New C(\overline{v})$ and $\ProofOf* \With \overline{v}$ by applying the inductive hypothesis on $\overline{v}$ and applying the relevant typing rule.
\end{proof}

\begin{lemma}\label{lemma:lifting-v-preserves-type}
    If $\refunreach{\TTau}$, then $\proves v : \TTau$ if and only if $\cproves v : \tplab{\TTau}{\ell}$ for any $\ell$.
\end{lemma}
\begin{proof}
        We prove the forward direction by induction on the typing rules of $v$.
        Variables and reference rules are vacuously true by not satisfying the premise.
        All constants are trivially true.
        Finally, classes and proof objects follow by applying the inductive hypothesis on its arguments $\overline{v}$ and using the same $\ell$ in $\bigcap \overline{\ell} = \ell$.

        The backward direction is similar to the forward direction, but inverting the argument for each case.
\end{proof}

\begin{lemma}\label{lemma:lifting-e-preserves-type}
    If $\refunreach{\TTau}$, $e$ is an expression in the surface language, $\cdot ; \Gamma \proves e : \TTau$, and $\Lift{\ell}{e}$ is defined, then $\cdot ; \Gamma \diamond \ell \cproves \Lift{\ell}{e} : \tplab{\TTau}{\ell}$.
\end{lemma}
\begin{proof}
    By induction on $e$.

    The case for values is proven by the forward direction of \Cref{lemma:lifting-v-preserves-type}, and the trivial case for variables.
    Any of the reference operations (allocation, dereferencing, or assignment) are vacuously true.
    In most other cases, we simply apply the result on subvalues in the expression to type with label $\ell$, and the premises for the typing rule follows.
    For instance, with \ruleref{TC-Lift} for method calls, raising the whole method call with $\ell$ means we raise all subvalues with the same $\ell$, and so the intersection of relevant labels is also $\ell$.
    Note that lifting is undefined on $\ProveN$ and compute-and-prove statements if $\ell \neq \Compute$, but we syntactically prevent lifting an expression in these case (i.e. prevent methods calls that nest prove statements).

    Finally, for the inductive cases, if $e$ is an if statement, then the inductive hypothesis proves $e_1$ and $e_2$ lift to $\tplab{\TTau}{\ell}$, and so the premises of the typing rule are satisfied.
    And if $e$ is a let statement, $\ell_1 = \ell_1'$ in \ruleref{TC-Let} and $e_2$ is typed with $\tplab{\TTau_2}{\ell}$ by the inductive hypothesis, and so the premises of the typing rule hold.
\end{proof}

\begin{lemma}\label{lemma:lowering-cp-types}
    If $\Sigma ; \cdot \cproves v : \tplab{\TTau}{\ComputeProve}$, then $\SigmaCP |_\Compute \proves \Lower{\Compute}{v} : \TTau$ and $\SigmaCP |_\Prove \proves \Lower{\Prove}{v} : \TTau$.
\end{lemma}
\begin{proof}
    By induction on $v$.

    For any constant $v$, $\Lower{\ell}{v} = v$ for any $\ell$, so $\cproves v : \tplab{\TTau}{\ComputeProve}$ implies $\proves v : \TTau$ by inverting \ruleref{TC-LiftValCP}.

    If $v$ is a reference, meaning $v = r_\ComputeProve(\iota_1, \iota_2)$, then we can invert $\Sigma ; \cdot \cproves v : \tplab{(\SRefTau{\TTau'})}{\ell}$ using \ruleref{TC-LiftValCP} to $\Theta(\SigmaCP) ; \cdot \proves \iota' : \TRefTau{\TTau'}$ where $\Theta(\iota_1, \iota_2) = \iota'$.
    Inverting again with \ruleref{T-Location} gives $(\Theta(\SigmaCP))(\iota') = \TTau'$, meaning $\SigmaCP(\iota_1, \iota_2) = \TTau'$ and therefore our conclusion holds.

    The other cases, of classes and $\ProofOfN$ objects, follow by the inductive hypothesis applied on their enclosed subvalues $\overline{v}$ and applying the relevant typing rule.
\end{proof}

\begin{lemma}\label{lemma:lowering-in-sync}
    If $\Sigma \cproves \sigmaC ; \sigmaP$ then for any $v$ such that $\Sigma ; \cdot \cproves v : \tplab{\STau}{\ComputeProve}$, then $\sigmaC ; \sigmaP \models \Lower{\Compute}{v} \approx \Lower{\Prove}{v}$.
\end{lemma}
\begin{proof}
    By induction on $v$.

    For any constant, $\Lower{\Compute}{v} = \Lower{\Prove}{v}$, and we're done by noting $\approx$ is reflexive.

    The theorem is vacuously true for variables.

    For any reference, it must be $v = r_\ComputeProve(\iota_1, \iota_2)$ for it to have type $\tplab{\TTau}{\ComputeProve}$.
    This means that $\Lower{\Compute}{v} = \iota_1$ and $\Lower{\Prove}{v} = \iota_2$.
    Note that since $(\iota_1, \iota_2) \in \dom(\SigmaCP)$ and $\iota_2 \neq \ast$ the last condition of configuration safety gives that $\sigmaC ; \sigmaP \models \sigmaC(\iota_1) \approx \sigmaP(\iota_2)$, which by definition of $\approx$ (\Cref{app:config-safety-def}) gives $\sigmaC ; \sigmaP \models \iota_1 \approx \iota_2$.

    For any class object with constructor inputs $\overline{v}$, the inductive hypothesis gives $\sigmaC ; \sigmaP \models \Lower{\Compute}{v_i} \approx \Lower{\Prove}{v_i}$, and so $\sigmaC ; \sigmaP \models \Lower{\Compute}{v} \approx \Lower{\Prove}{v}$ by definition.
    Similar reasoning holds for any \ProofOfN value, but where the values are precisely equal since the premise and the typing rule \ruleref{T-ProofOf} gives $\proves \overline{v_1} : \tplab{\overline{\TTau}}{\ComputeProve}$.
\end{proof}

\begin{lemma}\label{lemma:lifting-is-defined}
    If $\sigmaC ; \sigmaP \models v_1 \approx v_2$ then $\Lift{\Compute}{v_1} \sqcup \Lift{\Prove}{v_2}$ is defined.
\end{lemma}
\begin{proof}
    By induction on the definition of $\approx$.

    If $v_1$ and $v_2$ are constants, then the premise is true only if $v_1 = v_2$, and so $\Lift{\Compute}{v_1} \sqcup \Lift{\Prove}{v_2} = v_1 = v_2$.

    If $\sigmaC ; \sigmaP \models \iota_1 \approx \iota_2$ then $\Lift{\Compute}{\iota_1} \sqcup \Lift{\Prove}{\iota_2} = r_\ComputeProve(\iota_1, \iota_2)$ is defined.

    If $v_1 = \NewN~C(\overline{v_1})$ and $v_2 = \NewN~C(\overline{v_2})$ then by the inductive hypothesis $\Lift{\Compute}{v_{1i}} \sqcup \Lift{\Prove}{v_{2i}}$ exists, meaning $\Lift{\Compute}{v_1} \sqcup \Lift{\Prove}{v_2}$ is defined by definition.
    And similarly for \ProofOfN values.
\end{proof}

\begin{lemma}\label{lemma:theta-is-defined}
    If $(\iota_1, \iota_2) \in \dom(\SigmaCP)$ and  $\Sigma \cproves \sigmaC , \sigmaP$ then $\Lift{\Compute}{\sigmaC(\iota_1)} \sqcup \Lift{\Prove}{\sigmaP(\iota_2)}$ is defined.
\end{lemma}
\begin{proof}
    From the last condition of $\Sigma \cproves \sigmaC, \sigmaP$ and $(\iota_1, \iota_2) \in \dom(\SigmaCP)$ we get that $\sigmaC ; \sigmaP \models \sigmaC(\iota_1) \approx \sigmaP(\iota_2)$.
    Thus, by \Cref{lemma:lifting-is-defined} this means $\Lift{\Compute}{\sigmaC(\iota_1)} \sqcup \Lift{\Prove}{\sigmaP(\iota_2)}$ is defined.
\end{proof}

\begin{lemma}\label{lemma:joining-is-welltyped}
    Suppose $(\iota_\Compute, \iota_\Prove) \in \dom(\SigmaCP)$ and $\Sigma \cproves \sigmaC, \sigmaP$.
    If $v = \Lift{\Compute}{\sigmaC(\iota_\Compute)} \sqcup \Lift{\Prove}{\sigmaP(\iota_\Prove)}$ is defined, then $\SigmaCP \cproves v : \tplab{\SigmaCP(\iota_\Compute, \iota_\Prove)}{\ComputeProve}$.
\end{lemma}
\begin{proof}
    By induction on the definition of $\sqcup$:

    If $\sigmaC(\iota_\Compute) = \sigmaP(\iota_\Prove)$, then the join $v$ is the equal value and doesn't contain any locations, so $\cproves v : \TTau$ with $\TTau$ being the type of the equal value.
    Note this value could be a class or proof object, but all the fields/subvalues must be exactly equal and contain no references.

    If $\sigmaC(\iota_\Compute) = \iota_1$ and $\sigmaP(\iota_\Prove) = \iota_2$ then by the last condition of $\Sigma \cproves \sigmaC, \sigmaP$, we know $(\iota_1, \iota_2) \in \dom(\SigmaCP)$ and therefore $\SigmaCP \cproves r_\ComputeProve(\iota_1, \iota_2) : \tplab{\SigmaCP(\iota_\Compute, \iota_\Prove)}{\ComputeProve}$.

    If the join is over two class objects, the inductive hypothesis proves the condition on the sub-cases, and we use \ruleref{T-New} and \ruleref{TC-LiftValCP} to get the expected results.
    And similarly for proof objects with \ruleref{T-ProofOf} and \ruleref{TC-LiftValCP}.
\end{proof}

\begin{theorem}[Preservation of Inner \zkstrudel]
\label{thm:abs-preservation}
     Suppose $\Sigma ; \varnothing ; \Delta ; A \cproves \SExp : \tplab{\TTau}{\ell} \produces A'$ and $\Sigma \cproves \sigma_\Compute, \sigma_\Prove$ and $\Delta; A \cproves \rho$ and $(e, \sigmaC, \SigmaC, \SigmaP^{in}, \SigmaCP^{in}, \Delta^{in}, A^{in}, \dom(\SigmaCP |_\Compute)) \nestedok$ for some $\SigmaP^{in}, \SigmaCP^{in}, \Delta^{in}, A^{in}$.
     If $\cpconf*{\SExp} \cpstep \cpconf{\SExp'}{\sigmaC'}{\sigmaP'}{\rho'}$, then for some $\Sigma' \supseteq \Sigma$ and $A'' \supseteq A$, $\Sigma' ; \varnothing ; \Delta ; A'' \cproves \SExp' : \tplab{\TTau'}{\ell} \produces A'$ where $\TTau <: \TTau'$ and $\Sigma' \cproves \sigma_\Compute', \sigma_\Prove'$ and $ \Delta ; A'' \cproves \rho'$ and $(e', \sigmaC', \SigmaC', \SigmaP^{in'}, \SigmaCP^{in'}, \Delta^{in'}, A^{in'}, \dom(\SigmaCP' |_\Compute)) \nestedok$ for some $\SigmaP^{in'}, \SigmaCP^{in'}, \Delta^{in'}, A^{in'}$.
\end{theorem}
\begin{proof}
    By induction on the syntax of $e$.
    The definitions for configuration safety are defined in \Cref{app:config-safety-def} and the $\nestedok$ predicate is defined in \Cref{app:nested-preservation}.
    Note that the $\nestedok$ predicate is true by definition in all but the last three cases by simply using $\SigmaP^{in} = \SigmaCP^{in} = \varnothing$ to satisfy that premise of the definition (the last three cases involve $\SCnpN$ terms, the others don't by \Cref{lemma:cnp-always-head}).

    \begin{itemize}[itemsep=0.6em]
        \item If $e$ is a value, then the theorem is vacuously true.
    
        \item If $e = \SRefVal{\ell}{v}$, then it can step by either \ruleref{EC-RefC}, \ruleref{EC-RefP}, or \ruleref{EC-RefCP}.
        We consider each of those cases below:
        \begin{itemize}
            \item Suppose $\cpconf*{\SRefVal{\Compute}{v}} \cpstep \cpconf{r_\Compute(\iota)}{\sigmaC[\iota \mapsto \Lower{\Compute}{v}]}{\sigmaP}{\rho}$.
            If we let $\SigmaC' = \SigmaC \cup \{\iota \mapsto \TTau \}$, and $\SigmaP' = \SigmaP$ and $\SigmaCP' = \SigmaCP$, then $\Sigma' ; \Gamma ; \Delta ; A \cproves r_\Compute(\iota) \cproduces A$ from \ruleref{T-Location} as an axiom, into \ruleref{T-Val} and then \ruleref{TC-Lift} (or from \ruleref{T-Location} first into \ruleref{TC-LiftVal} then \ruleref{TC-Val}).

            We now prove each cases for configuration safety.
            By the inductive hypothesis, combined with $\iota \in \dom(\SigmaC')$ and $\iota \not\in \dom(\SigmaCP' |_\Compute)$ by construction, configuration safety case \rone is proven true.

            By the premise, for any $\iota' \in \dom(\SigmaC), \SigmaC \proves \sigmaC(\iota') : \SigmaC(\iota')$, so to prove it for $\dom(\SigmaC')$ we just have to prove it on the $\iota$ created in this step.
            Note that $\sigmaC(\iota) = \Lower{\Compute}{v}$.
            Given that $\Sigma ; \Gamma ; \Delta ; A \cproves \SRefVal{\Compute}{v} : \Lift{\Compute}{\SRefTau{\TTau}} \cproduces A$, inverting it via \ruleref{TC-Lift} and then \ruleref{T-Ref} gives that $\Sigma ; \Gamma \proves \Lower{\Compute}{v} : \TTau$.
            Noting that $v$ cannot be a variable means this can be proven with an empty $\Gamma$ and therefore configuration safety condition \rtwo is true.

            Finally, note that case \rthree is trivially true since $\sigmaP$ does not change, and cases \rfour, \rfive, \rsix, and \rseven are trivially true since $\SigmaCP' = \SigmaCP$.

            \item Suppose $\cpconf*{\SRefVal{\Prove}{v}} \cpstep \cpconf{r_\Prove(\iota)}{\sigmaC}{\sigmaP[\iota \mapsto \Lower{\Prove}{v}]}{\rho}$.
            If we let $\SigmaC' = \SigmaC$, $\SigmaP' = \SigmaP \cup \{\iota \mapsto \TTau\}$, and $\SigmaCP' = \SigmaCP$, then the proof follows similarly as in the previous case (when it was a \Compute reference instead of a \Prove reference).

            \item Suppose $\cpconf*{\SRefVal{\ComputeProve}{v}} \cpstep \cpconf{r_\ComputeProve(\iota_1, \iota_2)}{\sigmaC[\iota_1 \mapsto \Lower{\Compute}{v}]}{\sigmaP[\iota_2 \mapsto \Lower{\Prove}{v}]}{\rho}$.
            If we let $\SigmaC' = \SigmaC$, $\SigmaP' = \SigmaP$, and $\SigmaCP' = \SigmaCP \cup \{(\iota_1, \iota_2) \mapsto \TTau \}$, then $\Sigma' ; \Gamma ; \Delta ; A \cproves r_\ComputeProve(\iota_1, \iota_2) \cproduces A$ from \ruleref{T-Location} as an axiom into \ruleref{T-Val} and then \ruleref{TC-LiftCP}.
            Note that this requires a function $\Theta$ that maps pairs of locations in $\SigmaCP'$ to ghost locations, and therefore $\Theta(\iota_1, \iota_2)$ is the location used for the \ruleref{T-Location} axiom.

            Then since $\iota_1 \in \dom(\SigmaCP' |_\Compute)$ and $\iota_1 \not\in \dom(\SigmaC')$ by construction, and since $\iota_2 \in \dom(\SigmaCP' |_\Prove)$ and $\iota_2 \not\in \dom(\SigmaP')$ by construction, combined with the inductive hypothesis we get that configuration safety \rone is proven.

            Since $\SigmaC' = \SigmaC$ and $\SigmaP' = \SigmaP$, cases \rtwo and \rthree are trivial from the inductive hypothesis.
            
            From inverting the typing judgment from the premise using \ruleref{TC-LiftCP} and \ruleref{T-Deref}, then applying \ruleref{TC-LiftValCP}, we get that $\Sigma ; \Gamma \cproves v : \tplab{\TTau}{\ComputeProve}$.
            Using \Cref{lemma:lowering-cp-types}, and because $v$ is closed, we immediately get configuration safety conditions \rfour and \rfive.

            For \rsix, since $(\iota_1, \iota_2) \in \dom(\SigmaCP')$, we need to show $\sigmaC ; \sigmaP \models \sigmaC(\iota_1) \approx \sigmaP(\iota_2)$, which we get by \Cref{lemma:lowering-in-sync}.
            Note this requires either the value is $r_\ComputeProve(\iota_1', \iota_2')$, which satisfies the premise, or is a value that satisfies \Cref{lemma:c-and-p-types}.
            
            For \rseven, this would happen only if $v = r_\ComputeProve(\iota_1', \iota_2')$.
            $v$ can only be typed with \ruleref{TC-LiftValCP}, which requires a $\Theta$ mapping $\iota_1', \iota_2'$ to some $\iota_3'$ that is in the domain of $\Theta(\SigmaCP')$, meaning $(\iota_1', \iota_2') \in \dom(\SigmaCP')$.
        \end{itemize}

        \item If $e = \SDeref{v}$, then it can step by either \ruleref{EC-DerefC}, \ruleref{EC-DerefP}, or \ruleref{EC-DerefCP}.
        We consider each of those cases below:

        \begin{itemize}[itemsep=0.6em]
            \item Suppose $\cpconf*{\SDeref{r_\Compute(\iota)}} \cpstep \cpconf*{\Lift{\Compute}{\sigmaC(\iota)}}$.
            From well-typedness in the premise, e.g. $\Sigma ; \Gamma ; \Delta ; A \cproves \SDeref{r_\Compute(\iota)} : \tplab{\SRefTau{\TTau}}{\Compute} \cproduces A$, we can invert this with \ruleref{TC-Lift}.
            Inversion gives that $\SigmaC ; \Gamma |_\Compute \proves \SDeref{\iota} : \Lower{}{\tplab{\SRefTau{\TTau}}{\Compute}} \produces \Compute$.
            By inverting with \ruleref{T-Deref} and once more with \ruleref{T-Location}, we get that $\SigmaC(\iota) = \TTau$ for $\tplab{\TTau}{\Compute} = \Lift{\Compute}{\TTau}$.
            Thus, $\iota \in \dom(\SigmaC)$ and therefore by condition \rtwo of configuration safety, $\SigmaC \proves \sigmaC(\iota) : \SigmaC(\iota)$.
            This then can apply to \ruleref{TC-LiftVal} and then \ruleref{TC-Val} to get $\Sigma ; \Gamma ; \Delta ; A \cproves \Lift{\Compute}{\sigmaC(\iota)} : \tplab{\TTau}{\Compute} \cproduces A$.
            
            Because none of the stores change, all configuration properties are trivially preserved by induction.
            
            \item Suppose $\cpconf*{\SDeref{r_\Prove(\iota)}} \cpstep \cpconf*{\Lift{\Prove}{\sigmaP(\iota)}}$.
            The proof follows similarly as in the previous case (when dereferencing a \Compute reference instead of a \Prove reference).

            \item Suppose $\cpconf*{\SDeref{r_\ComputeProve(\iota_1, \iota_2)}} \cpstep \cpconf*{\Lift{\Compute}{\sigmaC(\iota_1)} \sqcup \Lift{\Prove}{\sigmaP(\iota_2)}}$.
            From well-typedness in the premise, $\Sigma ; \Gamma ; \Delta ; A \cproves \SDeref{r_\ComputeProve(\iota_1, \iota_2)} : (\tplab{\SRefTau{\TTau}}{\ComputeProve}) \cproduces A$, the only possible typing rule is \ruleref{TC-LiftCP}.
            Inverting with injective function $\Theta$ where $\Theta(\iota_1, \iota_2) = \iota_3$ then gives $\Theta(\SigmaCP) ; \Gamma \proves \TDeref{\iota_3} : \TRefTau{\TTau} \cproduces \ell$.
            By inverting \ruleref{T-Deref}, we get the premise that $\Theta(\SigmaCP)(\iota_3) = \TTau$, meaning $\iota_3 \in \dom(\Theta(\SigmaCP))$ and therefore $(\iota_1, \iota_2) \in \dom(\SigmaCP)$.
            By \Cref{lemma:theta-is-defined}, we know that $\Lift{\Compute}{\sigmaC(\iota_1)} \sqcup \Lift{\Prove}{\sigmaP(\iota_2)}$ is defined.
            Then, by \Cref{lemma:joining-is-welltyped}, we get that it also types with type $\tplab{(\SigmaCP(\iota_1, \iota_2))}{\ComputeProve}$.
            
            Because none of the stores change, all configuration properties are trivially preserved by induction.
        \end{itemize}

        \item If $e = \SAssign{v_1}{v_2}$, then it must step via either \ruleref{EC-AssignC}, \ruleref{EC-AssignP}, or \ruleref{EC-AssignCP}. 
        We consider each of the cases below:

        \begin{itemize}[itemsep=0.6em]
            \item Suppose $\cpconf*{\SAssign{r_\Compute(\iota)}{v}} \cpstep \cpconf{()}{\sigmaC[\iota \mapsto \Lower{\Compute}{v}]}{\sigmaP}{\rho}$.
            First, note that the typing condition is trivial, since $\Sigma ; \Gamma \cproves () : \tplab{\SUnit}{\ell'}$ for any $\ell'$.
            Also note that since the domains of the stores do not change, and the types of the values in the stores remain the same, so setting $\Sigma' = \Sigma$ satisfies all the configuration safety conditions.

            \item Suppose $\cpconf*{\SAssign{r_\Prove(\iota)}{v}} \cpstep \cpconf{()}{\sigmaC}{\sigmaP[\iota \mapsto \Lower{\Prove}{v}]}{\rho}$.
            Just like in the \Compute case, the typing condition is trivial and setting $\Sigma' = \Sigma$ satisfies the configuration safety conditions.

            \item Suppose $\cpconf*{\SAssign{r_\ComputeProve(\iota_1, \iota_2)}{v}} \cpstep \cpconf{()}{\sigmaC[\iota_1 \mapsto \Lower{\Compute}{v}]}{\sigmaP[\iota_2 \mapsto \Lower{\Prove}{v}]}{\rho}$.
            Like before, the typing condition is trivial, and setting $\Sigma' = \Sigma$ satisfies configuration safety conditions \rone to \rfive.
            Case \rsix is true from \Cref{lemma:lowering-in-sync}, and case \rseven is true using the same logic as the case where $\SRefVal{\ComputeProve}{v}$ steps.
        \end{itemize}

        \item Suppose $\cpconf*{\SWitAssign*} \cpstep \cpconf{()}{\sigmaC}{\sigmaP}{\rho[x \mapsto v]}$.
        The typing condition is trivial, since $\Sigma ; \Gamma \cproves () : \tplab{\Unit}{\ell'}$ for any $\ell'$.
        Keeping $\Sigma' = \Sigma$, we get by induction that $\Sigma' \cproves \sigmaC, \sigmaP$.
        Then, $\Delta ; A' \cproves \rho'$ where $A' = A \cup \{x\}$, since $\dom(\rho') = \dom(\rho) \cup \{x\} = A \cup \{x\} = A'$, and by inverting \ruleref{TC-InputAssign}, we get that $\Delta(x) = \TTau$ and $\Sigma ; \Gamma \cproves \rho'(x) : \Lift{\Compute}{\TTau}$. 
        Note that $\rho'(x)$ cannot be a variable, so $\Sigma ; \cdot \cproves \rho'(x) : \Lift{\Compute}{\TTau}$.
        Similarly, computed input variables cannot be references (see \ruleref{T-ComputeAndProve}), which comes from $\refunreach{\TTau}$.
        This gives us that $\cdot ; \cdot \cproves \rho'(x) : \tplab{\TTau}{\ell}$.
        Therefore, by inverting with \ruleref{TC-LiftVal}, we get $\proves \rho'(x) : \TTau$.

        \item Suppose $\cpconf*{\SWitDeref*} \cpstep \cpconf*{\rho(x)}$.
        Inverting \ruleref{TC-InputDeref} gives $\Delta(x) = \TTau$ and $x \in A$.
        Combined with the premise $\Delta ; A \cproves \rho$, we know that $\proves \rho(x) : \TTau$.
        Therefore, by \ruleref{TC-LiftValCP}, $\Sigma ; \Gamma ; \Delta ; A \cproves \rho(x) : \tplab{\TTau}{\ComputeProve}$.
        
        None of the stores change, so $\Sigma \cproves \sigmaC, \sigmaP$ and $\Delta ; A \cproves \rho$ are trivially preserved.
    
        \item For the cases for casting, field access, prove and verify statements, none of the stores change and type preservation occurs similar to these cases in \Cref{thm:core-preservation}.
        
        \item Suppose $e = v.m_\ell(\overline{u})$ is a method call, so $\Sigma ; \Gamma ; \Delta ; A \cproves e : \tplab{\TTau}{\ell} \produces A$ and $\cpconf*{(\New C(\overline{v})).m_\ell(\overline{u})} \cpstep \cpconf*{\subst*{\Lift*{e_m}}{{\overline{x}}{\overline{u}}{\This}{\New C(\overline{v})}}}$.
        By \Cref{lemma:lifting-e-preserves-type}, since $\cdot ; \overline{x} : \overline{\TTau_1}, \This : C \proves e : \TTau$, we know $\cdot ; \overline{x} : \tplab{\overline{\TTau_1}}{\ell}, \This : \tplab{C}{\ell} \proves \Lift*{e} : \tplab{\TTau}{\ell}$.
        Thus, the substitution also types by \Cref{thm:subst-pres-surf-typing}: $\cdot ; \Gamma \proves \subst*{\Lift*{e_m}}{{\overline{x}}{\overline{u}}{\This}{\New C(\overline{v})}} : \tplab{\TTau}{\ell}$.
        And since none of the stores change, $\Sigma \cproves \sigmaC, \sigmaP$ and $\Delta ; A \cproves \rho$ are preserved.

        \item Suppose $e = \SLetIn{x \ty \tplab{\TTau_1}{\ell}}{e_1}{e_2}$.
        Then it can step via either \ruleref{EC-Eval} or \ruleref{EC-Lift} and \ruleref{E-Let}.
        \begin{itemize}[itemsep=0.6em]
            \item In the first case, $\cpconf*{e} \cpstep \cpconf{\SLetIn{x\ty\tplab{\TTau_1}{\ell}}{e_1'}{e_2}}{\sigmaC'}{\sigmaP'}{\rho'}$.
            By the inductive hypothesis, there is some $\Sigma', A''$ such that $\Sigma' \cproves \sigmaC', \sigmaP'$ and $\Delta ; A'' \cproves \rho'$ and $\Sigma' ; \Gamma ; \Delta ; A'' \cproves e_1' : \tplab{\TTau_1'}{\ell} \produces A_1$.
            Thus, we get that $\Sigma' ; \Gamma ; \Delta ; A'' \cproves \SLetIn{x\ty\tplab{\TTau_1}{\ell}}{e_1'}{e_2} : \tplab{\TTau}{\ell'} \produces A$.

            \item In the second case, $\cpconf*{e} \cpstep \cpconf*{\subst{e_2}{x}{v}}$. 
            Inverting \ruleref{TC-Let} means $\Sigma ; \Gamma \cproves v : \tplab{\TTau_1}{\ell}$ and so $\Sigma ; \Gamma ; \Delta ; A \cproves \subst{e_2}{x}{v} : \tplab{\TTau_2}{\ell'} \produces A_2$ since it is a safe substitution (\Cref{thm:subst-pres-surf-typing}).
            And since none of the stores change, $\Sigma \cproves \sigmaC, \sigmaP$ and $\Delta ; A \cproves \rho$ are preserved.
        \end{itemize}

        \item Suppose $e = \SCnp*[e_{in}]$.
        Then it can only step via \ruleref{EC-ComputeAndProveInit} to $e' = \SCnpAdmin*{e_{in}'}{\varnothing}{\varnothing}$, where $\cpconf*{e} \cpstep \cpconf{e'}{\sigmaC'}{\sigmaP}{\rho}$ where $\sigmaC' = \sigmaC[\overline{\iota_{yp}} \mapsto \bot, \overline{\iota_{ys}} \mapsto \bot]$.
        We can still choose $\Sigma' = \Sigma$ and get $\Sigma' \cproves \sigmaC', \sigmaP$ as we just relegate the newly allocated computed input locations as untracked in configuration safety (note it is tracked in $\varphi$).
        Similar reasoning means $\Delta ; A \cproves \rho$ is preserved.
        Also, $e'$ is well-typed by \ruleref{TC-Combined} because it is directly implied by the last two premises of \ruleref{TC-ComputeAndProve}.

        Because $e'$ is itself a $\SCnpN$ term, we need to additionally prove $\nestedok$ predicate is true.
        If $\SigmaP^{in} = \SigmaCP^{in} = \varnothing$ and then we can prove $\SigmaC, \SigmaP^{in}, \SigmaCP^{in} \cproves \sigmaC, \sigmaP^{in}$ (since $\sigmaP^{in}$, defined by the $\SCnpN$ term, is $\varnothing$) as well as $\Delta^{in}, A^{in} \cproves \rho^{in}$ where $\Delta^{in}$ is as defined in \ruleref{TC-ComputeAndProve} and $A^{in} = \varnothing$ and $\rho^{in} = \varnothing$ from the $\SCnpN$ term. 
        Also, $e_{in}'$ types via the premise of \ruleref{TC-Combined}.
        Finally, using \Cref{lemma:cnp-always-head} and since $\SCnpN$ cannot appear in the surface syntax, we know $\SCnpN \not\in e_{in}'$ and therefore $(e', \sigmaC, \SigmaC, \varnothing, \varnothing, \varnothing, \varnothing, \dom(\SigmaCP |_\Compute) \cup \varnothing) \nestedok$ is true. 
        All of these then prove $(e, \sigmaC, \SigmaC, \SigmaP^{in}, \SigmaCP^{in}, \Delta^{in}, A^{in}, \dom(\SigmaCP |_\Compute)) \nestedok$, which finishes preservation for this case.

        \item Suppose $e = \SCnpAdmin*{e_{in}}{\sigmaP^{in}}{\rho^{in}}$ and using \ruleref{EC-ComputeAndProveStep} it steps $\cpconf*{e} \cpstep \cpconf{e'}{\sigmaC'}{\sigmaP}{\rho}$, meaning $e' = \SCnpAdmin*{e_{in}'}{\sigmaP^{in'}}{\rho^{in'}}$ from $\cpconf{e_{in}}{\sigmaC}{\sigmaP^{in}}{\rho^{in}} \cpstep \cpconf{e_{in}'}{\sigmaC'}{\sigmaP^{in'}}{\rho^{in'}}$.
        Using the contexts $\SigmaP^{in}, \SigmaCP^{in}, \Delta^{in}, A^{in}$ from typing $e_{in}$ on \ruleref{TC-Combined} for the predicate $(e, \sigmaC, \SigmaC, \SigmaP^{in}, \SigmaCP^{in}, \Delta^{in}, A^{in}, \dom(\SigmaCP |_\Compute)) \nestedok$, we can apply the inductive hypothesis to prove all of the preservation conditions hold on $e_{in}'$: namely, $\SigmaC, \SigmaP^{in'}, \SigmaCP^{in'} \cproves \sigmaC', \sigmaP^{in'}$ and $\Delta^{in}, A^{in'} \cproves \rho^{in'}$ and $\nestedok$ holds on $e_{in}'$ with $\sigmaC, \SigmaC$, and some other contexts.
        This all proves $(e', \sigmaC', \SigmaC', \SigmaP^{in'}, \SigmaCP^{in'}, \Delta^{in'}, A^{in'}, \SigmaCP' |_\Compute) \nestedok$.
        Also, since $\sigmaP$ is unchanged, letting $\SigmaP' = \SigmaP$ and $\SigmaCP' = \SigmaCP$ gives $\SigmaC', \SigmaP', \SigmaCP' \cproves \sigmaC', \sigmaP$ because the changed parts of the store $\sigmaC'$ are disjoint from the subset corresponding to $\dom(\SigmaCP)$ by the disjoint conditions in $\nestedok$ (specifically, the recursive set $S$ that collects the domains of shared memory at each step and ensures future shared memory is disjoint).
        And since $\rho$ is unchanged, $\Delta ; A \cproves \rho$ is preserved.
        Finally, the $\nestedok$ predicate directly implies the premises of \ruleref{TC-Combined}, and so $e'$ is also well-typed.

        \item Suppose $e = \SCnpAdmin*{e_{in}}{\sigmaP^{in}}{\rho^{in}}$ and it steps according to \ruleref{EC-ComputeAndProveTrue}, meaning $\cpconf*{e} \cpstep \cpconf{e'}{\sigmaC'}{\sigmaP}{\rho}$ where $e' = \SProofOfUsing{\alpha}{\overline{v}{::}\rho(\overline{x_p})}$.
        Note that this types because the sub-values are closed and so we can use \ruleref{T-ProofOf} and \ruleref{TC-LiftVal} to type $e'$.
        Since $\sigmaC'$ doesn't update existing locations and only adds new ones, we can keep the same $\Sigma$ that types $e$ and $\Sigma \cproves \sigmaC', \sigmaP$ and $\Delta; A \cproves \rho$ doesn't change.
        Finally, because $\SCnpN \not\in e'$, we can set $\SigmaP^{in} = \SigmaCP^{in} = \Delta^{in} = A^{in} = \varnothing$ to get $(e', \sigmaC', \SigmaC, \SigmaP^{in}, \SigmaCP^{in}, \Delta^{in}, A^{in}, \dom(\SigmaCP |_\Compute)) \nestedok$.
    \end{itemize}
\end{proof}

\begin{theorem}[Progress of Inner \zkstrudel]
\label{thm:abs-progress}
    Suppose that $\Sigma ; \varnothing ; \Delta ; A \cproves \SExp : \tplab{\TTau}{\ell} \produces A'$.
    Then either
    \begin{enumerate}[label=(\roman*)]
        \item $\SExp$ is a value;
        \item for any set of stores $(\sigma_\Compute, \sigma_\Prove, \rho)$ such that $\Sigma \cproves \sigma_\Compute, \sigma_\Prove$ and $\Delta ; A \cproves \rho$ (\Cref{app:config-safety-def}), and for any set of contexts $(\SigmaP^{in}, \SigmaCP^{in}, \Delta^{in}, A^{in})$ such that $(e, \sigmaC, \SigmaC, \SigmaP^{in}, \SigmaCP^{in}, \Delta^{in}, A^{in}, \dom(\SigmaCP |_\Compute)) \nestedok$ (\Cref{app:nested-preservation}),  there is some $\SExp',\sigma_\Compute', \sigma_\Prove', \rho'$ such that $\cpconf*{\SExp} \cpstep \cpconf{\SExp'}{\sigmaC'}{\sigmaP'}{\rho'}$;
        \item $\SExp = \Eval{\SProveUsing*[e']}$ where $\conf{\subst*{e'}{{\overline{x}}{\overline{v}}{\overline{y}}{\overline{u}}}}{\varnothing} \tstep^* \conf{\False}{\_}$ or it diverges; 
        \item $\SExp = \Eval{\SCast{C}{(\SNew{D}{\SVals*})}}$ where $D \not<: C$
        \item $\SExp = \Eval{\SCnpAdmin*{\False}{\tilde{\sigmaP}}{\tilde{\rho}}}$; 
        \item $\SExp = \Eval{\SDeref{r_\Compute(\iota)}}$ where $\sigmaC(\iota) = \bot$; or
        \item $\SExp = \Eval{\SCnpAdmin*{\SExp'}{\tilde{\sigmaP}}{\tilde{\rho}}}$ where $\SExp'$ falls into one of failure conditions \rthree to \rseven.
    \end{enumerate}
\end{theorem}
\begin{proof}
    By induction on $e$.
    Note that case \rtwo means we take a step on a normal expression, 
    and \rthree to \rseven specify the conditions where the semantics could get stuck.
    \begin{itemize}[itemsep=0.6em]
        \item If $e$ is a value, then we are in case \rone and are done.
        
        \item Suppose $e = \SRefVal{\ell}{v}$ and is well-typed. 
        $e$ is not a value, nor does it fit into the cases of \rthree to \rseven, so we need to prove it can step.
        We can simply pick the relevant rule based on $\ell$ and pick fresh locations in order to take a step.

        \item Suppose $e = \SDeref{v}$, which means we are either in case \rtwo or \rsix.
        Given $e$ is well-typed, it can be typed in three cases:
        \begin{itemize}
            \item If $e$ is typed with \ruleref{TC-Lift}, then inverting with \ruleref{T-Deref} gives that $v$ (typed as a $\TRefN$) must be a location (since $\Gamma$ is empty). 
            If $\ell = \Compute$, then either $v = \bot$, so we're in case \rsix, or $v \in \dom(\SigmaC)$ meaning $\Sigma \cproves \sigmaC, \sigmaP$, we know that $v \in \dom(\SigmaC) \subseteq \dom(\sigmaC)$ and therefore can step with \ruleref{EC-DerefC}.
            If $\ell = \Prove$, then similar logic shows that it can step with \ruleref{EC-DerefP}.

            \item If $e$ is typed with \ruleref{TC-LiftCP}, then this means the inverted value being dereferenced is a location in $\Theta(\SigmaCP)$.
            This means there is some $\iota_1, \iota_2$ such that $\Theta(\iota_1, \iota_2) = v$, and therefore $(\iota_1, \iota_2) \in \SigmaCP$.
            Firstly, this means that $\iota_1 \in \dom(\sigmaC)$ and $\iota_2 \in \dom(\sigmaP)$.
            Then, by \Cref{lemma:theta-is-defined}, we get that the joined term $\Lift{\Compute}{\sigmaC(\iota_1)} \sqcup \Lift{\Prove}{\sigmaP(\iota_2)}$ is defined and can therefore take a step with \ruleref{EC-DerefCP}.
        \end{itemize}
        
        \item Suppose $e = \SAssign{v_1}{v_2}$ which is not a value, nor does it fit into the cases of \rthree to \rsix, so we need to prove it can step.
        Given $e$ is well-typed, it can be typed in two cases:
        \begin{itemize}
            \item If $e$ is typed with \ruleref{TC-Lift}, then inverting this rule and then inverting \ruleref{T-Assign} results in the value $v$ (where $\Lift{\ell}{v} = v_1$) which is typed as \TRefN is a location $\iota$ (since $\Gamma$ is empty).
            If $\ell = \Compute$ then $\iota \in \dom(\SigmaC)$, and therefore by $\Sigma \cproves \sigmaC, \sigmaP$ we know that $\iota \in \dom(\SigmaC) \subseteq \dom(\sigmaC)$ and therefore can step with \ruleref{EC-AssignC}.
            If $\ell = \Prove$ then similar logic shows that it can step with \ruleref{EC-DerefP}.

            \item If $e$ is typed with \ruleref{TC-LiftCP}, then the inverted value being dereferenced is a location in $\Theta(\SigmaCP)$, meaning there is some $\iota_1, \iota_2$ such that $\Theta(\iota_1, \iota_2) = v_1$ and therefore $(\iota_1, \iota_2) \in \SigmaCP$.
            This then means that $\iota_1 \in \sigmaC$ and $\iota_2 \in \sigmaP$ so it can step with \ruleref{EC-AssignCP}.
        \end{itemize}

        \item Suppose $e = \SWitAssign*$ and is well-typed.
        By inverting the typing rule \ruleref{TC-InputAssign}, we know $x \not\in A$.
        By preservation conditions $\Delta ; A \cproves \rho$, this means $x \not\in \dom(\rho)$, meaning it can step with \ruleref{EC-InputAssign}.

        \item Suppose $e = \SWitDeref*$ and is well-typed.
        By inverting the typing rule \ruleref{TC-InputDeref}, we get $x \in A$.
        By preservation conditions $\Delta ; A \cproves \rho$, this means $x \in \dom(\rho)$, meaning it can step to $\rho(x)$ with \ruleref{EC-InputDeref}.

        \item If $e = \SLetIn{x}{e_1}{e_2}$, then by applying the inductive hypothesis on $e_1$ gives it is either a value, it can step to $e_1'$ by \rtwo, or gets stuck according to \rthree to \rseven.
        In the first case, we can take a step with \ruleref{EC-Lift} and \ruleref{E-Let}.
        In the second case, we take a step with \ruleref{EC-Eval}.
        For the remaining, they fall under their respective cases.

        \item Suppose $e = \SCast{C}{v}$ and is well-typed.
        Then either $D <: C$ and it steps, or it gets stuck according to \rfive.
        
        \item Suppose $e = \SCall{v}{m}{\ell}{v}$ and is well-typed.
        Note that $\Lift{\ell}{e}$ is always defined, since no raw locations can appear in a method body.
        Thus it can step with \ruleref{EC-Lift} on method calls.

        \item The case for field access, \SIfN, \SProveN, and \SVerifyN all follow similarly to their respective cases in the proof of \corelang progress (\Cref{thm:core-progress}).
        
        \item If $e$ is a compute-and-prove block, then it can step via \ruleref{EC-ComputeAndProveInit}.
        
        \item If $e$ is a \SCnpN term, by the premise of well-typedness, the body $e'$ of the \CnpN term also is well-typed.
        We can therefore apply the inductive hypothesis to get that $e'$ is either \rone a value, \rtwo steps, \rthree gets stuck on a proof, \rfour gets stuck when evaluating a combined block, \rfive gets stuck on cast, \rsix gets stuck on an uninitialized dereference, or \rseven gets struck in another nested \CnpN term.
        For case \rone in the inductive hypothesis, we know that $e$ must type to $\Bool$, meaning it must be either $\True$ or $\False$ (since it cannot be an open variable).
        If $\True$, we can step $e$ with \ruleref{EC-ComputeAndProveTrue}.
        If $\False$, then we get stuck as is allowed.
        For \rtwo with the inductive hypothesis, using the step $e'$ takes inside, we can take a step with \ruleref{EC-ComputeAndProveStep}.
        The inductive hypothesis applies here because the predicate $\nestedok$ implies our preservation conditions $\Sigma' \cproves \sigmaC, \sigmaP'$ and $\Delta', A'' \cproves \rho'$ for the inner expression $e'$.
        All other failure cases \rthree to \rseven for the inductive hypothesis fall into case \rseven for the current expression.
    \end{itemize}
\end{proof}

\abssoundness*
\begin{proof}
    We explicitly define all the cases:
    \begin{enumerate}[label=(\roman*)]
        \item $\SExp'$ is a value;
        \item there is some $\SExp'',\sigma_\Compute'', \sigma_\Prove'', \rho''$ such that $\cpconf{\SExp'}{\sigmaC'}{\sigmaP'}{\rho'} \cpstep \cpconf{\SExp''}{\sigmaC''}{\sigmaP''}{\rho''}$;
        \item $\SExp' = \Eval{\SProveUsing*[e'']}$ where $\conf{\subst*{e''}{{\overline{x}}{\overline{v}}{\overline{y}}{\overline{u}}}}{\varnothing} \tstep^* \conf{\False}{\_}$ or it diverges; 
        \item $\SExp' = \Eval{\SCnpAdmin*{\SFalse}{\tilde{\sigmaP}}{\tilde{\rho}}}$;
        \item $\SExp' = \Eval{\SDeref{r_\Compute(\iota)}}$ where $\sigmaC(\iota) = \bot$; 
        \item $\SExp' = \Eval{\SCast{C}{(\SNew{D}{\SVals*})}}$ where $D \not<: C$; or
        \item $\SExp' = \Eval{\SCnpAdmin*{\SExp''}{\tilde{\sigmaP}}{\tilde{\rho}}}$ where $\SExp''$ itself gets stuck from any of \rthree to \rseven.
  \end{enumerate}
    By induction, applying progress (\Cref{thm:abs-progress}) and preservation (\Cref{thm:abs-preservation}).
\end{proof}

\begin{theorem}[Preservation of Full \zkstrudul]\label{thm:extended-preservation}
    If $\Sigma ; \Gamma \proves e : \TTau$ and $\Sigma \proves \sigma$ and $\langle e \divi \sigma \rangle \tstep \langle e' \divi \sigma' \rangle$ and $(e, \sigma, \Sigma, \SigmaP^{in}, \SigmaCP^{in}, \Delta^{in}, A^{in}, \varnothing) \nestedok$ for some $\SigmaP^{in}, \SigmaCP^{in}, \Delta^{in}, A^{in}$ then for some $\Sigma' \supseteq \Sigma$, $\Sigma' ; \Gamma \proves e' : \TTau'$ where $\TTau <: \TTau'$ and $\Sigma' \proves \sigma'$ and $(e', \sigma', \Sigma', \SigmaP^{in'}, \SigmaCP^{in'}, \Delta^{in'}, A^{in'}, \varnothing) \nestedok$ for some $\SigmaP^{in'}, \SigmaCP^{in'}, \Delta^{in'}, A^{in'}$.
\end{theorem}
\begin{proof}
    By induction on $e$.
    \begin{itemize}[itemsep=0.6em]
        \item All cases come from the subset corresponding to \corelang come directly from \Cref{thm:core-preservation}.
        The cases to handle compute-and-prove blocks (and their intermediate administrative syntax terms) follows very similarly to the last three cases of the proof of \Cref{thm:abs-preservation}.
    \end{itemize}
\end{proof}

\begin{theorem}[Progress of Full \zkstrudul]\label{thm:extended-progress}
    Suppose that $\Sigma ; \varnothing \proves e : \TTau$.
    Then either 
    \begin{enumerate}[label=(\roman*)]
        \item $e$ is a value;
        \item for any store $\sigma$ where $\Sigma \proves \sigma$, and for any set of contexts $(\SigmaP^{in}, \SigmaCP^{in}, \Delta^{in}, A^{in})$ such that the predicate $(e, \sigma, \Sigma, \SigmaP^{in}, \SigmaCP^{in}, \Delta^{in}, A^{in}, \varnothing) \nestedok$ (\Cref{app:nested-preservation}), there is some $e'$ and $\sigma'$ such that $\langle e \divi \sigma \rangle \tstep \langle e' \divi \sigma' \rangle$;
        \item $e = \Eval{\ProveWith*[e']}$ where $\conf{\subst*{e'}{{\overline{x}}{\overline{v}}{\overline{y}}{\overline{u}}}}{\varnothing} \tstep^* \conf{\False}{\_}$ or it diverges; 
        \item $e = \Eval{\TCnpAdmin*{\False}{\sigmaP}{\rho}}$;
        \item $e = \Eval{\TDeref{\iota}}$ where $\sigma(\iota) = \bot$; 
        \item $e = \Eval{(C)(\New D(\overline{v}))}$ where $D \not<: C$; or
        \item $e = \Eval{\TCnpAdmin*{\SExp'}{\tilde{\sigmaP}}{\tilde{\rho}}}$ where $\stuckpredzk{\SExp'}{\sigma}{\sigmaP}{\rho}$.
    \end{enumerate}
\end{theorem}
\begin{proof}
    By induction on $e$.
    \begin{itemize}[itemsep=0.6em]        
        \item All cases from the subset corresponding to \corelang come directly from \Cref{thm:core-progress}.
        The cases to handle compute-and-prove blocks (and their intermediate administrative syntax terms) follows very similarly to the last two cases of the proof of \Cref{thm:abs-progress}, using the $\nestedok$ predicate to ensure a step can be taken.
    \end{itemize}
\end{proof}

\begin{theorem}[Type Soundness of Full \zkstrudul]
\label{thm:extended-soundness}
  If $\proves e : \TTau$ and $\conf{e}{\varnothing} \tstep^* \conf*{e'}$,
  then one of the following holds.
  \begin{enumerate}[nosep,label={(\roman*)},leftmargin=*]
    \item $e'$ is a value.
    \item There is some~$e''$ and~$\sigma'$ such that $\conf*{e'} \tstep \conf{e''}{\sigma'}$.
    \item $e' = \Eval{\TProveUsing*[e'']}$ where either $\conf{\subst*{e''}{{\overline{x}}{\overline{v}}{\overline{y}}{\overline{u}}}}{\varnothing}$
    diverges or it steps to~$\False$; or
    \item $e' = \Eval{\TDeref{\iota}}$ where $\sigma(\iota) = \bot$
    \item $e' = \Eval{\TCast{C}{(\TNew{D}{\TVals*})}}$ where $D \not<: C$
    \item $e' = \Eval{\TCnpAdmin*{\False}{\sigmaP}{\rho}}$;
    \item $e' = \Eval{\TCnpAdmin*{\SExp''}{\tilde{\sigmaP}}{\tilde{\rho}}}$ where $\stuckpredzk{\SExp''}{\sigma}{\sigmaP}{\rho}$
  \end{enumerate}
\end{theorem}
\begin{proof}
    By induction, applying progress (\Cref{thm:extended-progress}) and preservation (\Cref{thm:extended-preservation}).
\end{proof}

\subsection{Confluence}
\begin{theorem}[Determinism of $\tstep$]\label{thm:target-determinism}
    If $\conf*{e} \tstep \conf{e'}{\sigma'}$ and $\conf*{e} \tstep \conf{e''}{\sigma''}$, then $e' = e''$ and $\sigma' = \sigma''$, both up to permuting location names.
\end{theorem}
\begin{proof}
    By induction on the semantic steps for $\tstep$.
    For all syntactic terms except for If, Let, and Verify, there is at most one semantic rule that can apply.
    In the case of If, the rule to be chosen of the two is based on the conditional guard.
    In the case of Let, if $e_1$ is a value, we apply the second Let rule, otherwise we invoke the inductive hypothesis on $e_1$.
    In the case of verify, the premises of the rules are disjoint.
\end{proof}

\begin{lemma}\label{lemma:step-preserves-lifting}
    For any expression $e$ and labels $\ell$ and $\ell'$ such that $\Lift{\ell}{e} = \Lift{\ell'}{e}$, if $\conf{e}{\varnothing} \tstep \conf{e'}{\varnothing}$ then $\Lift{\ell}{e'} = \Lift{\ell'}{e'}$.
\end{lemma}
\begin{proof}
    By cases on $\tstep$.
    For the two \IfN statement rules, applying the definition of $\Lift{\ell}{e}$ gives $\Lift{\ell}{e_i} = \Lift{\ell'}{e_i}$ for each subterm $i \in \{1,2\}$.
    The cases involving rules using references are vacuously true, as we demand the rules do not use locations meaningfully.
    The remaining rules involve decomposing equality between lifts on $e$ to equality between lifts on subterms of $e$, and the step either steps to a subterm or a constant value.
\end{proof}

\begin{theorem}[Determinism of $\cpstep$]\label{thm:surface-determinism}
    If $\cpconf*{e} \cpstep \cpconf{e'}{\sigmaC'}{\sigmaP'}{\rho'}$ and $\cpconf*{e} \cpstep \cpconf{e''}{\sigmaC''}{\sigmaP''}{\rho''}$, then $e' = e''$ and $\sigmaC' = \sigmaC''$ and $\sigmaP' = \sigmaP''$ and $\rho' = \rho''$, up to permuting location names.
\end{theorem}
\begin{proof}
    By direct induction on the semantic steps for $\cpstep$.
    For all rules lifted from \ruleref{EC-Lift}, determinism follows from \Cref{thm:target-determinism} and \Cref{lemma:step-preserves-lifting}.
    We call out the case for method calls in particular:
    they utilize the syntactic label $\ell$ to determine where to lift, which would violate \Cref{lemma:step-preserves-lifting} if not.
    For all other syntactic terms except for Let, there is at most one semantic rule that can apply.
    This includes references, which we assume are created deterministically.
    In the case of Let, if $e_1$ is a value, we apply the second Let rule, otherwise we invoke the inductive hypothesis on $e_1$.
\end{proof}

\begin{theorem}[Confluence of Concurrent Semantics]\label{thm:confluence-conc-target}
  The concurrent semantics $\pstep$ satisfies single-step diamond confluence:
  \[
    \begin{array}{rcl}
      % top eqn
      & 
      \pconf* {e_\Compute} {e_\Prove} 
      & \\

      % top row labels
      \vspace{1mm}
      {\text{\tiny{[\Compute-Step]}}} \hspace{-2mm} \targetconcsemSW \hspace{-10mm} & & \hspace{-12mm} \targetconcsemSE \hspace{0mm} {\text{\tiny{[\Prove-Step]}}} \\

      % left eqn
      \pconf {\ComputeHL{e_\Compute'}} {e_\Prove} {\ComputeHL{\sigmaC'}} {\sigmaP} {\ComputeHL{\rho'}}
      \hspace{-15mm} & & \hspace{-15mm}

      % right eqn
      \pconf {e_\Compute} {\ProveHL{e_\Prove'}} {\sigmaC} {\ProveHL{\sigmaP'}} {\rho} \\

      % bottom row labels
      \vspace{1mm} {\text{\tiny{[\Prove-Step]}}} \hspace{-2mm} \targetconcsemSE \hspace{-10mm} & & \hspace{-12mm} \targetconcsemSW \hspace{0mm} {\text{\tiny{[\Compute-Step]}}} \\

      % bottom eqn
      & 
      \pconf {\ComputeHL{e_\Compute'}} {\ProveHL{e_\Prove'}} {\ComputeHL{\sigmaC'}} {\ProveHL{\sigmaP'}} {\ComputeHL{\rho'}}
      & \\
    \end{array}
  \]
\end{theorem}
\begin{proof}
    More formally, we prove that if $\pconf* {e_\Compute}{e_\Prove} \pstep_{\trace_1} \pconf {e_{1,\Compute}}{e_{1,\Prove}}{\sigma_{1,\Compute}}{\sigma_{1,\Prove}}{\rho_1}$ and $\pconf {e_\Compute}{e_\Prove}{\sigma_\Compute}{\sigma_\Prove}{\rho} \pstep_{\trace_2} \pconf {e_{2,\Compute}}{e_{2,\Prove}}{\sigma_{2,\Compute}}{\sigma_{2,\Prove}}{\rho_2}$ then there is some ($e_{3,\Compute}$, $e_{3,\Prove}$), ($\sigma_{3,\Compute}$, $\sigma_{3,\Prove}$), and $\rho_3$ such that $\pconf {e_{1,\Compute}}{e_{1,\Prove}}{\sigma_{1,\Compute}}{\sigma_{1,\Prove}}{\rho_1} \pstep_{\trace_2} \pconf {e_{3,\Compute}}{e_{3,\Prove}}{\sigma_{3,\Compute}}{\sigma_{3,\Prove}}{\rho_3}$ and $\pconf {e_{2,\Compute}}{e_{2,\Prove}}{\sigma_{2,\Compute}}{\sigma_{2,\Prove}}{\rho_2} \pstep_{\trace_1} \pconf {e_{3,\Compute}}{e_{3,\Prove}}{\sigma_{3,\Compute}}{\sigma_{3,\Prove}}{\rho_3}$ and $\trace_1 :: \trace_2 \traceq \trace_2 :: \trace_1$.

    We prove this by case-analysis over the semantic definition of \Cref{app:target-conc-lang}.
    \begin{itemize}
        \item Suppose that the rule \ruleref{EP-Let1C} and \ruleref{EP-Let1P} can be applied to $e_\Compute$ and $e_\Prove$. If the same rule is applied to both, then using determinism (\Cref{thm:target-determinism}), both of these syntactic pairs are confluent. Suppose that different rules are applied to $e_\Compute$ and $e_\Prove$. This pair is shown to be confluent by the following diagram.
        \begin{center}
            \footnotesize
            $$
                \begin{array}{rcl}
                % top eqn
                & 
                \pconf {\LetIn{x:t}{e_{\Compute 1}}{e_{\Compute 2}}} {\LetIn{x}{e_{\Prove 1}}{e_{\Prove 2}}} {\sigma_\Compute} {\sigma_\Prove} {\rho}
                & \\

                % top row labels
                \vspace{1mm}
                {\text{\tiny{\ruleref{EP-Let1C}}}} \hspace{-2mm} \targetconcsemSW \hspace{-10mm} & & \hspace{-12mm} \targetconcsemSE \hspace{0mm} {\text{\tiny{\ruleref{EP-Let1P}}}} \\

                % left eqn
                \pconf {\ComputeHL{\LetIn{x:t}{e_{\Compute 1}'}{e_{\Compute 2}}}} {\LetIn{x:t}{e_{\Prove 1}}{e_{\Prove 2}}} {\ComputeHL{\sigma_\Compute'}} {\sigma_\Prove} {\ComputeHL{\rho'}} 
                \hspace{-32mm} & & \hspace{-32mm}

                % right eqn
                \pconf {\LetIn{x:t}{e_{\Compute 1}}{e_{\Compute 2}}} {\ProveHL{\LetIn{x:t}{e_{\Prove 1}'}{e_{\Prove 2}}}} {\sigma_\Compute} {\ProveHL{\sigma_\Prove'}} {\rho}
                \\

                % bottom row labels
                \vspace{1mm} {\text{\tiny{\ruleref{EP-Let1P}}}} \hspace{-2mm} \targetconcsemSE \hspace{-10mm} & & \hspace{-12mm} \targetconcsemSW \hspace{0mm} {\text{\tiny{\ruleref{EP-Let1C}}}} \\

                % bottom eqn
                & 
                \pconf {\ComputeHL{\LetIn{x:t}{e_{\Compute 1}'}{e_{\Compute 2}}}} {\ProveHL{\LetIn{x:t}{e_{\Prove 1}'}{e_{\Prove 2}}}} {\ComputeHL{\sigma_\Compute'}} {\ProveHL{\sigma_\Prove'}} {\ComputeHL{\rho'}}
                & \\
                \end{array}
            $$
        \end{center}
        In the diagram, the modified elements from executing the \Compute step are denoted in yellow, and the modified elements from the \Prove step are denoted in pink. Note that the two rules modify different components that do not interfere with each other, allowing them to converge at the same state simply by applying the other semantic rule.
        
        \item Suppose that the rule \ruleref{EP-StepC} and \ruleref{EP-StepP} can be applied to $e_\Compute$ and $e_\Prove$. We note that one pre-requisite of this syntactic form is that $\langle e_\ell \divi \sigma_\ell \rangle \tstep \langle e_\ell' \divi \sigma_\ell' \rangle$ for $\ell = \Compute, \Prove$. The statement $\langle e_\ell \divi \sigma_\ell \rangle \tstep \langle e_\ell' \divi \sigma_\ell' \rangle$ may be referencing any of the semantic rules in \Cref{app:core-semantics}. Some of the rules in \Cref{app:core-semantics} modify the store element, while others don't. If the store remains unmodified, then the proof continues below. If the store is modified, we apply \Cref{thm:target-determinism} to show that the modified stores converge to the same store.

        Suppose that the same rule is applied to both $e_\Compute$ and $e_\Prove$. Then the pair is confluent by \Cref{thm:target-determinism}. Suppose that different rules are applied to $e_\Compute$ and $e_\Prove$. For this case to be applicable, we know that $\langle e_\Compute \divi \sigma_\Compute \rangle \tstep \langle e_\Compute' \divi \sigma_\Compute' \rangle$ and $\langle e_\Prove \divi \sigma_\Prove \rangle \tstep \langle e_\Prove' \divi \sigma_\Prove' \rangle$ must be true (regardless of what step is being taken in the target semantics). We can prove confluence by applying the semantic rule in both orders.
        $$
            \begin{array}{rcl}
            % top eqn
            & 
            \pconf {e_\Compute} {e_\Prove} {\sigma_\Compute} {\sigma_\Prove} {\rho} 
            & \\

            % top row labels
            \vspace{1mm}
            {\text{\tiny{\ruleref{EP-StepC}}}} \hspace{-2mm} \targetconcsemSW \hspace{-10mm} & & \hspace{-12mm} \targetconcsemSE \hspace{0mm} {\text{\tiny{\ruleref{EP-StepP}}}} \\

            % left eqn
            \pconf {\ComputeHL{e_\Compute'}} {e_\Prove} {\ComputeHL{\sigma_\Compute'}}  {\sigma_\Prove} {\rho}
            \hspace{-15mm} & & \hspace{-15mm}
        
            % right eqn
            \pconf {e_\Compute} {\ProveHL{e_\Prove'}} {\sigma_\Compute} {\ProveHL{\sigma_\Prove'}} {\rho} \\

            % bottom row labels
            \vspace{1mm} {\text{\tiny{\ruleref{EP-StepP}}}} \hspace{-2mm} \targetconcsemSE \hspace{-10mm} & & \hspace{-12mm} \targetconcsemSW \hspace{0mm} {\text{\tiny{\ruleref{EP-StepC}}}} \\

            % bottom eqn
            & 
            \pconf {\ComputeHL{e_\Compute'}} {\ProveHL{e_\Prove'}} {\ComputeHL{\sigma_\Compute'}} {\ProveHL{\sigma_\Prove'}} {\rho} 
            & \\
            \end{array}
        $$
        
        \item Suppose that \ruleref{EP-WitAssign} is applicable to $e_\Compute$. We note that this rule can only be done by the \Compute component. Then there exists only one pair of syntactic form to prove which is if both branches step \ruleref{EP-WitAssign}. This is confluent by \Cref{thm:target-determinism}.
        
        \item Suppose that rule \ruleref{EP-WitDerefC} and \ruleref{EP-WitDerefP} can be applied to $e_\Compute$ and $e_\Prove$. Suppose that the same rule is applied to both $e_\Compute$ and $e_\Prove$. Then the pair is confluent by \Cref{thm:target-determinism}. Suppose that different rules are applied to $e_\Compute$ and $e_\Prove$. We require that $x \in \dom(\rho)$. Confluence follows from the diagram below.
        $$
            \begin{array}{rcl}
            % top eqn
            & 
            \pconf {\witderef x} {\witderef x} {\sigma_\Compute} {\sigma_\Prove} {\rho}
            & \\

            % top row labels
            \vspace{1mm}
            {\text{\tiny{\ruleref{EP-WitDerefC}}}} \hspace{-2mm} \targetconcsemSW \hspace{-10mm} & & \hspace{-12mm} \targetconcsemSE \hspace{0mm} {\text{\tiny{\ruleref{EP-WitDerefP}}}} \\

            % left eqn
            \pconf {\ComputeHL{\rho(x)}} {\witderef x} {\sigma_\Compute} {\sigma_\Prove} {\rho}
            \hspace{-15mm} & & \hspace{-15mm}
        
            % right eqn
            \pconf {\witderef x} {\ProveHL{\rho(x)}} {\sigma_\Compute} {\sigma_\Prove} {\rho}
            \\

            % bottom row labels
            \vspace{1mm} {\text{\tiny{\ruleref{EP-WitDerefP}}}} \hspace{-2mm} \targetconcsemSE \hspace{-10mm} & & \hspace{-12mm} \targetconcsemSW \hspace{0mm} {\text{\tiny{\ruleref{EP-WitDerefC}}}} \\

            % bottom eqn
            & 
            \pconf {\ComputeHL{\rho(x)}} {\ProveHL{\rho(x)}} {\sigma_\Compute} {\sigma_\Prove} {\rho}
            & \\
            \end{array}
        $$
    \end{itemize}

    Note that in all cases, traces can only appear on the compute side, so at least one of $t_1$ and $t_2$ is $\Nulltrace$, meaning $\trace_1 :: \trace_2 \traceq \trace_2 :: \trace_1$.
\end{proof}

\subsection{Bisimulation}
Note that all proofs up to \Cref{thm:adequacy-of-combined} are on expressions from the subset of the language that do not contain (nested) compute and prove blocks.
Then, \Cref{thm:adequacy-of-combined} will separately prove adequacy if there are potentially nested compute-and-prove blocks.
Because method calls can introduce nested compute-and-prove blocks, we initially assume they've been compiled before method calls.
For shorthand we sometimes use $\storetuple$ to represent $(\sigma_\Compute, \sigma_\Prove, \rho)$ and $\storetuple'$ for $(\sigma_\Compute', \sigma_\Prove', \rho'$).

\begin{lemma}[Bisimulation of Annotated Semantics]\label{lemma:annotated-bisim}
    For any expression $e$ and its corresponding annotated version $\hat{e}$ using the typing derivation to assign annotations to values,
    then $\cpconf*{e} \cpstep[t] \cpconf{e'}{\sigmaC'}{\sigmaP'}{\rho'}$ if and only if $\cpconf*{\hat{e}} \cpstep[t] \cpconf{\hat{e}'}{\sigmaC'}{\sigmaP'}{\rho'}$
\end{lemma}
\begin{proof}
    By induction on $e$, a direct lock-step bisimulation is possible by linking the corresponding rules.
\end{proof}

\begin{lemma}[Projected Annotated Expressions]\label{lemma:annotated-proj}
    The projections of any expression $e$ and its corresponding annotated version $\hat{e}$ are identical.
    That is, $\denL{e} = \denL{\hat{e}}$.
\end{lemma}
\begin{proof}
    By induction on $e$.
    We show some cases, where the rest can be inferred.
    \begin{itemize}
        \item ($e = v$)

        If $v$ is typed at label~$\ell'$, then $\hat{e} = (v)_\annellp$.
        Then either~$\denL{e} = \denL{\hat{e}} = v$, or $\denL{e} = \denL{\hat{e}} = ()$.

        \item ($e = \Let x\ty\tplab{\TTau}{\ell'} = e_1 \In e_2$)

        If $\ell \subseteq \ell'$, then
        $\denL{e} = (\Let x = \denL{e_1} \In \denL{e_2}) = (\Let x = \denL{\hat{e_1}} \In \denL{\hat{e_2}}) = \denL{\hat{e}}$.
        Otherwise, either 
        $\denL{e} = \denL{e_1} ; \denL{e_2} = \denL{\hat{e_1}} ; \denL{\hat{e_2}} = \denL{\hat{e}}$
        or $\denL{e} = \denL{e_2} = \denL{\hat{e_2}} = \denL{\hat{e}}$,
        depending on whether~$\denL{e_1} = \denL{\hat{e_1}}$ is a value or not.
        
        \item ($e = \If v \Then e_1 \Else e_2$)

        Supposing $v$ is typed at label~$\ell'$, if~$\ell \subseteq \ell'$ then
        \begin{align*}
            \denL{e} &= \If v \Then \denL{e_1} \Else \denL{e_2}\\
            &= \If v \Then \denL{\hat{e_1}} \Else \denL{\hat{e_2}}\\
            &= \denL{\hat{e}}.
        \end{align*}
        Otherwise $\denL{e} = \denL{e_1} = \denL{\hat{e_1}} = \denL{\hat{e}}$.
    \end{itemize}
\end{proof}

\begin{lemma}\label{lemma:raising-type}
    If $e$ is an expression where $\Sigma ; \Gamma \proves e : \TTau$, and $\sigma$ is a substitution mapping variables to values with matching types,
    then $\Sigma ; \Gamma ; \Delta ; A \proves \LiftAnn*{e}[\sigma] : \tplab{\TTau}{\ell}$.
\end{lemma}
\begin{proof}
    By induction on the definition of $\LiftAnn*{\cdot}$ in \Cref{app:ann-lifting-def} and then applying the relevant typing rule.
\end{proof}

\begin{lemma}\label{lemma:value-substitution-proj}
    If $\sigma$ is a substitution mapping variables to values, $e$ is an annotated expression,
    and $\denL{e[\sigma]}$ is a value,
    then $\denL{e}$, given that it exists, is also a value.
\end{lemma}
\begin{proof}
    By induction on $e$.
    We show some cases, where the rest can be inferred.
    \begin{itemize}
        \item ($e = (v)_\annellp$)

        This case is trivial since~$\denL{(v)_\annellp}$ always projects to either~$v$ or~$()$,
        both of which are values.

        \item ($e = \Let x\ty\tplab{\TTau}{\ell'} = e_1 \In e_2$)

        In order for $\denL{e[\sigma]}$ to be a value, it must be the case
        that $\ell \not\subseteq \ell'$, and both $\denL{e_1[\sigma]}$ and $\denL{e_2[\sigma]}$ are values.
        But then by induction both $\denL{e_1}$ and $\denL{e_2}$ are values, so
        $\denL{e} = \denL{e_2}$, which is a value.
        
        \item ($e = \If (v)_\annellp \Then e_1 \Else e_2$)

        In order for $\denL{e[\sigma]}$ to be a value, it must be the case
        that $\ell \not\subseteq \ell'$, and that $\denL{e_1[\sigma]} = \denL{e_2[\sigma]}$ is a value.
        Then by induction $\denL{e_1} = \denL{e_2}$ is a value, so
        $\denL{e} = \denL{e_1} = \denL{e_2}$, which is a value.
        
        \item ($e = \Provet \alpha = \CircuitDefault \With \overline{v}[\overline{w}]$)

        It must be the case that $\ell = \Prove$ and $\denL{e[\sigma]} = ()$, so clearly $\denL{e} = ()$ is a value.
        The same argument applies to $\Verify (p)_\annellp \Proves \alpha \With \overline{v}$.
    \end{itemize}
\end{proof}

\begin{lemma}[Simultaneous Substitution Preserves Projections]
    If $\sigma_1$ is a substitution mapping variables to values, $e$ an annotated expression,
    $\denL{\sigma_1(x)} = \sigma_2(x)$ for each variable~$x$,
    and $\denL{e}$ is defined,
    then $\denL{e[\sigma_1]} = \denL{e}[\sigma_2]$, where both sides are also defined.
\end{lemma}
\begin{proof}
    By induction on $e$.
    We show some cases, where the rest can be inferred.
    \begin{itemize}
        \item ($e = v$)

        The cases when $e$ is a value are all straightforward, noting that the
        action of a substitution does not change any labels in the expression, so
        identical cases are chosen for the left and right-hand projections.

        \item ($\Let x\ty\tplab{\TTau}{\ell'} = e_1 \In e_2$)

            If $\ell \subsetl \ell'$ then by induction
            \begin{align*}
                \denL{(\Let x\ty\tplab{\TTau}{\ell'} = e_1 \In e_2)[\sigma_1]} &= \denL{\Let x\ty\tplab{\TTau}{\ell'} = e_1[\sigma_1] \In e_2[x \mapsto x, y \mapsto \sigma_1(y)]} \\
                &= \Let x\ty\TTau = \denL{e_1[\sigma_1]} \In \denL{e_2[x \mapsto x, y \mapsto \sigma_1(y)]} \\
                &= \Let x\ty\TTau = \denL{e_1}[\sigma_2] \In \denL{e_2}[x \mapsto x, y \mapsto \sigma_2(y)] \\
                &= (\Let x\ty\TTau = \denL{e_1} \In \denL{e_2})[\sigma_2] \\
                &= \denL{\Let x\ty\tplab{\TTau}{\ell'} = e_1 \In e_2}[\sigma_2].
            \end{align*}
            Otherwise suppose $\ell \not\subsetl \ell'$ and $\denL{e_1} = v$.
            By induction, $\denL{e_1[\sigma_1]} = v[\sigma_2]$, which must also be a value.
            Therefore
            \begin{align*}
                \denL{(\Let x\ty\tplab{\TTau}{\ell'} = e_1 \In e_2)[\sigma_1]} &= \denL{e_2[\sigma_1]} \\
                &= \denL{e_2}[\sigma_2] \\
                &= \denL{\Let x\ty\tplab{\TTau}{\ell'} = e_1 \In e_2}[\sigma_2].
            \end{align*}
            Finally, suppose that $\ell \not\subsetl \ell'$ but $\denL{e_1}$ is not a value.
            Then by \Cref{lemma:value-substitution-proj}, $\denL{e_1[\sigma_1]}$ is also not a value, so
            \begin{align*}
                \denL{(\Let x\ty\tplab{\TTau}{\ell'} = e_1 \In e_2)[\sigma_1]} &= \denL{e_1[\sigma_1]} ; \denL{e_2[\sigma_1]} \\
                &= \denL{e_1}[\sigma_2] ; \denL{e_2}[\sigma_2] \\
                &= \denL{\Let x\ty\tplab{\TTau}{\ell'} = e_1 \In e_2}[\sigma_2].
            \end{align*}

        \item ($\If (v)_\annellp \Then e_1 \Else e_2$)
        
            If $\ell \subsetl \annellp$, then
            \begin{align*}
                \denL{(\If (v)_\annellp \Then e_1 \Else e_2)[\sigma_1]} &= \denL{\If (v[\sigma_1])_\annellp \Then e_1[\sigma_1] \Else e_2[\sigma_1]} \\
                &= \If v[\sigma_2] \Then \denL{e_1[\sigma_1]} \Else \denL{e_2[\sigma_1]} \\
                &= \If v[\sigma_2] \Then \denL{e_1}[\sigma_2] \Else \denL{e_2}[\sigma_2] \\
                &= \denL{\If (v)_\annellp \Then e_1 \Else e_2}[\sigma_2].
            \end{align*}
            Otherwise if $\ell \not\subsetl \annellp$ and $\denL{e_1} = \denL{e_2}$, then
            by induction $\denL{e_1[\sigma_1]} = \denL{e_1}[\sigma_2] = \denL{e_2}[\sigma_2] = \denL{e_2[\sigma_1]}$, so
            \begin{align*}
                \denL{(\If (v)_\annellp \Then e_1 \Else e_2)[\sigma_1]} &= \denL{e_1[\sigma_1]} \\
                &= \denL{e_1}[\sigma_2] \\
                &= \denL{\If (v)_\annellp \Then e_1 \Else e_2}[\sigma_2].
            \end{align*}
            Otherwise the projection is undefined, which violates the assumption.
    \end{itemize}
\end{proof}

\begin{corollary}[Substitution Preserves Projections]\label{lemma:lessthan-substitution1}
    For any annotated $e$ and $v$, if $\denL{e}$ is defined, then $\denL{e[x \mapsto v]} = \denL{e}[x \mapsto \denL{v}]$,
    and both sides are defined.
\end{corollary}

\begin{lemma}\label{lemma:raise-lessthan-unit}
    Suppose $e$ is an expression and $\sigma$ is a substitution mapping variables to values. 
    If $\ell' \not\subsetl \ell$ and $\ell \not\subsetl \ell'$ then $\denL{\LiftAnn{\ell'}{e}[\sigma]} = ()$.
\end{lemma}
\begin{proof}
    By induction on $e$.

    For all values, the raising operation gives them an annotated label $\ell'$: $\LiftAnn{\ell'}{v} = (\Lift{\ell'}{v}[\sigma])_\annellp$.
    Thus,$\denL{\LiftAnn{\ell'}{v}[\sigma]} = ()$.

    There are no computed input operations since $e$ comes from the outer language.

    For all but $\LetN$ and $\IfN$ the raising operation adds an annotated label to the value that determines projection, so they also project to $()$.

    Suppose $e = \IfThenElse{v}{e_1}{e_2}$, meaning $\LiftAnn{\ell'}{e}[\sigma] = \If (v)_\annellp \Then \LiftAnn{\ell'}{e_1}[\sigma] \Else \LiftAnn{\ell'}{e_2}[\sigma]$.
    By the inductive hypothesis, $\denL{\LiftAnn{\ell'}{e_1}[\sigma]} = \denL{\LiftAnn{\ell'}{e_2}[\sigma]} = ()$.
    Therefore, since $\ell \not\subsetl \annellp$, we get $\denL{\LiftAnn{\ell'}{e}[\sigma]} = ()$.

    Finally, suppose that $e = \LetIn{x\ty\TTau}{e_1}{e_2}$, meaning $\LiftAnn{\ell'}{e}[\sigma] = \Let x\ty\tplab{\TTau}{\ell'} = \LiftAnn{\ell'}{e_1}[\sigma] \In \LiftAnn{\ell'}{e_2}[\sigma]$.
    By the inductive hypothesis, $\denL{\LiftAnn{\ell'}{e_1}[\sigma]} = \denL{\LiftAnn{\ell'}{e_2}[\sigma]} = ()$.
    Therefore, since $\ell \not\subsetl \ell'$, we get $\denL{\LiftAnn{\ell'}{e}[\sigma]} = ()$.
\end{proof}

\begin{lemma}\label{lemma:raise-lessthan-all}
    Suppose $e$ is an expression and $\sigma$ is a substitution mapping variables to values.
    If $\ell = \ell'$ then $\denL{\LiftAnn{\ell'}{e}[\sigma]} = e[\sigma]$.
\end{lemma}
\begin{proof}
    By induction on $e$.

    For all values, the raising operation gives them an annotated label $\ell'$: $\LiftAnn{\ell'}{v}[\sigma] = (\Lift{\ell'}{v}[\sigma])_\annellp$.
    Thus,$\denL{\LiftAnn{\ell'}{v}[\sigma]} = v[\sigma]$.

    There are no computed input operations since $e$ comes from the outer language.

    For all but $\LetN$ and $\IfN$ the raising operation adds an annotated label to the value that determines projection, so they also project to non-units.

    Suppose $e = \IfThenElse{v}{e_1}{e_2}$, meaning $\LiftAnn{\ell'}{e}[\sigma] = \If (v)_\annellp \Then \LiftAnn{\ell'}{e_1}[\sigma] \Else \LiftAnn{\ell'}{e_2}[\sigma]$.
    By the inductive hypothesis, $\denL{\LiftAnn{\ell'}{v}[\sigma]} = v[\sigma]$ and $\denL{\LiftAnn{\ell'}{e_1}[\sigma]} = e_1[\sigma]$ and $\denL{\LiftAnn{\ell'}{e_2}[\sigma]} = e_2[\sigma]$.
    Thus, since $\ell' = \ell$, $\denL{\LiftAnn{\ell'}{e}[\sigma]} = \IfThenElse{v[\sigma]}{e_1[\sigma]}{e_2[\sigma]} = (\IfThenElse{v}{e_1}{e_2})[\sigma]$.

    Finally, suppose that $e = \LetIn{x\ty\TTau}{e_1}{e_2}$, meaning $\LiftAnn{\ell'}{e}[\sigma] = \Let x\ty\tplab{\TTau}{\ell'} = \LiftAnn{\ell'}{e_1}[\sigma] \In \LiftAnn{\ell'}{e_2}[\sigma]$.
    By the inductive hypothesis, $\denL{\LiftAnn{\ell'}{e_1}[\sigma]} = e_1[\sigma]$ and $\denL{\LiftAnn{\ell'}{e_2}[\sigma]} = e_2[\sigma]$.
    Thus, since $\ell' = \ell$, $\denL{\LiftAnn{\ell'}{e}[\sigma]} = \LetIn{x\ty\TTau}{e_1[\sigma]}{e_2[\sigma]} = e[\sigma]$.
\end{proof}

\begin{lemma}[Annotated Label Preservation]\label{lemma:annotated-preservation}
    If $\Sigma ; \Gamma ; \Delta ; A \proves e : \tplab{\TTau}{\ell'} \produces A'$ and $\langle e \divi \storetuple \rangle \cpstep^* \langle (v)_\annell \divi \storetuple' \rangle$ then $\annell = \ell'$.
\end{lemma}
\begin{proof}
    By induction on $e$.

    \begin{itemize}[itemsep=0.6em]
    
        \item
        If $e$ is a value, then the term takes 0 steps.
        In this case, the annotated label is the type of the value, so we're done.

        \item
        If $e = \Refsub{C} (v)_\annell$, then $\langle e \divi \storetuple \rangle \cpstep \langle (r_\Compute(\iota))_\annc \divi \storetuple' \rangle$.
        Note that $\ell' = \Compute$ by the typing rule for $\Reft$ so $\ell' = \annell$.
        Similar reasoning proves the cases for the other two $\Reft$.

        \item
        If $e = \bang (r_\Compute(\iota))_\annc$ then $\langle e \divi \storetuple \rangle \cpstep \langle (\sigma_\Compute(\iota))_\annc \divi \storetuple \rangle$.
        Note that $\ell' = \tplab{(\Reft \tplab{\TTau}{\Compute})}{\Compute}$ by the typing rule for deref, so $\ell' = \annell$.
        Similar reasoning proves the cases for the other two dereferences and computed input dereferencing.

        \item
        If $e = (r_\Compute(\iota))_\Compute \coloneq (v)_\annell$, then $\langle e \divi \storetuple \rangle \cpstep \langle ()_\annellp \divi \storetuple' \rangle$.
        This $\annellp$ can be anything, so we set $\annellp = \ell' = \Compute$.
        Similar reasoning proves the cases for the other two assignments, and for computed input assignment.

        \item
        If $e = (C)(v)_\annellp$ then $\langle e \divi \storetuple \rangle \cpstep \langle (v)_\annellp \divi \storetuple \rangle$.
        Here, $\ell' = \annellp$ by the casting rule, and $\annellp = \annell$ by annotating the AST.
        Similar reasoning applies for the cases when $e$ is a field access, a proof cast, or verify.

        \item
        If $e = (v)_\annellp.m(\overline{w})$ then $\langle e \divi \storetuple \rangle \cpstep \langle \Lift*{e_m}[\overline{x} \mapsto \overline{w}, \This \mapsto (v)_\annellp] \divi \storetuple \rangle$ where $e_m$ is the method body and $\overline{x}$ are its inputs.
        By \Cref{lemma:raising-type} and the inductive hypothesis, $\annellp = \annell = \ell'$.

        \item
        If $e = \Provet \alpha = \CircuitDefault \With \overline{v}[\overline{w}]$ then $\langle e \divi \storetuple \rangle \cpstep \langle (\pi)_\Compute \divi \storetuple \rangle$.
        The typing rule requires $\ell'$ to be $\Compute$. 

        \item
        If $e = \IfThenElse{(\True)_\annellp}{e_1}{e_2}$ then $\langle e \divi \storetuple \rangle \cpstep \langle e_1 \divi \storetuple \rangle$.
        Note that by the typing rule for $\IfN$ $\Sigma ; \Gamma ; \Delta ; A \proves e_1 : \tplab{\TTau}{\ell'} \produces A'$.
        By the inductive hypothesis, $\langle e_1 \divi \storetuple \rangle \cpstep^* \langle (v)_\annellp \divi \storetuple' \rangle$, meaning $\langle e \divi \storetuple \rangle \cpstep^* \langle (v)_\annellp \divi \storetuple' \rangle$.
        Similar reasoning proves the case when the condition is \False.

        \item
        If $e = \LetIn{x\ty\tplab{\TTau}{\ell'}}{e_1}{e_2}$.
        We know that $\langle e \divi \storetuple \rangle \cpstep^* \langle (v)_\annell \divi \storetuple' \rangle$.
        This means $\langle e_1 \divi \storetuple \rangle \cpstep^* \langle (v_1)_\annellp \divi \storetuple'' \rangle$, meaning $\langle e \divi \storetuple \rangle \cpstep^* \langle \LetIn{x\ty\tplab{\TTau}{\ell'}}{(v_1)_\annellp}{e_2} \divi \storetuple'' \rangle \cpstep \langle e[(x)_\annellp \mapsto (v_1)_\annellp] \divi \storetuple'' \rangle$.
        By the inductive hypothesis, $\langle e[(x)_\annellp \mapsto (v_1)_\annellp] \divi \storetuple'' \rangle \cpstep^* (v)_\annell \divi \storetuple'$, meaning $\ell' = \annell$.

    \end{itemize}
\end{proof}

\begin{theorem}[Single Step Concurrent Completeness]\label{thm:singlestepbisimcomplete}
    If $e$ is an annotated expression where $\cpconf*{e} \cpstep_{{\trace_1}} \cpconf{e'}{\sigmaC'}{\sigmaP'}{\rho'}$, then there is some $e'_\Compute$ and $e'_\Prove$ such that $\pconf*{\denL[\Compute]{e}}{\denL[\Prove]{e}} \pstep^{\{1,2\}}_{{\trace_2}} \pconf{e_\Compute'}{e_\Prove'}{\sigmaC'}{\sigmaP'}{\rho'}$, where $\denC{e'} \lessthan e'_\Compute$ and $\denP{e'} \lessthan e'_\Prove$ and ${\trace_1} \traceq {\trace_2}$.
    In the case where only one step is taken, either $\denP{e}$ or $\denC{e}$ remains the same.
    In the case where two steps are taken, then $\pconf*{\denL[\Compute]{e}}{\denL[\Prove]{e}} \pstep \pconf{e_\Compute'}{\denL[\Prove]{e}}{\sigmaC'}{\sigmaP}{\rho'} \pstep \pconf{e_\Compute'}{e_\Prove'}{\sigmaC'}{\sigmaP'}{\rho'}$.
\end{theorem}
\begin{proof}
    By induction on the semantic rules $\cpstep$.
    We elide the annotated label $(v)_\annell$ when it's not used or is obvious from context.
    \begin{itemize}[itemsep=0.6em]
        \item (Let1)
        
            Suppose that $\langle {\LetIn{x\ty\tplab{\TTau}{\ell}}{e_1}{e_2}} \divi \storetuple \rangle \cpstep_\trace \langle {\LetIn{x\ty\tplab{\TTau}{\ell}}{e_1'}{e_2}} \divi \storetuple' \rangle$ from $\langle {e_1} \divi \storetuple \rangle \cpstep_\trace \langle {e_1'} \divi \storetuple' \rangle$.

            First, consider the case where $\ell = \ComputeProve$.
            In this case, $\denC{{\LetIn{x\ty\tplab{\TTau}{\ComputeProve}}{e_1}{e_2}}} = {\Let x\ty\TTau = } \denC{{e_1}} {\In} \denC{{e_2}}$ and $\denP{{\LetIn{x\ty\tplab{\TTau}{\ComputeProve}}{e_1}{e_2}}} = {\Let x\ty\TTau = } \denP{{e_1}} {\In} \denP{{e_2}}$.
            By the inductive hypothesis, $\pconf*{\denC{e_1}}{\denP{e_2}} \pstep^+_\trace \pconf{e_{1\Compute}'}{e_{1\Prove}'}{\sigmaC'}{\sigmaP'}{\rho'}$, where $\denC{{e_1'}} \lessthan e_{1\Compute}'$ and $\denP{{e_1'}} \lessthan e_{1\Prove}'$ and $\trace \traceq \trace$.
            So $\pconf*{\denC{e}}{\denP{e}} \pstep_\trace^+ \pconf{{\Let x\ty\TTau = } e_{1\Compute}' {\In} \denC{{e_2}}}{{\Let x\ty\TTau = } e_{1\Prove}' {\In} \denP{{e_2}}}{\sigmaC'}{\sigmaP'}{\rho'}$.
            Since $\denC{{e_1'}} \lessthan e_{1\Compute}'$ and $\denP{{e_1'}} \lessthan e_{1\Prove}'$, so too are $\denC{e'} \lessthan e'_\Compute$ and $\denP{e'} \lessthan e'_\Prove$ by definition of $\lessthan$.

            The case where $\ell = \Compute$ follows similar reasoning, except $\denP{e} = \denP{{e_2}}$ if $\denP{{e_1}} = v$; or $\denP{{e_1}} ; \denP{{e_2}}$ if not.
            In the first case, only the compute side takes a step, since the prove side is a value.
            In the second case, the logic follows the same as above if $\denP{{e_1'}}$ is not a value.
            If it is a value, then we apply an additionall $\lessthan$ to get the desired result.

            Similarly, the case where $\ell = \Prove$ is similar to the case of $\ell = \Compute$.

        \item (Let2)
            Suppose that $\langle {\LetIn{x\ty\tplab{\TTau}{\ell}}{v}{e'}} \divi \storetuple \rangle \cpstep_\Nulltrace \langle {e[x \mapsto v]} \divi \storetuple \rangle$.

            First, consider the case where $\ell = \ComputeProve$.
            Then, $\denC{{\LetIn{x\ty\tplab{\TTau}{\ell}}{v}{e'}}} = {\Let x\ty\TTau =} \denC{v} {\In} \denC{e'}$ and $\denP{{\LetIn{x\ty\tplab{\TTau}{\ell}}{v}{e'}}} = {\Let x\ty\TTau =} \denP{v} {\In} \denP{e'}$.
            So $\pconf*{\denC{e}}{\denP{e}} \pstep_\Nulltrace \pconf*{\denC{e'}{[x \mapsto}\denC{v}{]}}{\denP{e}} \pstep_\Nulltrace \pconf*{\denC{e'}{[x \mapsto}\denC{v}{]}}{\denP{e'}{[x \mapsto}\denC{v}{]}}$.
            Finally, $\denC{{e'[x \mapsto v]}} = \denC{e'}{[x \mapsto}\denC{v}{]}$ and $\denC{{e'[x \mapsto v]}} = \denP{e'}{[x \mapsto}\denC{v}{]}$ by \Cref{lemma:lessthan-substitution1}.

            Now, suppose that $\ell = \Compute$, meaning $\denC{{\LetIn{x\ty\tplab{\TTau}{\ell}}{v}{e'}}} = {\Let x\ty\TTau =} \denC{v} {\In} \denC{e'}$ and $\denP{{\LetIn{x\ty\tplab{\TTau}{\ell}}{v}{e'}}} = \denP{e'}$ since $\denP{(v)_\Compute} = ()$ by \Cref{lemma:annotated-preservation}, the type system, and definition of $\denL{\cdot}$.
            So $\pconf*{\denC{e}}{\denP{e}} \pstep_\Nulltrace \pconf*{\denC{e'}{[x \mapsto}\denC{v}{]}}{\denP{e'}}$.
            Note that ${x}$ is not free in $\denP{{e'}}$, so $\denP{{e'}} = \denP{{e'}}[{x \mapsto ()}]$.
            The proof finishes by applying \Cref{lemma:lessthan-substitution1} as in the previous case but just on the compute side.

            The proof follows similarly for $\ell = \Prove$ as with $\ell = \Compute$, where $\Compute$ and $\Prove$ are switched.

            In all cases, $\Nulltrace \traceq \Nulltrace :: \Nulltrace$.
        \item (IfT)
        
            Suppose $\langle {\IfThenElse{(v)_\annell}{e_1}{e_2}} \divi \storetuple \rangle \cpstep_\Nulltrace \langle {e_1} \divi \storetuple \rangle$.

            First, consider the case where $\annell = \ComputeProve$, meaning $\denC{{\IfThenElse{(v)_\annell}{e_1}{e_2}}} = {\If v \Then} \denC{{e_1}} {\Else} \denC{{e_2}}$ and $\denP{{\IfThenElse{(v)_\annell}{e_1}{e_2}}} = {\If v \Then} \denP{{e_1}} {\Else} \denP{{e_2}}$.
            Then, $\pconf*{\denC{e}}{\denP{e}} \pstep_\Nulltrace \pconf*{\denC{e_1}}{\denP{e}} \pstep_\Nulltrace \pconf*{\denC{e_1}}{\denP{e_1}}$.

            Next, consider when $\annell = \Compute$, meaning $\denC{{\IfThenElse{(v)_\annell}{e_1}{e_2}}} = {\If v \Then} \denC{{e_1}} {\Else} \denC{{e_2}}$ and $\denP{{\IfThenElse{(v)_\annell}{e_1}{e_2}}} = \denP{{e_1}}$.
            Then $\pconf*{\denC{e}}{\denP{e}} \pstep_\Nulltrace \pconf*{\denC{e_1}}{\denP{e_1}}$.

            The case when $\annell = \Prove$ is similar to the one for $\annell = \Compute$.

            In all cases, $\Nulltrace \traceq \Nulltrace :: \Nulltrace$.

        \item (IfF)
        
            Similar logic to IfT case, but with $e_2$ instead of $e_1$.
        \item (RefC)
        
        Suppose $\langle {\Refsub{\Compute} v} \divi \storetuple \rangle \cpstep_{\opalloc(\iota, v)} \langle {r_\Compute(\iota)} \divi \storetuple' \rangle$ with premises $\iota \not\in \dom(\sigma_\Compute) \cup \dom(\sigma_\Prove)$ and $\storetuple' = \storetuple[\Compute \mid \iota \mapsto v]$.
        Here, $\denC{{\Refsub{\Compute} v}} = {\Reft v}$ and $\denP{{\Refsub{\Compute} v}} = {()}$.
        Thus, $\pconf*{\denC{e}}{\denP{e}} \pstep_{\opalloc(\iota, v)} \pconf{\iota}{()}{\sigmaC[\iota \mapsto v]}{\sigmaP}{\rho}$.

        \item (RefP)
        
        Similar logic to RefC case, except traces produce $\Nulltrace$ instead of $\opalloc(\iota, v)$.

        \item (RefCP)
        
        Suppose $\langle {\Refsub{\ComputeProve} v} \divi \storetuple \rangle \cpstep_{\opalloc(\iota_\Compute, v)} \langle {r_\ComputeProve(\iota_\Compute, \iota_\Prove)} \divi \storetuple' \rangle$ with premises $\iota_\Compute, \iota_\Prove \not\in \dom(\sigma_\Compute) \cup \dom(\sigma_\Prove)$ and $\storetuple' = \storetuple[\Compute \mid \iota_\Compute \mapsto v][\Prove \mid \iota_\Prove \mapsto v]$.
        Then $\denC{{\Refsub{\ComputeProve} v}} = \denP{{\Refsub{\ComputeProve} v}} = {\Reft v}$, so $\pconf*{\denC{e}}{\denP{e}} \pstep_{\opalloc(\iota_\Compute, v)} \pconf{\iota_\Compute}{\Reft{v}}{\sigmaC[\iota_\Compute \mapsto v]}{\sigmaP}{\rho} \pstep_\Nulltrace \pconf{\iota_\Compute}{\iota_\Prove}{\sigmaC[\iota_\Compute \mapsto v]}{\sigmaP[\iota_\Prove \mapsto v]}{\rho}$.
        Finally, $\opalloc(\iota_\Compute, v) \traceq \opalloc(\iota_\Compute, v) :: \Nulltrace$.

        \item (DerefC)
        
        Suppose $\langle {\bang r_\Compute(\iota)} \divi \storetuple \rangle \cpstep_\Nulltrace \langle v \divi \storetuple \rangle$ where $\iota \in \dom(\sigma_\Compute)$ and $v = \sigma_\Compute(\iota)$.
        We know $\denC{{\bang r_\Compute(\iota)}} = {\bang \iota}$ and $\denP{{\bang r_\Compute(\iota)}} = {()}$, so $\pconf*{\denC{e}}{\denP{e}} \pstep_\Nulltrace \pconf*{\sigmaC(\iota)}{()}$.

        \item (DerefP)
        
        Similar logic to DerefC case.

        \item (DerefCP)
        
        Similar logic to DerefC case, but take two steps (one for each of \Compute and \Prove) and finally $\Nulltrace \traceq \Nulltrace :: \Nulltrace$.

        \item (AssignC)
        
        Suppose $\langle {r_\Compute(\iota) \coloneq v} \divi \storetuple \rangle \cpstep_{\opset(\iota, v)} \langle {()} \divi \storetuple' \rangle$ where $\iota \in \dom(\sigma_\Compute)$ and $\storetuple' = \storetuple[\Compute \mid \iota \mapsto v]$.
        We know that $\denC{{r_\Compute(\iota) \coloneq v}} = {\iota \coloneq v}$ and $\denP{{r_\Compute(\iota) \coloneq v}} = {()}$, so $\pconf*{\denC{e}}{\denP{e}} \pstep_{\opset(\iota, v)} \pconf{()}{()}{\sigmaC[\iota \mapsto v]}{\sigmaP}{\rho}$.

        \item (AssignP)
        
        Similar logic to AssignC case, except traces produce $\Nulltrace$ instead of $\opset(\iota, v)$.

        \item (AssignCP)
        
        Suppose $\langle {r_\ComputeProve(\iota_\Compute, \iota_\Prove) \coloneq v} \divi \storetuple \rangle \cpstep_{\opset(\iota_\Compute, v)} \langle {()} \divi \storetuple' \rangle$ where $\iota_\Compute \in \dom(\sigma_\Compute)$ and $\iota_\Prove \in \dom(\sigma_\Prove)$ and $\storetuple' = \storetuple[\Compute \mid \iota_\Compute \mapsto v][\Prove \mid \iota_\Prove \mapsto v]$.
        Note that $\denC{{r_\ComputeProve(\iota_\Compute, \iota_\Prove) \coloneq v}} = {\iota_\Compute \coloneq v}$ and $\denP{{r_\ComputeProve(\iota_\Compute, \iota_\Prove) \coloneq v}} = {\iota_\Prove \coloneq v}$.
        Thus, $\pconf*{\denC{e}}{\denP{e}} \pstep_{\opset(\iota_\Compute, v)} \pconf{()}{\denP{e}}{\sigmaC[\iota_\Compute \mapsto v]}{\sigmaP}{\rho} \pstep_\Nulltrace \pconf{()}{()}{\sigmaC[\iota_\Compute \mapsto v]}{\sigmaP[\iota_\Prove \mapsto v]}{\rho}$.
        Finally, $\opset(\iota_\Compute, v) \traceq \opset(\iota_\Compute, v) :: \Nulltrace$.

        \item (Cast)
        
        Suppose $\langle {(C)(\Newl D(\overline{v}))} \divi \storetuple \rangle \cpstep_\Nulltrace \langle {\Newl D(\overline{v})} \divi \storetuple \rangle$ where $D <: C$.
        If $\ell = \Compute$, then $\denC{{(C)(\Newl D(\overline{v}))}} = {(C)(\Newl D(\overline{v}))}$ and $\denP{{(C)(\Newl D(\overline{v}))}} = {()}$, and so $\pconf*{\denC{e}}{\denP{e}} \pstep_\Nulltrace \pconf*{\Newl D(\overline{v})}{()}$.
        Similar logic proves the case where $\ell = \Prove$.
        And if $\ell = \ComputeProve$ then $\pconf*{\denC{e}}{\denP{e}} \pstep_\Nulltrace \pconf*{\Newl D(\overline{v})}{\denP{e}} \pstep_\Nulltrace \pconf*{\Newl D(\overline{v})}{\Newl D(\overline{v})}$, and $\Nulltrace \traceq \Nulltrace :: \Nulltrace$.
        
        \item (Field)
        
        Follows similar logic to (Cast).

        \item (Call)
        
        Suppose $\langle {(\Newl C(\overline{v})).m(\overline{w})} \divi \storetuple \rangle \cpstep_\Nulltrace \langle {\Lift*{e}[\overline{x} \mapsto (\overline{w})_\annell, \This \mapsto (\Newsub{\ell} C(\overline{v}))_\annell]} \divi \storetuple \rangle$.
        If $\ell = \Compute$, then $\denC{{(\Newl C(\overline{v})).m(\overline{w})}} = {\New C(\overline{v}).m(\overline{w})}$ and $\denP{{(\Newl C(\overline{v})).m(\overline{w})}} = ()$.
        Thus, $\pconf*{\denC{e}}{\denP{e}} \pstep_\Nulltrace \pconf*{e[\overline{x} \mapsto \overline{w}, \This \mapsto \New C(\overline{v})]}{()}$.
        Finally, \Cref{lemma:raise-lessthan-all} proves the terms match on the \Compute side and \Cref{lemma:raise-lessthan-unit} proves they match on the \Prove side. Similar logic proves the case for $\ell = \Prove$ (swapping which lemmas are used on which procedure), and the case for $\ell = \ComputeProve$ is executing two steps, one on each side, and concluding $\Nulltrace \traceq \Nulltrace :: \Nulltrace$.

        \item (Input Assignment) 

        Suppose $\langle {x \witassign v} \divi \storetuple \rangle \cpstep_{\opset(x,v)} \langle {()} \divi \storetuple' \rangle$ from the premise that $\rho(x) = (v',0)$ and $\storetuple' = \storetuple[\rho \mid {x} \mapsto ({v}, 1)]$.
        We know $\denC{{x \witassign v}} = {x \witassign v}$ and $\denP{{x \witassign v}} = {()}$, so $\pconf*{\denC{e}}{\denP{e}} \pstep_{\opset(x,v)} \pconf{()}{()}{\sigmaC}{\sigmaP}{\rho[x \mapsto v]}$.

        \item (Input Deref)
        
        Suppose $\langle {\witderef x} \divi \storetuple \rangle \cpstep_\Nulltrace \langle v \divi \storetuple \rangle$ with the premise $\rho({x}) = (v, 1)$.
        We know $\denC{{\witderef x}} = \denP{{\witderef x}} = {\witderef x}$, so $\pconf*{\denC{e}}{\denP{e}} \pstep^2_\Nulltrace \pconf*{\rho(x)}{\rho(x)}$. 
        Finally, $\Nulltrace \traceq \Nulltrace :: \Nulltrace$.

        \item (Prove)
        
        Suppose $\langle {\Provet \alpha = \CircuitDefault \With \overline{v}[\overline{w}]} \divi \storetuple \rangle \cpstep_{\opgen(\alpha, \overline{v}, \overline{w})} \langle {\ProofOfN~\alpha~\UsingN~\overline{v}} \divi \storetuple \rangle$.
        Since $\Provet$ statements only exist on the compute side, $\denC{{\Provet \alpha = \CircuitDefault \With \overline{v}[\overline{w}]}}$ projects to itself, and $\denP{{\Provet \alpha = \CircuitDefault \With \overline{v}[\overline{w}]}} = {()}$.
        Then, $\pconf*{\denC{e}}{\denP{e}} \pstep_{\opgen(\alpha, \overline{v}, \overline{w})} \pconf*{\ProofOfN~\alpha~\UsingN~\overline{v}}{()}$.
        The proof finishes by noting the traces are equivalent.

        \item (VerifyT)
        
        Suppose $\langle {\Verifyt (\ProofOfN~\alpha~\UsingN~\overline{v}) \Proves \alpha \With \overline{v}} \divi \storetuple \rangle \cpstep_{\opverif(\alpha, \overline{v})} \langle \True \divi \storetuple \rangle$.
        If $\pi$ is labeled with $\Compute$, then $\pconf*{\denC{e}}{\denP{e}} \pstep_{\opverif(\alpha, \overline{v})} \pconf*{\True}{()}$.
        Similar reasoning applies to the $\pi$ labeled with $\Prove$ case.
        If $\pi$ is labeled with $\ComputeProve$, then $\pconf*{\denC{e}}{\denP{e}} \pstep_{\opverif(\alpha, \overline{v})} \pconf*{\True}{\denP{e}} \pstep_\Nulltrace \pconf*{\True}{\True}$.
        Finally, note that $\opverif(\alpha, \overline{v}) \traceq \opverif(\alpha, \overline{v}) :: \Nulltrace$.

        \item (VerifyF)
        
        Follows similar logic to (VerifyT), except the expressions evaluate to false rather than true.
    \end{itemize}
\end{proof}

\begin{theorem}[Concurrent Single-Step Completeness]\label{thm:single-bisim-complete}
  If $e$ is an annotated expression where $\langle e \divi (\sigma_\Compute, \sigma_\Prove, \rho) \rangle \cpstep \langle e' \divi (\sigma_\Compute', \sigma_\Prove', \rho') \rangle$,
  then there is some $e'_\Compute$ and $e'_\Prove$ such that
  $\pconf*{\denC{e}}{\denP{e}} \pstep^{\{1, 2\}} \pconf{e_\Compute'}{e_\Prove'}{\sigmaC'}{\sigmaP'}{\rho'}$, where
  $\denC{e'} \lessthan e'_\Compute$ and $\denP{e'} \lessthan e'_\Prove$.
\end{theorem}
\begin{proof}
    Corollary of \Cref{thm:singlestepbisimcomplete}.
\end{proof}

\begin{lemma}[$\lessthan$ is Reflexive]\label{lemma:less-reflexive}
    For all (annotated) expressions $e$, $e \lessthan e$.
\end{lemma}
\begin{proof}
    Follows directly by induction on $e$.
\end{proof}

\begin{lemma}[$\lessthan$ is Transitive]\label{lemma:less-transitive}
    If $e_1 \lessthan e_2$ and $e_2 \lessthan e_3$ then $e_1 \lessthan e_3$.
\end{lemma}
\begin{proof}
    By induction on the definition of $\lessthan$ in \Cref{app:lessthan-def}. 
    Note all the rules are either reflexive, or (for let statements, if statements, and sequencing) a direct application of the inductive hypothesis gives the desired result.
\end{proof}

\begin{lemma}\label{lemma:less-redex}
    If $e \lessthan e'$, then $e' = v ; e''$ and $e \lessthan e''$, or $e = \Eval{e_1}$ and $e' = \Eval{e_2}$ and $e_1 = e_2$ (i.e. their head matches).
\end{lemma}
\begin{proof}
    By cases on the definition of $\lessthan$ in \Cref{app:lessthan-def}.
    The first condition of $e' = v ; e''$ and $e \lessthan e''$ applies to the very first sequencing rule, and the second condition of head equality applies to all the other rules.
\end{proof}

\begin{lemma}
    If $e \lessthan e'$, then for all substitutions $\sigma$ mapping variables to values, $e[\sigma] \lessthan e'[\sigma]$.
\end{lemma}
\begin{proof}
    By induction on the definition of $\lessthan$ in \Cref{app:lessthan-def}.
    Most cases are reflexive axioms and follow directly.
    The $\IfN$ and $\LetN$ and sequencing cases are inductively defined, and follow from applying the inductive hypothesis.
\end{proof}

\begin{corollary}\label{lemma:less-substitution}
    If $e \lessthan e'$, then $e[x \mapsto v] \lessthan e'[x \mapsto v]$.
\end{corollary}

\begin{lemma}
    Suppose $e_\Compute \lessthan e_\Compute'$ and $e_\Prove \lessthan e_\Prove'$.
    If $\pconf*{e_\Compute}{e_\Prove} \pstep_{\trace_1} \pconf{e_\Compute''}{e_\Prove}{\sigmaC'}{\sigmaP}{\rho'}$, then there exists $e_\Compute'''$ such that $\pconf*{e_\Compute'}{e_\Prove'} \pstep_{\trace_2}^+ \pconf{e_\Compute'''}{e_\Prove'}{\sigmaC'}{\sigmaP}{\rho'}$ and $e_\Compute'' \lessthan e_\Compute'''$ and $\trace_1 \traceq \trace_2$.
    Alternatively, if $\pconf*{e_\Compute}{e_\Prove} \pstep_{\trace_1} \pconf{e_\Compute}{e_\Prove''}{\sigmaC}{\sigmaP'}{\rho}$, then there exists $e_\Prove'''$ such that $\pconf*{e_\Compute'}{e_\Prove'} \pstep_{\trace_2}^+ \pconf{e_\Compute'}{e_\Prove'''}{\sigmaC}{\sigmaP'}{\rho}$ and $e_\Prove'' \lessthan e_\Prove'''$ and $\trace_1 \traceq \trace_2$.
\end{lemma}
\begin{proof}
    By induction on the rules of $\pstep$. 
    For notation, we shorthand $\Nulltrace^k$ to mean a sequence of $\Nulltrace$ repeated $k$ times: $\Nulltrace :: \Nulltrace :: \cdots :: \Nulltrace$.

    \begin{itemize}[itemsep=0.6em]
        \item (E-Step-C)
            We'll consider three cases: that of (E-Ref) and (E-IfT) and (E-Let2).
            All other cases follow similar logic to (E-Ref).
            In all of these subcases, $\pconf*{e_\Compute}{e_\Prove} \pstep \pconf{e_\Compute''}{e_\Prove}{\sigmaC'}{\sigmaP}{\rho'}$, i.e. where only the compute side takes a step.

            Suppose $e_\Compute = \Reft v$, meaning $\pconf*{e_\Compute}{e_\Prove} \pstep_{\trace_1} \pconf{\iota}{e_\Prove}{\sigmaC'}{\sigmaP}{\rho}$.
            Since $e_\Compute \lessthan e_\Compute'$, by \Cref{lemma:less-redex}, $e_\Compute' = v_1 ; v_2 ; \cdots ; v_k ; \Reft v$ for some finite $k \geq 0$.
            So $\pconf*{e_\Compute'}{e_\Prove'} \pstep_{\Nulltrace^k}^k \pconf*{\Reft{v}}{e_\Prove'} \pstep_{\trace_1} \pconf{\iota}{e_\Prove'}{\sigmaC'}{\sigmaP}{\rho}$.
            Thus $e_\Compute''' = e_\Compute'' = \iota$ and $t_2 = \Nulltrace^k :: \trace_1 \traceq t_1$.

            Suppose instead that $e_\Compute = \IfThenElse{\True}{e_1}{e_2}$, meaning  $\pconf*{e_\Compute}{e_\Prove} \pstep_{\Nulltrace} \pconf*{e_1}{e_\Prove}$.
            Since $e_\Compute \lessthan e_\Compute'$, by \Cref{lemma:less-redex}, $e_\Compute' = v_1 ; v_2 ; \cdots ; v_k ; \IfThenElse{\True}{e_1'}{e_2'}$ where $e_1 \lessthan e_1'$ and $e_2 \lessthan e_2'$.
            So $\pconf*{e_\Compute'}{e_\Prove'} \pstep_{\Nulltrace^k}^k \pconf*{\IfThenElse{\True}{e_1'}{e_2'}}{e_\Prove'} \pstep_{\Nulltrace} \pconf*{e_1'}{e_\Prove'}$.
            Thus $e_\Compute'' = e_1 \lessthan e_1' = e_\Compute'''$ and $t_2 = \Nulltrace^k :: \trace_1 \traceq t_1$.
            
            Finally, suppose instead that $e_\Compute = \LetIn{x\ty\TTau}{v}{e}$, meaning $\pconf*{e_\Compute}{e_\Prove} \pstep_{\Nulltrace} \pconf*{e[x \mapsto v]}{e_\Prove}$.
            Since $e_\Compute \lessthan e_\Compute'$, by \Cref{lemma:less-redex}, $e_\Compute' = v_1 ; \cdots ; v_k ; \LetIn{x\ty\TTau}{e_1'}{e'}$ where $v \lessthan e_1'$ and $e \lessthan e'$.
            So $\pconf*{e_\Compute'}{e_\Prove'} \pstep_{\Nulltrace^k}^k \pconf*{\LetIn{x\ty\TTau}{e_1'}{e'}}{e_\Prove'}$.
            Notice that if $v \lessthan e_1'$ then by \Cref{lemma:less-redex} that $e_1' = v_1' ; \cdots ; v_j' ; v'$ where $v \lessthan v'$ meaning $v = v'$ by definition of $\lessthan$.
            Thus, $\pconf*{\LetIn{x\ty\TTau}{e_1'}{e'}}{e_\Prove'} \pstep_{\Nulltrace^j}^j \pconf*{\LetIn{x\ty \TTau}{v}{e'}}{e_\Prove'} \pstep_\Nulltrace \pconf*{e'[x \mapsto v]}{e_\Prove'}$.
            Finally, by \Cref{lemma:less-substitution}, $e[x\mapsto v] \lessthan e'[x \mapsto v]$, and $\Nulltrace^k :: \Nulltrace^j :: \Nulltrace \traceq \Nulltrace$.

        \item (E-Step-P)
        
            Follows the same logic as the (E-Step-C) case, but now on the prove side and where the traces are all sequences of $\Nulltrace$.
        \item (E-Let1-C)
        
            Suppose $e_\Compute = \LetIn{x\ty\TTau}{e_1}{e}$, meaning  $\pconf*{e_\Compute}{e_\Prove} \pstep_{\trace_1} \pconf{\LetIn{x\ty\TTau}{e_2}{e}}{e_\Prove}{\sigmaC'}{\sigmaP}{\rho'}$ from $\pconf*{e_1}{e_\Prove} \pstep_{\trace_1} \pconf{e_2}{e_\Prove}{\sigmaC'}{\sigmaP}{\rho'}$.
            Since $e_\Compute \lessthan e_\Compute'$, by \Cref{lemma:less-redex}, $e_\Compute' = v_1 ; \cdots ; v_k ; \LetIn{x\ty\TTau}{e_1'}{e'}$ where $e_1 \lessthan e_1'$ and $e \lessthan e'$ by definition of $\lessthan$.
            So $\pconf*{e_\Compute'}{e_\Prove'} \pstep_{\Nulltrace^k}^k \pconf*{\LetIn{x\ty\TTau}{e_1'}{e'}}{e_\Prove'}$.
            By the inductive hypothesis, there is some $e_2'$ such that $\pconf*{e_1'}{e_\Prove'} \pstep_{\trace_2} \pconf{e_2'}{e_\Prove'}{\sigmaC'}{\sigmaP}{\rho'}$ where $e_2 \lessthan e_2'$ and $\trace_1 \traceq \trace_2$.
            Then by applying the Let rule, we get $\pconf*{\LetIn{x\ty\TTau}{e_1'}{e'}}{e_\Prove'} \pstep_{\trace_2} \pconf{\LetIn{x\ty\TTau}{e_2'}{e'}}{e_\Prove'}{\sigmaC'}{\sigmaP}{\rho'}$.
            By definition, $\LetIn{x\ty\TTau}{e_2}{e} \lessthan \LetIn{x\ty\TTau}{e_2'}{e'}$.
        \item (E-Let1-P)
        
            Follows the same logic as (E-Let1-C), but now on the prove side and where the traces are all sequences of $\Nulltrace$. 
        \item (E-WitAssign)

            Follows similar logic as the (E-Ref) subcase of (E-Step-C).
        \item (E-WitDeref-C)
        
            Follows similar logic as the (E-Ref) subcase of (E-Step-C), but without changes to the store.
        \item (E-WitDeref-P)
        
            Follows similar logic as the (E-Ref) subcase of (E-Step-C), but without changes to the store, and all traces are sequences of $\Nulltrace$ on the prove side.
    \end{itemize}
\end{proof}

\begin{corollary}[Less-Than Completeness]\label{cor:lessthan-complete}
    Suppose $e_\Compute \lessthan e_\Compute'$ and $e_\Prove \lessthan e_\Prove'$.
    If $\pconf*{e_\Compute}{e_\Prove} \pstep_{\trace_1} \pconf{e_\Compute''}{e_\Prove''}{\sigmaC'}{\sigmaP'}{\rho'}$, then there exists $e_\Compute''', e_\Prove'''$ such that $\pconf*{e_\Compute'}{e_\Prove'} \pstep_{\trace_2}^* \pconf{e_\Compute'''}{e_\Prove'''}{\sigmaC'}{\sigmaP'}{\rho'}$ and $e_\Compute'' \lessthan e_\Compute'''$ and $e_\Prove'' \lessthan e_\Prove'''$ and $\trace_1 \traceq \trace_2$.
\end{corollary}

\begin{theorem}[Concurrent Completeness with Traces]\label{thm:fullbisimcomplete}
    If $e$ is an annotated expression where $\langle e \divi (\sigmaC, \sigmaP, \rho) \rangle \cpstep^*_{{\trace_1}} \langle e' \divi (\sigmaC', \sigmaP', \rho') \rangle$, then there is some $e'_\Compute$ and $e'_\Prove$ such that $\pconf*{\denC{e}}{\denP{e}} \pstep^*_{{\trace_2}} \pconf{e_\Compute'}{e_\Prove'}{\sigmaC'}{\sigmaP'}{\rho'}$, where $\denC{e'} \lessthan e'_\Compute$ and $\denP{e'} \lessthan e'_\Prove$ and ${\trace_1} \traceq {\trace_2}$.
\end{theorem}
\begin{proof}
    By induction.
    Specifically, combining \Cref{thm:singlestepbisimcomplete} and \Cref{cor:lessthan-complete}, then applying transitivity from \Cref{lemma:less-transitive} as shown:

    \begin{center}
        \begin{tikzpicture}[
            double_arrow/.style = {-{Straight Barb}{Straight Barb}}
            ]
            \node (Surf1) at (0,0) {$e$\strut};
            \node (Targ1) at (0,-6) {$\den{e}$\strut};
            \node (Surf2) at (3,0) {$e_1$\strut};
            \node (Targ2) at (3,-4.5) {$\den{e_1}$\strut};
            \node (Less2) at (3,-6) {$(e_{\Compute1}, e_{\Prove1})$\strut};
            \node (Surf3) at (6,0) {$e_2$\strut};
            \node (Targ3) at (6,-3) {$\den{e_2}$\strut};
            \node (Less3) at (6,-4.5) {$(e_{\Compute2}, e_{\Prove2})$\strut};
            \node (Bott3) at (6,-6) {$(e_{\Compute1}', e_{\Prove1}')$\strut};
            \node (SurfN) at (10,0) {$e'$\strut};
            \node (TargN) at (10,-1.5) {$\den{e'}$\strut};
            \node (DownN) at (10,-3) {$(e_{\Compute3}', e_{\Prove3}')$\strut};
            \node (LessN) at (10,-4.5) {$(e_{\Compute2}', e_{\Prove2}')$\strut};
            \node (BottN) at (10,-6) {$(e_{\Compute1}'', e_{\Prove1}'')$\strut};

            \draw[|->] (Surf1.south) -- (Targ1.north) node[midway,left]{$\den{\cdot}$};
            \draw[|->] (Surf2.south) -- (Targ2.north) node[midway,left]{$\den{\cdot}$};
            \draw[|->] (Surf3.south) -- (Targ3.north) node[midway,left]{$\den{\cdot}$};
            \draw[|->] (SurfN.south) -- (TargN.north) node[midway,left]{$\den{\cdot}$};
            \draw[|->] (Targ2.south) -- (Less2.north) node[midway,left]{$\lessthan$};
            \draw[|->] (Targ3.south) -- (Less3.north) node[midway,left]{$\lessthan$};
            \draw[|->] (Less3.south) -- (Bott3.north) node[midway,left]{$\lessthan$};
            \draw[|->] (TargN.south) -- (DownN.north) node[midway,left]{$\lessthan$};
            \draw[|->] (DownN.south) -- (LessN.north) node[midway,left]{$\lessthan$};
            \draw[|->] (LessN.south) -- (BottN.north) node[midway,left]{$\lessthan$};

            \draw[-{{Straight Barb}[scale=0.9]{Straight Barb}[scale=0.9]},shorten >=3pt,sourcecolor,semithick,dashed] ($(Surf1.east)$) -- ($(Surf2.west)$);
            \draw[-{{Straight Barb}[scale=0.9]{Straight Barb}[scale=0.9]},shorten >=3pt,sourcecolor,semithick,dashed] ($(Surf2.east)$) -- ($(Surf3.west)$);
            \draw[-{{Straight Barb}[scale=0.9]{Straight Barb}[scale=0.9]},shorten >=3pt,sourcecolor,semithick,dashed] ($(Surf3.east)$) -- ($(SurfN.west)$) node[pos=0.94,right]{\textcolor{black}{${}^*$}};

            \draw[-{{Straight Barb}[scale=0.9]},shorten >=3pt,targetcolor,semithick,dashed] ($(Targ1.east) + (0,0.09)$) -- ($(Less2.west) + (0.1,0.09)$) node[midway,above]{\textcolor{black}{\Cref{thm:singlestepbisimcomplete}}};
            \draw[-{{Straight Barb}[scale=0.9]},shorten >=3pt,targetcolor,semithick,dashed] ($(Targ1.east) - (0,0.09)$) -- ($(Less2.west) + (0.1,-0.09)$);

            \draw[-{{Straight Barb}[scale=0.9]},shorten >=3pt,targetcolor,semithick,dashed] ($(Targ2.east) + (0,0.09)$) -- ($(Less3.west) + (0.1,0.09)$) node[midway,above]{\textcolor{black}{\Cref{thm:singlestepbisimcomplete}}};
            \draw[-{{Straight Barb}[scale=0.9]},shorten >=3pt,targetcolor,semithick,dashed] ($(Targ2.east) - (0,0.09)$) -- ($(Less3.west) + (0.1,-0.09)$);

            \draw[-{{Straight Barb}[scale=0.9]},shorten >=3pt,targetcolor,semithick,dashed] ($(Targ3.east) + (0,0.09)$) -- ($(DownN.west) + (0.1,0.09)$) node[midway,above]{\textcolor{black}{Induction}} node[pos=0.92,right]{\textcolor{black}{${}^*$}};
            \draw[-{{Straight Barb}[scale=0.9]},shorten >=3pt,targetcolor,semithick,dashed] ($(Targ3.east) - (0,0.09)$) -- ($(DownN.west) + (0.1,-0.09)$);

            \draw[-{{Straight Barb}[scale=0.9]},shorten >=3pt,targetcolor,semithick,dashed] ($(Less3.east) + (0,0.09)$) -- ($(LessN.west) + (0.1,0.09)$) node[midway,above]{\textcolor{black}{\Cref{cor:lessthan-complete}}} node[pos=0.92,right]{\textcolor{black}{${}^*$}};
            \draw[-{{Straight Barb}[scale=0.9]},shorten >=3pt,targetcolor,semithick,dashed] ($(Less3.east) - (0,0.09)$) -- ($(LessN.west) + (0.1,-0.09)$);

            \draw[-{{Straight Barb}[scale=0.9]},shorten >=3pt,targetcolor,semithick,dashed] ($(Less2.east) + (0,0.09)$) -- ($(Bott3.west) + (0.1,0.09)$) node[midway,above]{\textcolor{black}{\Cref{cor:lessthan-complete}}} node[pos=0.86,right]{\textcolor{black}{${}^*$}};
            \draw[-{{Straight Barb}[scale=0.9]},shorten >=3pt,targetcolor,semithick,dashed] ($(Less2.east) - (0,0.09)$) -- ($(Bott3.west) + (0.1,-0.09)$);

            \draw[-{{Straight Barb}[scale=0.9]},shorten >=3pt,targetcolor,semithick,dashed] ($(Bott3.east) + (0,0.09)$) -- ($(BottN.west) + (0.1,0.09)$) node[midway,above]{\textcolor{black}{\Cref{cor:lessthan-complete}}} node[pos=0.92,right]{\textcolor{black}{${}^*$}};
            \draw[-{{Straight Barb}[scale=0.9]},shorten >=3pt,targetcolor,semithick,dashed] ($(Bott3.east) - (0,0.09)$) -- ($(BottN.west) + (0.1,-0.09)$);
        \end{tikzpicture}
    \end{center}
\end{proof}

\bisimcomplete*
\begin{proof}
    Corollary of~\Cref{thm:fullbisimcomplete}.
\end{proof}

\begin{lemma}
    Suppose $e_\Compute \lessthan e_\Compute'$ and $e_\Prove \lessthan e_\Prove'$.
    If $\pconf*{e_\Compute'}{e_\Prove'} \pstep_{\trace_1} \pconf{e_\Compute''}{e_\Prove'}{\sigmaC'}{\sigmaP}{\rho'}$ then there exists a $e_\Compute'''$ such that $\pconf*{e_\Compute}{e_\Prove} \pstep_{\trace_2}^{\{0,1\}} \pconf{e_\Compute'''}{e_\Prove}{\sigmaC'}{\sigmaP}{\rho'}$ where $e_\Compute''' \lessthan e_\Compute''$ and $\trace_1 \traceq \trace_2$.
    Alternatively, if $\pconf*{e_\Compute'}{e_\Prove'} \pstep_{\trace_1} \pconf{e_\Compute'}{e_\Prove''}{\sigmaC}{\sigmaP'}{\rho}$ then there exists a $e_\Prove'''$ such that $\pconf*{e_\Compute}{e_\Prove} \pstep_{\trace_2}^{\{0,1\}} \pconf{e_\Compute}{e_\Prove'''}{\sigmaC}{\sigmaP'}{\rho}$ where $e_\Prove''' \lessthan e_\Prove''$ and $\trace_1 \traceq \trace_2$.
\end{lemma}
\begin{proof}
    By induction on the rules of $\pstep$.

    \begin{itemize}[itemsep=0.6em]
        \item (E-Step-C)
            We'll consider three cases: that of (E-Ref) and (E-IfT) and (E-Let2).
            All other cases follow similar logic to (E-Ref).
            In all of these subcases, $\pconf*{e_\Compute'}{e_\Prove'} \pstep \pconf{e_\Compute''}{e_\Prove'}{\sigmaC'}{\sigmaP}{\rho'}$, i.e. where only the compute side takes a step.

            Suppose $e_\Compute = \Reft v$, meaning by \Cref{lemma:less-redex}, $e_\Compute' = v_1 ; v_2 ; \cdots ; v_k ; \Reft v$ for some finite $k \geq 0$.
            If $k = 0$ then $e_\Compute = e_\Compute'$ so the exact same step can be taken and $\trace_1 = \trace_2$.
            If $k > 0$ then we are done by simply taking 0 steps, as $\Reft v \lessthan v_2 ; \cdots v_k ; \Reft v$ and $\Nulltrace = \trace_1 \traceq \Nulltrace$.

            Suppose instead that $e_\Compute = \IfThenElse{\True}{e_1}{e_2}$, which means by \Cref{lemma:less-redex} that $e_\Compute' = v_1 ; v_2 ; \cdots ; v_k ; \IfThenElse{\True}{e_1'}{e_2'}$ where $e_1 \lessthan e_1'$ and $e_2 \lessthan e_2'$.
            If $k = 0$ then $\pconf*{e_\Compute'}{e_\Prove'} \pstep_{\Nulltrace} \pconf*{e_1'}{e_\Prove'}$ and so we can take $\pconf*{e_\Compute}{e_\Prove} \pstep_{\Nulltrace} \pconf*{e_1}{e_\Prove}$.
            By the inductive hypothesis, $e_1 \lessthan e_1'$, and $\Nulltrace \traceq \Nulltrace$.
            If $k > 0$ then we are done by simply taking 0 steps, as $\IfThenElse{\True}{e_1}{e_2} \lessthan v_2 ; \cdots ; v_k ; \IfThenElse{\True}{e_1'}{e_2'}$ since $e_1 \lessthan e_1'$ and $e_2 \lessthan e_2'$.
            
            Finally, suppose instead that $e_\Compute = \LetIn{x\ty\TTau}{v}{e}$, meaning by \Cref{lemma:less-redex}, $e_\Compute' = v_1 ; \cdots ; v_k ; \LetIn{x\ty\TTau}{e_1'}{e'}$ where $v \lessthan e_1'$ and $e \lessthan e'$.
            If $k = 0$ and $e_1'$ is a value (which means it's equal to $v$) then $\pconf*{e_\Compute'}{e_\Prove'} \pstep_{\Nulltrace} \pconf*{e'[x \mapsto v]}{e_\Prove'}$ and so we can take $\pconf*{e_\Compute}{e_\Prove} \pstep_{\Nulltrace} \pconf*{e[x \mapsto v]}{e_\Prove}$, and since $e \lessthan e'$, we know $e[x \mapsto v] \lessthan e'[x \mapsto v]$ by \Cref{lemma:less-substitution}.
            Otherwise, we are done by simply taking 0 steps.

        \item (E-Step-P)
        
            Follows the same logic as the (E-Step-C) case, but now on the prove side and where the traces are all sequences of $\Nulltrace$.
        \item (E-Let1-C)
        
            Suppose $e_\Compute = \LetIn{x\ty\TTau}{e_1}{e}$, which by \Cref{lemma:less-substitution} means $e_\Compute' = v_1 ; \cdots ; v_k ; \LetIn{x\ty\TTau}{e_1'}{e'}$ where $e_1 \lessthan e_1'$ and $e \lessthan e'$.
            If $k = 0$ then $e_\Compute' = \LetIn{x\ty\TTau}{e_1'}{e'}$, so $\pconf*{e_\Compute'}{e_\Prove'} \pstep_{\trace_1} \pconf{\LetIn{x\ty\TTau}{e_2'}{e'}}{e_\Prove'}{\sigmaC'}{\sigmaP}{\rho}$ from $\pconf*{e_1'}{e_\Prove'} \pstep_{\trace_1} \pconf{e_2'}{e_\Prove'}{\sigmaC'}{\sigmaP}{\rho'}$.
            By the inductive hypothesis there is some $e_2$ such that $\pconf*{e_1}{e_\Prove} \pstep_{\trace_2} \pconf{e_2}{e_\Prove}{\sigmaC'}{\sigmaP}{\rho'}$ where $e_2 \lessthan e_2'$ and $\trace_1 \traceq \trace_2$.
            The, by applying the $\LetN$ rule, we get $\pconf*{\LetIn{x\ty\TTau}{e_1}{e}}{e_\Prove} \pstep_{\trace_2} \pconf{\LetIn{x\ty\TTau}{e_2}{e}}{e_\Prove}{\sigmaC'}{\sigmaP}{\rho'}$.
            By definition, $\LetIn{x\ty\TTau}{e_2}{e} \lessthan \LetIn{x\ty\TTau}{e_2'}{e'}$.
        \item (E-Let1-P)
        
            Follows the same logic as (E-Let1-C), but now on the prove side and where the traces are all sequences of $\Nulltrace$. 
        \item (E-WitAssign)

            Follows similar logic as the (E-Ref) subcase of (E-Step-C).
        \item (E-WitDeref-C)
        
            Follows similar logic as the (E-Ref) subcase of (E-Step-C), but without changes to the store.
        \item (E-WitDeref-P)
        
            Follows similar logic as the (E-Ref) subcase of (E-Step-C), but without changes to the store, and all traces are sequences of $\Nulltrace$ on the prove side.
    \end{itemize}
\end{proof}

\begin{corollary}[Less-Than Soundness]\label{cor:lessthan-sound}
    Suppose $e_\Compute \lessthan e_\Compute'$ and $e_\Prove \lessthan e_\Prove'$.
    If $\pconf*{e_\Compute'}{e_\Prove'} \pstep_{\trace_1} \pconf{e_\Compute''}{e_\Prove''}{\sigmaC'}{\sigmaP'}{\rho'}$ then there exists a $e_\Compute''', e_\Prove'''$ such that $\pconf*{e_\Compute}{e_\Prove} \pstep_{\trace_2}^{\{0,1\}} \pconf{e_\Compute'''}{e_\Prove'''}{\sigmaC'}{\sigmaP'}{\rho'}$ where $e_\Compute''' \lessthan e_\Compute''$ and $e_\Prove''' \lessthan e_\Prove''$ and $\trace_1 \traceq \trace_2$.
\end{corollary}

\begin{lemma}[Termination of Source Programs]\label{lemma:source-termination}
    If~$\SExp$ is an annotated well-typed expression where~$\denC{\SExp}$ and~$\denP{\SExp}$ terminate with respect to $\tstep$, then~$\SExp$ must also terminate with respect to~$\cpstep$.
\end{lemma}
\begin{proof}
    By structural induction on $e$.

    Base case: The cases where a projected expression evaluates to a value in zero or one steps.
    In the former case, we are done.
    All semantic rules that evaluate to a value in one step (e.g. everything except let, if, method calls) terminate by definition.
    The only interesting case is $e = (\New C(\overline{v})).m_\ell(\overline{u})$, as recursive method calls is the only way in the language to lead to non-termination.
    Note that its projection $\denL{e} = (\New C(\overline{v})).m(\overline{u})$ or $()$.
    In the latter case, we are done.
    In the former case, by the premises that the projections terminate, the method body terminates.
    This means that raising the method body to \zkstrudel will also terminate by applying the related sequence of steps.

    Inductive case: Consider an expression $e$ inductvely built from either let or if statements.
    In both cases, the inductive hypothesis applies on the evaluated subexpression $e_1$ and $e_2$, and by \Cref{thm:abs-progress} we know these can take a step and therefore must terminate.
\end{proof}

\begin{theorem}[Single Step Soundness with Traces]\label{thm:singlestepbisimsound}
    If $e$ is an annotated expression where $\pconf*{\denC{e}}{\denP{e}} \pstep_{{\trace_1}} \pconf{e_\Compute'}{e_\Prove'}{\sigmaC'}{\sigmaP'}{\rho'}$, then there is some $e''$, $e''_\Compute$, and $e''_\Prove$ such that $\langle e \divi \storetuple \rangle \cpstep^*_{{\trace_2}} \langle e'' \divi \storetuple'' \rangle$ and $\pconf{e_\Compute'}{e_\Prove'}{\sigmaC'}{\sigmaP'}{\rho'} \pstep^*_{{\trace_3}} \pconf{e_\Compute''}{e_\Prove''}{\sigmaC''}{\sigmaP''}{\rho''}$, where $\denC{e''} \lessthan e''_\Compute$ and $\denP{e''} \lessthan e''_\Prove$ and ${\trace_2} \traceq {\trace_1} :: {\trace_3}$.
\end{theorem}
\begin{proof}
    The proof follows the strategy outlined in \Cref{sec:adequacy}.

    Suppose that $\pconf*{\denC{e}}{\denP{e}} \pstep_{\trace_1} \pconf{e_\Compute'}{e_\Prove'}{\sigmaC'}{\sigmaP'}{\rho'}$.
    By \Cref{lemma:source-termination}, we know $\langle e \divi \storetuple \rangle \cpstep^* \langle v \divi \storetuple'' \rangle$.

    In the first case, if $e$ is a value, then the theorem if vacuously true.
    So assume we take at least one step in $\langle e \divi \storetuple \rangle \cpstep^*_{\trace_a} \langle v \divi \storetuple'' \rangle$.
    By completeness (\Cref{thm:fullbisimcomplete}), we know $\pconf*{\denC{e}}{\denP{e}} \pstep^*_{\trace_b} \pconf{e_\Compute''}{e_\Prove''}{\sigmaC''}{\sigmaP''}{\rho''}$ where $\denC{v} \lessthan e_\Compute''$ and $\denP{v} \lessthan e_\Prove''$ and $\trace_a \traceq \trace_b$.
    Then, by confluence (\Cref{thm:confluence-conc-target}), $\pconf{e_\Compute'}{e_\Prove'}{\sigmaC'}{\sigmaP'}{\rho'} \pstep^*_{\trace_c} \pconf{e_\Compute''}{e_\Prove''}{\sigmaC''}{\sigmaP''}{\rho''}$ and $\trace_1 :: \trace_c = \trace_b \traceq \trace_a$.
\end{proof}

\begin{theorem}[Concurrent Soundness with Traces]\label{thm:fullbisimsoundness}
    If $e$ is an annotated expression where $\pconf*{\denC{e}}{\denP{e}} \pstep^*_{{\trace_1}} \pconf{e_\Compute'}{e_\Prove'}{\sigmaC'}{\sigmaP'}{\rho'}$, then there is some $e''$, $e''_\Compute$, and $e''_\Prove$ such that $\langle e \divi \storetuple \rangle \cpstep^*_{{\trace_2}} \langle e'' \divi \storetuple'' \rangle$ and $\pconf{e_\Compute'}{e_\Prove'}{\sigmaC'}{\sigmaP'}{\rho'} \pstep^*_{{\trace_3}} \pconf{e_\Compute''}{e_\Prove''}{\sigmaC''}{\sigmaP''}{\rho''}$, where $\denC{e''} \lessthan e''_\Compute$ and $\denP{e''} \lessthan e''_\Prove$ and ${\trace_2} \traceq {\trace_1} :: {\trace_3}$.
\end{theorem}
\begin{proof}
    By induction, applying \Cref{thm:singlestepbisimsound} and \Cref{cor:lessthan-sound}.
\end{proof}

\bisimsound*
\begin{proof}
    Corollary of \Cref{thm:fullbisimsoundness} by ignoring reasoning about traces.
\end{proof}

\subsection{Adequacy}
\begin{theorem}[Reordering of Interleaving is Valid with Traces]\label{thm:bubblesort-trace}
    If $\pconf*{e_\Compute}{e_\Prove} \pstep^*_{\trace_1} \pconf{v_\Compute}{v_\Prove}{\sigmaC'}{\sigmaP'}{\rho'}$ then $\pconf*{e_\Compute}{e_\Prove} \pstep^*_{\trace_2} \pconf{v_\Compute}{e_\Prove}{\sigmaC'}{\sigmaP}{\rho'} \pstep^*_{\trace_3} \pconf{v_\Compute}{v_\Prove}{\sigmaC'}{\sigmaP'}{\rho'}$ where $\trace_1 \traceq \trace_2$ and $\trace_3 \traceq \Nulltrace$.
\end{theorem}
\begin{proof}
    We first group the semantic rules of $\pstep$ into two groups.
    We will then define a comparator order on the semantic rules of $\pstep$ for the purposes of reordering the rules.

    The first group consists of the following set of rules: \ruleref{EP-Let1C}, \ruleref{EP-StepC}, \ruleref{EP-WitAssign}, and \ruleref{EP-WitDerefC}.
    The second group consists of: \ruleref{EP-Let1P}, \ruleref{EP-StepP}, and \ruleref{EP-WitDerefP}.
    Notice how the steps in the first group only modify $e_\Compute, \sigma_\Compute,$ and $\rho$ in $\langle e_\Compute \divi e_\Prove \divi (\sigma_\Compute, \sigma_\Prove, \rho) \rangle$, while the steps in the second group only modify $e_\Prove$ and $\sigma_\Prove$.
    Now, suppose we have a pair of semantic rules $(r_1, r_2)$.
    We say that $r_1 \sqsubseteq r_2$ if $r_1$ is a rule in the first group and $r_2$ is a rule in the second group, or if they're both in the same group.

    By the theorem premise, $\pconf*{e_\Compute}{e_\Prove} \pstep^*_{\trace_1} \pconf{v_\Compute}{v_\Prove}{\sigmaC'}{\sigmaP'}{\rho'}$, so we take this sequence of $k \geq 0$ semantic steps of $\pstep$ and bubble sort this sequence according to the comparator order $\sqsubseteq$ we defined.
    Specifically, if we find any pair of consecutive steps $(r_1, r_2)$ that violates the order, meaning $r_1 \not\sqsubseteq r_2$, then we swap the two steps: $(r_2, r_1)$.
    
    Suppose the steps are $\pconf{e_\Compute'}{e_\Prove'}{\sigmaC'}{\sigmaP'}{\rho'} \pstep \pconf{e_\Compute''}{e_\Prove''}{\sigmaC''}{\sigmaP''}{\rho''} \pstep \pconf{e_\Compute'''}{e_\Prove'''}{\sigmaC'''}{\sigmaP'''}{\rho'''}$ where $r_1$ is the first step and $r_2$ is the second step. 
    From the starting state of $r_1$, it's obvious that $r_1$ is a possible semantic execution.
    We also know $r_2$ is a possible semantic execution, since the same rule of $r_2$ can be applied, but where $e_\Compute'$ is used in place of $e_\Compute''$, $\sigma_\Compute'$ in place of $\sigma_\Compute''$, and $\rho'$ in place of $\rho''$.
    Then by \Cref{thm:confluence-conc-target}, executing $r_1$ after $r_2$ gives the same result as the original execution ordering.
    This means that swapping the two semantic steps still results in the same output $\pconf{e_\Compute'}{e_\Prove'}{\sigmaC'}{\sigmaP'}{\rho'} \pstep^2 \pconf{e_\Compute'''}{e_\Prove'''}{\sigmaC'''}{\sigmaP'''}{\rho'''}$.

    Now after completely bubble sorting the sequence of semantic steps, we are left with all semantic steps in the first group, followed by all semantic steps in the second group.
    Note that the steps in the first group only modify $e_\Compute, \sigma_\Compute,$ and $\rho$ in $\langle e_\Compute \divi e_\Prove \divi (\sigma_\Compute, \sigma_\Prove, \rho) \rangle$, while the steps in the second group only modify $e_\Prove$ and $\sigma_\Prove$.

    Finally, recall that only empty trace events are emitted on the prove side, meaning $\trace_3 \traceq \Nulltrace$.
    Since we did not reorder the execution between compute-side operations, we can therefore conclude that $\trace_1 \traceq \trace_2$.
\end{proof}

\begin{theorem}[Reordering of Interleaving is Valid]\label{thm:bubblesort}
  If $\langle e_\Compute \divi e_\Prove \divi (\sigma_\Compute, \sigma_\Prove, \rho) \rangle \pstep^* \langle v_\Compute \divi v_\Prove \divi (\sigma_\Compute', \sigma_\Prove', \rho') \rangle$, then $\langle e_\Compute \divi e_\Prove \divi (\sigma_\Compute, \sigma_\Prove, \rho) \rangle \pstep^* \langle v_\Compute \divi e_\Prove \divi (\sigma_\Compute', \sigma_\Prove, \rho') \rangle \pstep^* \langle v_\Compute \divi v_\Prove \divi (\sigma_\Compute', \sigma_\Prove', \rho') \rangle$.
\end{theorem}
\begin{proof}
    Corollary of \Cref{thm:bubblesort-trace} by ignoring reasoning about traces.
\end{proof}

\begin{lemma}\label{lemma:equal-nowit-witnessproject}
    If $e$ is an expression that doesn't have any computed inputs, then $\witdenL{e} = e$.
\end{lemma}
\begin{proof}
    By induction on $e$, using the relevant case from the definition of $\witdenL{\cdot}$.
\end{proof}

\begin{lemma}\label{lemma:equal-subs-witnessproject}
    Given an expression $e$, if $x$ is not a computed input variable then $\witdenL{e[x \mapsto v]} = \witdenL{e}[x \mapsto v]$.
\end{lemma}
\begin{proof}
    By induction on $e$, using the relevant case from the definition of $\witdenL{\cdot}$.
\end{proof}

\begin{lemma}[Single Step Compute-Side Computed Input Translation Completeness]\label{lemma:witness-compute-complete}
    If $\pconf*{e_\Compute}{e_\Prove} \pstep_\trace \pconf{e_\Compute'}{e_\Prove}{\sigmaC'}{\sigmaP}{\rho'}$, and $\varphi$ is a mapping from all computed input variables $x$ to locations $\iota$, then $\langle \witdenC{e_\Compute} \divi \sigma_\Compute \cup \witdenR{\rho} \rangle \tstep_{\witdenR{\trace}} \langle \witdenC{e'_\Compute} \divi \sigma_\Compute' \cup \witdenR{\rho'} \rangle$.
\end{lemma}
\begin{proof}
    By induction on $\pstep$.

    \begin{itemize}[itemsep=0.6em]
        \item Suppose $e_\Compute = {\Reft v}$, and so $\pconf*{e_\Compute}{e_\Prove} \pstep_{\opalloc(\iota, v)} \pconf{\iota}{e_\Prove}{\sigmaC[\iota \mapsto v]}{\sigmaP}{\rho}$.
        Note that $\witdenC{e_\Compute} = {\Reft v}$, and $\langle {\Reft v} \divi \sigma_\Compute \cup \witdenR{\rho} \rangle \tstep_{\opalloc(\iota, v)} \langle {\iota} \divi (\sigma_\Compute \cup \witdenR{\rho})[\iota \mapsto v] \rangle$.
        The proof for this case finishes by observing that $\witdenC{{\iota}} = {\iota}$ and $(\sigma_\Compute \cup \witdenR{\rho})[\iota \mapsto v] = \sigma_\Compute[\iota \mapsto v] \cup \witdenR{\rho} = \sigma_\Compute' \cup \witdenR{\rho}$.

        Many other cases follow similar logic: deref, assign, cast, field, prove, proof cast, and verify.

        \item Suppose $e_\Compute = {v.m(\overline{w})}$, meaning $\pconf*{e_\Compute}{e_\Prove} \pstep_\Nulltrace \pconf*{e[\overline{x'} \mapsto \overline{w}, \ThisN \mapsto v]}{e_\Prove}$, where $e$ is the method body and $\overline{x'}$ are the input names.
        Note that $\witdenC{e_\Compute} = {v.m(\overline{w})}$, and $\langle {v.m(\overline{w})} \divi \sigma_\Compute \cup \witdenR{\rho} \rangle \tstep_\Nulltrace \langle {e[\overline{x'}\mapsto\overline{w}, \This \mapsto v]} \divi \sigma_\Compute \cup \witdenR{\rho} \rangle$.
        Since $e$ is a method body defined using the outer language, meaning it cannot define or use any computed input variables, we can apply \Cref{lemma:equal-nowit-witnessproject} and conclude that $\witdenC{{e[\overline{x'}\mapsto\overline{w}, \This \mapsto v]}} = {e[\overline{x'}\mapsto\overline{w}, \This \mapsto v]}$.
        
        \item Suppose $e_\Compute = {\LetIn{x\ty\TTau}{e_1}{e_2}}$, where $e_1$ is not a value.
        This means $\pconf*{e_\Compute}{e_\Prove} \pstep_\trace \pconf{\LetIn{x\ty\TTau}{e_1'}{e_2}}{e_\Prove}{\sigmaC'}{\sigmaP}{\rho'}$ from $\pconf*{e_1}{e_\Prove} \pstep_\trace \pconf{e_1'}{e_\Prove}{\sigmaC'}{\sigmaP}{\rho'}$.
        By the inductive hypothesis, we know $\langle \witdenC{{e_1}} \divi \sigma_\Compute \cup \witdenR{\rho} \rangle \tstep_\trace \langle \witdenC{{e_1'}} \divi \sigma_\Compute' \cup \witdenR{\rho'} \rangle$.
        Thus, $\langle {\LetIn{x\ty\TTau}{\witdenC{{e_1}}}{\witdenC{{e_2}}}} \divi \sigma_\Compute \cup \witdenR{\rho} \rangle \tstep_\trace \langle {\LetIn{x\ty\TTau}{\witdenC{{e_1'}}}{\witdenC{{e_2}}}} \divi \sigma_\Compute' \cup \witdenR{\rho'} \rangle$.
        Finally, ${\LetIn{x\ty\TTau}{\witdenC{{e_1'}}}{\witdenC{{e_2}}}} = \witdenC{{\LetIn{x\ty\TTau}{e_1'}{e_2}}}$.

        \item Suppose $e_\Compute = {\LetIn{x\ty\TTau}{v}{e}}$.
        Then $\pconf*{\sigmaC}{\sigmaP} \pstep_\Nulltrace \pconf*{e[x \mapsto v]}{e_\Prove}$.
        Note that $\witdenC{e_\Compute} = {\LetIn{x\ty\TTau}{v}{\witdenC{e}}}$, and so $\langle \witdenC{e_\Compute} \divi \sigma_\Compute \cup \witdenR{\rho} \rangle \tstep_\Nulltrace \langle {\witdenC{e}[x \mapsto v]} \divi \sigma_\Compute \cup \witdenR{\rho} \rangle$.
        The proof finishes by applying \Cref{lemma:equal-subs-witnessproject}.
        
        \item Suppose $e_\Compute = {\IfThenElse{\True}{e_1}{e_2}}$, meaning $\pconf*{e_\Compute}{e_\Prove} \pstep_\Nulltrace \pconf*{e_1}{e_\Prove}$.
        Note that $\witdenC{e_\Compute} = {\IfThenElse{\True}{\witdenC{{e_1}}}{\witdenC{{e_2}}}}$, so $\langle \witdenC{e_\Compute} \divi \sigma_\Compute \cup \witdenR{\rho} \rangle \tstep_\Nulltrace \langle \witdenC{{e_1}} \divi \sigma_\Compute \cup \witdenR{\rho} \rangle$.

        Similar reasoning proves the case for $e_\Compute = {\IfThenElse{\False}{e_1}{e_2}}$.

        \item If $e_\Compute = {x \witassign v}$, then $\pconf*{e_\Compute}{e_\Prove} \pstep_{\opset(x,v)} \pconf{()}{e_\Prove}{\sigmaC}{\sigmaP}{\rho[x \mapsto v]}$.
        Here, $\witdenC{e_\Compute} = {\varphi(x) \coloneq v}$, and so $\langle {\varphi(x) \coloneq v} \divi \sigma_\Compute \cup \witdenR{\rho} \rangle \tstep_{\opset(\varphi(x), v)} \langle {()} \divi (\sigma_\Compute \cup \witdenR{\rho})[\varphi(x) \mapsto v] \rangle$.
        Now, $\witdenC{{()}} = {()}$, and $\sigma_\Compute' \cup \witdenR{\rho'} = \sigma_\Compute \cup \witdenR{\rho[x \mapsto v]} = (\sigma_\Compute \cup \witdenR{\rho})[\varphi(x) \mapsto v]$.
        Finally, the trace $\witdenR{\opset(x,v)} = \opset(\varphi(x),v)$.

        \item If $e_\Compute = {\witderef x}$ then $\pconf*{e_\Compute}{e_\Prove} \pstep_\Nulltrace \pconf*{\rho(x)}{e_\Prove}$.
        Here, $\witdenC{{\witderef x}} = {\bang \varphi(x)}$, and so $\langle {\bang \varphi(x)} \divi \sigma_\Compute \cup \witdenR{\rho} \rangle \tstep_\Nulltrace \langle (\sigma_\Compute \cup \witdenR{\rho}){(\varphi(x))} \divi \sigma_\Compute \cup \witdenR{\rho} \rangle$.
        Finally, $\witdenC{{\varphi(x)}} = \witdenR{\rho}{(\varphi(x))}$.
    \end{itemize}
\end{proof}

\begin{lemma}[Single Step Compute-Side Computed Input Translation Soundness]\label{lemma:witness-compute-sound}
    Suppose $e_\Compute$ is an expression, and $\varphi$ is a mapping from all computed input variables $x$ to locations $\iota$.
    If $\langle \witdenC{e_\Compute} \divi \sigma' \rangle \tstep_{\trace_1} \langle e' \divi \sigma_\Compute' \cup \witdenR{\rho'} \rangle$, then there exists $e'_\Compute,\sigma_\Compute', \rho', \trace_2$ such that $\pconf*{e_\Compute}{e_\Prove} \pstep_{\trace_2} \pconf{e_\Compute'}{e_\Prove}{\sigmaC'}{\sigmaP}{\rho'}$ where $\witdenC{e'_\Compute} = e'$ and $\sigma' = \sigma_\Compute' \cup \witdenR{\rho'}$ and $\trace_1 = \witdenR{\trace_2}$.
\end{lemma}
\begin{proof}
    The proof follows similarly in structure to the one in \Cref{lemma:witness-compute-complete}, where the same semantic rules are to prove the simulation.

    \begin{itemize}
        \item Suppose $e_\Compute = {\Reft v}$, meaning $\witdenC{e_\Compute} = {\Reft v}$.
        Then $\langle {\Reft v} \divi \sigma_\Compute \cup \witdenR{\rho} \rangle \tstep_{\opalloc(\iota, v)} \langle {\iota} \divi \sigma_\Compute[\iota \mapsto v] \cup \witdenR{\rho} \rangle$.
        We can simulate this step with $\pconf*{e_\Compute}{e_\Prove} \pstep_{\opalloc(\iota, v)} \pconf{\iota}{e_\Prove}{\sigmaC[\iota \mapsto v]}{\sigmaP}{\rho}$.
        Finally, $\witdenC{{\iota}} = {\iota}$ and $\sigma_\Compute[\iota \mapsto v] \cup \witdenR{\rho} = \sigma_\Compute' \cup \witdenR{\rho}$.

        \item Suppose $e_\Compute = {\LetIn{x\ty\TTau}{e_1}{e_2}}$, meaning $\witdenC{e_\Compute} = {\LetIn{x\ty\TTau}{\witdenC{{e_1}}}{\witdenC{{e_2}}}}$.
        Then $\langle \witdenC{e_\Compute} \divi \sigma_\Compute \cup \witdenR{\rho} \rangle \tstep_\trace \langle {\LetIn{x\ty\TTau}{e'}{\witdenC{{e_2}}}} \divi \sigma' \rangle$ from $\langle \witdenC{{e_1}} \divi \sigma_\Compute \cup \witdenR{\rho} \rangle \tstep_\trace \langle {e'} \divi \sigma' \rangle$.
        By the inductive hypothesis, we know that $\pconf*{e_1}{e_\Prove} \pstep_{\trace'} \pconf{e_1'}{e_\Prove}{\sigmaC'}{\sigmaP'}{\rho'}$ where $\witdenC{{e_1'}} = {e'}$ and $\sigma' = \sigma_\Compute' \cup \witdenR{\rho'}$ and $\trace = \witdenR{\trace'}$.
        This then means that $\pconf*{e_\Compute}{e_\Prove} \pstep_{\trace'} \pconf{\LetIn{x\ty\TTau}{e_1'}{e_2}}{e_\Prove}{\sigmaC'}{\sigmaP'}{\rho'}$.
        
    \end{itemize}
\end{proof}

\begin{theorem}[Compute-Side Bisimulation of Computed Input Translation]\label{thm:witness-compute-bisimulation}
    If $\varphi$ is a mapping from all computed input variables $x$ to locations $\iota$, then $\pconf*{e_\Compute}{e_\Prove} \pstep^*_\trace \pconf{e_\Compute'}{e_\Prove}{\sigmaC'}{\sigmaP}{\rho'}$ if and only if $\langle \witdenC{e_\Compute} \divi \sigma_\Compute \cup \witdenR{\rho} \rangle \tstep_{\witdenR{\trace}}^* \langle \witdenC{e'_\Compute} \divi \sigma_\Compute' \cup \witdenR{\rho'} \rangle$.
\end{theorem}
\begin{proof}
    By induction on the number of steps, using \Cref{lemma:witness-compute-complete} for the forward direction ($\Rightarrow$), and \Cref{lemma:witness-compute-sound} for the backwards direction ($\Leftarrow$).
\end{proof}

\begin{lemma}[Single Step Prove-Side Computed Input Translation Completeness]\label{lemma:witness-prove-complete}
    If $\pconf*{e_\Compute}{e_\Prove} \pstep_\Nulltrace \pconf{e_\Compute}{e_\Prove'}{\sigmaC}{\sigmaP'}{\rho}$, and $\varphi$ is a mapping from all computed input variables $x$ to locations $\iota$, denoting $\overline{\varphi} = \varphi(\dom(\varphi))$, then $\langle \witdenP{e_\Prove}[\overline{\varphi} \mapsto \bang \overline{\varphi}] \divi \sigma_\Prove \cup \witdenR{\rho} \rangle \tstep_{\Nulltrace} \langle \witdenP{e'_\Prove}[\overline{\varphi} \mapsto \bang \overline{\varphi}] \divi \sigma_\Prove' \cup \witdenR{\rho} \rangle$.
\end{lemma}
\begin{proof}
    The proof follows similarly in structure to the one in \Cref{lemma:witness-compute-complete}, where the same semantic rules are to prove the simulation.
    We show the logic for three cases here, and the rest can be easily inferred.
    Note that certain rules (such as computed input assign and method calls) from \Cref{lemma:witness-compute-complete} do not appear in the cases of this proof.

    \begin{itemize}
        \item Suppose $e_\Prove = {\Reft v}$ and $\pconf*{e_\Compute}{e_\Prove} \pstep_{\opalloc(\iota, v)} \pconf{e_\Compute}{\iota}{\sigmaC}{\sigmaP[\iota \mapsto v]}{\rho}$.
        Then $\witdenP{e_\Prove}[\overline{\varphi} \mapsto \bang \overline{\varphi}] = {\Reft v}[\overline{\varphi} \mapsto \bang \overline{\varphi}] = {\Reft v}$.
        This last equality comes from the fact that it's not possible in our language to define a reference of the computed input (the type system prevents this).
        Finally, $\langle \witdenP{e_\Prove}[\overline{\varphi} \mapsto \bang \overline{\varphi}] \divi \sigma_\Prove \cup \witdenR{\rho} \rangle \tstep_{\opalloc(\iota, v)} \langle {\iota} \divi \sigma_\Prove[\iota \mapsto v] \cup \witdenR{\rho} \rangle$.

        \item Suppose $e_\Prove = {\LetIn{x\ty\TTau}{e_1}{e_2}}$ and $\pconf*{e_\Compute}{e_\Prove} \pstep_\trace \pconf{e_\Compute}{\LetIn{x\ty\TTau}{e_1'}{e_2}}{\sigmaC}{\sigmaP'}{\rho}$ from $\pconf*{e_\Compute}{e_1} \pstep_\trace \pconf{e_\Compute}{e_1'}{\sigmaC}{\sigmaP'}{\rho}$.
        By the inductive hypothesis, $\langle \witdenP{{e_1}}[\overline{\varphi} \mapsto \bang \overline{\varphi}] \divi \sigma_\Prove \cup \witdenR{\rho} \rangle \tstep_{\witdenR{\trace}} \langle \witdenP{{e_1'}}[\overline{\varphi} \mapsto \bang \overline{\varphi}] \divi \sigma_\Prove' \cup \witdenR{\rho} \rangle$.
        Then, $\langle {\LetIn{x\ty\TTau}{\witdenP{{e_1}}{\color{black}[\overline{\varphi} \mapsto \bang \overline{\varphi}]}}{\witdenP{{e_2}}}}[\overline{\varphi} \mapsto \bang \overline{\varphi}] \divi \sigma_\Prove \cup \witdenR{\rho} \rangle \tstep_{\witdenR{\rho}} \langle {\LetIn{x\ty\TTau}{\witdenP{{e_1'}}{\color{black} [\overline{\varphi} \mapsto \bang \overline{\varphi}]}}{\witdenP{{e_2}}}}[\overline{\varphi} \mapsto \bang \overline{\varphi}] \divi \sigma_\Prove' \cup \witdenR{\rho} \rangle$.

        \item Suppose $e_\Prove = {\witderef x}$ and $\pconf*{e_\Compute}{e_\Prove} \pstep_\Nulltrace \pconf*{e_\Compute}{\rho(x)}$.
        Then $\witdenP{e_\Prove}[\overline{\varphi} \mapsto \bang \overline{\varphi}] = {\varphi(x)}[\overline{\varphi} \mapsto \bang \overline{\varphi}] = {\bang \varphi(x)}$.
        Finally, $\langle \witdenP{e_\Prove}[\overline{\varphi} \mapsto \bang \overline{\varphi}] \divi \sigma_\Prove \cup \witdenR{\rho} \rangle \tstep_\Nulltrace \langle {\rho(x)} \divi \sigma_\Prove \cup \witdenR{\rho} \rangle$. 
    \end{itemize}
\end{proof}

\begin{lemma}[Single Step Prove-Side Computed Input Translation Soundness]\label{lemma:witness-prove-sound}
    Suppose $e_\Prove$ is an expression, and $\varphi$ is a mapping from all computed input variables $x$ to locations $\iota$.
    If $\langle \witdenP{e_\Prove}[\overline{\varphi} \mapsto \bang \overline{\varphi}] \divi \sigma_\Prove \cup \witdenR{\rho} \rangle \tstep_{\Nulltrace} \langle e' \divi \sigma_\Prove' \cup \witdenR{\rho} \rangle$, denoting $\overline{\varphi} = \varphi(\dom(\varphi))$, then $\pconf*{e_\Compute}{e_\Prove} \pstep_\Nulltrace \pconf{e_\Compute}{e_\Prove'}{\sigmaC}{\sigmaP'}{\rho}$ where $\witdenP{e'_\Prove}[\overline{\varphi} \mapsto \bang \overline{\varphi}] = e'$.
\end{lemma}
\begin{proof}
    The proof follows similarly in structure to the one in \Cref{lemma:witness-compute-sound} with the semantic rules from \Cref{lemma:witness-prove-complete}.
\end{proof}

\begin{theorem}[Prove-Side Bisimulation of Computed Input Translation]\label{thm:witness-prove-bisimulation}
    If $\varphi$ is a mapping from all computed input variables $x$ to locations $\iota$, denoting $\overline{\varphi} = \varphi(\dom(\varphi))$, then $\pconf*{e_\Compute}{e_\Prove} \pstep^*_\Nulltrace \pconf{e_\Compute}{e_\Prove'}{\sigmaC}{\sigmaP'}{\rho}$ if and only if $\langle \witdenP{e_\Prove}[\overline{\varphi} \mapsto \bang \overline{\varphi}] \divi \sigma_\Prove \rangle \tstep^*_{\Nulltrace} \langle \witdenP{e'_\Prove}[\overline{\varphi} \mapsto \bang \overline{\varphi}] \divi \sigma_\Prove' \rangle$.
\end{theorem}
\begin{proof}
    By induction on the number of steps, using \Cref{lemma:witness-prove-complete} for the forward direction ($\Rightarrow$), and \Cref{lemma:witness-prove-sound} for the backwards direction ($\Leftarrow$).
\end{proof}

\begin{lemma}[Nesting Sub-Projections Match Projection]\label{lemma:projections-match}
    For any $e$ that is well-typed and where projection is defined, if $\varphi$ is a mapping from all computed input variables $x$ to locations $\iota$, denoting $\overline{\varphi} = \varphi(\dom(\varphi))$, then $\witdenL{\denL{e}} = \denFull{e[\dom(\varphi) \mapsto \overline{\varphi}]}$.
\end{lemma}
\begin{proof}
    By induction on $e$, considering each projection case defined for $\denL{\cdot}$.
\end{proof}

\begin{lemma}[Substitution Preserves Adequacy]\label{lemma:subst-adequacy}
    If for any $e_c, v_c, \hat{e_c}, \overline{\sigma}, \overline{\sigma_1}$ we have $\langle e_c \divi \overline{\sigma_1} \rangle \cpstep^* \langle v_c \divi \overline{\sigma_1}' \rangle$ if and only if $\langle \hat{e_c} \divi \overline{\sigma_1} \rangle \cpstep^* \langle v_c \divi \overline{\sigma_1}' \rangle$, then $\langle e \divi \overline{\sigma} \rangle \cpstep^* \langle v \divi \overline{\sigma}' \rangle$ if and only if $\langle e[e_c \mapsto \hat{e_c}] \divi \overline{\sigma} \rangle \cpstep^* \langle v \divi \overline{\sigma}' \rangle$.
\end{lemma}
\begin{proof}
    By induction on $e$.

    For any $e$ that does not have nested subexpressions, the conclusion follows immediately from the premise: either it substitutes an equivalent value $v_c$, substitutes the whole expression, or it leaves $e$ unchanged.

    That leaves sequencing, branching, prove, and compute-and-prove.
    Prove and compute-and-prove follow from applying the inductive hypothesis to its body $e$; sequencing and branching follow by applying the inductive hypothesis on its subexpressions $e_1$ and $e_2$.
\end{proof}

\begin{lemma}[Bisimulation of Lifting]\label{lemma:lifting-bisim}
    For any $e$ in \corelang, then $\conf{e}{\sigma} \tstep_\trace^* \conf{e'}{\sigma'}$ if and only if $\cpconf{\Lift{\Compute}{e}}{\sigma}{\sigmaP}{\rho} \cpstep^*_{\witdenR[\cdot]{\trace}} \cpconf{\Lift{\Compute}{e'}}{\sigma'}{\sigmaP}{\rho}$.
\end{lemma}
\begin{proof}
    By induction on $e$, where you directly link each syntactic form that can take a step in \corelang with its matching compute-side rule on the lifted expression in \zkstrudel.
\end{proof}

\begin{theorem}[Adequacy of Combined Blocks]\label{thm:adequacy-of-combined}
    For any compute-and-prove block $e_{cnp} = \TCnp*$ that is well-typed, then $\langle {e_{cnp}} \divi \sigma \rangle \tstep^*_{\trace_1} \langle {\TProofOfUsing*} \divi \sigma' \rangle$ if and only if $\langle \fullComp{{e_{cnp}}} \divi \sigma \rangle \tstep_{\trace_2}^* \langle {\TProofOfUsing*} \divi \sigma' \rangle$.
    In either direction, $\trace_1 \traceq \trace_2$.
\end{theorem}
\begin{proof}
    Because compute-and-prove blocks can nest, we prove the theorem via induction on the maximum nesting depth of compute-and-prove blocks in $e$.
    The base case corresponds to the case outline in \Cref{sec:adequacy}, and follows the general strategy in the following diagram:
    
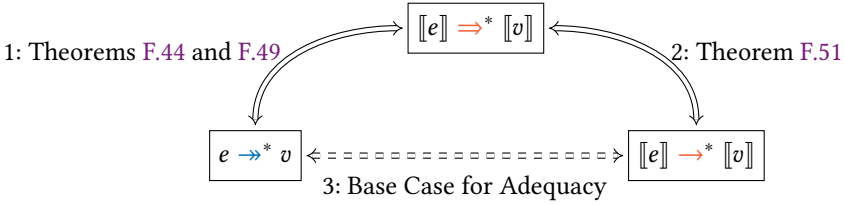
\begin{center}
        \begin{tikzpicture}[
            node/.style={outer sep=2pt,draw,black},
            ]
            \node[node] (SurfSequential) {$\SExp \cpstep^* \SVal$\strut};
            \node[node,right=12em of SurfSequential] (Sequential) {$\den{\SExp} \tstep^* \den{\SVal}$\strut};
            \coordinate (midpoint) at ($(SurfSequential.north)!.5!(Sequential.north)$);
            \node[node,above=2.5em of midpoint] (Concurrent) {$\den{\SExp} \pstep^* \den{\SVal}$\strut};

            \draw[implies-implies,double equal sign distance] (SurfSequential) to[out=90,in=180] node[midway,left = 5pt]{1: \Cref{thm:fullbisimcomplete,thm:fullbisimsoundness}} (Concurrent);
            \draw[implies-implies,double equal sign distance] (Concurrent) to[out=0,in=90] node[midway,right = 5pt]{2: \Cref{thm:bubblesort}} (Sequential);
            \draw[implies-implies,double equal sign distance,dashed] (SurfSequential) -- (Sequential) node[midway,below = 4pt]{3: Base Case for Adequacy};
        \end{tikzpicture}
    \end{center}
    The inductive cases will then prove that any expression with $n+1$ nestings can be adequately translated to an expression with $n$ nestings.

    The following proofs prove the relation between the annotated semantics and the target semantics.
    The proof is complete by applying \Cref{lemma:annotated-bisim} and \Cref{lemma:annotated-proj} to then transitively get adequacy between the original surface and target semantics.

    Base Case: suppose $k=0$ meaning there is no nested blocks anywhere within execution.
    The proof is split into two parts for completeness $(\Rightarrow)$ and soundness $(\Leftarrow)$.
    
    \newlist{bidir}{itemize}{1}
    \begin{bidir}
        \item[($\Rightarrow$)]
            For this direction, suppose that $\langle {e_{cnp}} \divi \sigma \rangle \tstep^*_{\trace_1} \langle {\TProofOfUsing*} \divi \sigma''' \rangle$.
            This can only happen (by determinism \Cref{thm:target-determinism} and \Cref{thm:surface-determinism}) by the following sequence of steps, where $e' = e[\overline{y_p} \mapsto \overline{v_p}, \overline{y_s} \mapsto \overline{v_s}]$:
            \begin{align*}
                \langle e_{cnp} \divi \sigma \rangle & \tstep_{\topalloc(\overline{\varphi}, \bot)} \langle \TCnpAdmin*{e'}{\varnothing}{\varnothing} \divi \sigma[\overline{\varphi} \mapsto \bot] \rangle \\
                & \tstep^*_{\witdenR{\strace}} \langle \TCnpAdmin*{\STrue}{\sigmaP}{\rho} \divi \sigma'' \rangle \\
                & \tstep_{\topgen(\alpha, \overline{v}, \overline{u})} \langle \TProofOfUsing* \divi \sigma''[\overline{\varphi} \mapsto \rho(\dom(\varphi))] \rangle
            \end{align*}
            The first step is via \ruleref{E-ComputeAndProveInit}, then \ruleref{E-ComputeAndProveStep} for the middle steps, and the last one is via \ruleref{E-ComputeAndProveTrue}.

            Recall the definition of compiling a combined block:
            \[\begin{array}{rl}
                & \fullComp{\TCnp*[e]} = \\
                &\hspace{2em} \TLetIn{\overline{w_p}, \overline{w_s} \ty \Reft \overline{\TTau_{x_p}}, \Reft \overline{\TTau_{x_s}}}{\TAlloc{\overline{\TTau_{x_p}}, \overline{\TTau_{x_s}}}}{\big( 
                    \subst*{\denC{\subst*{e}{{\overline{y_p}}{\overline{v_p}}{\overline{y_s}}{\overline{v_s}}}}}{{\overline{x_p}}{\overline{w_p}}{\overline{x_s}}{\overline{w_s}}} \Tseq
                } \\
                &\hspace{4em} \TProveUsing{\alpha}{\exists \overline{y_p}, \overline{x_p} [\overline{y_s}, \overline{x_s}]. \denP{e}}{\overline{v_p}, \TDeref{\overline{w_p}}[\overline{v_s}, \TDeref{\overline{w_s}}]} \big)
            \end{array}\]
            Notice the \ruleref{E-ComputeAndProveInit} is simulated by evaluating the $\TLetN$ statement that allocates a set of location $\overline{\varphi}$ and keeps them uninitialized to $\bot$.
            
            The main portion of proof is with the \ruleref{E-ComputeAndProveStep} steps in the middle, which follows the diagram from \Cref{sec:adequacy}.
            Specifically, these steps evaluate $\cpconf{e'}{\sigma'}{\varnothing}{\varnothing} \cpstep^* \cpconf{\TTrue}{\sigma''}{\sigmaP}{\rho}$ where $\sigma' = \sigma[\overline{\varphi} \mapsto \bot]$.
            From these steps, completeness between the annotated semantics and concurrent semantics (\Cref{thm:fullbisimcomplete}) gives that $\pconf{\denL[\Compute]{e'}}{\denL[\Prove]{e'}}{\sigma'}{\varnothing}{\varnothing} \pstep^* \pconf{e_\Compute'}{e_\Prove'}{\sigma''}{\sigmaP}{\rho}$ where $\denL[\Compute]{\STrue_\annell} \lessthan e_\Compute'$ and $\denL[\Prove]{\STrue_\annell} \lessthan e_\Prove'$.
            Note that $\Prove \subsetl \annell$ by the combined block being well-typed and \Cref{thm:abs-soundness}, meaning $e_\Prove' = \TTrue$ and $e_\Compute' = v$ for some value $v$.
            Then, by the reordering theorem (\Cref{thm:bubblesort-trace}), we get $\pconf{\denL[\Compute]{e'}}{\denL[\Prove]{e'}}{\sigma'}{\varnothing}{\varnothing} \pstep^* \pconf{v}{\denL[\Prove]{e'}}{\sigma''}{\varnothing}{\rho} \pstep^* \pconf{v}{\TTrue}{\sigma''}{\sigmaP}{\rho}$.

            We can now separate reasoning based on the first set of steps (compute-side) and the second set of steps (prove-side).
            Using \Cref{thm:witness-compute-bisimulation} along with the compute-side steps proves that $\langle \witdenC{\denL[\Compute]{e'}} \divi \sigma' \rangle \tstep^* \langle \witdenC{\denL[\Compute]{v}} \divi \sigma'' \cup \{\overline{\varphi} \mapsto \rho(\dom(\varphi))\} \rangle$.
            Then, by \Cref{lemma:projections-match}, we get that $\langle \denC{e'[\dom(\varphi) \mapsto \overline{\varphi}]} \divi \sigma' \rangle \tstep^* \langle v \divi \sigma'' \cup \{\overline{\varphi} \mapsto \rho(\dom(\varphi))\} \rangle$.
            Note that this corresponds to the expression $\subst*{\denC{\subst*{e}{{\overline{y_p}}{\overline{v_p}}{\overline{y_s}}{\overline{v_s}}}}}{{\overline{x_p}}{\overline{w_p}}{\overline{x_s}}{\overline{w_s}}}$ in the full compilation, as $e' = \subst*{e}{{\overline{y_p}}{\overline{v_p}}{\overline{y_s}}{\overline{v_s}}}$ and $\overline{w_p}$ and $\overline{w_s}$ are exactly the raw locations from $\varphi$, after evaluating the $\TLetN$ statement and performing the substitution.

            Similarly, using \Cref{thm:witness-prove-bisimulation} along with the prove-side steps gives $\langle \witdenP{\denL[\Prove]{e'}}[\dom(\varphi) \mapsto \overline{\varphi}] \divi \varnothing \rangle \tstep^* \langle \witdenP{\denL[\Prove]{\TTrue}} \divi \sigmaP \rangle$.
            Then, by \Cref{lemma:projections-match}, we get that $\langle \denP{e'}[\dom(\varphi) \mapsto \TDeref{\overline{\varphi}}] \divi \varnothing \rangle \tstep^* \langle \TTrue \divi \sigmaP \rangle$.

            After executing the $\TLetN$ statements and the compute-side steps, we first get $\langle \fullComp{{e_{cnp}}} \divi \sigma \rangle \tstep^* \langle \TProveUsing{\alpha = \exists \overline{y_p},\overline{x_p}[\overline{y_s}, \overline{x_s}]}{\denP{e}}{\overline{v_p}, \TDeref{\overline{\iota_{xp}}}[\overline{v_s}, \TDeref{\overline{\iota_{xp}}}]} \divi \sigma''' \rangle$, where $\sigma'' \cup \{\overline{\iota_{xp}} \mapsto \rho(\overline{x_p}), \overline{\iota_{xs}} \mapsto \rho(\overline{x_s})\}$ and $\varphi = \{\overline{x_p} \mapsto \overline{\iota_{xp}}, \overline{x_s} \mapsto \overline{\iota_{xs}}\}$.
            Note that in these middle steps, the traces through this line of reasoning are $\traceq$ to $\witdenR{\strace}$ due to the invoked theorems and lemmas.
            Finally, the prove-side steps mean we can apply \ruleref{T-Prove} and finish with $\langle \TProofOfUsing* \divi \sigma''' \rangle$, where $\overline{v} = \overline{v_p}, \sigma(\overline{\iota_{xp}})$, which matches with the one from executing the surface semantics.
            Also note that at this step the stores match.
            To finish the proof, we observe the traces match: first $\topalloc(\overline{\varphi}, \bot)$, followed by a trace equivalent to ${\witdenR{\strace}}$, and finally $\topgen(\alpha, \overline{v}, \overline{u})$ where the public and private inputs $\overline{v}, \overline{u}$ are equal in the trace just like in the produced proof.
        \item[($\Leftarrow$)]
            For this direction, suppose that $\langle \fullComp{{e_{cnp}}} \divi \sigma \rangle \tstep^* \langle \TProofOfUsing* \divi \sigma''' \rangle$.
            The proof for this direction follows very similarly, and applying soundness instead of completeness.
            This can only happen (by determinism \Cref{thm:target-determinism} and \Cref{thm:surface-determinism}) by the following sequence of steps, where $e' = e[\overline{y_p} \mapsto \overline{v_p}, \overline{y_s} \mapsto \overline{v_s}]$ and $e'' = e'[\overline{x_p} \mapsto \overline{\iota_{xp}}, \overline{x_s} \mapsto \overline{\iota_{xs}}]$ and $\varphi = \{\overline{x_p} \mapsto \overline{\iota_{xp}}, \overline{x_s} \mapsto \overline{\iota_{xs}}\}$:
            \begin{align*}
                \langle \fullComp{{e_{cnp}}} \divi \sigma \rangle & \tstep^*_{\topalloc(\overline{\varphi}, \bot)} \langle \denC{e''} ; \TProveUsing{\alpha=\exists \overline{y_p}, \overline{x_p}[\overline{y_s}, \overline{x_s}]}{\denP{e}}{\overline{v_p}, \TDeref{\overline{\iota_{xp}}}[\overline{v_s}, \TDeref{\overline{\iota_{xs}}}]} \divi \sigma' \rangle \\
                & \tstep^*_{\trace} \langle \TProveUsing{\alpha=\exists \overline{y_p}, \overline{x_p}[\overline{y_s}, \overline{x_s}]}{\denP{e}}{\overline{v_p}, \TDeref{\overline{\iota_{xp}}}[\overline{v_s}, \TDeref{\overline{\iota_{xs}}}]} \divi \sigma''' \rangle \\
                & \tstep_{\topgen(\alpha, \overline{v}, \overline{u})} \langle \TProofOfUsing{\alpha}{\overline{v_p}, \sigma'''(\overline{\iota_{xs}})} \divi \sigma''' \rangle
            \end{align*}
            Like with the forward direction, the \ruleref{E-ComputeAndProveInit} rule simulates that first set of steps involving allocating uninitialized references.
            Then combining \Cref{lemma:projections-match} with \Cref{thm:witness-compute-bisimulation} gives for any $e_\Prove$ and $\sigma$ that $\pconf{\denL[\Compute]{e'}}{e_\Prove}{\sigma'}{\sigma}{\varnothing} \pstep^* \pconf{\denL[\Compute]{v}}{e_\Prove}{\sigma''}{\sigma}{\rho}$ where $\rho = \dom(\varphi) \mapsto \overline{\varphi}$.
            This $\sigma''$ looks a lot like $\sigma'''$ except the recently allocated locations for computed inputs are still $\bot$.
            Similarly, combining \Cref{lemma:projections-match} with \Cref{thm:witness-prove-bisimulation} gives for any $e_\Compute, \sigma, \rho$ that $\pconf{e_\Compute}{\denL[\Prove]{e'}}{\sigma}{\varnothing}{\rho} \pstep^* \pconf{e_\Compute}{\denL[\Prove]{\TTrue}}{\sigma}{\sigmaP}{\rho}$.
            Combining these two sequentially by setting the univerally quantified expressions and stores to agree gives $\pconf{\denL[\Compute]{e'}}{\denL[\Prove]{e'}}{\sigma'}{\varnothing}{\varnothing} \pstep^* \pconf{\denL[\Compute]{v}}{\denL[\Prove]{\TTrue}}{\sigma''}{\sigmaP}{\rho}$.
            Then by \Cref{thm:fullbisimsoundness} we get that $\cpconf{e}{\sigma'}{\varnothing}{\varnothing} \cpstep^* \cpconf{e''}{\sigma'''}{\sigmaP}{\rho}$ where $\denL[\Prove]{\TTrue} \lessthan e''$, meaning $e'' = \TTrue_\annell$ where $\Prove \subsetl \annell$.
            This then means the inner expression of the $\TCnpN$ term evaluates to $\TTrue$, which means it can take a final step via \ruleref{E-ComputeAndProveTrue} to reach a $\TProofOfUsing{\alpha}{\overline{v_p}, \sigma''(\overline{\iota_{xs}})}$ value.
            Note also the store changes from $\sigma''$ to $\sigma'''$ by setting the computed inputs to their computed value.
            Finally, like in the forward directions, the traces do match up: first $\topalloc(\overline{\varphi}, \bot)$, followed by a trace ${\witdenR{\strace}}$ equivalent to $\trace$, and finally $\topgen(\alpha, \overline{v}, \overline{u})$ where the public and private inputs $\overline{v}, \overline{u}$ are equal in the trace just like in the produced proof.
    \end{bidir}

    Inductive Case: suppose the expression $e$ has $n$ nested compute-and-prove blocks.
    The inner-most one, by definition, does not have any nested compute-and-prove blocks.
    Proving this inner-most compute-and-prove is adequate to its projection $\Lift{\Compute}{\fullComp{e_{in}}}$ follows very similar logic to proving the base case: follow the three sets of compute-and-prove semantic rules to evaluate the combined block and link to the full projection (through the parallel semantics) accordingly.
    One key difference is that $e_{cnp}$ is in \zkstrudul, so the simulation uses those combined rules (e.g. \ruleref{EC-ComputeAndProveInit}) as opposed to the ones in full \zkstrudul.
    Additionally, because $\denL{\cdot}$ returns a \corelang program, we also lift $\Lift{\Compute}{\cdot}$ to the compute side to turn it into a \zkstrudul expression, which by \Cref{lemma:lifting-bisim} transitively gives us the result we want.

    Then, by \Cref{lemma:subst-adequacy}, we know that substituting that nested compute-and-prove block with its lifted projection but otherwise running the expression down to a proof is adequate to just executing the inner combined block via its semantics in \zkstrudul.
    However, this resulting expression after substitution removes one layer of nesting, meaning we can apply the inductive hypothesis to get our desired result.
\end{proof}

\begin{theorem}[Adequacy with Traces]\label{thm:adequacy-with-traces}
    For any expression $e$ that may contain a combined block, $\langle e \divi \sigma \rangle \tstep_{\trace_1}^* \langle v \divi \sigma' \rangle$ if and only if~$\langle \fullComp{e} \divi \sigma \rangle \tstep_{\trace_2}^* \langle \fullComp{v} \divi \sigma' \rangle$ and $\trace_1 \traceq \trace_2$.
\end{theorem}
\begin{proof}
    By induction on $e$.
    For most terms, $\fullComp{e} = e$, and therefore are trivial.
    For $\LetN$ and $\IfN$, we apply the $\fullComp{\cdot}$ definition and apply the inductive hypothesis to the subterms to get the result.
    For method calls, \Cref{lemma:subst-adequacy} and \Cref{thm:adequacy-of-combined} allows us to replace all combined blocks with their compilation and still preserve adequacy.
    Finally, \Cref{thm:adequacy-of-combined} gives us the non-trivial adequacy result for combined blocks.
\end{proof}

\adequacy*
\begin{proof}
    Corollary of \Cref{thm:adequacy-with-traces} without reasoning about traces.
\end{proof}

\begin{theorem}[Traceful Adequacy of Whole Programs]\label{thm:whole-adequacy-with-traces}
    Suppose $\SFullProg$ is a whole program containing a list of class definitions and an expression $e$ to execute.
    Then if $e'$ is the expression to execute from $\fullComp{\SFullProg}$, the following is true: $\langle e \divi \sigma \rangle \tstep[\trace_1]^* \langle v \divi \sigma' \rangle$ if and only if~$\langle e' \divi \sigma \rangle \tstep[\trace_2]^* \langle v \divi \sigma' \rangle$ and $\trace_1 \traceq \trace_2$.
\end{theorem}
\begin{proof}
    Follows directly from \Cref{thm:adequacy-with-traces}, noting that method calls are themselves adequate and are the only expressions outside of $e$ that could contain compute-and-prove blocks.
\end{proof}

\begin{lemma}\label{lemma:compilation-over-context}
    For any source context $\SContextN$ in the subset \corelang of \zkstrudul, and for any partial program $\PartialP$ in \zkstrudul, $\fullComp{\PartialP} \link \SContextN$ is syntactically equivalent to $\fullComp{\PartialP \link \SContextN}$.
\end{lemma}
\begin{proof}
    By induction on the full program, noting that $\fullComp{\SContextN} = \SContextN$ since compute-and-prove blocks do not appear in \corelang.
\end{proof}

\rrhp*
\begin{proof}
    Since the context $\TContextN$ comes from \corelang, and \corelang is a subset of full \zkstrudul, we can simply define the backtranslation to $\SContextN$ as the identity function: $\SContextN = \TContextN$.

    The bidirectional comes from compilation being well-typed and preserves types of all target expressions and, for the sake of reasoning about behaviors, proves the linked whole programs can be typed checked--and so compiling a partial program is meaningful--or both behaviors are undefined.
    We define the whole program as $\SFullProg = \PartialP \link \SContextN$ and using \Cref{lemma:compilation-over-context}, $\fullComp{\SFullProg} = \fullComp{\PartialP \link \SContextN} = \fullComp{\PartialP} \link \TContextN$.
    With this, by the forward direction of \Cref{thm:whole-adequacy-with-traces}, we know $\SBehav{\SFullProg} \subseteq \TBehav{\fullComp{\SFullProg}}$, and from the backwards direction we know $\SBehav{\SFullProg} \supseteq \TBehav{\fullComp{\SFullProg}}$, meaning we can prove they are equal.
    Note that the two sets $\SBehav{\SFullProg}$ and $\TBehav{\fullComp{\SFullProg}}$ are considered equal if for every corresponding pair of traces are equivalent up to stuttering, i.e. according to $\traceq$.
\end{proof}

\end{document}